\documentclass[a4paper,nobind]{ociamthesis}

\usepackage{graphicx}

\usepackage{amsmath,amsthm,amssymb,amsfonts,soul}
\usepackage{braket}
\usepackage{footnote}
\usepackage{epigraph}

\correctionstrue

\usepackage{iftex}

\usepackage[utf8]{inputenc}
\usepackage{amsmath,amssymb,amsthm}
\usepackage{mathtools}
\usepackage{graphicx}
\usepackage{caption, subcaption}
\usepackage{braket}
\usepackage{textcomp}
\usepackage{verbatim}
\usepackage{adjustbox}
\usepackage{bbold}
\usepackage{chemarrow}

\usepackage{charter}
\usepackage{tikz}
\usepackage{graphicx}
\usepackage{amsmath}
\usepackage{amssymb}   

\usepackage{stmaryrd}    
\usepackage{float}

\usepackage{xspace} 
\usepackage{color} 
\usepackage{epsfig}

\newtheorem{theorem}{Theorem}[section]

\theoremstyle{definition}
\newtheorem*{theorem*}{Theorem}
\newtheorem{lemma}[theorem]{Lemma} 
\newtheorem{proposition}[theorem]{Proposition}

\newtheorem{defn}[theorem]{Definition}
 
\newtheorem{example}[theorem]{Example} 
\newtheorem{examples}[theorem]{Examples}
\newtheorem{example*}[theorem]{Example*}
\newtheorem{examples*}[theorem]{Examples*}
\newtheorem{remark}[theorem]{Remark}
\newtheorem{remark*}[theorem]{Remark*}

\newtheorem{convention}[theorem]{Convention}

\newtheorem{disclaimer}[theorem]{Disclaimer}

\newtheorem{thesis}[theorem]{Thesis}
\newtheorem{definition}[theorem]{Definition}

\usepackage{tikzfig}
\usepackage{keycommand} 

\usepackage[ruled,vlined]{algorithm2e}
\usepackage{endnotes}

\usepackage{tikz-cd}
\usepackage{tikz}
\usepackage{pgfplots}

\usepackage{bussproofs}

\usepackage{stackrel}

\ifpdf
\usepackage{underscore}         
\usepackage[T1]{fontenc}        
\else
\usepackage{breakurl}           
\fi
\usepackage{authblk}

\usepackage{url}

\usepackage[english]{babel}
\usetikzlibrary{babel}

\usepackage[backend=biber, sorting=none, style=numeric-comp,   maxcitenames=1, maxbibnames=99, giveninits=true, doi=false,
url=false, eprint=false, isbn=false]{biblatex}
\DeclareUnicodeCharacter{0301}{\'{e}}
\DeclareUnicodeCharacter{0302}{\'{e}}
\DeclareUnicodeCharacter{0300}{\'{e}}
\DeclareUnicodeCharacter{0304}{\'{e}}

\title{String Diagrams for Quantum Foundations,\\ Computing and Natural Language Processing}
\author{Muhammad Hamza Waseem}
\college{Magdalen College}

\degree{Doctor of Philosophy}
\degreedate{Trinity 2024}

\begin{document}

\setlength{\textbaselineskip}{22pt plus2pt}

\setlength{\frontmatterbaselineskip}{17pt plus1pt minus1pt}

\setlength{\baselineskip}{\textbaselineskip}

\setcounter{secnumdepth}{2}
\setcounter{tocdepth}{2}

\begin{romanpages}

\maketitle

\begin{acknowledgements}
\[
\includegraphics[width=7cm]{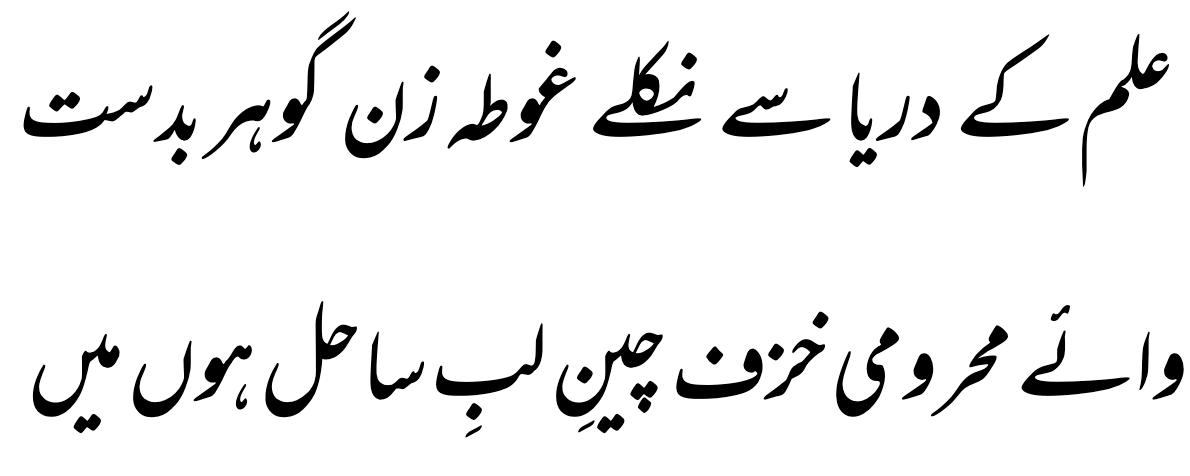}
\]
\centerline{\textit{``From the sea of knowledge, divers emerge with pearls in hand.}}
\centerline{\textit{Alas, O deprivation! I am but a mere collector of pebbles on the shore.''}}
 \\
 \\
Like Iqbal, I feel like an ordinary pebble-collector at the end of my DPhil. And yet, the chance even to touch the shore of the sea of knowledge has been an extraordinary privilege. No one writes a thesis alone. This is my attempt to acknowledge the debts I owe to the people and communities who have supported my journey.

I cannot imagine writing this acknowledgements section without mentioning the profound impact Sabieh Anwar had on my academic path. About a decade ago, I was a freshman who had reluctantly chosen electrical engineering, more out of pragmatism and societal pressure than genuine passion, which for me lay in physics. As I learned about Sabieh’s trajectory, I realised it was possible to pivot toward physics. In my second year of undergrad, I had the privilege of working with Sabieh on science outreach and popularisation. Later, I managed to undertake a physics thesis under his supervision, despite bureaucratic challenges at my home institute, since Sabieh was at a different university. That experience, I believe, was the defining factor that eventually led me to Oxford for a DPhil in Physics.

While looking for DPhil projects and research groups, I ran into some scepticism at Oxford, mainly because I did not hold a formal degree in physics. But Alexy Karenowska saw things differently. Having made the transition from engineering to physics herself, she offered me a place in her group. Looking back, I think I hit the jackpot. With Alexy, I found not just an excellent and supportive research supervisor, but also someone who genuinely appreciated my passion for science outreach and teaching. Throughout my DPhil, I collaborated with her on several related projects.

In Alexy's group, I was originally meant to study the quantum properties of electronic spin waves through experiments. But then the Covid-19 pandemic hit, and all the lab work became nearly impossible. To make things worse, I caught Covid myself and was ill for almost six months; this was well before vaccines were around. Travel restrictions meant I was stuck in Pakistan for over a year.

Still, I kept up with teaching and outreach remotely. Through it all, Alexy, my co-supervisor John Gregg, and my College Advisor Giles Barr were consistently encouraging. John and Giles, especially, made sure I always had teaching opportunities at Magdalen College. That eventually led to a Lectureship, something I am incredibly grateful for. I additionally owe a lot to my students at Oxford, who constantly challenged my understanding of physics and may have taught me more than I taught them.

When lockdowns eased, I realised my heart was not really in experimental physics anymore. I felt a strong pull toward theoretical work, especially in quantum foundations. That is when I reached out to Bob Coecke. He had just left academia, but generously agreed to supervise me at Quantinuum, a quantum computing company. I have learned so much from Bob not just in physics, maths, and philosophy, but also in how to build community and collaborate. He introduced me to the world of string diagrams and helped me see the world through a relational lens.

I am truly grateful to Alexy for backing my transition completely and continuing to believe in me throughout. This thesis is the result of work guided by both Alexy and Bob, and I really could not have asked for better supervisors.

I have likewise been deeply fortunate to have wonderful mentors beyond research. The late Vicky Neale gave me thoughtful and generous feedback on my teaching philosophy and practice. I benefited a great deal from co-teaching with the late Derek Goldrei, whose tutorials were masterclasses in mathematical clarity and wit. Sian Tedaldi always had my back on the outreach front, and Soufia Siddiqi was someone I could turn to for academic, semi-academic, and even non-academic advice whenever I needed it.

Over the course of my DPhil, I got to work on a project that aimed to make quantum physics more accessible. In the process, I gained insight into inclusive teaching from Selma Coecke, Caterina Puca, and Lia Yeh. I am thankful as well to the late Ian Shipsey for supporting the project within the Department of Physics and the MPLS Division.

To all my collaborators and co-authors: thank you for everything you have taught me. I have also had many inspiring conversations with colleagues and friends. I am probably forgetting some names, but I especially want to thank Johannes Fankhauser, Maria Violaris, Nicetu Tibau Vidal, Lia Yeh, Sadia Naeem, and Robin Lorenz. I am equally appreciative of my research group colleagues, Finlay Ryburn, Sally Lord, and Will Henderson. My work thrived in the intellectually rich and welcoming environment of the Quantinuum office in Oxford. Every visit left me with fresh ideas and motivation.

My sincere appreciation to Sean Tull, Harny Wang, Jonathon Liu, and Vincent Wang-Ma{\'s}cianica for reading parts of this thesis and offering helpful feedback. I am especially indebted to Anna Pearson for reading the whole thing and offering thoughtful, constructive critique throughout.

I was fortunate to be financially supported by the Rhodes Scholarship for the first three years of my DPhil. I am grateful to the Rhodes Trust and the selection committee for believing in me.

The Rhodes community gave me so much: space for reflection, learning, and friendship. Special thanks to the Rhodes House staff, especially Mary Eaton (Registrar) and Maureen Freed (Scholar Wellbeing Adviser). To my fellow Rhodes Scholars and friends---Khansa Maria, Arslan Chaudhry, Sana Naeem, Siobhan Tobin, Namrata Ramesh, Terence Tsui, Krishnendu Ray, and Suhas Mahesh---thank you for the deep, thoughtful, and often heartwarming conversations. You have all left a mark on me.

Magdalen College supported me in all kinds of ways: great opportunities for teaching and outreach, amazing library access, and the ever-welcoming Dining Hall and Old Kitchen Bar. I feel honoured to have been part of this college, which also helped finance the latter years of my DPhil.

On a personal note, I have been blessed with the steady support of friends throughout this journey.

Mohsin Javed has been a constant through all my highs and lows in Oxford. His advice and insight, both intellectual and emotional, have helped me in countless ways. I am also thankful to Sarah Salim, Amal Javed, and Zaha Javed for their enduring friendship and love.

Abdul Afzal encouraged me to dig into philosophy and the foundations of quantum physics. Much of my philosophical education started with his book recommendations. Our wide-ranging discussions on consciousness, politics, love, and friendship were always a source of joy and insight.

My friendship with Ayesha Ali deeply enriched my perspective on life and relationships, and truly transformed me as a person. I cannot thank her enough for seeing light and goodness in me, especially at times when I struggled to see them myself.

Hassaan Saleem has also been a great friend: always encouraging, occasionally irritating, and often hilarious. I am glad we have stayed connected despite being on opposite sides of the Atlantic.

Mahdi Murtaza made the latter half of my time at Magdalen especially fun. Thanks for dragging me to formal dinners and Magdalen Iftars, and for showing up with home-cooked food.

I am truly thankful to Duaa Jamshaid for her kindness, her encouragement, our many conversations on physics and society, and for always cheering me on. I am beyond grateful for your support and friendship.

Over these years, I was fortunate to spend time with friends in Pakistan---Khadija Maryam, Shahryar Khan, Faizan-e-Ilahi, Tooba Malik, Shumail Hassan, Waqar Ali, and Mohsina Asif---and those in Oxford: Rishika Sahgal, Ninad Rajgopal, Bhavna Verma, Vikaran Khanna, and Shai de Vries. Thank you for looking after me and being such a meaningful part of this journey.

And of course, my friends from the Khwarizmi Science Society, Salman Mahmood Qazi, Nimra Khurram, and Charisma Wafee, thank you for keeping me in the loop on all things science, science communication, and occasionally science miscommunication during my visits back to Pakistan.

But perhaps the greatest privilege of all has been having a family that is unconditionally behind me in my quest for knowledge. Without the prayers and sacrifices of my mother, Saima, and my father, Waseem, and their constant encouragement, I would not have achieved anything worthwhile. My brothers, Hammad and Haider, have celebrated my journey every step of the way. My Dadi and Phupho have always been a source of strength and love. And I am sure my late Dada would have been proud to see this chapter of my life.

The city of Oxford has been a wellspring of dreams, ideas, and inspiration. These years have been the happiest of my life so far, and I cannot leave without leaving a piece of my heart behind. Oxford has truly become my second home. In the words of Ghalib, I say to this city:

\[
\includegraphics[width=7cm]{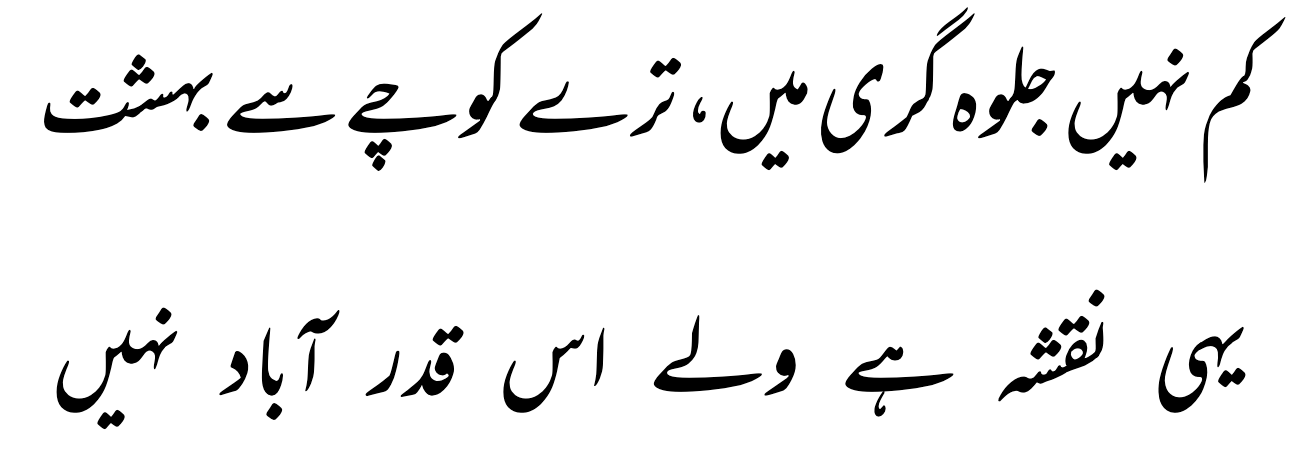}
\]
\centerline{\textit{``Not less in splendour is Paradise than thy street.}}
\centerline{\textit{This very design exists, indeed, yet not so flourishing.''}}
\end{acknowledgements}

\begin{abstract}
	
Applied category theory provides powerful mathematical tools for modelling processes and their composition. Symmetric monoidal categories, which involve series and parallel composition, are particularly well-suited for describing the composition of processes in space and time. Also called process theories, they admit string diagrams, which constitute a visually intuitive, mathematically rigorous, expressive and flexible syntax that is applicable to wide-ranging scientific domains. 

In this thesis, we employ string diagrams to investigate a selection of topics in the areas of quantum foundations, computing, and natural language processing. We report three main contributions:

\begin{itemize}

    \item We formalise constructor theory as a process theory. In the context of quantum physics, we also demonstrate the conflict between constructor-theoretic principles of locality and composition. Moreover, we argue that if the principle of locality is rejected, categorical quantum mechanics (CQM) can be conceived as a constructor theory of quantum physics.

    \item We develop a formalism for wave-based logic circuits with phase encoding. We motivate the formalism using the example of spin-wave circuits, and then demonstrate its utility in design, analysis and optimisation of Boolean logic circuits. 

    \item We investigate the elimination of inter-language grammatical bureaucracy in the distributional compositional circuits (DisCoCirc) framework. In particular, we develop a hybrid grammar for a restricted fragment of the Urdu language, and show that Urdu text endowed with this hybrid grammar maps surjectively to DisCoCirc text circuits. Furthermore, we show that for the same language fragment, Urdu and English text circuits become the same up to gate-level translation. 
\end{itemize}

The aforementioned work supports the view that a process-relational outlook in science is well-supported by applied category-theoretic tools, particularly string diagrams. 
\end{abstract}

\flushbottom

\tableofcontents

\end{romanpages}

\flushbottom

\newcommand{\textdef}[1]{\textbf{#1}}
\newcommand{\obj}[1]{\ensuremath{\text{obj}\!\left(#1\right)}}
\newcommand{\states}[2]{\ensuremath{\text{states}_{#1}\!\left(#2\right)}}
\newcommand{\id}[1]{\ensuremath{\text{id}_{#1}}}
\newcommand{\substr}[1]{\ensuremath{\texttt{#1}}}
\newcommand{\smc}[1]{\ensuremath{\textbf{#1}}}
\newcommand{\Rel}{\smc{Rel}}
\newcommand{\EndoRel}{\smc{EndoRel}}
\newcommand{\tmapsto}[1]{\ensuremath{\stackrel{\task{#1}}{\mapsto}}}
\newcommand{\suchthat}[2]{\ensuremath{\left\{#1\;\middle|\;#2\right\}}}
\newcommand{\task}[1]{\ensuremath{\mathfrak{#1}}}
\newcommand{\precond}[2]{\ensuremath{{#1}|_{#2}}}
\newcommand{\postcond}[2]{\ensuremath{{#1}|^{#2}}}
\newcommand{\ppcond}[3]{\ensuremath{{#1}|_{#2}^{#3}}}
\newcommand{\indtask}[1]{\ensuremath{\left\lfloor #1 \right\rfloor}}
\newcommand{\possibletasks}[1]{\ensuremath{#1^\checkmark}}
\newcommand{\quotask}[3]{#1|_{#2}^{#3}}

\newcommand\nuclear{\omding\char195\xspace}
\newcommand\goodskull{\omding\char194\xspace}
\newcommand\frowny{\raisebox{-0.5mm}{\scalebox{1.4}{\omding\char197}}\xspace}
\newcommand\smily{\raisebox{-0.5mm}{\scalebox{1.4}{\omding\char196}}\xspace}

\newcommand*{\threesim}{%
  \mathrel{\vcenter{\offinterlineskip
  \hbox{$\sim$}\vskip-.35ex\hbox{$\sim$}\vskip-.35ex\hbox{$\sim$}}}}

\newcommand{\C}{\ensuremath{\mathbb{C}}}

\newcommand{\smallsum}{\textstyle{\sum\limits_i}\xspace}
\newcommand{\smallsumi}[1]{\textstyle{\sum\limits_{#1}}\xspace}

\newcommand{\emptydiag}{\,\tikz{\node[style=empty diagram] (x) {};}\,}
\newcommand{\scalar}[1]{\,\tikz{\node[style=scalar] (x) {$#1$};}\,}
\newcommand{\dscalar}[1]{\,\tikz{\node[style=scalar, doubled] (x) {$#1$};}\,}
\newkeycommand{\boxmap}[style=small box][1]{\,\tikz{\node[style=\commandkey{style}] (x) {$#1$};\draw(0,-1.25)--(x)--(0,1.25);}\,}
\newkeycommand{\dboxmap}[style=dbox][1]{\,\tikz{\node[style=\commandkey{style}] (x) {$#1$};\draw[boldedge] (0,-1.25)--(x)--(0,1.25);}\,}
\newcommand{\dhadamard}{\,\tikz{\node[style=dhadamard] (x) {$H$};\draw[boldedge] (0,-1.25)--(x)--(0,1.25);}\,}
\newcommand{\hadamard}{\,\tikz{\node[style=hadamard] (x) {$H$};\draw[] (0,-0.75)--(x)--(0,0.75);}\,}

\newcommand{\tboxmap}[3]{\,\begin{tikzpicture}
  \begin{pgfonlayer}{nodelayer}
    \node [style=small box] (0) at (0, 0) {$#1$};
    \node [style=none] (1) at (0, -0.5) {};
    \node [style=none] (2) at (0, -1.25) {};
    \node [style=none] (3) at (0, 1.25) {};
    \node [style=none] (4) at (0, 0.5) {};
    \node [style=right label] (5) at (0.25, -1) {$#2$};
    \node [style=right label] (6) at (0.25, 1) {$#3$};
  \end{pgfonlayer}
  \begin{pgfonlayer}{edgelayer}
    \draw (2.center) to (1.center);
    \draw (4.center) to (3.center);
  \end{pgfonlayer}
\end{tikzpicture}}
\newcommand{\boxstate}[1]{\,\tikz{\node[style=small box] (x) {$#1$};\draw(x)--(0,1.25);}\,}
\newcommand{\boxeffect}[1]{\,\tikz{\node[style=small box] (x) {$#1$};\draw(0,-1.25)--(x);}\,}
\newcommand{\boxmapTWOtoONE}[1]{\,\tikz{\node[style=small box] (x) {$#1$};\draw(-0.8,-1)--(x)--(0,1);\draw (0.8,-1)--(x);}\,}
\newcommand{\boxmapONEtoTWO}[1]{\,\tikz{\node[style=small box] (x) {$#1$};\draw(0,-1)--(x)--(-1,1);\draw (x)--(1,1);}\,}

\newcommand{\doubleop}{\ensuremath{\textsf{double}}\xspace}

\newcommand{\kscalar}[1]{\,\tikz{\node[style=kscalar] (x) {$#1$};}\,}
\newcommand{\kscalarconj}[1]{\,\tikz{\node[style=kscalarconj] (x) {$#1$};}\,}

\newcommand{\dmap}[1]{\,\tikz{\node[style=dmap] (x) {$#1$};\draw[boldedge] (0,-1.3)--(x)--(0,1.3);}\,}
\newcommand{\dmapdag}[1]{\,\tikz{\node[style=dmapdag] (x) {$#1$};\draw[boldedge] (0,-1.3)--(x)--(0,1.3);}\,}
\newcommand{\dmaptrans}[1]{\,\tikz{\node[style=dmaptrans] (x) {$#1$};\draw[boldedge] (0,-1.3)--(x)--(0,1.3);}\,}
\newcommand{\dmapconj}[1]{\,\tikz{\node[style=dmapconj] (x) {$#1$};\draw[boldedge] (0,-1.3)--(x)--(0,1.3);}\,}

\newcommand{\map}[1]{\,\tikz{\node[style=map] (x) {$#1$};\draw(0,-1.3)--(x)--(0,1.3);}\,}

\newcommand{\maptypes}[3]{\begin{tikzpicture}
  \begin{pgfonlayer}{nodelayer}
    \node [style=map] (0) at (0, 0) {$#1$};
    \node [style=none] (1) at (0, -1.5) {};
    \node [style=none] (2) at (0, 1.5) {};
    \node [style=right label] (3) at (0.25, -1.25) {$#2$};
    \node [style=right label] (4) at (0.25, 1) {$#3$};
  \end{pgfonlayer}
  \begin{pgfonlayer}{edgelayer}
    \draw (1.center) to (0);
    \draw (0) to (2.center);
  \end{pgfonlayer}
\end{tikzpicture}}

\newcommand{\mapdag}[1]{\,\tikz{\node[style=mapdag] (x) {$#1$};\draw(0,-1.3)--(x)--(0,1.3);}\,}

\newcommand{\maptrans}[1]{\,\tikz{\node[style=maptrans] (x) {$#1$};\draw(0,-1.3)--(x)--(0,1.3);}\,}

\newcommand{\mapconj}[1]{\,\tikz{\node[style=mapconj] (x) {$#1$};\draw(0,-1.3)--(x)--(0,1.3);}\,}

\newcommand{\mapONEtoTWO}[1]{\,\begin{tikzpicture}
  \begin{pgfonlayer}{nodelayer}
    \node [style=map] (0) at (0, 0) {$#1$};
    \node [style=none] (1) at (0, -0.5) {};
    \node [style=none] (2) at (-0.75, 0.5) {};
    \node [style=none] (3) at (-0.75, 1.25) {};
    \node [style=none] (4) at (0, -1.25) {};
    \node [style=none] (5) at (0.75, 0.5) {};
    \node [style=none] (6) at (0.75, 1.25) {};
  \end{pgfonlayer}
  \begin{pgfonlayer}{edgelayer}
    \draw (4.center) to (1.center);
    \draw (2.center) to (3.center);
    \draw (5.center) to (6.center);
  \end{pgfonlayer}
\end{tikzpicture}\,}

\newcommand{\classicalstate}[1]{%
\,\begin{tikzpicture}
  \begin{pgfonlayer}{nodelayer}
    \node [style=dkpoint] (0) at (0, -0.75) {$#1$};
    \node [style=white dot] (1) at (0, 0.25) {};
    \node [style=none] (2) at (0, 1) {};
  \end{pgfonlayer}
  \begin{pgfonlayer}{edgelayer}
    \draw [style=boldedge] (0) to (1);
    \draw (1) to (2.center);
  \end{pgfonlayer}
\end{tikzpicture}\,}

\newkeycommand{\wirelabel}[style=white label]{\,\tikz{\node[style=\commandkey{style}]{};}\,\xspace}

\newkeycommand{\pointmap}[style=point][1]{\,\tikz{\node[style=\commandkey{style}] (x) at (0,-0.3) {$#1$};\node [style=none] (3) at (0, 0.9) {};\node [style=none] (3) at (0, -0.9) {};\draw (x)--(0,0.7);}\,}

\newkeycommand{\point}[style=point,estyle=][1]{\,\tikz{\node[style=\commandkey{style}] (x) at (0,-0.05) {$#1$};\draw [\commandkey{estyle}] (x)--(0,1.05);}\,}

\newcommand{\pointdag}[1]{\,\tikz{\node[style=copoint] (x) at (0,0.05) {$#1$};\draw (x)--(0,-1.05);}\,}

\newcommand{\graypointmap}[1]{\pointmap[style=gray point]{#1}}

\newkeycommand{\copointmap}[style=copoint][1]{\,\tikz{\node[style=\commandkey{style}] (x) at (0,0.3) {$#1$};\draw (x)--(0,-0.7);}\,}

\newcommand{\graycopointmap}[1]{\copointmap[style=gray copoint]{#1}}

\newcommand{\dpointmap}[1]{\,\tikz{\node[style=point,doubled] (x) at (0,-0.3) {$#1$};\draw[boldedge] (x)--(0,0.7);}\,}

\newcommand{\dcopointmap}[1]{\,\tikz{\node[style=copoint, doubled] (x) at (0,0.3) {$#1$};\draw[boldedge] (x)--(0,-0.7);}\,}

\newcommand{\graydcopointmap}[1]{\,\tikz{\node[style=gray copoint, doubled] (x) at (0,0.3) {$#1$};\draw[boldedge] (x)--(0,-0.7);}\,}

\newkeycommand{\dpoint}[style=point][1]{\,\tikz{\node[style=\commandkey{style},doubled] (x) at (0,-0.3) {$#1$};\draw[boldedge] (x)--(0,0.7);}\,}
\newcommand{\dcopoint}[1]{\,\tikz{\node[style=copoint, doubled] (x) at (0,0.3) {$#1$};\draw[boldedge] (x)--(0,-0.7);}\,}

\newkeycommand{\kpoint}[style=kpoint,estyle=][1]{\,\tikz{\node[style=\commandkey{style}] (x) at (0,-0.05) {$#1$};\draw [\commandkey{estyle}] (x)--(0,1.05);}\,}

\newkeycommand{\kpointgray}[style=gray kpoint,estyle=][1]{\,\tikz{\node[style=\commandkey{style}] (x) at (0,-0.05) {$#1$};\draw [\commandkey{estyle}] (x)--(0,1.05);}\,}

\newcommand{\kpointgrayconj}[1]{\,\tikz{\node[style=gray kpointconj] (x) at (0,-0.05) {$#1$};\draw (x)--(0,1.05);}\,}

\newcommand{\kpointgrayadj}[1]{\,\tikz{\node[style=gray kpointadj] (x) at (0,0.05) {$#1$};\draw (x)--(0,-1.05);}\,}

\newcommand{\kpointconj}[1]{\,\tikz{\node[style=kpoint conjugate] (x) at (0,-0.05) {$#1$};\draw (x)--(0,1.05);}\,}

\newcommand{\kpointdag}[1]{\,\tikz{\node[style=kpoint adjoint] (x) at (0,0.05) {$#1$};\draw (x)--(0,-1.05);}\,}
\newcommand{\kpointadj}[1]{\kpointdag{#1}}
\newcommand{\kpointtrans}[1]{\,\tikz{\node[style=kpoint transpose] (x) at (0,0.05) {$#1$};\draw (x)--(0,-1.05);}\,}

\newkeycommand{\typedkpoint}[style=kpoint,edgestyle=][2]{%
\,\begin{tikzpicture}
  \begin{pgfonlayer}{nodelayer}
    \node [style=none] (0) at (0, 1) {};
    \node [style=\commandkey{style}] (1) at (0, -0.75) {$#2$};
    \node [style=right label] (2) at (0.25, 0.5) {$#1$};
  \end{pgfonlayer}
  \begin{pgfonlayer}{edgelayer}
    \draw [\commandkey{edgestyle}] (1) to (0.center);
  \end{pgfonlayer}
\end{tikzpicture}\,}

\newkeycommand{\typedkpointdag}[style=kpoint adjoint,edgestyle=][2]{%
\,\begin{tikzpicture}
  \begin{pgfonlayer}{nodelayer}
    \node [style=none] (0) at (0, -1) {};
    \node [style=\commandkey{style}] (1) at (0, 0.75) {$#2$};
    \node [style=right label] (2) at (0.25, -0.5) {$#1$};
  \end{pgfonlayer}
  \begin{pgfonlayer}{edgelayer}
    \draw [\commandkey{edgestyle}] (0.center) to (1);
  \end{pgfonlayer}
\end{tikzpicture}\,}

\newcommand{\typedpoint}[2]{\,\begin{tikzpicture}
  \begin{pgfonlayer}{nodelayer}
    \node [style=none] (0) at (0, 1) {};
    \node [style=point] (1) at (0, -0.5) {$#2$};
    \node [style=right label] (2) at (0.25, 0.5) {$#1$};
  \end{pgfonlayer}
  \begin{pgfonlayer}{edgelayer}
    \draw (1) to (0.center);
  \end{pgfonlayer}
\end{tikzpicture}\,}

\newkeycommand{\bistate}[style=kpoint][1]{\,\begin{tikzpicture}
  \begin{pgfonlayer}{nodelayer}
    \node [style=\commandkey{style}, minimum width=1 cm, inner sep=2pt] (0) at (0, -0.25) {$#1$};
    \node [style=none] (1) at (-0.75, 0) {};
    \node [style=none] (2) at (0.75, 0) {};
    \node [style=none] (3) at (-0.75, 0.88) {};
    \node [style=none] (4) at (0.75, 0.88) {};
  \end{pgfonlayer}
  \begin{pgfonlayer}{edgelayer}
    \draw (1.center) to (3.center);
    \draw (2.center) to (4.center);
  \end{pgfonlayer}
\end{tikzpicture}\,}

\newkeycommand{\bistatebig}[style=kpoint][1]{\,\begin{tikzpicture}
  \begin{pgfonlayer}{nodelayer}
    \node [style=\commandkey{style}, minimum width=, minimum width=1.26 cm, inner sep=2pt] (0) at (0, -0.25) {$#1$};
    \node [style=none] (1) at (-1, 0) {};
    \node [style=none] (2) at (1, 0) {};
    \node [style=none] (3) at (-1, 0.88) {};
    \node [style=none] (4) at (1, 0.88) {};
  \end{pgfonlayer}
  \begin{pgfonlayer}{edgelayer}
    \draw (1.center) to (3.center);
    \draw (2.center) to (4.center);
  \end{pgfonlayer}
\end{tikzpicture}\,}

\newkeycommand{\dbistate}[style=dkpoint][1]{\,\begin{tikzpicture}
  \begin{pgfonlayer}{nodelayer}
    \node [style=\commandkey{style}, minimum width=1 cm, inner sep=2pt] (0) at (0, -0.25) {$#1$};
    \node [style=none] (1) at (-0.75, 0) {};
    \node [style=none] (2) at (0.75, 0) {};
    \node [style=none] (3) at (-0.75, 1.0) {};
    \node [style=none] (4) at (0.75, 1.0) {};
  \end{pgfonlayer}
  \begin{pgfonlayer}{edgelayer}
    \draw [boldedge] (1.center) to (3.center);
    \draw [boldedge] (2.center) to (4.center);
  \end{pgfonlayer}
\end{tikzpicture}\,}

\newcommand{\bistateadj}[1]{\,\begin{tikzpicture}[yshift=-3mm]
  \begin{pgfonlayer}{nodelayer}
    \node [style=kpointadj, minimum width=1 cm, inner sep=2pt] (0) at (0, 0.5) {$#1$};
    \node [style=none] (1) at (-0.75, 0.25) {};
    \node [style=none] (2) at (0.75, 0.25) {};
    \node [style=none] (3) at (-0.75, -0.63) {};
    \node [style=none] (4) at (0.75, -0.63) {};
  \end{pgfonlayer}
  \begin{pgfonlayer}{edgelayer}
    \draw (1.center) to (3.center);
    \draw (2.center) to (4.center);
  \end{pgfonlayer}
\end{tikzpicture}\,}

\newcommand{\dbistateadj}[1]{\,\begin{tikzpicture}[yshift=-3mm]
  \begin{pgfonlayer}{nodelayer}
    \node [style=kpointadj, doubled, minimum width=1 cm, inner sep=2pt] (0) at (0, 0.5) {$#1$};
    \node [style=none] (1) at (-0.75, 0.25) {};
    \node [style=none] (2) at (0.75, 0.25) {};
    \node [style=none] (3) at (-0.75, -0.5) {};
    \node [style=none] (4) at (0.75, -0.5) {};
  \end{pgfonlayer}
  \begin{pgfonlayer}{edgelayer}
    \draw [boldedge] (1.center) to (3.center);
    \draw [boldedge] (2.center) to (4.center);
  \end{pgfonlayer}
\end{tikzpicture}\,}

\newkeycommand{\bistatebraket}[style1=kpoint,style2=kpointadj][2]{\,\begin{tikzpicture}
  \begin{pgfonlayer}{nodelayer}
    \node [style=\commandkey{style1}, minimum width=1 cm, inner sep=2pt] (0) at (0, -0.75) {$#1$};
    \node [style=none] (1) at (-0.75, -0.5) {};
    \node [style=none] (2) at (0.75, -0.5) {};
    \node [style=none] (3) at (-0.75, 0.5) {};
    \node [style=none] (4) at (0.75, 0.5) {};
    \node [style=\commandkey{style2}, minimum width=1 cm, inner sep=2pt] (5) at (0, 0.75) {$#2$};
  \end{pgfonlayer}
  \begin{pgfonlayer}{edgelayer}
    \draw (1.center) to (3.center);
    \draw (2.center) to (4.center);
  \end{pgfonlayer}
\end{tikzpicture}\,}

\newcommand{\binop}[2]{\,\begin{tikzpicture}
  \begin{pgfonlayer}{nodelayer}
    \node [style=none] (0) at (0.75, -1.25) {};
    \node [style=right label, xshift=1 mm] (1) at (0.75, -1) {$#1$};
    \node [style=none] (2) at (0, 1.25) {};
    \node [style=white dot] (3) at (0, 0) {$#2$};
    \node [style=none] (4) at (-0.75, -1.25) {};
    \node [style=right label] (5) at (-0.5, -1) {$#1$};
    \node [style=right label] (6) at (0.25, 1) {$#1$};
  \end{pgfonlayer}
  \begin{pgfonlayer}{edgelayer}
    \draw [bend right=15, looseness=1.00] (0.center) to (3);
    \draw [bend left=15, looseness=1.00] (4.center) to (3);
    \draw (3) to (2.center);
  \end{pgfonlayer}
\end{tikzpicture}\,}

\newkeycommand{\weight}[style=dkpoint][1]{\,\begin{tikzpicture}
  \begin{pgfonlayer}{nodelayer}
    \node [style=\commandkey{style}] (0) at (0, -0.5) {$#1$};
    \node [style=upground] (1) at (0, 0.75) {};
  \end{pgfonlayer}
  \begin{pgfonlayer}{edgelayer}
    \draw [style=boldedge] (0) to (1);
  \end{pgfonlayer}
\end{tikzpicture}\,}

\newcommand{\weightmap}[1]{\,\begin{tikzpicture}
  \begin{pgfonlayer}{nodelayer}
    \node [style=none] (0) at (0, -1.25) {};
    \node [style=dmap] (1) at (0, 0) {$#1$};
    \node [style=upground] (2) at (0, 1.5) {};
  \end{pgfonlayer}
  \begin{pgfonlayer}{edgelayer}
    \draw [style=boldedge] (0.center) to (1);
    \draw [style=boldedge] (1) to (2);
  \end{pgfonlayer}
\end{tikzpicture}\,}

\newcommand{\weightdag}[1]{\,\begin{tikzpicture}
  \begin{pgfonlayer}{nodelayer}
    \node [style=downground] (0) at (0, -0.75) {};
    \node [style=dkpointdag] (1) at (0, 0.75) {$#1$};
  \end{pgfonlayer}
  \begin{pgfonlayer}{edgelayer}
    \draw [style=boldedge] (0) to (1);
  \end{pgfonlayer}
\end{tikzpicture}\,}

\newcommand{\mapbent}[1]{\,\begin{tikzpicture}
	\begin{pgfonlayer}{nodelayer}
		\node [style=map] (0) at (1, 0) {$#1$};
		\node [style=none] (1) at (1, 1.25) {};
		\node [style=none] (2) at (-0.5, -0.5) {};
		\node [style=none] (3) at (-0.5, 1.25) {};
		\node [style=none] (4) at (1, -0.5) {};
	\end{pgfonlayer}
	\begin{pgfonlayer}{edgelayer}
		\draw (0) to (1.center);
		\draw [in=-90, out=-90, looseness=1.75] (4.center) to (2.center);
		\draw (2.center) to (3.center);
		\draw (4.center) to (0);
	\end{pgfonlayer}
\end{tikzpicture}\,}

\newcommand{\qproc}[1]{\left(\vphantom{\widehat X}#1\right)}

\newcommand{\dmapdiscard}[1]{\,\begin{tikzpicture}
  \begin{pgfonlayer}{nodelayer}
    \node [style=dmap] (0) at (0, 0) {$#1$};
    \node [style=none] (1) at (0, -1.25) {};
    \node [style=upground] (2) at (0, 1.50) {};
  \end{pgfonlayer}
  \begin{pgfonlayer}{edgelayer}
    \draw [style=boldedge] (0) to (1.center);
    \draw [style=boldedge] (0) to (2);
  \end{pgfonlayer}
\end{tikzpicture}\,}

\newkeycommand{\dkpoint}[style=dkpoint][1]{\,\tikz{\node[style=\commandkey{style}] (x) at (0,-0.1) {$#1$};\draw[boldedge] (x)--(0,1);}\,}
\newkeycommand{\dkpointadj}[style=dkpointadj][1]{\,\tikz{\node[style=\commandkey{style}] (x) at (0,0.1) {$#1$};\draw[boldedge] (x)--(0,-1);}\,}
\newkeycommand{\dkpointtrans}[style=dkpointtrans][1]{\,\tikz{\node[style=\commandkey{style}] (x) at (0,0.1) {$#1$};\draw[boldedge] (x)--(0,-1);}\,}
\newkeycommand{\dkpointconj}[style=dkpointconj][1]{\,\tikz{\node[style=\commandkey{style}] (x) at (0,-0.1) {$#1$};\draw[boldedge] (x)--(0,1);}\,}

\newkeycommand{\blackkpoint}[style=black kpoint][1]{\,\tikz{\node[style=\commandkey{style}] (x) at (0,-0.1) {$#1$};\draw (x)--(0,1);}\,}
\newkeycommand{\blackkpointadj}[style=black kpointadj][1]{\,\tikz{\node[style=\commandkey{style}] (x) at (0,0.1) {$#1$};\draw (x)--(0,-1);}\,}
\newkeycommand{\blackdkpoint}[style=black dkpoint][1]{\,\tikz{\node[style=\commandkey{style}] (x) at (0,-0.1) {$#1$};\draw[boldedge] (x)--(0,1);}\,}
\newkeycommand{\blackdkpointadj}[style=black dkpointadj][1]{\,\tikz{\node[style=\commandkey{style}] (x) at (0,0.1) {$#1$};\draw[boldedge] (x)--(0,-1);}\,}

\newcommand{\trace}{\,\begin{tikzpicture}
  \begin{pgfonlayer}{nodelayer}
    \node [style=none] (0) at (0, -0.60) {};
    \node [style=upground] (1) at (0, 0.40) {};
  \end{pgfonlayer}
  \begin{pgfonlayer}{edgelayer}
    \draw [style=none] (0.center) to (1);
  \end{pgfonlayer}
\end{tikzpicture}\,}

\newcommand{\spider}[1]{\,\begin{tikzpicture}
  \begin{pgfonlayer}{nodelayer}
    \node [style=#1] (0) at (0, 0) {};
    \node [style=none] (1) at (1.5, 1.25) {};
    \node [style=none] (2) at (-1, 1.25) {};
    \node [style=none] (3) at (1.25, -1.25) {};
    \node [style=none] (4) at (-0.75, -1.25) {};
    \node [style=none] (5) at (0.25, 1) {$\cdot\cdot\cdot\cdot\cdot$};
    \node [style=none] (6) at (0.25, -1) {$\cdot\cdot\cdot$};
    \node [style=none] (7) at (-1.5, 1.25) {};
    \node [style=none] (8) at (-1.25, -1.25) {};
  \end{pgfonlayer}
  \begin{pgfonlayer}{edgelayer}
    \draw [style=swap, in=135, out=-90, looseness=0.75] (2.center) to (0);
    \draw [style=swap, in=-90, out=45, looseness=0.75] (0) to (1.center);
    \draw [style=swap, in=90, out=-45, looseness=0.75] (0) to (3.center);
    \draw [style=swap, in=90, out=-135, looseness=0.75] (0) to (4.center);
    \draw [style=swap, in=-153, out=90, looseness=0.50] (8.center) to (0);
    \draw [style=swap, in=149, out=-90, looseness=0.50] (7.center) to (0);
  \end{pgfonlayer}
\end{tikzpicture}\,}

\newcommand\discard\trace


\newcommand{\maxmix}{\,\begin{tikzpicture}
  \begin{pgfonlayer}{nodelayer}
    \node [style=none] (0) at (-0.25, 0.60) {};
    \node [style=downground] (1) at (-0.25, -0.40) {};
  \end{pgfonlayer}
  \begin{pgfonlayer}{edgelayer}
    \draw [style=none] (0.center) to (1);
  \end{pgfonlayer}
\end{tikzpicture}\,}

\newcommand{\namedeq}[1]{\,\begin{tikzpicture}
  \begin{pgfonlayer}{nodelayer}
    \node [style=none] (0) at (0, 0.75) {\footnotesize $#1$};
    \node [style=none] (1) at (0, 0) {$=$};
  \end{pgfonlayer}
\end{tikzpicture}\,}
\newcommand{\namedmapsto}[1]{\,\begin{tikzpicture}
  \begin{pgfonlayer}{nodelayer}
    \node [style=none] (0) at (0, 0.75) {\footnotesize $#1$};
    \node [style=none] (1) at (0, 0) {$\mapsto$};
  \end{pgfonlayer}
\end{tikzpicture}\,}

\newcommand{\scalareq}{\ensuremath{\approx}\xspace}
\newcommand{\scalareqalt}{\ensuremath{\threesim}\xspace}
%


\newkeycommand{\pointketbra}[style1=copoint,style2=point,estyle=][2]{\,%
\begin{tikzpicture}
  \begin{pgfonlayer}{nodelayer}
    \node [style=none] (0) at (0, -1.75) {};
    \node [style=\commandkey{style1}] (1) at (0, -0.75) {$#1$};
    \node [style=\commandkey{style2}] (2) at (0, 0.75) {$#2$};
    \node [style=none] (3) at (0, 1.75) {};
  \end{pgfonlayer}
  \begin{pgfonlayer}{edgelayer}
    \draw [\commandkey{estyle}] (0.center) to (1);
    \draw [\commandkey{estyle}] (2) to (3.center);
  \end{pgfonlayer}
\end{tikzpicture}\,}

\newkeycommand{\twocopointketbra}[style1=copoint,style2=copoint,style3=point][3]{\,%
\begin{tikzpicture}
  \begin{pgfonlayer}{nodelayer}
    \node [style=none] (0) at (-0.75, -1.75) {};
    \node [style=none] (1) at (0.75, -1.75) {};
    \node [style=\commandkey{style1}] (2) at (-0.75, -0.75) {$#1$};
    \node [style=\commandkey{style2}] (3) at (0.75, -0.75) {$#2$};
    \node [style=\commandkey{style3}] (4) at (0, 0.75) {$#3$};
    \node [style=none] (5) at (0, 1.75) {};
  \end{pgfonlayer}
  \begin{pgfonlayer}{edgelayer}
    \draw (0.center) to (2);
    \draw (1.center) to (3);
    \draw (4) to (5.center);
  \end{pgfonlayer}
\end{tikzpicture}\,}

\newcommand{\graytwocopointketbra}{\twocopointketbra[style1=gray copoint,style2=gray copoint,style3=gray point]}

\newkeycommand{\twopointketbra}[style1=copoint,style2=point,style3=point][3]{\,%
\begin{tikzpicture}
  \begin{pgfonlayer}{nodelayer}
    \node [style=none] (0) at (0, -1.75) {};
    \node [style=none] (1) at (-0.75, 1.75) {};
    \node [style=\commandkey{style1}] (2) at (0, -0.75) {$#1$};
    \node [style=\commandkey{style2}] (3) at (-0.75, 0.75) {$#2$};
    \node [style=\commandkey{style3}] (4) at (0.75, 0.75) {$#3$};
    \node [style=none] (5) at (0.75, 1.75) {};
  \end{pgfonlayer}
  \begin{pgfonlayer}{edgelayer}
    \draw (0.center) to (2);
    \draw (1.center) to (3);
    \draw (4) to (5.center);
  \end{pgfonlayer}
\end{tikzpicture}\,}

\newcommand{\graytwopointketbra}{\twopointketbra[style1=gray copoint,style2=gray point,style3=gray point]}

\newcommand{\idwire}{\,%
\begin{tikzpicture}
  \begin{pgfonlayer}{nodelayer}
    \node [style=none] (0) at (0, -1.5) {};
    \node [style=none] (3) at (0, 1.5) {};
  \end{pgfonlayer}
  \begin{pgfonlayer}{edgelayer}
    \draw (0.center) to (3.center);
  \end{pgfonlayer}
\end{tikzpicture}\,}

\newcommand{\didwire}{\,%
\begin{tikzpicture}
  \begin{pgfonlayer}{nodelayer}
    \node [style=none] (0) at (0, -1.5) {};
    \node [style=none] (3) at (0, 1.5) {};
  \end{pgfonlayer}
  \begin{pgfonlayer}{edgelayer}
    \draw [boldedge] (0.center) to (3.center);
  \end{pgfonlayer}
\end{tikzpicture}\,}

\newcommand{\didwirelong}{\,%
\begin{tikzpicture}
  \begin{pgfonlayer}{nodelayer}
    \node [style=none] (0) at (0, -2) {};
    \node [style=none] (3) at (0, 2) {};
  \end{pgfonlayer}
  \begin{pgfonlayer}{edgelayer}
    \draw [boldedge] (0.center) to (3.center);
  \end{pgfonlayer}
\end{tikzpicture}\,}

\newcommand{\shortidwire}{\,%
\begin{tikzpicture}
  \begin{pgfonlayer}{nodelayer}
    \node [style=none] (0) at (0, -1) {};
    \node [style=none] (3) at (0, 1) {};
  \end{pgfonlayer}
  \begin{pgfonlayer}{edgelayer}
    \draw (0.center) to (3.center);
  \end{pgfonlayer}
\end{tikzpicture}\,}

\newkeycommand{\pointbraket}[style1=point,style2=copoint,estyle=][2]{\,%
\begin{tikzpicture}
  \begin{pgfonlayer}{nodelayer}
    \node [style=\commandkey{style1}] (0) at (0, -0.625) {$#1$};
    \node [style=\commandkey{style2}] (1) at (0, 0.625) {$#2$};
  \end{pgfonlayer}
  \begin{pgfonlayer}{edgelayer}
    \draw [\commandkey{estyle}] (0) to (1);
  \end{pgfonlayer}
\end{tikzpicture}\,}

\newcommand{\graypointbraket}[2]{\,%
\begin{tikzpicture}
  \begin{pgfonlayer}{nodelayer}
    \node [style=gray point] (0) at (0, -0.625) {$#1$};
    \node [style=gray copoint] (1) at (0, 0.625) {$#2$};
  \end{pgfonlayer}
  \begin{pgfonlayer}{edgelayer}
    \draw (0) to (1);
  \end{pgfonlayer}
\end{tikzpicture}\,}

\newcommand{\dpointbraket}[2]{\kpointbraket[style1=dpoint,style2=dcopoint,estyle=boldedge]{#1}{#2}}

\newcommand{\dkpointbraket}[2]{\kpointbraket[style1=dkpoint,style2=dkpointdag,estyle=boldedge]{#1}{#2}}

\newcommand{\dpointbraketx}[2]{\kpointbraketx[style1=dpoint,style2=dcopoint,estyle=boldedge]{#1}{#2}}

\newcommand{\dkpointbraketx}[2]{\kpointbraketx[style1=dkpoint,style2=dkpointdag,estyle=boldedge]{#1}{#2}}

\newkeycommand{\kpointketbra}[style1=kpointdag,style2=kpoint,estyle=][2]{\,%
\begin{tikzpicture}
  \begin{pgfonlayer}{nodelayer}
    \node [style=none] (0) at (0, -1.9) {};
    \node [style=\commandkey{style1}] (1) at (0, -0.95) {$#1$};
    \node [style=\commandkey{style2}] (2) at (0, 0.95) {$#2$};
    \node [style=none] (3) at (0, 1.9) {};
  \end{pgfonlayer}
  \begin{pgfonlayer}{edgelayer}
    \draw [\commandkey{estyle}] (0.center) to (1);
    \draw [\commandkey{estyle}] (2) to (3.center);
  \end{pgfonlayer}
\end{tikzpicture}\,}

\newcommand{\dkpointketbra}[2]{\kpointketbra[style1=dkpointdag,style2=dkpoint,estyle=boldedge]{#1}{#2}}

\newcommand{\kpointketbratwocols}[2]{\,%
\begin{tikzpicture}
  \begin{pgfonlayer}{nodelayer}
    \node [style=none] (0) at (0, -1.9) {};
    \node [style=kpointdag] (1) at (0, -0.85) {$#1$};
    \node [style=kpoint,fill=gray!40!white] (2) at (0, 0.85) {$#2$};
    \node [style=none] (3) at (0, 1.9) {};
  \end{pgfonlayer}
  \begin{pgfonlayer}{edgelayer} 
    \draw (0.center) to (1);
    \draw (2) to (3.center);
  \end{pgfonlayer}
\end{tikzpicture}\,}

\newcommand{\boxpointmap}[2]{\,%
\begin{tikzpicture}
  \begin{pgfonlayer}{nodelayer}
    \node [style=point] (0) at (0, -1) {$#1$};
    \node [style=map] (1) at (0, 0.5) {$#2$};
    \node [style=none] (2) at (0, 1.75) {};
  \end{pgfonlayer}
  \begin{pgfonlayer}{edgelayer}
    \draw (0) to (1);
    \draw (1) to (2.center);
  \end{pgfonlayer}
\end{tikzpicture}%
\,}

\newcommand{\boxtranspointmap}[2]{\,%
\begin{tikzpicture}
  \begin{pgfonlayer}{nodelayer}
    \node [style=point] (0) at (0, -1.50) {$#1$};
    \node [style=maptrans] (1) at (0, 0) {$#2$};
    \node [style=none] (2) at (0, 1.25) {};
  \end{pgfonlayer}
  \begin{pgfonlayer}{edgelayer}
    \draw (0) to (1);
    \draw (1) to (2.center);
  \end{pgfonlayer}
\end{tikzpicture}%
\,}

\newkeycommand{\kpointmap}[style1=kpoint,style2=map][2]{\,%
\begin{tikzpicture}
  \begin{pgfonlayer}{nodelayer}
    \node [style=\commandkey{style1}] (0) at (0, -1.1) {$#1$};
    \node [style=\commandkey{style2}] (1) at (0, 0.5) {$#2$};
    \node [style=none] (2) at (0, 1.75) {};
  \end{pgfonlayer}
  \begin{pgfonlayer}{edgelayer}
    \draw (0) to (1);
    \draw (1) to (2.center);
  \end{pgfonlayer}
\end{tikzpicture}%
\,}

\newkeycommand{\kpointmaptrans}[style1=kpointconj,style2=mapadj][2]{\,%
\begin{tikzpicture}
  \begin{pgfonlayer}{nodelayer}
    \node [style=\commandkey{style1}] (0) at (0, -1.1) {$#1$};
    \node [style=\commandkey{style2}] (1) at (0, 0.5) {$#2$};
    \node [style=none] (2) at (0, 1.75) {};
  \end{pgfonlayer}
  \begin{pgfonlayer}{edgelayer}
    \draw (0) to (1);
    \draw (1) to (2.center);
  \end{pgfonlayer}
\end{tikzpicture}%
\,}

\newkeycommand{\kpointmapadj}[style1=kpointadj,style2=map][2]{\,%
\begin{tikzpicture}
  \begin{pgfonlayer}{nodelayer}
    \node [style=\commandkey{style1}] (0) at (0, 1.1) {$#1$};
    \node [style=\commandkey{style2}] (1) at (0, -0.5) {$#2$};
    \node [style=none] (2) at (0, -1.75) {};
  \end{pgfonlayer}
  \begin{pgfonlayer}{edgelayer}
    \draw (0) to (1);
    \draw (1) to (2.center);
  \end{pgfonlayer}
\end{tikzpicture}%
\,}

\newkeycommand{\dkpointmap}[style1=dkpoint,style2=dmap][2]{\,%
\begin{tikzpicture}
  \begin{pgfonlayer}{nodelayer}
    \node [style=\commandkey{style1}] (0) at (0, -1.2) {$#1$};
    \node [style=\commandkey{style2}] (1) at (0, 0.4) {$#2$};
    \node [style=none] (2) at (0, 1.7) {};
  \end{pgfonlayer}
  \begin{pgfonlayer}{edgelayer}
    \draw[style=boldedge] (0) to (1);
    \draw[style=boldedge] (1) to (2.center);
  \end{pgfonlayer}
\end{tikzpicture}%
\,}

\newkeycommand{\dkpointmapx}[style1=dkpoint,style2=dmap][2]{\,%
\begin{tikzpicture}
  \begin{pgfonlayer}{nodelayer}
    \node [style=\commandkey{style1}] (0) at (0, -1.4) {$#1$};
    \node [style=\commandkey{style2}] (1) at (0, 0.6) {$#2$};
    \node [style=none] (2) at (0, 1.9) {};
  \end{pgfonlayer}
  \begin{pgfonlayer}{edgelayer}
    \draw[style=boldedge] (0) to (1);
    \draw[style=boldedge] (1) to (2.center);
  \end{pgfonlayer}
\end{tikzpicture}%
\,}

\newcommand{\boxpointmapdag}[2]{\,%
\begin{tikzpicture}
  \begin{pgfonlayer}{nodelayer}
    \node [style=point] (0) at (0, -1.50) {$#1$};
    \node [style=mapdag] (1) at (0, 0) {$#2$};
    \node [style=none] (2) at (0, 1.25) {};
  \end{pgfonlayer}
  \begin{pgfonlayer}{edgelayer}
    \draw (0) to (1);
    \draw (1) to (2.center);
  \end{pgfonlayer}
\end{tikzpicture}%
\,}

\newcommand{\boxcopointmap}[2]{\,%
\begin{tikzpicture}
  \begin{pgfonlayer}{nodelayer}
    \node [style=none] (0) at (0, -1.75) {};
    \node [style=map] (1) at (0, -0.5) {$#1$};
    \node [style=copoint] (2) at (0, 1) {$#2$};
  \end{pgfonlayer}
  \begin{pgfonlayer}{edgelayer}
    \draw (0.center) to (1);
    \draw (1) to (2);
  \end{pgfonlayer}
\end{tikzpicture}%
\,}

\newkeycommand{\boxcopointmapdag}[style1=map,style2=copoint][2]{\,%
\begin{tikzpicture}
  \begin{pgfonlayer}{nodelayer}
    \node [style=none] (0) at (0, -1.75) {};
    \node [style=\commandkey{style1}] (1) at (0, -0.5) {$#1$};
    \node [style=\commandkey{style2}] (2) at (0, 1) {$#2$};
  \end{pgfonlayer}
  \begin{pgfonlayer}{edgelayer}
    \draw (0.center) to (1);
    \draw (1) to (2);
  \end{pgfonlayer}
\end{tikzpicture}%
\,}

\newkeycommand{\kpointdagmap}[style1=map,style2=kpointdag][2]{\,%
\begin{tikzpicture}
  \begin{pgfonlayer}{nodelayer}
    \node [style=none] (0) at (0, -1.75) {};
    \node [style=\commandkey{style1}] (1) at (0, -0.5) {$#1$};
    \node [style=\commandkey{style2}] (2) at (0, 1.1) {$#2$};
  \end{pgfonlayer}
  \begin{pgfonlayer}{edgelayer}
    \draw (0.center) to (1);
    \draw (1) to (2);
  \end{pgfonlayer}
\end{tikzpicture}%
\,}

\newcommand{\kpointdagmapdag}[2]{\,%
\begin{tikzpicture}
  \begin{pgfonlayer}{nodelayer}
    \node [style=none] (0) at (0, -1.75) {};
    \node [style=mapdag] (1) at (0, -0.5) {$#1$};
    \node [style=kpointdag] (2) at (0, 1.1) {$#2$};
  \end{pgfonlayer}
  \begin{pgfonlayer}{edgelayer}
    \draw (0.center) to (1);
    \draw (1) to (2);
  \end{pgfonlayer}
\end{tikzpicture}%
\,}

\newkeycommand{\sandwichmap}[style1=point,style2=map,style3=copoint][3]{\,%
\begin{tikzpicture}
  \begin{pgfonlayer}{nodelayer}
    \node [style=\commandkey{style1}] (0) at (0, -1.50) {$#1$};
    \node [style=\commandkey{style2}] (1) at (0, 0) {$#2$};
    \node [style=\commandkey{style3}] (2) at (0, 1.50) {$#3$};
  \end{pgfonlayer}
  \begin{pgfonlayer}{edgelayer}
    \draw (0) to (1);
    \draw (1) to (2);
  \end{pgfonlayer}
\end{tikzpicture}%
}

\newkeycommand{\kpointsandwichmap}[style1=kpoint,style2=map,style3=kpointdag][3]{\,%
\begin{tikzpicture}
  \begin{pgfonlayer}{nodelayer}
    \node [style=\commandkey{style1}] (0) at (0, -1.5) {$#1$};
    \node [style=\commandkey{style2}] (1) at (0, 0) {$#2$};
    \node [style=\commandkey{style3}] (2) at (0, 1.5) {$#3$};
  \end{pgfonlayer}
  \begin{pgfonlayer}{edgelayer}
    \draw (0) to (1);
    \draw (1) to (2);
  \end{pgfonlayer}
\end{tikzpicture}%
}

\newcommand{\dkpointsandwichmap}[3]{\,%
\begin{tikzpicture}
  \begin{pgfonlayer}{nodelayer}
    \node [style=dkpoint] (0) at (0, -1.5) {$#1$};
    \node [style=dmap] (1) at (0, 0) {$#2$};
    \node [style=dkpointdag] (2) at (0, 1.5) {$#3$};
  \end{pgfonlayer}
  \begin{pgfonlayer}{edgelayer}
    \draw [boldedge] (0) to (1);
    \draw [boldedge] (1) to (2);
  \end{pgfonlayer}
\end{tikzpicture}%
}

\newcommand{\dkpointsandwichmapx}[3]{\,%
\begin{tikzpicture}
  \begin{pgfonlayer}{nodelayer}
    \node [style=dkpoint] (0) at (0, -1.85) {$#1$};
    \node [style=dmap] (1) at (0, 0) {$#2$};
    \node [style=dkpointdag] (2) at (0, 1.85) {$#3$};
  \end{pgfonlayer}
  \begin{pgfonlayer}{edgelayer}
    \draw [boldedge] (0) to (1);
    \draw [boldedge] (1) to (2);
  \end{pgfonlayer}
\end{tikzpicture}%
}

\newcommand{\kpointsandwichmapdag}[3]{\,%
\begin{tikzpicture}
  \begin{pgfonlayer}{nodelayer}
    \node [style=kpoint] (0) at (0, -1.5) {$#1$};
    \node [style=mapdag] (1) at (0, 0) {$#2$};
    \node [style=kpointdag] (2) at (0, 1.5) {$#3$};
  \end{pgfonlayer}
  \begin{pgfonlayer}{edgelayer}
    \draw (0) to (1);
    \draw (1) to (2);
  \end{pgfonlayer}
\end{tikzpicture}%
}

\newcommand{\sandwichmapdag}[3]{\,%
\begin{tikzpicture}
  \begin{pgfonlayer}{nodelayer}
    \node [style=point] (0) at (0, -1.50) {$#1$};
    \node [style=mapdag] (1) at (0, 0) {$#2$};
    \node [style=copoint] (2) at (0, 1.50) {$#3$};
  \end{pgfonlayer}
  \begin{pgfonlayer}{edgelayer}
    \draw (0) to (1);
    \draw (1) to (2);
  \end{pgfonlayer}
\end{tikzpicture}%
}

\newcommand{\sandwichmaptrans}[3]{\,%
\begin{tikzpicture}
  \begin{pgfonlayer}{nodelayer}
    \node [style=point] (0) at (0, -1.50) {$#1$};
    \node [style=maptrans] (1) at (0, 0) {$#2$};
    \node [style=copoint] (2) at (0, 1.50) {$#3$};
  \end{pgfonlayer}
  \begin{pgfonlayer}{edgelayer}
    \draw (0) to (1);
    \draw (1) to (2);
  \end{pgfonlayer}
\end{tikzpicture}%
}

\newkeycommand{\sandwichmapconj}[style1=point,style2=mapconj,style3=copoint][3]{\,%
\begin{tikzpicture}
  \begin{pgfonlayer}{nodelayer}
    \node [style=\commandkey{style1}] (0) at (0, -1.50) {$#1$};
    \node [style=\commandkey{style2}] (1) at (0, 0) {$#2$};
    \node [style=\commandkey{style3}] (2) at (0, 1.50) {$#3$};
  \end{pgfonlayer}
  \begin{pgfonlayer}{edgelayer}
    \draw (0) to (1);
    \draw (1) to (2);
  \end{pgfonlayer}
\end{tikzpicture}%
}

\newkeycommand{\sandwichtwo}[style1=point,style2=map,style3=map,style4=copoint][4]{\,%
\begin{tikzpicture}
  \begin{pgfonlayer}{nodelayer}
    \node [style=\commandkey{style2}] (0) at (0, -1) {$#2$};
    \node [style=\commandkey{style3}] (1) at (0, 1) {$#3$};
    \node [style=\commandkey{style1}] (2) at (0, -2.6) {$#1$};
    \node [style=\commandkey{style4}] (3) at (0, 2.6) {$#4$};
  \end{pgfonlayer}
  \begin{pgfonlayer}{edgelayer}
    \draw (2) to (0);
    \draw (0) to (1);
    \draw (1) to (3);
  \end{pgfonlayer}
\end{tikzpicture}\,}

\newkeycommand{\longbraket}[style1=point,style2=copoint][2]{\,%
\begin{tikzpicture}
  \begin{pgfonlayer}{nodelayer}
    \node [style=\commandkey{style1}] (0) at (0, -2.6) {$#1$};
    \node [style=\commandkey{style2}] (1) at (0, 2.6) {$#2$};
  \end{pgfonlayer}
  \begin{pgfonlayer}{edgelayer}
    \draw (0) to (1);
  \end{pgfonlayer}
\end{tikzpicture}\,}

\newkeycommand{\onb}[style=point]{\ensuremath{\{\,\tikz{\node[style=\commandkey{style}] (x) at (0,-0.3) {$j$};\draw (x)--(0,0.7);}\,\}}\xspace}

\newcommand{\whiteonb}{\onb[style=point]}

\newcommand{\grayonb}{\onb[style=gray point]}

\newcommand{\redonb}{\onb[style=red point]}

\newcommand{\greenonb}{\onb[style=green point]}

\newcommand{\whiteonbi}{\ensuremath{\left\{\raisebox{1mm}{\pointmap{i}}\right\}_i}\xspace}
\newcommand{\grayonbi}{\ensuremath{\left\{\raisebox{1mm}{\graypointmap{i}}\right\}_i}\xspace}
\newcommand{\greyonbi}{\ensuremath{\left\{\raisebox{1mm}{\graypointmap{i}}\right\}_i}\xspace}

\newcommand{\redpointmap}[1]{\,\tikz{\node[style=red point] (x) at (0,-0.3) {$#1$};\draw (x)--(0,0.7);}\,}

\newcommand{\redcopointmap}[1]{\,\tikz{\node[style=red copoint] (x) at (0,0.3) {$#1$};\draw (x)--(0,-0.7);}\,}

\newcommand{\greenpointmap}[1]{\,\tikz{\node[style=green point] (x) at (0,-0.3) {$#1$};\draw (x)--(0,0.7);}\,}

\newcommand{\greencopointmap}[1]{\,\tikz{\node[style=green copoint] (x) at (0,0.3) {$#1$};\draw (x)--(0,-0.7);}\,}

\newcommand{\meas}{\ensuremath{\,\begin{tikzpicture}
  \begin{pgfonlayer}{nodelayer}
    \node [style=white dot] (0) at (0, 0) {};
    \node [style=none] (1) at (0, 0.75) {};
    \node [style=none] (2) at (0, -0.75) {};
  \end{pgfonlayer}
  \begin{pgfonlayer}{edgelayer}
    \draw [style=swap] (1.center) to (0);
    \draw [style=boldedge] (0) to (2.center);
  \end{pgfonlayer}
\end{tikzpicture}\,}\xspace}

\newcommand{\encode}{\ensuremath{\,\begin{tikzpicture}
  \begin{pgfonlayer}{nodelayer}
    \node [style=none] (0) at (0, -0.75) {};
    \node [style=none] (1) at (0, 0.75) {};
    \node [style=white dot] (2) at (0, 0) {};
  \end{pgfonlayer}
  \begin{pgfonlayer}{edgelayer}
    \draw [style=swap] (2.center) to (0);
    \draw [style=boldedge] (1) to (2.center);
  \end{pgfonlayer}
\end{tikzpicture}\,}\xspace}



\newcommand{\grayphasemult}[2]{\,\begin{tikzpicture}
	\begin{pgfonlayer}{nodelayer}
		\node [style=gray dot] (0) at (0, 0.5) {};
		\node [style=gray phase dot] (1) at (0.75, -0.75) {$#1$};
		\node [style=gray phase dot] (2) at (-0.75, -0.75) {$#2$};
		\node [style=none] (3) at (0, 1.5) {};
	\end{pgfonlayer}
	\begin{pgfonlayer}{edgelayer}
		\draw [in=-165, out=90, looseness=1.00] (2) to (0);
		\draw [in=-15, out=90, looseness=1.00] (1) to (0);
		\draw (0) to (3.center);
	\end{pgfonlayer}
\end{tikzpicture}\,}

\newcommand{\grayphasemultbis}[2]{\,\begin{tikzpicture}
	\begin{pgfonlayer}{nodelayer}
		\node [style=gray dot] (0) at (0, 0.5) {};
		\node [style=point] (1) at (0.75, -0.75) {$#1$};
		\node [style=point] (2) at (-0.75, -0.75) {$#2$};
		\node [style=none] (3) at (0, 1.5) {};
	\end{pgfonlayer}
	\begin{pgfonlayer}{edgelayer}
		\draw [in=-165, out=90, looseness=1.00] (2) to (0);
		\draw [in=-15, out=90, looseness=1.00] (1) to (0);
		\draw (0) to (3.center);
	\end{pgfonlayer}
\end{tikzpicture}\,}

\newcommand{\phasepointmult}[2]{\,\begin{tikzpicture}
  \begin{pgfonlayer}{nodelayer}
    \node [style=white dot] (0) at (0, 0.5) {};
    \node [style=white dot] (1) at (-0.75, -0.75) {$#1$};
    \node [style=none] (2) at (0, 1.5) {};
    \node [style=white dot] (3) at (0.75, -0.75) {$#2$};
  \end{pgfonlayer}
  \begin{pgfonlayer}{edgelayer}
    \draw [in=-165, out=90, looseness=1.00] (1) to (0);
    \draw [in=-15, out=90, looseness=1.00] (3) to (0);
    \draw (0) to (2.center);
  \end{pgfonlayer}
\end{tikzpicture}\,}

\newcommand{\grayphasepoint}[1]{\phasepoint[style=gray phase dot]{#1}}
\newcommand{\grayphasecopoint}[1]{\phasecopoint[style=gray phase dot]{#1}}

\newcommand{\graysquarepoint}[1]{\begin{tikzpicture}
    \begin{pgfonlayer}{nodelayer}
        \node [style=gray square point] (0) at (0, -0.5) {};
        \node [style=none] (1) at (0, 0.75) {};
        \node [style=none] (2) at (0.75, -0.5) {$#1$};
    \end{pgfonlayer}
    \begin{pgfonlayer}{edgelayer}
        \draw (0) to (1.center);
    \end{pgfonlayer}
\end{tikzpicture}}

\newkeycommand{\phase}[style=white phase dot][1]{\,\begin{tikzpicture}
    \begin{pgfonlayer}{nodelayer}
        \node [style=none] (0) at (0, 1) {};
        \node [style=\commandkey{style}] (2) at (0, -0) {$#1$}; 
        \node [style=none] (3) at (0, -1) {};
    \end{pgfonlayer}
    \begin{pgfonlayer}{edgelayer}
        \draw (2) to (0.center);
        \draw (3.center) to (2);
    \end{pgfonlayer}
\end{tikzpicture}\,}

\newkeycommand{\dphase}[style=white phase ddot][1]{\,\begin{tikzpicture}
    \begin{pgfonlayer}{nodelayer}
        \node [style=none] (0) at (0, 1.25) {};
        \node [style=\commandkey{style}] (2) at (0, -0) {$#1$};
        \node [style=none] (3) at (0, -1.25) {};
    \end{pgfonlayer}
    \begin{pgfonlayer}{edgelayer}
        \draw [boldedge] (2) to (0.center);
        \draw [boldedge] (3.center) to (2);
    \end{pgfonlayer}
\end{tikzpicture}\,}

\newcommand{\dphasegray}[1]{\dphase[style=gray phase ddot]{#1}}

\newkeycommand{\dphasepoint}[style=white phase ddot][1]{\,\begin{tikzpicture}[yshift=-3mm]
  \begin{pgfonlayer}{nodelayer}
    \node [style=none] (0) at (0, 1) {};
    \node [style=\commandkey{style}] (1) at (0, 0) {$#1$};
  \end{pgfonlayer}
  \begin{pgfonlayer}{edgelayer}
    \draw [style=boldedge] (0.center) to (1);
  \end{pgfonlayer}
\end{tikzpicture}\,}

\newkeycommand{\dphasepointgray}[style=gray phase ddot][1]{\,\begin{tikzpicture}[yshift=-3mm]
  \begin{pgfonlayer}{nodelayer}
    \node [style=none] (0) at (0, 1) {};
    \node [style=\commandkey{style}] (1) at (0, 0) {$#1$};
  \end{pgfonlayer}
  \begin{pgfonlayer}{edgelayer}
    \draw [style=boldedge] (0.center) to (1);
  \end{pgfonlayer}
\end{tikzpicture}\,}

\newkeycommand{\dphasecopoint}[style=white phase ddot][1]{\,\begin{tikzpicture}[yshift=3mm,yscale=-1]
  \begin{pgfonlayer}{nodelayer}
    \node [style=none] (0) at (0, 1) {};
    \node [style=\commandkey{style}] (1) at (0, 0) {$#1$};
  \end{pgfonlayer}
  \begin{pgfonlayer}{edgelayer}
    \draw [style=boldedge] (0.center) to (1);
  \end{pgfonlayer}
\end{tikzpicture}\,}

\newkeycommand{\dphasepointsm}[style=white phase ddot][1]{\,\begin{tikzpicture}[yshift=-3mm]
  \begin{pgfonlayer}{nodelayer}
    \node [style=none] (0) at (0, 1) {};
    \node [style=\commandkey{style}] (1) at (0, 0) {\footnotesize\!$#1$\!};
  \end{pgfonlayer}
  \begin{pgfonlayer}{edgelayer}
    \draw [style=boldedge] (0.center) to (1);
  \end{pgfonlayer}
\end{tikzpicture}\,}

\newcommand{\measure}[1]{\,\begin{tikzpicture}
    \begin{pgfonlayer}{nodelayer}
        \node [style=none] (0) at (0, 0.75) {};
        \node [style=#1] (2) at (0, -0) {};
        \node [style=none] (3) at (0, -0.75) {};
    \end{pgfonlayer}
    \begin{pgfonlayer}{edgelayer}
        \draw (2) to (0.center);
        \draw[doubled] (3.center) to (2);
    \end{pgfonlayer}
\end{tikzpicture}\,}

\newcommand{\nondemmeasure}[1]{\,\begin{tikzpicture}
	\begin{pgfonlayer}{nodelayer}
		\node [style=none] (0) at (1, 1) {};
		\node [style=none] (1) at (1, 1.25) {};
		\node [style=none] (2) at (0.75, 0.5) {};
		\node [style=#1] (3) at (0, -0.25) {};
		\node [style=none] (4) at (0, 1.25) {};
		\node [style=none] (5) at (0, -1.25) {};
	\end{pgfonlayer}
	\begin{pgfonlayer}{edgelayer}
		\draw (0.center) to (1.center);
		\draw (3) to (2.center);
		\draw [in=-90, out=45, looseness=1.00] (2.center) to (0.center);
		\draw [style=boldedge] (5.center) to (3);
		\draw [style=boldedge] (3) to (4.center);
	\end{pgfonlayer}
\end{tikzpicture}\,}

\newcommand{\prepare}[1]{\,\begin{tikzpicture}
    \begin{pgfonlayer}{nodelayer}
        \node [style=none] (0) at (0, 0.75) {};
        \node [style=#1] (2) at (0, -0) {};
        \node [style=none] (3) at (0, -0.75) {};
    \end{pgfonlayer}
    \begin{pgfonlayer}{edgelayer}
        \draw[doubled] (2) to (0.center);
        \draw (3.center) to (2);
    \end{pgfonlayer}
\end{tikzpicture}\,}

\newcommand{\whitephase}[1]{\phase[style=white phase dot]{#1}}
\newcommand{\whitedphase}[1]{\dphase[style=white phase ddot]{#1}}
\newcommand{\greyphase}[1]{\phase[style=grey phase dot]{#1}}
\newcommand{\greydphase}[1]{\dphase[style=grey phase ddot]{#1}}

\newkeycommand{\phasepoint}[style=white phase dot][1]{\,\begin{tikzpicture}[yshift=-3mm]
  \begin{pgfonlayer}{nodelayer}
    \node [style=none] (0) at (0, 1) {};
    \node [style=\commandkey{style}] (1) at (0, 0) {$#1$};
  \end{pgfonlayer}
  \begin{pgfonlayer}{edgelayer}
    \draw (0.center) to (1);
  \end{pgfonlayer}
\end{tikzpicture}\,}

\newkeycommand{\phasecopoint}[style=white phase dot][1]{\,\begin{tikzpicture}[yshift=3mm]
  \begin{pgfonlayer}{nodelayer}
    \node [style=none] (0) at (0, -1) {};
    \node [style=\commandkey{style}] (1) at (0, 0) {$#1$};
  \end{pgfonlayer}
  \begin{pgfonlayer}{edgelayer}
    \draw (0.center) to (1);
  \end{pgfonlayer}
\end{tikzpicture}\,}

\newcommand{\pointcopointmap}[2]{\,\tikz{
  \node[style=copoint] (x) at (0,-0.7) {$#2$};\draw (x)--(0,-1.5);
  \node[style=point] (y) at (0,0.7) {$#1$};\draw (0,1.5)--(y);
}\,}

\newcommand{\innerprodmap}[2]{\,\tikz{
\node[style=copoint] (y) at (0,0.625) {$#1$};
\node[style=point] (x) at (0,-0.625) {$#2$};
\draw (x)--(y);}\,}

\newkeycommand{\kpointbraket}[style1=kpoint,style2=kpoint adjoint,estyle=][2]{\,%
\begin{tikzpicture}
  \begin{pgfonlayer}{nodelayer}
    \node [style=\commandkey{style1}] (0) at (0, -0.735) {$#1$};
    \node [style=\commandkey{style2}] (1) at (0, 0.735) {$#2$};
  \end{pgfonlayer}
  \begin{pgfonlayer}{edgelayer}
    \draw [\commandkey{estyle}] (0) to (1);
  \end{pgfonlayer}
\end{tikzpicture}\,}

\newkeycommand{\kkpointbraket}[style1=kpoint,style2=kpoint adjoint,estyle=][2]{\,%
\begin{tikzpicture}
  \begin{pgfonlayer}{nodelayer}
    \node [style=\commandkey{style1}] (0) at (0, -0.685) {$#1$};
    \node [style=\commandkey{style2}] (1) at (0, 0.685) {$#2$};
  \end{pgfonlayer}
  \begin{pgfonlayer}{edgelayer}
    \draw [\commandkey{estyle}] (0) to (1);
  \end{pgfonlayer}
\end{tikzpicture}\,}

\newkeycommand{\kpointbraketx}[style1=kpoint,style2=kpoint adjoint,estyle=][2]{\,%
\begin{tikzpicture}
  \begin{pgfonlayer}{nodelayer}
    \node [style=\commandkey{style1}] (0) at (0, -0.885) {$#1$};
    \node [style=\commandkey{style2}] (1) at (0, 0.885) {$#2$};
  \end{pgfonlayer}
  \begin{pgfonlayer}{edgelayer}
    \draw [\commandkey{estyle}] (0) to (1);
  \end{pgfonlayer}
\end{tikzpicture}\,}

\newcommand{\kpointbraketconj}[2]{\,%
\begin{tikzpicture}
  \begin{pgfonlayer}{nodelayer}
    \node [style=kpoint conjugate] (0) at (0, -0.735) {$#1$};
    \node [style=kpoint transpose] (1) at (0, 0.735) {$#2$};
  \end{pgfonlayer}
  \begin{pgfonlayer}{edgelayer}
    \draw (0) to (1);
  \end{pgfonlayer}
\end{tikzpicture}\,}

\newcommand{\wginnerprodmap}[2]{\,\begin{tikzpicture}
    \begin{pgfonlayer}{nodelayer}
        \node [style=point] (0) at (0, -0.75) {$#2$};
        \node [style=gray copoint] (1) at (0, 0.75) {$#1$};
    \end{pgfonlayer}
    \begin{pgfonlayer}{edgelayer}
        \draw (0) to (1);
    \end{pgfonlayer}
\end{tikzpicture}\,}

\def\alpvec{[\vec{\alpha}]}
\def\alpprevec{\vec{\alpha}}
\def\betvec{[\vec{\beta}_j]}
\def\betprevec{\vec{\beta}_j}

\def\blackmu{\mu_{\smallblackdot}} 
\def\graymu {\mu_{\smallgraydot}}
\def\whitemu{\mu_{\smallwhitedot}}

\def\blacketa{\eta_{\smallblackdot}} 
\def\grayeta {\eta_{\smallgraydot}}
\def\whiteeta{\eta_{\smallwhitedot}}

\def\blackdelta{\delta_{\smallblackdot}} 
\def\graydelta {\delta_{\smallgraydot}}
\def\whitedelta{\delta_{\smallwhitedot}}

\def\blackepsilon{\epsilon_{\smallblackdot}} 
\def\grayepsilon {\epsilon_{\smallgraydot}}
\def\whiteepsilon{\epsilon_{\smallwhitedot}}

\def\whitePhi{\Phi_{\!\smallwhitedot}}
\def\graymeas{m_{\!\smallgraydot}}

\newcommand{\Owg}{\ensuremath{{\cal O}_{\!\smallwhitedot\!,\!\smallgraydot}}}
\newcommand{\whiteK}{\ensuremath{{\cal  K}_{\!\smallwhitedot}}}
\newcommand{\grayK}{\ensuremath{{\cal  K}_{\!\smallgraydot}}}

\newcommand{\roundket}[1]{\ensuremath{\left|  #1 \right)}}
\newcommand{\ketbra}[2]{\ensuremath{\ket{#1}\!\bra{#2}}}

\newcommand{\ketGHZ} {\ket{\textit{GHZ}\,}}
\newcommand{\ketW}   {\ket{\textit{W\,}}}
\newcommand{\ketBell}{\ket{\textit{Bell\,}}}
\newcommand{\ketEPR} {\ket{\textit{EPR\,}}}
\newcommand{\ketGHZD}{\ket{\textrm{GHZ}^{(D)}}}
\newcommand{\ketWD}  {\ket{\textrm{W}^{(D)}}}

\newcommand{\braGHZ} {\bra{\textit{GHZ}\,}}
\newcommand{\braW}   {\bra{\textit{W\,}}}
\newcommand{\braBell}{\bra{\textit{Bell\,}}}
\newcommand{\braEPR} {\bra{\textit{EPR\,}}}

\newcommand{\catC}{\ensuremath{\mathcal{C}}\xspace}
\newcommand{\catCop}{\ensuremath{\mathcal{C}^{\mathrm{op}}}\xspace}
\newcommand{\catD}{\ensuremath{\mathcal{D}}\xspace}
\newcommand{\catDop}{\ensuremath{\mathcal{D}^{\mathrm{op}}}\xspace}

\newcommand{\catSet}{\ensuremath{\textrm{\bf Set}}\xspace}
\newcommand{\catRel}{\ensuremath{\textrm{\bf Rel}}\xspace}
\newcommand{\catFRel}{\ensuremath{\textrm{\bf FRel}}\xspace}
\newcommand{\catVect}{\ensuremath{\textrm{\bf Vect}}\xspace}
\newcommand{\catFVect}{\ensuremath{\textrm{\bf FVect}}\xspace}
\newcommand{\catFHilb}{\ensuremath{\textrm{\bf FHilb}}\xspace}
\newcommand{\catHilb}{\ensuremath{\textrm{\bf Hilb}}\xspace}
\newcommand{\catSuperHilb}{\ensuremath{\textrm{\bf SuperHilb}}\xspace}
\newcommand{\catAb}{\ensuremath{\textrm{\bf Ab}}\xspace}
\newcommand{\catTop}{\ensuremath{\textrm{\bf Top}}\xspace}
\newcommand{\catCHaus}{\ensuremath{\textrm{\bf CHaus}}\xspace}
\newcommand{\catHaus}{\ensuremath{\textrm{\bf Haus}}\xspace}
\newcommand{\catGraph}{\ensuremath{\textrm{\bf Graph}}\xspace}
\newcommand{\catMat}{\ensuremath{\textrm{\bf Mat}}\xspace}
\newcommand{\catGr}{\ensuremath{\textrm{\bf Gr}}\xspace}
\newcommand{\catSpek}{\ensuremath{\textrm{\bf Spek}}\xspace}


\tikzstyle{cdiag}=[matrix of math nodes, row sep=3em, column sep=3em, text height=1.5ex, text depth=0.25ex,inner sep=0.5em]
\tikzstyle{arrow above}=[transform canvas={yshift=0.5ex}]
\tikzstyle{arrow below}=[transform canvas={yshift=-0.5ex}]

\newcommand{\csquare}[8]{
\begin{tikzpicture}
    \matrix(m)[cdiag,ampersand replacement=\&]{
    #1 \& #2 \\
    #3 \& #4  \\};
    \path [arrs] (m-1-1) edge node {$#5$} (m-1-2)
                 (m-2-1) edge node {$#6$} (m-2-2)
                 (m-1-1) edge node [swap] {$#7$} (m-2-1)
                 (m-1-2) edge node {$#8$} (m-2-2);
\end{tikzpicture}
}

\newcommand{\NWbracket}[1]{%
\draw #1 +(-0.25,0.5) -- +(-0.5,0.5) -- +(-0.5,0.25);}
\newcommand{\NEbracket}[1]{%
\draw #1 +(0.25,0.5) -- +(0.5,0.5) -- +(0.5,0.25);}
\newcommand{\SWbracket}[1]{%
\draw #1 +(-0.25,-0.5) -- +(-0.5,-0.5) -- +(-0.5,-0.25);}
\newcommand{\SEbracket}[1]{%
\draw #1 +(0.25,-0.5) -- +(0.5,-0.5) -- +(0.5,-0.25);}
\newcommand{\THETAbracket}[2]{%
\draw [rotate=#1] #2 +(0.25,0.5) -- +(0.5,0.5) -- +(0.5,0.25);}

\newcommand{\posquare}[8]{
\begin{tikzpicture}
    \matrix(m)[cdiag,ampersand replacement=\&]{
    #1 \& #2 \\
    #3 \& #4  \\};
    \path [arrs] (m-1-1) edge node {$#5$} (m-1-2)
                 (m-2-1) edge node [swap] {$#6$} (m-2-2)
                 (m-1-1) edge node [swap] {$#7$} (m-2-1)
                 (m-1-2) edge node {$#8$} (m-2-2);
    \NWbracket{(m-2-2)}
\end{tikzpicture}
}

\newcommand{\pbsquare}[8]{
\begin{tikzpicture}
    \matrix(m)[cdiag,ampersand replacement=\&]{
    #1 \& #2 \\
    #3 \& #4  \\};
    \path [arrs] (m-1-1) edge node {$#5$} (m-1-2)
                 (m-2-1) edge node [swap] {$#6$} (m-2-2)
                 (m-1-1) edge node [swap] {$#7$} (m-2-1)
                 (m-1-2) edge node {$#8$} (m-2-2);
  \SEbracket{(m-1-1)}
\end{tikzpicture}
}

\newcommand{\ctri}[6]{
    \begin{tikzpicture}[-latex]
        \matrix (m) [cdiag,ampersand replacement=\&] { #1 \& #2 \\ #3 \& \\ };
        \path [arrs] (m-1-1) edge node {$#4$} (m-1-2)
              (m-1-1) edge node [swap] {$#5$} (m-2-1)
              (m-2-1) edge node [swap] {$#6$} (m-1-2);
    \end{tikzpicture}
}

\newcommand{\carr}[3]{
\begin{tikzpicture}
    \matrix(m)[cdiag,ampersand replacement=\&]{
    #1 \& #3 \\};
    \path [arrs] (m-1-1) edge node {$#2$} (m-1-2);
\end{tikzpicture}
}

\newcommand{\crun}[5]{
\begin{tikzpicture}
    \matrix(m)[cdiag,ampersand replacement=\&]{
    #1 \& #3 \& #5 \\};
    \path [arrs] (m-1-1) edge node {$#2$} (m-1-2)
                 (m-1-2) edge node {$#4$} (m-1-3);
\end{tikzpicture}
}

\newcommand{\cspan}[5]{
\ensuremath{#1 \overset{#2}{\longleftarrow} #3
               \overset{#4}{\longrightarrow} #5}
}

\newcommand{\ccospan}[5]{
\ensuremath{#1 \overset{#2}{\longrightarrow} #3
               \overset{#4}{\longleftarrow} #5}
}

\newcommand{\cpair}[4]{
\begin{tikzpicture}
    \matrix(m)[cdiag,ampersand replacement=\&]{
    #1 \& #2 \\};
    \path [arrs] (m-1-1.20) edge node {$#3$} (m-1-2.160)
                 (m-1-1.-20) edge node [swap] {$#4$} (m-1-2.-160);
\end{tikzpicture}
}

\newcommand{\csquareslant}[9]{
\begin{tikzpicture}[-latex]
    \matrix(m)[cdiag,ampersand replacement=\&]{
    #1 \& #2 \\
    #3 \& #4  \\};
    \path [arrs] (m-1-1) edge node {$#5$} (m-1-2)
                 (m-2-1) edge node {$#6$} (m-2-2)
                 (m-1-1) edge node [swap] {$#7$} (m-2-1)
                 (m-1-2) edge node {$#8$} (m-2-2)
                 (m-1-2) edge node [swap] {$#9$} (m-2-1);
\end{tikzpicture}
}


\usetikzlibrary{decorations.markings}
\usetikzlibrary{shapes.geometric}

\pgfdeclarelayer{edgelayer}
\pgfdeclarelayer{nodelayer}
\pgfsetlayers{edgelayer,nodelayer,main}

\tikzstyle{none}=[inner sep=0pt]
\definecolor{hexcolor0xff0000}{rgb}{1.000,0.000,0.000}
\definecolor{hexcolor0x000000}{rgb}{0.000,0.000,0.000}
\definecolor{hexcolor0x00ff00}{rgb}{0.000,1.000,0.000}
\definecolor{hexcolor0x000000}{rgb}{0.000,0.000,0.000}
\definecolor{hexcolor0xffff00}{rgb}{1.000,1.000,0.000}
\definecolor{hexcolor0xffffff}{rgb}{1.000,1.000,1.000}

\tikzstyle{rn}=[circle,fill=hexcolor0xff0000,draw=hexcolor0x000000,line width=0.8 pt]
\tikzstyle{gn}=[circle,fill=hexcolor0x00ff00,draw=hexcolor0x000000,line width=0.8 pt]
\tikzstyle{yn}=[circle,fill=hexcolor0xffff00,draw=hexcolor0x000000,line width=0.8 pt]
\tikzstyle{wn}=[circle,fill=hexcolor0xffffff,draw=hexcolor0x000000,line width=0.8 pt]
\tikzstyle{wnthick}=[circle,fill=hexcolor0xffffff,draw=hexcolor0x000000,line width=2.500]

\tikzstyle{simple}=[-,draw=hexcolor0x000000,line width=2.000]
\tikzstyle{arrow}=[-,draw=hexcolor0x000000,postaction={decorate},decoration={markings,mark=at position .5 with {\arrow{>}}},line width=2.000]
\tikzstyle{tick}=[-,draw=hexcolor0x000000,postaction={decorate},decoration={markings,mark=at position .5 with {\draw (0,-0.1) -- (0,0.1);}},line width=2.000]
\tikzstyle{halfthickness}=[-,draw=hexcolor0x000000,line width=0.500]
\tikzstyle{thick}=[-,draw=hexcolor0x000000,line width=2.500]
\tikzstyle{thicker}=[-,draw=hexcolor0x000000,line width=4.000]

\tikzstyle{env}=[copoint,regular polygon rotate=0,minimum width=0.2cm, fill=black]

\tikzstyle{probs}=[shape=semicircle,fill=white,draw=black,shape border rotate=180,minimum width=1.2cm]

%
%


\tikzstyle{every picture}=[baseline=-0.25em,scale=0.5]
\tikzstyle{dotpic}=[] 
\tikzstyle{diredges}=[every to/.style={diredge}]
\tikzstyle{math matrix}=[matrix of math nodes,left delimiter=(,right delimiter=),inner sep=2pt,column sep=1em,row sep=0.5em,nodes={inner sep=0pt},text height=1.5ex, text depth=0.25ex]


\tikzstyle{inline text}=[text height=1.5ex, text depth=0.25ex,yshift=0.5mm]
\tikzstyle{label}=[font=\footnotesize,text height=1.5ex, text depth=0.25ex,yshift=0.5mm]
\tikzstyle{left label}=[label,anchor=east,xshift=1.5mm]
\tikzstyle{right label}=[label,anchor=west,xshift=-1.5mm]

\tikzstyle{braceedge}=[decorate,decoration={brace,amplitude=2mm,raise=-1mm}]
\tikzstyle{small braceedge}=[decorate,decoration={brace,amplitude=1mm,raise=-1mm}]
\tikzstyle{singleedgegray}=[gray]

\tikzstyle{doubled}=[line width=1.6pt] 
\tikzstyle{boldedge}=[doubled,shorten <=-0.17mm,shorten >=-0.17mm]
\tikzstyle{boldedgegray}=[doubled,gray,shorten <=-0.17mm,shorten >=-0.17mm]

\tikzstyle{semidoubled}=[line width=1.4pt] 
\tikzstyle{semiboldedgegray}=[semidoubled,gray,shorten <=-0.17mm,shorten >=-0.17mm]

\tikzstyle{boldedgedashed}=[very thick,dashed,shorten <=-0.17mm,shorten >=-0.17mm]
\tikzstyle{vboldedgedashed}=[doubled,dashed,shorten <=-0.17mm,shorten >=-0.17mm]
\tikzstyle{left hook arrow}=[left hook-latex]
\tikzstyle{right hook arrow}=[right hook-latex]
\tikzstyle{sembracket}=[line width=0.5pt,shorten <=-0.07mm,shorten >=-0.07mm]

\tikzstyle{causal edge}=[->,thick,gray]
\tikzstyle{causal nondir}=[thick,gray]
\tikzstyle{timeline}=[thick,gray, dashed]

\tikzstyle{cedge}=[<->,thick,gray!70!white]

\tikzstyle{empty diagram}=[draw=gray!40!white,dashed,shape=rectangle,minimum width=1cm,minimum height=1cm]
\tikzstyle{empty diagram small}=[draw=gray!50!white,dashed,shape=rectangle,minimum width=0.6cm,minimum height=0.5cm]




\tikzstyle{dot}=[inner sep=0mm,minimum width=2mm,minimum height=2mm,draw,shape=circle]
\tikzstyle{ddot}=[inner sep=0mm, doubled, minimum width=2.5mm,minimum height=2.5mm,draw,shape=circle]

\tikzstyle{black dot}=[dot,fill=black]
\tikzstyle{white dot}=[dot,fill=white,,text depth=-0.2mm]
\tikzstyle{green dot}=[white dot] 
\tikzstyle{gray dot}=[dot,fill=gray!40!white,,text depth=-0.2mm]
\tikzstyle{red dot}=[gray dot] 


\tikzstyle{black ddot}=[ddot,fill=black]
\tikzstyle{white ddot}=[ddot,fill=white]
\tikzstyle{gray ddot}=[ddot,fill=gray!40!white]

\tikzstyle{gray edge}=[gray!40!white]

\tikzstyle{small dot}=[inner sep=0.5mm,minimum width=0pt,minimum height=0pt,draw,shape=circle]

\tikzstyle{small black dot}=[small dot,fill=black]
\tikzstyle{small white dot}=[small dot,fill=white]
\tikzstyle{small gray dot}=[small dot,fill=gray!40!white]

\tikzstyle{causal dot}=[inner sep=0.4mm,minimum width=0pt,minimum height=0pt,draw=white,shape=circle,fill=gray!40!white]


\tikzstyle{phase dimensions}=[minimum size=5mm,font=\footnotesize,rectangle,rounded corners=2.5mm,inner sep=0.2mm,outer sep=-2mm]
\tikzstyle{dphase dimensions}=[minimum size=5mm,font=\footnotesize,rectangle,rounded corners=2.5mm,inner sep=0.2mm,outer sep=-2mm]
\tikzstyle{phase dimensions small}=[minimum size=3.0mm,font=\footnotesize,rectangle,rounded corners=1.5mm,inner sep=0.2mm,outer sep=-1.2mm]

\tikzstyle{white phase dot}=[dot,fill=white,phase dimensions]
\tikzstyle{white phase ddot}=[ddot,fill=white,dphase dimensions]
\tikzstyle{green phase ddot}=[ddot,fill=green,dphase dimensions]
\tikzstyle{white phase dot small}=[dot,fill=white,phase dimensions small]
\tikzstyle{gray phase dot small}=[dot,fill=gray!40!white,phase dimensions small]

\tikzstyle{white rect ddot}=[draw=black,fill=white,doubled,minimum size=5mm,font=\footnotesize,rectangle,rounded corners=2.5mm,inner sep=0.2mm]
\tikzstyle{gray rect ddot}=[draw=black,fill=gray!40!white,doubled,minimum size=6mm,font=\footnotesize,rectangle,rounded corners=3mm]

\tikzstyle{gray phase dot}=[dot,fill=gray!40!white,phase dimensions]
\tikzstyle{gray phase ddot}=[ddot,fill=gray!40!white,dphase dimensions]
\tikzstyle{red phase ddot}=[ddot,fill=red,dphase dimensions]

\tikzstyle{grey phase dot}=[gray phase dot]
\tikzstyle{grey phase ddot}=[gray phase ddot]

\tikzstyle{small phase dimensions}=[minimum size=4mm,font=\tiny,rectangle,rounded corners=2mm,inner sep=0.2mm,outer sep=-2mm]
\tikzstyle{small dphase dimensions}=[minimum size=4mm,font=\tiny,rectangle,rounded corners=2mm,inner sep=0.2mm,outer sep=-2mm]

\tikzstyle{small gray phase dot}=[dot,fill=gray!40!white,small phase dimensions]
\tikzstyle{small gray phase ddot}=[ddot,fill=gray!40!white,small dphase dimensions]


\tikzstyle{small map}=[draw,shape=rectangle,minimum height=4mm,minimum width=4mm,fill=white]

\tikzstyle{cnot}=[fill=white,shape=circle,inner sep=-1.4pt]

\tikzstyle{asym hadamard}=[fill=white,draw,shape=NEbox,inner sep=0.6mm,font=\footnotesize,minimum height=4mm]
\tikzstyle{asym hadamard conj}=[fill=white,draw,shape=NWbox,inner sep=0.6mm,font=\footnotesize,minimum height=4mm]
\tikzstyle{asym hadamard dag}=[fill=white,draw,shape=SEbox,inner sep=0.6mm,font=\footnotesize,minimum height=4mm]

\tikzstyle{hadamard}=[fill=white,draw,inner sep=0.6mm,font=\footnotesize,minimum height=4mm,minimum width=4mm]
\tikzstyle{small hadamard}=[fill=white,draw,inner sep=0.6mm,minimum height=1.5mm,minimum width=1.5mm]
\tikzstyle{dhadamard}=[hadamard,doubled]
\tikzstyle{small dhadamard}=[small hadamard,doubled]
\tikzstyle{small dhadamard rotate}=[small hadamard,doubled,rotate=45]
\tikzstyle{antipode}=[white dot,inner sep=0.3mm,font=\footnotesize]

\tikzstyle{scalar}=[diamond,draw,inner sep=0.5pt,font=\small]
\tikzstyle{dscalar}=[diamond,doubled, draw,inner sep=0.5pt,font=\small]

\tikzstyle{small box}=[rectangle,inline text,fill=white,draw,minimum height=5mm,yshift=-0.5mm,minimum width=5mm,font=\small]
\tikzstyle{small gray box}=[small box,fill=gray!30]
\tikzstyle{medium box}=[rectangle,inline text,fill=white,draw,minimum height=5mm,yshift=-0.5mm,minimum width=10mm,font=\small]
\tikzstyle{square box}=[small box] 
\tikzstyle{medium gray box}=[small box,fill=gray!30]
\tikzstyle{semilarge box}=[rectangle,inline text,fill=white,draw,minimum height=5mm,yshift=-0.5mm,minimum width=12.5mm,font=\small]
\tikzstyle{large box}=[rectangle,inline text,fill=white,draw,minimum height=5mm,yshift=-0.5mm,minimum width=15mm,font=\small]
\tikzstyle{large gray box}=[small box,fill=gray!30]

\tikzstyle{Bayes box}=[rectangle,fill=black,draw, minimum height=3mm, minimum width=3mm]

\tikzstyle{gray square point}=[small box,fill=gray!50]

\tikzstyle{dphase box white}=[dhadamard]
\tikzstyle{dphase box gray}=[dhadamard,fill=gray!50!white]

\tikzstyle{point}=[regular polygon,regular polygon sides=3,draw,scale=0.75,inner sep=-0.5pt,minimum width=9mm,fill=white,regular polygon rotate=180]
\tikzstyle{copoint}=[regular polygon,regular polygon sides=3,draw,scale=0.75,inner sep=-0.5pt,minimum width=9mm,fill=white]
\tikzstyle{dpoint}=[point,doubled]
\tikzstyle{dcopoint}=[copoint,doubled]

\tikzstyle{wide copoint}=[fill=white,draw,shape=isosceles triangle,shape border rotate=90,isosceles triangle stretches=true,inner sep=0pt,minimum width=1.5cm,minimum height=6.12mm]
\tikzstyle{wide point}=[fill=white,draw,shape=isosceles triangle,shape border rotate=-90,isosceles triangle stretches=true,inner sep=0pt,minimum width=1.5cm,minimum height=6.12mm,yshift=-0.0mm]
\tikzstyle{wide point plus}=[fill=white,draw,shape=isosceles triangle,shape border rotate=-90,isosceles triangle stretches=true,inner sep=0pt,minimum width=1.74cm,minimum height=7mm,yshift=-0.0mm]

\tikzstyle{wide dpoint}=[fill=white,doubled,draw,shape=isosceles triangle,shape border rotate=-90,isosceles triangle stretches=true,inner sep=0pt,minimum width=1.5cm,minimum height=6.12mm,yshift=-0.0mm]
\tikzstyle{wide dcopoint}=[fill=white,doubled,draw,shape=isosceles triangle,shape border rotate=90,isosceles triangle stretches=true,inner sep=0pt,minimum width=1.5cm,minimum height=6.12mm,yshift=-0.0mm]

\tikzstyle{tinypoint}=[regular polygon,regular polygon sides=3,draw,scale=0.55,inner sep=-0.15pt,minimum width=6mm,fill=white,regular polygon rotate=180]

\tikzstyle{white point}=[point]
\tikzstyle{white dpoint}=[dpoint]
\tikzstyle{green point}=[white point] 
\tikzstyle{white copoint}=[copoint]
\tikzstyle{gray point}=[point,fill=gray!40!white]
\tikzstyle{gray dpoint}=[gray point,doubled]
\tikzstyle{red point}=[gray point] 
\tikzstyle{gray copoint}=[copoint,fill=gray!40!white]
\tikzstyle{gray dcopoint}=[gray copoint,doubled]

\tikzstyle{white point guide}=[regular polygon,regular polygon sides=3,font=\scriptsize,draw,scale=0.65,inner sep=-0.5pt,minimum width=9mm,fill=white,regular polygon rotate=180]

\tikzstyle{black point}=[point,fill=black,font=\color{white}]
\tikzstyle{black copoint}=[copoint,fill=black,font=\color{white}]

\tikzstyle{tiny gray point}=[tinypoint,fill=gray!40!white]

\tikzstyle{diredge}=[->]
\tikzstyle{ddiredge}=[<->]
\tikzstyle{rdiredge}=[<-]
\tikzstyle{thickdiredge}=[->, very thick]
\tikzstyle{pointer edge}=[->,very thick,gray]
\tikzstyle{pointer edge part}=[very thick,gray]
\tikzstyle{dashed edge}=[dashed]
\tikzstyle{thick dashed edge}=[very thick,dashed]
\tikzstyle{thick gray dashed edge}=[thick dashed edge,gray!40]
\tikzstyle{thick map edge}=[very thick,|->]


\makeatletter

\providecommand{\boxshape}[3]{%
	\pgfdeclareshape{#1}{
		\inheritsavedanchors[from=rectangle]
		\inheritanchorborder[from=rectangle]
		\inheritanchor[from=rectangle]{center}
		\inheritanchor[from=rectangle]{north}
		\inheritanchor[from=rectangle]{south}
		\inheritanchor[from=rectangle]{west}
		\inheritanchor[from=rectangle]{east}
		
		\backgroundpath{
			\southwest \pgf@xa=\pgf@x \pgf@ya=\pgf@y
			\northeast \pgf@xb=\pgf@x \pgf@yb=\pgf@y
			
			\@tempdima=#2
			\@tempdimb=#3
			
			\pgfpathmoveto{\pgfpoint{\pgf@xa - 5pt + \@tempdima}{\pgf@ya}}
			\pgfpathlineto{\pgfpoint{\pgf@xa - 5pt - \@tempdima}{\pgf@yb}}
			\pgfpathlineto{\pgfpoint{\pgf@xb + 5pt + \@tempdimb}{\pgf@yb}}
			\pgfpathlineto{\pgfpoint{\pgf@xb + 5pt - \@tempdimb}{\pgf@ya}}
			\pgfpathclose
		}
	}
}

\boxshape{NEbox}{0pt}{5pt}
\boxshape{SEbox}{0pt}{-5pt}
\boxshape{NWbox}{5pt}{0pt}
\boxshape{SWbox}{-5pt}{0pt}
\boxshape{EBox}{-3pt}{3pt}
\boxshape{WBox}{3pt}{-3pt}
\makeatother

\tikzstyle{cloud}=[shape=cloud,draw,minimum width=1.5cm,minimum height=1.5cm]

\tikzstyle{map}=[draw,shape=NEbox,inner sep=2pt,minimum height=6mm,fill=white]
\tikzstyle{dashedmap}=[draw,dashed,shape=NEbox,inner sep=2pt,minimum height=6mm,fill=white]
\tikzstyle{mapdag}=[draw,shape=SEbox,inner sep=2pt,minimum height=6mm,fill=white]
\tikzstyle{mapadj}=[draw,shape=SEbox,inner sep=2pt,minimum height=6mm,fill=white]
\tikzstyle{maptrans}=[draw,shape=SWbox,inner sep=2pt,minimum height=6mm,fill=white]
\tikzstyle{mapconj}=[draw,shape=NWbox,inner sep=2pt,minimum height=6mm,fill=white]

\tikzstyle{langmap}=[draw,shape=NEbox,inner sep=2pt,minimum height=2.4mm,minimum width=3.2mm,fill=white]
\tikzstyle{langmaptrans}=[draw,shape=SWbox,inner sep=2pt,minimum height=2.4mm,minimum width=3.2mm,fill=white]

\tikzstyle{medium map}=[draw,shape=NEbox,inner sep=2pt,minimum height=6mm,fill=white,minimum width=7mm]
\tikzstyle{medium map dag}=[draw,shape=SEbox,inner sep=2pt,minimum height=6mm,fill=white,minimum width=7mm]
\tikzstyle{medium map adj}=[draw,shape=SEbox,inner sep=2pt,minimum height=6mm,fill=white,minimum width=7mm]
\tikzstyle{medium map trans}=[draw,shape=SWbox,inner sep=2pt,minimum height=6mm,fill=white,minimum width=7mm]
\tikzstyle{medium map conj}=[draw,shape=NWbox,inner sep=2pt,minimum height=6mm,fill=white,minimum width=7mm]
\tikzstyle{semilarge map}=[draw,shape=NEbox,inner sep=2pt,minimum height=6mm,fill=white,minimum width=9.5mm]
\tikzstyle{semilarge map trans}=[draw,shape=SWbox,inner sep=2pt,minimum height=6mm,fill=white,minimum width=9.5mm]
\tikzstyle{semilarge map adj}=[draw,shape=SEbox,inner sep=2pt,minimum height=6mm,fill=white,minimum width=9.5mm]
\tikzstyle{semilarge map dag}=[draw,shape=SEbox,inner sep=2pt,minimum height=6mm,fill=white,minimum width=9.5mm]
\tikzstyle{semilarge map conj}=[draw,shape=NWbox,inner sep=2pt,minimum height=6mm,fill=white,minimum width=9.5mm]
\tikzstyle{large map}=[draw,shape=NEbox,inner sep=2pt,minimum height=6mm,fill=white,minimum width=12mm]
\tikzstyle{large map conj}=[draw,shape=NWbox,inner sep=2pt,minimum height=6mm,fill=white,minimum width=12mm]
\tikzstyle{very large map}=[draw,shape=NEbox,inner sep=2pt,minimum height=6mm,fill=white,minimum width=17mm]

\tikzstyle{medium dmap}=[draw,doubled,shape=NEbox,inner sep=2pt,minimum height=6mm,fill=white,minimum width=7mm]
\tikzstyle{medium dmap dag}=[draw,doubled,shape=SEbox,inner sep=2pt,minimum height=6mm,fill=white,minimum width=7mm]
\tikzstyle{medium dmap adj}=[draw,doubled,shape=SEbox,inner sep=2pt,minimum height=6mm,fill=white,minimum width=7mm]
\tikzstyle{medium dmap trans}=[draw,doubled,shape=SWbox,inner sep=2pt,minimum height=6mm,fill=white,minimum width=7mm]
\tikzstyle{medium dmap conj}=[draw,doubled,shape=NWbox,inner sep=2pt,minimum height=6mm,fill=white,minimum width=7mm]
\tikzstyle{semilarge dmap}=[draw,doubled,shape=NEbox,inner sep=2pt,minimum height=6mm,fill=white,minimum width=9.5mm]
\tikzstyle{semilarge dmap trans}=[draw,doubled,shape=SWbox,inner sep=2pt,minimum height=6mm,fill=white,minimum width=9.5mm]
\tikzstyle{semilarge dmap adj}=[draw,doubled,shape=SEbox,inner sep=2pt,minimum height=6mm,fill=white,minimum width=9.5mm]
\tikzstyle{semilarge dmap dag}=[draw,doubled,shape=SEbox,inner sep=2pt,minimum height=6mm,fill=white,minimum width=9.5mm]
\tikzstyle{semilarge dmap conj}=[draw,doubled,shape=NWbox,inner sep=2pt,minimum height=6mm,fill=white,minimum width=9.5mm]
\tikzstyle{large dmap}=[draw,doubled,shape=NEbox,inner sep=2pt,minimum height=6mm,fill=white,minimum width=12mm]
\tikzstyle{large dmap conj}=[draw,doubled,shape=NWbox,inner sep=2pt,minimum height=6mm,fill=white,minimum width=12mm]
\tikzstyle{large dmap trans}=[draw,doubled,shape=SWbox,inner sep=2pt,minimum height=6mm,fill=white,minimum width=12mm]
\tikzstyle{large dmap adj}=[draw,doubled,shape=SEbox,inner sep=2pt,minimum height=6mm,fill=white,minimum width=12mm]
\tikzstyle{large dmap dag}=[draw,doubled,shape=SEbox,inner sep=2pt,minimum height=6mm,fill=white,minimum width=12mm]
\tikzstyle{very large dmap}=[draw,doubled,shape=NEbox,inner sep=2pt,minimum height=6mm,fill=white,minimum width=19.5mm]

\tikzstyle{muxbox}=[draw,shape=rectangle,minimum height=3mm,minimum width=3mm,fill=white]
\tikzstyle{dmuxbox}=[muxbox,doubled]

\tikzstyle{box}=[draw,shape=rectangle,inner sep=2pt,minimum height=6mm,minimum width=6mm,fill=white]
\tikzstyle{dbox}=[draw,doubled,shape=rectangle,inner sep=2pt,minimum height=6mm,minimum width=6mm,fill=white]
\tikzstyle{dmap}=[draw,doubled,shape=NEbox,inner sep=2pt,minimum height=6mm,fill=white]
\tikzstyle{dmapdag}=[draw,doubled,shape=SEbox,inner sep=2pt,minimum height=6mm,fill=white]
\tikzstyle{dmapadj}=[draw,doubled,shape=SEbox,inner sep=2pt,minimum height=6mm,fill=white]
\tikzstyle{dmaptrans}=[draw,doubled,shape=SWbox,inner sep=2pt,minimum height=6mm,fill=white]
\tikzstyle{dmapconj}=[draw,doubled,shape=NWbox,inner sep=2pt,minimum height=6mm,fill=white]

\tikzstyle{ddmap}=[draw,doubled,dashed,shape=NEbox,inner sep=2pt,minimum height=6mm,fill=white]
\tikzstyle{ddmapdag}=[draw,doubled,dashed,shape=SEbox,inner sep=2pt,minimum height=6mm,fill=white]
\tikzstyle{ddmapadj}=[draw,doubled,dashed,shape=SEbox,inner sep=2pt,minimum height=6mm,fill=white]
\tikzstyle{ddmaptrans}=[draw,doubled,dashed,shape=SWbox,inner sep=2pt,minimum height=6mm,fill=white]
\tikzstyle{ddmapconj}=[draw,doubled,dashed,shape=NWbox,inner sep=2pt,minimum height=6mm,fill=white]

\boxshape{sNEbox}{0pt}{3pt}
\boxshape{sSEbox}{0pt}{-3pt}
\boxshape{sNWbox}{3pt}{0pt}
\boxshape{sSWbox}{-3pt}{0pt}
\tikzstyle{smap}=[draw,shape=sNEbox,fill=white]
\tikzstyle{smapdag}=[draw,shape=sSEbox,fill=white]
\tikzstyle{smapadj}=[draw,shape=sSEbox,fill=white]
\tikzstyle{smaptrans}=[draw,shape=sSWbox,fill=white]
\tikzstyle{smapconj}=[draw,shape=sNWbox,fill=white]

\tikzstyle{dsmap}=[draw,dashed,shape=sNEbox,fill=white]
\tikzstyle{dsmapdag}=[draw,dashed,shape=sSEbox,fill=white]
\tikzstyle{dsmaptrans}=[draw,dashed,shape=sSWbox,fill=white]
\tikzstyle{dsmapconj}=[draw,dashed,shape=sNWbox,fill=white]

\boxshape{mNEbox}{0pt}{10pt}
\boxshape{mSEbox}{0pt}{-10pt}
\boxshape{mNWbox}{10pt}{0pt}
\boxshape{mSWbox}{-10pt}{0pt}
\tikzstyle{mmap}=[draw,shape=mNEbox]
\tikzstyle{mmapdag}=[draw,shape=mSEbox]
\tikzstyle{mmaptrans}=[draw,shape=mSWbox]
\tikzstyle{mmapconj}=[draw,shape=mNWbox]

\tikzstyle{mmapgray}=[draw,fill=gray!40!white,shape=mNEbox]
\tikzstyle{smapgray}=[draw,fill=gray!40!white,shape=sNEbox]

\makeatletter
\pgfdeclareshape{cornerpoint}{
\inheritsavedanchors[from=rectangle] 
\inheritanchorborder[from=rectangle]
\inheritanchor[from=rectangle]{center}
\inheritanchor[from=rectangle]{north}
\inheritanchor[from=rectangle]{south}
\inheritanchor[from=rectangle]{west}
\inheritanchor[from=rectangle]{east}
\backgroundpath{
\southwest \pgf@xa=\pgf@x \pgf@ya=\pgf@y
\northeast \pgf@xb=\pgf@x \pgf@yb=\pgf@y

\pgfmathsetmacro{\pgf@shorten@left}{\pgfkeysvalueof{/tikz/shorten left}}
\pgfmathsetmacro{\pgf@shorten@right}{\pgfkeysvalueof{/tikz/shorten right}}

\pgfpathmoveto{\pgfpoint{0.5 * (\pgf@xa + \pgf@xb)}{\pgf@ya - 5pt}}
\pgfpathlineto{\pgfpoint{\pgf@xa - 8pt + \pgf@shorten@left}{\pgf@yb - 1.5 * \pgf@shorten@left}}
\pgfpathlineto{\pgfpoint{\pgf@xa - 8pt + \pgf@shorten@left}{\pgf@yb}}
\pgfpathlineto{\pgfpoint{\pgf@xb + 8pt - \pgf@shorten@right}{\pgf@yb}}
\pgfpathlineto{\pgfpoint{\pgf@xb + 8pt - \pgf@shorten@right}{\pgf@yb - 1.5 * \pgf@shorten@right}}
\pgfpathclose
}
}

\pgfdeclareshape{cornercopoint}{
\inheritsavedanchors[from=rectangle] 
\inheritanchorborder[from=rectangle]
\inheritanchor[from=rectangle]{center}
\inheritanchor[from=rectangle]{north}
\inheritanchor[from=rectangle]{south}
\inheritanchor[from=rectangle]{west}
\inheritanchor[from=rectangle]{east}
\backgroundpath{
\southwest \pgf@xa=\pgf@x \pgf@ya=\pgf@y
\northeast \pgf@xb=\pgf@x \pgf@yb=\pgf@y

\pgfmathsetmacro{\pgf@shorten@left}{\pgfkeysvalueof{/tikz/shorten left}}
\pgfmathsetmacro{\pgf@shorten@right}{\pgfkeysvalueof{/tikz/shorten right}}

\pgfpathmoveto{\pgfpoint{0.5 * (\pgf@xa + \pgf@xb)}{\pgf@yb + 5pt}}
\pgfpathlineto{\pgfpoint{\pgf@xa - 8pt + \pgf@shorten@left}{\pgf@ya + 1.5 * \pgf@shorten@left}}
\pgfpathlineto{\pgfpoint{\pgf@xa - 8pt + \pgf@shorten@left}{\pgf@ya}}
\pgfpathlineto{\pgfpoint{\pgf@xb + 8pt - \pgf@shorten@right}{\pgf@ya}}
\pgfpathlineto{\pgfpoint{\pgf@xb + 8pt - \pgf@shorten@right}{\pgf@ya + 1.5 * \pgf@shorten@right}}
\pgfpathclose
}
}

\pgfdeclareshape{langpoint}{
\inheritsavedanchors[from=rectangle] 
\inheritanchorborder[from=rectangle]
\inheritanchor[from=rectangle]{center}
\inheritanchor[from=rectangle]{north}
\inheritanchor[from=rectangle]{south}
\inheritanchor[from=rectangle]{west}
\inheritanchor[from=rectangle]{east}
\backgroundpath{
\southwest \pgf@xa=\pgf@x \pgf@ya=\pgf@y
\northeast \pgf@xb=\pgf@x \pgf@yb=\pgf@y

\pgfmathsetmacro{\pgf@shorten@left}{\pgfkeysvalueof{/tikz/shorten left}}
\pgfmathsetmacro{\pgf@shorten@right}{\pgfkeysvalueof{/tikz/shorten right}}

\pgfpathmoveto{\pgfpoint{0.5 * (\pgf@xa + \pgf@xb)}{\pgf@ya - 2pt}}
\pgfpathlineto{\pgfpoint{\pgf@xa - 8pt}{\pgf@yb - 3 * \pgf@shorten@left + 5pt}} 
\pgfpathlineto{\pgfpoint{\pgf@xa - 8pt}{\pgf@yb -1pt}}
\pgfpathlineto{\pgfpoint{\pgf@xb + 8pt}{\pgf@yb -1pt}}
\pgfpathlineto{\pgfpoint{\pgf@xb + 8pt}{\pgf@yb - 3 * \pgf@shorten@left + 5pt}}
\pgfpathclose
}
}

\pgfdeclareshape{langcopoint}{
\inheritsavedanchors[from=rectangle] 
\inheritanchorborder[from=rectangle]
\inheritanchor[from=rectangle]{center}
\inheritanchor[from=rectangle]{north}
\inheritanchor[from=rectangle]{south}
\inheritanchor[from=rectangle]{west}
\inheritanchor[from=rectangle]{east}
\backgroundpath{
\southwest \pgf@xa=\pgf@x \pgf@ya=\pgf@y
\northeast \pgf@xb=\pgf@x \pgf@yb=\pgf@y

\pgfmathsetmacro{\pgf@shorten@left}{\pgfkeysvalueof{/tikz/shorten left}}
\pgfmathsetmacro{\pgf@shorten@right}{\pgfkeysvalueof{/tikz/shorten right}}

\pgfpathmoveto{\pgfpoint{0.5 * (\pgf@xa + \pgf@xb)}{\pgf@yb +0pt}}
\pgfpathlineto{\pgfpoint{\pgf@xa - 8pt}{\pgf@ya + 3 * \pgf@shorten@left - 5pt}} 
\pgfpathlineto{\pgfpoint{\pgf@xa - 8pt}{\pgf@ya + 1pt}}
\pgfpathlineto{\pgfpoint{\pgf@xb + 8pt}{\pgf@ya + 1pt}}
\pgfpathlineto{\pgfpoint{\pgf@xb + 8pt}{\pgf@ya + 3 * \pgf@shorten@left - 5pt}}
\pgfpathclose
}
}

\pgfdeclareshape{langrect}{
\inheritsavedanchors[from=rectangle] 
\inheritanchorborder[from=rectangle]
\inheritanchor[from=rectangle]{center}
\inheritanchor[from=rectangle]{north}
\inheritanchor[from=rectangle]{south}
\inheritanchor[from=rectangle]{west}
\inheritanchor[from=rectangle]{east}
\backgroundpath{
\southwest \pgf@xa=\pgf@x \pgf@ya=\pgf@y
\northeast \pgf@xb=\pgf@x \pgf@yb=\pgf@y

\pgfmathsetmacro{\pgf@shorten@left}{\pgfkeysvalueof{/tikz/shorten left}}
\pgfmathsetmacro{\pgf@shorten@right}{\pgfkeysvalueof{/tikz/shorten right}}

\pgfpathmoveto{\pgfpoint{\pgf@xa - 8pt}{\pgf@ya + 3 * \pgf@shorten@left - 5pt}} 
\pgfpathlineto{\pgfpoint{\pgf@xa - 8pt}{\pgf@ya + 1pt}}
\pgfpathlineto{\pgfpoint{\pgf@xb + 8pt}{\pgf@ya + 1pt}}
\pgfpathlineto{\pgfpoint{\pgf@xb + 8pt}{\pgf@ya + 3 * \pgf@shorten@left - 5pt}}
\pgfpathclose
}
}

\makeatother

\pgfkeyssetvalue{/tikz/shorten left}{0pt}
\pgfkeyssetvalue{/tikz/shorten right}{0pt}

\tikzstyle{kpoint common}=[draw,fill=white,inner sep=1pt,minimum height=4mm]

\tikzstyle{langstate}=[shape=langcopoint,shorten left=5pt,kpoint common,font=\footnotesize]
\tikzstyle{langeffect}=[shape=langpoint,shorten left=5pt,kpoint common,font=\footnotesize]
\tikzstyle{langstatedash}=[shape=langcopoint,dashed, shorten left=5pt,kpoint common,font=\footnotesize]
\tikzstyle{langeffectdash}=[shape=langpoint,dashed, shorten left=5pt,kpoint common,font=\footnotesize]
\tikzstyle{langbox}=[shape=langrect,shorten left=5pt,kpoint common,font=\footnotesize] 

\tikzstyle{kpoint}=[shape=cornerpoint,shorten left=5pt,kpoint common]
\tikzstyle{kpoint adjoint}=[shape=cornercopoint,shorten left=5pt,kpoint common]

\tikzstyle{kpoint conjugate}=[shape=cornerpoint,shorten right=5pt,kpoint common]
\tikzstyle{kpoint transpose}=[shape=cornercopoint,shorten right=5pt,kpoint common]
\tikzstyle{kpoint symm}=[shape=cornerpoint,shorten left=5pt,shorten right=5pt,kpoint common]

\tikzstyle{black kpoint}=[shape=cornerpoint,shorten left=5pt,kpoint common,fill=black,font=\color{white}]
\tikzstyle{black kpoint adjoint}=[shape=cornercopoint,shorten left=5pt,kpoint common,fill=black,font=\color{white}]
\tikzstyle{black kpointadj}=[shape=cornercopoint,shorten left=5pt,kpoint common,fill=black,font=\color{white}]

\tikzstyle{black dkpoint}=[shape=cornerpoint,shorten left=5pt,kpoint common,fill=black, doubled,font=\color{white}]
\tikzstyle{black dkpoint adjoint}=[shape=cornercopoint,shorten left=5pt,kpoint common,fill=black, doubled,font=\color{white}]
\tikzstyle{black dkpointadj}=[shape=cornercopoint,shorten left=5pt,kpoint common,fill=black, doubled,font=\color{white}]

\tikzstyle{kpointdag}=[kpoint adjoint]
\tikzstyle{kpointadj}=[kpoint adjoint]
\tikzstyle{kpointconj}=[kpoint conjugate]
\tikzstyle{kpointtrans}=[kpoint transpose]

\tikzstyle{big kpoint}=[kpoint, minimum width=1.2 cm, minimum height=8mm, inner sep=4pt, text depth=3mm]

\tikzstyle{wide kpoint}=[kpoint, minimum width=1 cm, inner sep=2pt]
\tikzstyle{wide kpointdag}=[kpointdag, minimum width=1 cm, inner sep=2pt]
\tikzstyle{wide kpointconj}=[kpointconj, minimum width=1 cm, inner sep=2pt]
\tikzstyle{wide kpointtrans}=[kpointtrans, minimum width=1 cm, inner sep=2pt]

\tikzstyle{gray kpoint}=[kpoint,fill=gray!50!white]
\tikzstyle{gray kpointdag}=[kpointdag,fill=gray!50!white]
\tikzstyle{gray kpointadj}=[kpointadj,fill=gray!50!white]
\tikzstyle{gray kpointconj}=[kpointconj,fill=gray!50!white]
\tikzstyle{gray kpointtrans}=[kpointtrans,fill=gray!50!white]

\tikzstyle{gray dkpoint}=[kpoint,fill=gray!50!white,doubled]
\tikzstyle{gray dkpointdag}=[kpointdag,fill=gray!50!white,doubled]
\tikzstyle{gray dkpointadj}=[kpointadj,fill=gray!50!white,doubled]
\tikzstyle{gray dkpointconj}=[kpointconj,fill=gray!50!white,doubled]
\tikzstyle{gray dkpointtrans}=[kpointtrans,fill=gray!50!white,doubled]

\tikzstyle{white label}=[draw,fill=white,rectangle,inner sep=0.7 mm]
\tikzstyle{gray label}=[draw,fill=gray!50!white,rectangle,inner sep=0.7 mm]
\tikzstyle{black label}=[draw,fill=black,rectangle,inner sep=0.7 mm]

\tikzstyle{dkpoint}=[kpoint,doubled]
\tikzstyle{wide dkpoint}=[wide kpoint,doubled]
\tikzstyle{dkpointdag}=[kpoint adjoint,doubled]
\tikzstyle{wide dkpointdag}=[wide kpointdag,doubled]
\tikzstyle{dkcopoint}=[kpoint adjoint,doubled]
\tikzstyle{dkpointadj}=[kpoint adjoint,doubled]
\tikzstyle{dkpointconj}=[kpoint conjugate,doubled]
\tikzstyle{dkpointtrans}=[kpoint transpose,doubled]

\tikzstyle{kscalar}=[kpoint common, shape=EBox, inner xsep=-1pt, inner ysep=3pt,font=\small]
\tikzstyle{kscalarconj}=[kpoint common, shape=WBox, inner xsep=-1pt, inner ysep=3pt,font=\small]


 \tikzstyle{upground}=[circuit ee IEC,ground,rotate=90,scale=2.5]
 \tikzstyle{downground}=[circuit ee IEC,ground,rotate=-90,scale=2.5]
 \tikzstyle{bigground}=[regular polygon,regular polygon sides=3,draw=gray,scale=0.50,inner sep=-0.5pt,minimum width=10mm,fill=gray]


\tikzstyle{arrs}=[-latex,font=\small,auto]
\tikzstyle{arrow plain}=[arrs]
\tikzstyle{arrow dashed}=[dashed,arrs]
\tikzstyle{arrow bold}=[very thick,arrs]
\tikzstyle{arrow hide}=[draw=white!0,-]
\tikzstyle{arrow reverse}=[latex-]
\tikzstyle{cdnode}=[]



\newcommand{\bigcounit}[1]{%
\,\begin{tikzpicture}[dotpic,scale=2,yshift=-1mm]
\node [#1] (a) at (0,0.25) {}; 
\draw (0,-0.2)--(a);
\end{tikzpicture}\,}
\newcommand{\bigunit}[1]{%
\,\begin{tikzpicture}[dotpic,scale=2,yshift=1.5mm]
\node [#1] (a) at (0,-0.25) {}; 
\draw (a)--(0,0.2);
\end{tikzpicture}\,}
\newcommand{\bigcomult}[1]{%
\,\begin{tikzpicture}[dotpic,scale=2,yshift=0.5mm]
	\node [#1] (a) {};
	\draw (-90:0.55)--(a);
	\draw (a) -- (45:0.6);
	\draw (a) -- (135:0.6);
\end{tikzpicture}\,}
\newcommand{\bigmult}[1]{%
\,\begin{tikzpicture}[dotpic,scale=2]
	\node [#1] (a) {};
	\draw (a) -- (90:0.55);
	\draw (a) (-45:0.6) -- (a);
	\draw (a) (-135:0.6) -- (a);
\end{tikzpicture}\,}

\newcommand{\dotcounit}[1]{%
\,\begin{tikzpicture}[dotpic,yshift=-1mm]
\node [#1] (a) at (0,0.35) {}; 
\draw (0,-0.3)--(a);
\end{tikzpicture}\,}
\newcommand{\dotunit}[1]{%
\,\begin{tikzpicture}[dotpic,yshift=1.5mm]
\node [#1] (a) at (0,-0.35) {}; 
\draw (a)--(0,0.3);
\end{tikzpicture}\,}
\newcommand{\dotcomult}[1]{%
\,\begin{tikzpicture}[dotpic,yshift=0.5mm]
	\node [#1] (a) {};
	\draw (-90:0.55)--(a);
	\draw (a) -- (45:0.6);
	\draw (a) -- (135:0.6);
\end{tikzpicture}\,}
\newcommand{\dotmult}[1]{%
\,\begin{tikzpicture}[dotpic]
	\node [#1] (a) {};
	\draw (a) -- (90:0.55);
	\draw (a) (-45:0.6) -- (a);
	\draw (a) (-135:0.6) -- (a);
\end{tikzpicture}\,}

\newcommand{\ddotmult}[1]{%
\,\begin{tikzpicture}[dotpic]
	\node [#1] (a) {};
	\draw [boldedge] (a) -- (90:0.55);
	\draw [boldedge] (a) (-45:0.6) -- (a);
	\draw [boldedge] (a) (-135:0.6) -- (a);
\end{tikzpicture}\,}



\newcommand{\dotidualiser}[1]{%
\begin{tikzpicture}[dotpic,yshift=1.5mm]
	\node [#1] (a) {};
	\draw [medium diredge] (a) to (-90:0.35);
	\draw [medium diredge] (a) to (90:0.35);
\end{tikzpicture}}
\newcommand{\dotdualiser}[1]{%
\begin{tikzpicture}[dotpic,yshift=1.5mm]
	\node [#1] (a) {};
	\draw [medium diredge] (-90:0.35) to (a);
	\draw [medium diredge] (90:0.35) to (a);
\end{tikzpicture}}
\newcommand{\dottickunit}[1]{%
\begin{tikzpicture}[dotpic,yshift=-1mm]
\node [#1] (a) at (0,0.35) {}; 
\draw [postaction=decorate,
       decoration={markings, mark=at position 0.3 with \edgetick},
       decoration={markings, mark=at position 0.85 with \edgearrow}] (a)--(0,-0.25);
\end{tikzpicture}}
\newcommand{\dottickcounit}[1]{%
\begin{tikzpicture}[dotpic,yshift=1mm]
\node [#1] (a) at (0,-0.35) {}; 
\draw [postaction=decorate,
       decoration={markings, mark=at position 0.8 with \edgetick},
       decoration={markings, mark=at position 0.45 with \edgearrow}] (0,0.25) -- (a);
\end{tikzpicture}}
\newcommand{\dotonly}[1]{%
\,\begin{tikzpicture}[dotpic]
\node [#1] (a) at (0,0) {};
\end{tikzpicture}\,}
\newcommand{\smalldotonly}[1]{%
\,\begin{tikzpicture}[dotpic,yshift=-0.15mm]
\node [#1] (a) at (0,0) {};
\end{tikzpicture}\,}
\newcommand{\dotthreestate}[1]{%
\,\begin{tikzpicture}[dotpic,yshift=2.5mm]
	\node [#1] (a) at (0,0) {};
	\draw (a) -- (0,-0.6);
	\draw [bend right] (a) to (-0.4,-0.6) (0.4,-0.6) to (a);
\end{tikzpicture}\,}
\newcommand{\dotcap}[1]{%
\,\begin{tikzpicture}[dotpic,yshift=2.5mm]
	\node [#1] (a) at (0,0) {};
	\draw [bend right,medium diredge] (a) to (-0.4,-0.6);
	\draw [bend left,medium diredge] (a) to (0.4,-0.6);
\end{tikzpicture}\,}
\newcommand{\dotcup}[1]{%
\,\begin{tikzpicture}[dotpic,yshift=4mm]
	\node [#1] (a) at (0,-0.6) {};
	\draw [bend right,medium diredge] (-0.4,0) to (a);
	\draw [bend left,medium diredge] (0.4,0) to (a);
\end{tikzpicture}\,}

\newcommand{\tick}{%
\,\,\begin{tikzpicture}[dotpic]
	\node [style=none] (a) at (0,0.35) {};
	\node [style=none] (b) at (0,-0.35) {};
	\draw [dirtickedge] (a) -- (b);
\end{tikzpicture}\,\,}

\newcommand{\lolli}{%
\,\begin{tikzpicture}[dotpic,yshift=-1mm]
	\path [use as bounding box] (-0.25,-0.25) rectangle (0.25,0.5);
	\node [style=dot] (a) at (0, 0.15) {};
	\node [style=none] (b) at (0, -0.25) {};
	\draw [medium diredge] (a) to (b.center);
	\draw [diredge, out=45, looseness=1.00, in=135, loop] (a) to ();
\end{tikzpicture}\,}

\newcommand{\cololli}{%
\,\begin{tikzpicture}[dotpic]
	\path [use as bounding box] (-0.25,-0.5) rectangle (0.25,0.5);
	\node [style=none] (a) at (0, 0.5) {};
	\node [style=dot] (b) at (0, 0) {};
	\draw [diredge, in=-45, looseness=2.00, out=-135, loop] (b) to ();
	\draw [medium diredge] (a.center) to (b);
\end{tikzpicture}\,}

\newcommand{\unit}{\dotunit{dot}}
\newcommand{\counit}{\dotcounit{dot}}
\newcommand{\mult}{\dotmult{dot}}
\newcommand{\comult}{\dotcomult{dot}}

\newcommand{\blackdot}{\dotonly{black dot}\xspace}
\newcommand{\smallblackdot}{\smalldotonly{smalldot}\xspace}
\newcommand{\blackunit}{\dotunit{black dot}\xspace}
\newcommand{\blackcap}{\dotcap{black dot}\xspace}
\newcommand{\blackcup}{\dotcup{black dot}\xspace}

\newcommand{\tickunit}{\dottickunit{dot}}
\newcommand{\tickcounit}{\dottickcounit{dot}}
\newcommand{\dualiser}{\dotdualiser{dot}}
\newcommand{\idualiser}{\dotidualiser{dot}}
\newcommand{\threestate}{\dotthreestate{dot}}

\newcommand{\blackobs}{\ensuremath{\mathcal O_{\!\smallblackdot}}\xspace}

\newcommand{\whitedot}{\dotonly{white dot}\xspace}
\newcommand{\smallwhitedot}{\smalldotonly{small white dot}\xspace}
\newcommand{\whiteunit}{\dotunit{white dot}\xspace}
\newcommand{\whitecounit}{\dotcounit{white dot}\xspace}
\newcommand{\whitemult}{\dotmult{white dot}\xspace}
\newcommand{\whitecomult}{\dotcomult{white dot}\xspace}
\newcommand{\whitetickunit}{\dottickunit{white dot}\xspace}
\newcommand{\whitetickcounit}{\dottickcounit{white dot}\xspace}
\newcommand{\whitecap}{\dotcap{white dot}\xspace}
\newcommand{\whitecup}{\dotcup{white dot}\xspace}

\newcommand{\greendot}{\dotonly{green dot}}
\newcommand{\greenunit}{\dotunit{green dot}}
\newcommand{\greencounit}{\dotcounit{green dot}}
\newcommand{\greenmult}{\dotmult{green dot}}
\newcommand{\greencomult}{\dotcomult{green dot}}
\newcommand{\greentickunit}{\dottickunit{green dot}}
\newcommand{\greentickcounit}{\dottickcounit{green dot}}
\newcommand{\greencap}{\dotcap{green dot}}
\newcommand{\greencup}{\dotcup{green dot}}

\newcommand{\whiteobs}{\ensuremath{\mathcal O_{\!\smallwhitedot}}\xspace}

\newcommand{\altwhitedot}{\dotonly{alt white dot}}
\newcommand{\altwhiteunit}{\dotunit{alt white dot}}
\newcommand{\altwhitecounit}{\dotcounit{alt white dot}}
\newcommand{\altwhitemult}{\dotmult{alt white dot}}
\newcommand{\altwhitecomult}{\dotcomult{alt white dot}}
\newcommand{\altwhitetickunit}{\dottickunit{alt white dot}}
\newcommand{\altwhitetickcounit}{\dottickcounit{alt white dot}}
\newcommand{\altwhitecap}{\dotcap{alt white dot}}
\newcommand{\altwhitecup}{\dotcup{alt white dot}}

\newcommand{\graydot}{\dotonly{gray dot}\xspace}
\newcommand{\smallgraydot}{\smalldotonly{small gray dot}\xspace}
\newcommand{\grayunit}{\dotunit{gray dot}\xspace}
\newcommand{\graycounit}{\dotcounit{gray dot}\xspace}
\newcommand{\graymult}{\dotmult{gray dot}\xspace}
\newcommand{\dgraymult}{\ddotmult{gray ddot}\xspace}
\newcommand{\graycomult}{\dotcomult{gray dot}\xspace}
\newcommand{\graytickunit}{\dottickunit{gray dot}\xspace}
\newcommand{\graytickcounit}{\dottickcounit{gray dot}\xspace}
\newcommand{\graycap}{\dotcap{gray dot}\xspace}
\newcommand{\graycup}{\dotcup{gray dot}\xspace}

\newcommand{\grayobs}{\ensuremath{\mathcal O_{\!\smallgraydot}}\xspace}

\newcommand{\blacktranspose}{\ensuremath{{\,\blackdot\!\textrm{\rm\,T}}}}
\newcommand{\whitetranspose}{\ensuremath{{\!\!\altwhitedot\!\!}}}
\newcommand{\graytranspose}{\ensuremath{{\,\graydot\!\textrm{\rm\,T}}}}

\newcommand{\whiteconjugate}{\ensuremath{{\!\!\altwhitedot\!\!}}}

\newcommand{\circl}{\begin{tikzpicture}[dotpic]
		\node [style=none] (a) at (-0.25, 0.25) {};
		\node [style=none] (b) at (0.25, 0.25) {};
		\node [style=none] (c) at (-0.25, -0.25) {};
		\node [style=none] (d) at (0.25, -0.25) {};
		\draw [in=45, out=135] (b.center) to (a.center);
		\draw [in=135, out=225] (a.center) to (c.center);
		\draw [in=225, out=-45] (c.center) to (d.center);
		\draw [style=diredge, in=-45, out=45] (d.center) to (b.center);
\end{tikzpicture}}

\definecolor{hexcolor0xa9a9a9}{rgb}{0.663,0.663,0.663} 
\tikzstyle{GrayLine}=[dashed,draw=hexcolor0xa9a9a9] 
\tikzstyle{gray}=[dashed,draw=hexcolor0xa9a9a9]

\def\bR{\begin{color}{red}}  
\def\bB{\begin{color}{blue}}
\def\bM{\begin{color}{magenta}}  
\def\bC{\begin{color}{cyan}}
\def\bW{\begin{color}{white}}
\def\bBl{\begin{color}{black}}
\def\bG{\begin{color}{green}}
\def\bY{\begin{color}{yellow}}
\def\e{\end{color}\xspace}
\newcommand{\bit}{\begin{itemize}}
\newcommand{\eit}{\end{itemize}\par\noindent}
\newcommand{\ben}{\begin{enumerate}}
\newcommand{\een}{\end{enumerate}\par\noindent}
\newcommand{\beq}{\begin{equation}}
\newcommand{\eeq}{\end{equation}\par\noindent}
\newcommand{\beqa}{\begin{eqnarray*}}
\newcommand{\eeqa}{\end{eqnarray*}\par\noindent}
\newcommand{\beqn}{\begin{eqnarray}}
\newcommand{\eeqn}{\end{eqnarray}\par\noindent}

\newcommand{\TODOb}[1]{\marginpar{\scriptsize\bR \textbf{TODO:} #1\e}}
\newcommand{\TODOj}[1]{\marginpar{\scriptsize\bG \textbf{TODO:} #1\e}}
\newcommand{\TODOv}[1]{\marginpar{\scriptsize\bB \textbf{TODO:} #1\e}}

\newcommand{\COMMb}[1]{\marginpar{\scriptsize\bM \textbf{COMM:} #1\e}}
\newcommand{\COMMj}[1]{\marginpar{\scriptsize\bG \textbf{COMM:} #1\e}}
\newcommand{\COMMv}[1]{\marginpar{\scriptsize\bB \textbf{COMM:} #1\e}}

\newcommand{\TODO}[1]{{\color{red}\noindent\textbf{#1}}}

\newcommand{\gendiagram}[1]{
	\begin{tikzpicture}
		\begin{pgfonlayer}{nodelayer}
			\node [style=none] (0) at (-0.75, 1) {};
			\node [style=none] (1) at (0.75, 1) {};
			\node [style=none] (2) at (-0.75, -1) {};
			\node [style=none] (3) at (0.75, -1) {};
			\node [style=none] (4) at (0, 0.75) {$\ldots$};
			\node [style=semilarge box] (5) at (0, -0) { #1 };
			\node [style=none] (6) at (0, -0.75) {$\ldots$};
		\end{pgfonlayer}
		\begin{pgfonlayer}{edgelayer}
			\draw (0.center) to (2.center);
			\draw (1.center) to (3.center);
		\end{pgfonlayer}
	\end{tikzpicture}
}

\newcommand{\hadastate}[1]{\,\tikz{\node[style=hadamard] (x) {$#1$};\draw(x)--(0,0.75);}\,}
\newcommand{\hadadot}{\dotonly{small hadamard}\xspace}
\newcommand{\hadaunit}{\dotunit{small hadamard}}
\newcommand{\hadacounit}{\dotcounit{small hadamard}}
\newcommand{\hadamult}{\dotmult{small hadamard}}
\newcommand{\hadacomult}{\dotcomult{small hadamard}}
\newcommand{\hadacap}{\dotcap{small hadamard}}
\newcommand{\hadacup}{\dotcup{small hadamard}}
\newcommand{\grayphase}[1]{\phase[style=gray phase dot]{#1}} 

\chapter{\label{chapintro}Introduction}

\epigraph{`... the aim of science is not things themselves... but the relations between things; outside those relations there is no reality knowable.'}  
{Henri Poincar{\'e},\\ Science and Hypothesis~\cite[p.~xxiv]{poincare1952science}}

\subsection*{Relationalism and process ontology}

Mainstream Western science subscribes to the metaphysical doctrine that reality is fundamentally a collection of static objects whose dynamic aspects are ontologically derivative~\cite{sep-process-philosophy}. The origins of this worldview can be traced back to the pre-Socratics, especially Democritus and Parmenides. 

According to Democritus, things are \textit{essentially} clusters of indivisible, eternal particles called `atoms'~\cite{sep-democritus}. Democritus's doctrine is still reflected today in different forms of reductionism. Some of its examples are reducing physical phenomena to elementary particles; dividing biological systems into organs, tissues, and cells; and describing mathematical structures using set-theoretic foundations.

Democritus's atomism was a response to---or, more precisely, a refinement of---Parmenides's teaching~\cite{sep-democritus} according to which reality is static; any change or motion in it is illusory~\cite{sep-parmenides}. Democritus taught that the properties of things are determined by how their constituents (\emph{i.e.} atoms) are arranged. While the atoms are unchanging, their spatial rearrangements result in changing the appearance and properties of things. In this sense, all change is essentially the rearrangement of atoms in space. At any moment, a particular atomic arrangement constitutes the \textit{now}.
Stacking these \textit{now}s in a sequence yields (the illusion of) temporal change. 

It is worth noting that theories in physics more or less follow the aforementioned conceptual paradigm. First, they describe the state of a physical system at a particular instant; then, they track its time-evolution. The former is \textit{kinematics}, and the latter is \textit{dynamics}. A \textit{state} is a particular arrangement of all the constituent entities that make the system. A \textit{dynamical law}\footnote{Quantum mechanics does not strictly fit into this paradigm precisely because it has two dynamical laws. This is usually called the \textit{measurement problem} of quantum mechanics.} characterises how the states are stacked up in time.

According to Heraclitus, a contemporary of Parmenides, everything in nature undergoes a constant flux. Reality is essentially composed of processes instead of static things~\cite{sep-heraclitus}. This is in stark contrast to the stasis of Parmenides. Heraclitus is considered the founding father of process ontology in the Western philosophical tradition.

In modern philosophy, Leibniz's relational approach presents a perspective in which the relations between things in space and time are considered more fundamental than the spatio-temporal properties of things~\cite{sep-leibniz-physics}. Heraclitus's process ontology and Leibniz's relationalism inspired Whitehead's research programme~\cite{rescher1996process}, the aim of which was to describe all natural phenomena in terms of processes and their relations, rather than objects or entities that are independent of one another~\cite{sep-whitehead}. Schr\"{o}dinger refined this approach by pointing out that some relationships give rise to novel phenomena that cannot be strictly reduced to the properties of relata~\cite{schrodinger1935discussion}. 

For over 2000 years, developments in science have been underpinned by a worldview that is Parmenidean and Democritean in spirit. Even though this approach has been remarkably successful, quantum physics has exposed its limitations. In particular, it has long been well known that it is hard to reconcile the following two aspects of quantum theory with the prevailing outlook: (1) the existence of two distinct dynamical laws that depend on the presence or absence of interaction (\emph{i.e.} the measurement problem); and (2) the impossibility to reduce all phenomena to the properties of individual entities (\emph{i.e.} entanglement). Merely philosophical problems for around a century, the advent of quantum computing and quantum technologies has breathed new life into these issues. Therefore, the time is ripe to give a fair chance to a process-relational view of reality as advocated by Whitehead and Schr\"{o}dinger; and see how far it can take us. The mathematical foundation of this research area is provided by a relatively new branch of mathematics called applied category theory.

Applied category theory provides effective mathematical tools to model processes and how they compose~\cite{fong2019invitation}. Symmetric monoidal categories, which involve series and parallel composition, are especially relevant in the context of process-relational science~\cite{yanofsky2024monoidal}, being particularly well-suited for describing the composition of processes in space and time. Very aptly called process theories, they admit string diagrams~\cite{penrose1971applications, joyal1991geometry, SelingerSurvey}, a syntax that is not only visually intuitive but also mathematically rigorous.

\subsection*{What are string diagrams?}

String diagrams are two-dimensional graphical representations of morphisms (processes) in symmetric monoidal categories~\cite{SelingerSurvey, piedeleu2023introduction}. They are a rigorous, intuitive and completely diagrammatic alternative to symbolic algebra that capture the compositionality of processes (denoted by boxes) between systems (denoted by wires). They are a very expressive language---flexible and versatile enough to be used at different levels of abstraction. 

As a syntax, two-dimensional string diagrams incorporate the topology of one-dimensional symbolic expressions, thereby turning abstract and usually complex algebraic transformations into simple, fun, and intuitive diagram rewrite rules. An example diagrammatic calculation of `sliding boxes' is given below.
\[
\input{interpretations/boxb10.tikz}
\] 
The diagrams in this chapter must be read from bottom to top and left to right. The algebraic expression corresponding to each diagram is written below it. 

When moving from a symbolic algebraic syntax to string diagrams, some essential and non-trivial equations between algebraic expressions become redundant. In other words, these equations come for free in the string-diagrammatic universe and, therefore, there is no need to explicitly impose them. For instance, associativity, $(h \circ g) \circ f = h \circ (g \circ f)$ and $(f \otimes g) \otimes h = f \otimes (g \otimes h)$, is implicit in string diagrams. The algebraic expressions on both sides of each equation map to the same diagram, resulting in a tautology. 
\[
\input{interpretations/boxb5.tikz} \quad \quad \quad \input{interpretations/boxb4.tikz}
\]
Another equation that becomes redundant in string diagrams is the \textit{(middle-two-)interchange law}: $(h \otimes i) \circ (f \otimes g) = (h \circ f) \otimes (i \circ g)$. Both left- and right-hand sides map to the following diagram.
\[
\input{interpretations/boxb6.tikz}
\]

In short, string diagrams not only provide simpler and more intuitive counterparts to algebraic calculations, but also phase out some algebraic equations altogether. 

\subsection*{String diagrams for the sciences}

A precursor to string diagrams appeared in the work of Richard Feynman in 1949~\cite{feynman1949space, feynman1949theory}. String-diagrammatic notation was invented by Roger Penrose~\cite{penrose1971applications} in order to keep track of indices of tensors. G\"{u}nter Hotz independently invented string diagrams in a categorical setting to model series and parallel composition of digital circuits~\cite{hotz1965algebraisierung}. Hotz's thesis~\cite{hotz1965algebraisierung} appeared earlier than Penrose's paper~\cite{penrose1971applications}, but Penrose was reportedly already using string diagrams during his graduate studies in the late 1950s~\cite{pavlovic2023programs}. The connection between categories and string diagrams was formalised by Andr{\'e} Joyal and Ross Street in the 1990s~\cite{joyal1991geometry}. Since then, they have been applied in a wide range of scientific domains.

Particularly in the last two decades, string diagrams have been used as a syntactic foundation for wide-ranging research domains of science, engineering, and technology~\cite{piedeleu2023introduction}. For example: quantum physics~\cite{coecke2006kindergarten, coecke2022kindergarden}, quantum computing~\cite{CKbook, CGbook, coecke2023basic}, electrical circuit theory~\cite{boisseau2021string}, digital circuits~\cite{ghica2016categorical, ghica2022compositional}, photonics~\cite{de2022quantum, clement2022lov}, linguistics~\cite{clark2008compositional, coecke2021mathematics, wang2023distilling}, control systems~\cite{baez2014categories, erbele2016categories}, causal models~\cite{lorenz2023causal}, probability theory~\cite{cho2019disintegration, fritz2021finetti}, machine learning~\cite{shiebler2021category, koziell2024hybrid}, game theory~\cite{hedges2015string, hedges2016compositionality, ghani2018compositional} etc. Moreover, the string diagrammatic approach to quantum theory, dubbed \textit{Quantum Picturalism}~\cite{ContPhys}, has found success in education both at university~\cite{CKbook} and high school~\cite{CGbook} levels. In fact, the results of a pedagogical study~\cite{dundar2023quantum} demonstrate that the Quantum Picturalism approach is particularly effective in teaching university-level quantum theory to high school students~\cite{dundarcoecke2025makingquantumworldaccessible}.   

To sum up, a process-relational outlook goes hand in hand with a mathematical language that is applicable to wide-ranging scientific areas. Additionally, by replacing symbolic reasoning with diagrammatic reasoning, this language makes science much more intuitive and accessible. 

\subsection*{This thesis}

This thesis is an exercise in \textit{applied} process-relational philosophy with string diagrams as the tool of choice. It aims to add to the growing body of evidence supporting the view that there is a great deal of mileage in using string diagrams as the standard mathematical formalism in the sciences. Reported here are three strands of research, revolving around the following questions.

\begin{itemize}

	\item The relatively new research area of constructor theory~\cite{deutsch2013constructor} bears family resemblance with the now well-established area of process theories. How are these two programmes related? Are the principles of constructor theory consistent with one another? Do these principles agree with the physical theories that they aim to characterise? Can quantum physics be formulated as a constructor theory? Can categorical quantum mechanics (CQM)---a full-blown process-theoretic formulation of quantum theory---be conceived as a constructor theory?
	
	\item Wave-based computation involves encoding information in some property of waves and performing computation via wave interference~\cite{mahmoud2020introduction}. How can string diagrams be used to develop a formalism for this computing paradigm? How can this formalism be used for design, analysis and optimisation of logic circuits in wave-based technologies?
	
	\item Distributional compositional circuits (DisCoCirc)	is a framework for modelling the grammar and meaning of natural languages in process-theoretic terms~\cite{coecke2021mathematics}. For a particular language, how can an arbitrary piece of text be mapped to DisCoCirc text circuits? Does DisCoCirc eliminate grammatical differences between different natural languages?
\end{itemize}

It is hoped that this thesis makes some progress towards answering---or at least elucidating the nuances of---the above questions. 

\section{Main Results}

The main contributions to the field described in this thesis are as follows. Except where explicitly mentioned, all the work is solely that of the author.

\begin{itemize}
	
	\item Section~\ref{secconsproc} presents constructor theory as a process theory. It takes the theory of conceivable tasks to be $\Rel$, and then shows that for a given choice of substrates, the set of possible tasks forms a sub-SMC of $\Rel$.\footnote{This work is in collaboration with Stefano Gogioso, Vincent Wang-Ma{\'s}cianica, Carlo Maria Scandolo, and Bob Coecke. It was originally published as a part of the article \textit{Constructor Theory as Process Theory}~\cite{Gogioso2023Constructors}. The author of this thesis is the third author of that publication. The specific contributions made by the author to this work are detailed in this thesis.}
	
    \item Section~\ref{seclocquant} gives an overview of the principle of locality and how it is instantiated in the Deutsch-Hayden formulation of quantum theory~\cite{deutsch2000information}. Then it uses the Deutsch-Hayden approach to demonstrate the conflict between the principles of locality and composition.
    
    \item Section~\ref{secconstCQM} discusses the overlap between constructor theory and process theories with reference to quantum physics. It argues that if the principle of locality is rejected, CQM can be conceived as a constructor theory of quantum physics. This argument is supported by multiple examples of constructor-theoretic ideas in quantum physics and computation.
	
	\item Section~\ref{secwavestring} develops a string-diagrammatic formalism for representing and reasoning with wave-based logic circuits that use phase encoding. Using spin-wave systems as an example of wave-based computation platform, Section~\ref{secappwave} shows how to apply this formalism for design, analysis, and optimisation of Boolean logic circuits.\footnote{This work, in collaboration with Alexy Karenowska, was originally published in the article \textit{Wave-based Computation with String Diagrams}~\cite{waseem2023stringwave}. The author of this thesis is the first author of the publication and played a leading role in all aspects of the work.}	
	
	\item Section~\ref{urdugram} develops a hybrid grammar for a restricted fragment of Urdu language.\footnote{This work is in collaboration with Vincent Wang-Ma{\'s}cianica, Jonathon Liu, and Bob Coecke. It first appeared in the article \textit{Language-independence of DisCoCirc's Text Circuits: English and Urdu} ~\cite{DiscocircUrdu}. The author of this thesis is the first author of the publication and played a leading role in all aspects of the work.} 
	
	\item Section~\ref{urducircuits} describes text diagrams and text circuits for Urdu, and shows that there exists a surjective map from the set of all Urdu text endowed with hybrid grammar to the set of all Urdu text circuits. Moreover, for the restricted language fragment considered in this thesis, a clear isomorphism between the hybrid grammars for English (developed in~\cite{wang2023distilling}) and Urdu is demonstrated. Using this isomorpism, it is shown that the text circuits for English and Urdu become the same, up to translating the labels on the gates.\footnote{This work is in collaboration with Vincent Wang-Ma{\'s}cianica, Jonathon Liu, and Bob Coecke. It first appeared in the article \textit{Language-independence of DisCoCirc's Text Circuits: English and Urdu} ~\cite{DiscocircUrdu}. The author of this thesis is the first author of the publication and played a leading role in all aspects of the work.} 
\end{itemize}

Apart from the aforementioned contributions and publications, the author of this thesis co-authored the following two papers and a book as well. These publications are:
\begin{itemize}
	\item Quantum Picturalism: Learning Quantum Theory in High School~\cite{dundar2023quantum}
	
	\item Making the Quantum World Accessible to Young Learners through Quantum Picturalism: An Experimental Study~\cite{dundarcoecke2025makingquantumworldaccessible}
	
	\item Quantum Mechanics in the Single-Photon Laboratory~\cite{waseem2020quantum} (and its second edition~\cite{anwar2024quantum})
\end{itemize}

\chapter{\label{chapinterpretations1}Quantum Theory from Processes and Composition}

\epigraph{‘I would like to make a confession which may seem immoral: I do not believe absolutely in Hilbert space any more.’}  
{John von Neumann,\\ letter to Garrett Birkhoff, 1935~\cite[p.~1]{redei1996john}}

\textit{This chapter provides the necessary background for the thesis, beginning with process theories and string diagrams, and leading into categorical quantum mechanics.}

\section{Introduction}

Categorical quantum mechanics (CQM) is a process-theoretic framework to represent, characterise and reason about quantum processes, their composition and their interaction~\cite{abramsky2004categorical, CKbook}. It was introduced in 2004~\cite{abramsky2004categorical}. There are a number of ways to think of CQM. Initially, it was conceived as a high-level language~\cite{CKbook} on top of Hilbert space~\cite{vonNeumann1932quantenmechanik}, analogous to a high-level programming language, with the matrices of Hilbert space being the analogue to machine code. Alternatively, it can be seen as a new quantum formalism~\cite{selby2021reconstructing}, as desired by von Neumann~\cite{redei1996john, coecke2000operational}, but one that takes compositionality as its first-class citizen---something that had been advocated for by Schr\"{o}dinger in 1935~\cite{schrodinger1935discussion, HardyComposition}. There is also a reconstruction theorem supporting this view~\cite{selby2021reconstructing}. Finally, thanks to an equivalence between diagrams and categories~\cite{penrose1971applications, joyal1991geometry, SelingerSurvey}, CQM can be seen as an intuitive graphical substitute for the Hilbert space formalism~\cite{Kindergarten, ContPhys, coecke2023basic} that transparently reveals new aspects of quantum theory that are deeply buried in the traditional formalism. Since diagrams, foregrounding the composition and interaction of quantum processes are at the heart of CQM, it is also called Quantum Picturalism (QP)~\cite{ContPhys, dundar2023quantum}. 

This \textit{diagrammatic turn} in quantum physics has become very successful in quantum technologies. For example, it has bridged research disciplines like quantum computation and computational linguistics~\cite{heunen2013quantum, coecke2020foundations,meichanetzidis2020quantum, lorenz2023qnlp}, and led to advances in quantum circuit optimisation~\cite{kissinger2020reducing, duncan2020graph} and compilation~\cite{cowtan2019phase, khesin2023graphical}, quantum error correction~\cite{de2020zx, gidney2019efficient}, quantum natural language processing~\cite{coecke2020foundations}, quantum machine learning~\cite{zhao2021analyzing,wang2022differentiating}, measurement-based quantum computing~\cite{kissinger2019universal, mcelvanney2022complete} and optical quantum computing~\cite{de2022quantum, clement2022lov}. 

Within the education sector, the QP approach has been used to teach a university level postgraduate course at the University of Oxford for the last 10+ years~\cite{CKbook}. More recently, with the publication of the book \textit{Quantum in Pictures}~\cite{CGbook}, the QP approach has made the teaching and learning of quantum theory accessible to a wider audience with very few mathematical prerequisites. A recent educational experiment~\cite{dundar2023quantum} showed that quantum theory can be successfully taught to high schoolers (\emph{i.e.} students aged 16-18) using the QP approach.\footnote{54 randomly selected students from the UK attended two-hour online classes, along with tutorials, weekly for eight weeks, and were then assessed using Oxford postgraduate quantum physics exam questions. The results showed that over 80\% students passed, with about half earning a distinction. See Ref.~\cite{dundarcoecke2025makingquantumworldaccessible} for more details.}

The aim of this chapter is to provide background material for the remainder of the thesis. It serves as a self-contained introduction to CQM.\footnote{For a more detailed exposition, the reader is referred to the textbook~\cite{CKbook}; most of the definitions and explanations in this chapter are based on that text.} Since CQM is a process theory, we begin by introducing process theories and string diagrams before presenting the key components of CQM. This material lays the foundation for the chapters that follow. This chapter is organised as follows. Section~\ref{secprocth} introduces process theories from both category-theoretic and string-diagrammatic perspectives. Sections~\ref{sectoqth},~\ref{secCQM} and~\ref{secZX} provide a gentle introduction to CQM. 

\section{Process Theory}\label{secprocth}

\begin{definition}[Definition~3.1 in~\cite{CKbook}]\label{def:ProcTheor}
A \em process theory \em is defined to consist of:
\bit
\item a collection of systems (or system-types, or interfaces) $S$  represented by wires,  
\item a collection of processes $P$ represented by  boxes, and where each process in $P$ has a number of input wires and a number of output wires taken from $S$, and  
\item a means of composing processes, represented by wiring boxes together, and the result of doing so should again be a process in $P$.
\eit  
\end{definition}

In other words, a process theory is a collection of processes that can be composed in a meaningful sense; \emph{i.e.} the composition of processes in the process theory should be a process in the same theory.

From a category-theoretic perspective, a process theory is a \textit{strict symmetric monoidal category}. Essentially, this is a symbolic way to axiomatise series composition $\circ$ and parallel composition $\otimes$. First we define a strict monoidal category.

\begin{definition}[Definition~3.45 in~\cite{CKbook}]\label{moncat}
	A \textit{strict monoidal category} $\mathcal{C}$ consists of:	
	\begin{itemize}
		\item a collection ob($\mathcal{C}$) of \textit{objects}
		\item for every pair of objects $A$, $B$, a set $\mathcal{C}$($A$, $B$) of \textit{morphisms}
		\item for every object $A$, a special identity morphism: $1_A \in \mathcal{C}(A, A)$
		\item a sequential composition operation for morphisms:  $\circ : \mathcal{C}(B, C) \times \mathcal{C}(A, B) \rightarrow \mathcal{C}(A, C)$
		\item a parallel composition operation for objects: $\otimes : \text{ob}(\mathcal{C}) \times \text{ob}(\mathcal{C}) \rightarrow \text{ob}(\mathcal{C})$ 
		\item a unit object: $I \in \text{ob}(\mathcal{C})$
		\item and a parallel composition operation for morphisms: \[\otimes : \mathcal{C}(A, B)  \times \mathcal{C}(C, D)  \rightarrow \mathcal{C} (A \otimes C , B \otimes D)\]
	\end{itemize}
	satisfying the following conditions:
	\begin{itemize}
		\item $\otimes$ is associative and unital on objects: $(A \otimes B) \otimes C = A \otimes (B \otimes C); A\otimes I = A = I \otimes A$
		\item $\otimes$ is associative and unital on morphisms: $(f \otimes g) \otimes h = f \otimes (g \otimes h); f\otimes 1_I = f = 1_I \otimes f$
		\item $\circ$ is associative and unital on morphisms: $(h \circ g) \circ f = h \circ (g \circ f); 1_B \circ f = f = f \circ 1_A$
		\item $\otimes$ and $\circ$ satisfy the \textit{interchange law}: $
		(g_1 \otimes g_2) \circ (f_1 \otimes f_2) = (g_1 \circ f_1) \otimes (g_2 \circ f_2)$
	\end{itemize}
\end{definition}

Compare the above definition with the definition of process theory (Definition~\ref{def:ProcTheor}). \textit{Types} or \textit{system-types} in process theory correspond to \text{objects} in category theory; \textit{processes} correspond to \textit{morphisms}. The set $\mathcal{C}$($A$, $B$) therefore should be considered the set of all processes with input type $A$ and output type $B$. A member of this set can be denoted by either $f \in \mathcal{C}(A, B)$ or $f: A \rightarrow B$. A collection of such sets with sequential composition gives a \textit{category}. A \textit{monoidal category} is a category with both sequential and parallel composition. The `strict' part of the above definition will be explained shortly. Adding a swap morphism to a monoidal category makes it symmetric.

\begin{definition}[Definition~3.46 in~\cite{CKbook}]\label{symmoncat}
	A \textit{strict symmetric monoidal category} is a strict monoidal category with a \textit{swap morphism}: $\sigma_{A, B} : A \otimes B \rightarrow B \otimes A$
	defined for all objects $A$, $B$, satisfying:
	\[
	\sigma_{B, A} \circ \sigma_{A, B} = 1_{A \otimes B} \quad \quad \quad \sigma_{A, I} = 1_A 
	\]
	\[
	(f \otimes g) \circ \sigma_{A, B} = \sigma_{B', A'} \circ (g \otimes f) 
	\]
	\[
	(1_B \otimes \sigma_{A, C}) \circ (\sigma_{A, B} \otimes 1_C) = \sigma_{A, B \otimes C}
	\]
\end{definition}
A monoidal category is called strict when $\otimes$ is associative and unital exactly. On the other hand, in the non-strict monoidal category, associativity and unitality of $\otimes$ are isomorphisms rather than identities.

\begin{definition}[Definition~3.47 in~\cite{CKbook}]
	An object $A$ is \textit{isomorphic} to an object $B$, denoted $A \cong B$
	if and only if there exists a pair of morphisms:
	\[
	f : A \rightarrow B \quad \quad \quad  f^{-1} : B \rightarrow A
	\]
	such that:
	\[
	f^{-1} \circ f = 1_A \quad \quad \quad f \circ f^{-1} = 1_B
	\]
	The morphism $f$ is then called an \textit{isomorphism}.
\end{definition}

Non-strict monoidal categories are the more common species of monoidal categories. Their unitality and associativity of parallel composition are isomorphisms rather than identities:
\[
(A \otimes B) \otimes C \cong A \otimes (B \otimes C) \quad \quad \quad A \otimes I \cong A \cong I \otimes A
\]
These are called \textit{structural isomorphisms} and must be included in the definition of a category in order to define non-strict monoidal category. Additionally, a number of equations called \textit{coherence equations} must be added to the definition to ensure that the behaviour of structural isomorphisms is sensible when they are composed. The following theorem gives the relation between strict and non-strict (symmetric) monoidal categories.

\begin{theorem}[Coherence theorem; Theorem~3.48 in~\cite{CKbook}]
	Every (symmetric) monoidal category $\mathcal{C}$  is equivalent to a strict (symmetric) monoidal category $\mathcal{C'}$.
\end{theorem}

Equivalence between two categories mean that they are the same for all practical purposes. The above theorem justifies considering $(A \otimes B) \otimes C$, $A \otimes (B \otimes C)$ and $A \otimes B \otimes C$ the same. In other words, we do not need to worry about parentheses in associativity of $\otimes$. Likewise, the theorem justifies considering $A \otimes I$, $A$ and $I \otimes A$ the same.

\begin{definition}
	String diagrams are two-dimensional graphical representations of morphisms (processes) in symmetric monoidal categories.\footnote{In the literature, the term \textit{string diagrams} is used for various kinds of diagrams corresponding to different notions of monoidal categories with some additional structure~\cite{SelingerSurvey, piedeleu2023introduction}. For instance, in the book~\cite{CKbook}, the term \textit{circuit diagrams} is used for symmetric monoidal categories to disambiguate from the term \textit{string diagrams}, reserved for symmetric monoidal categories with some additional structure. In this thesis, to avoid potential confusion and to maintain consistency, we use the term \textit{string diagrams} for symmetric monoidal categories. Specifically, we adopt the following correspondence:
		\[
		\text{string diagrams} \leftrightarrow \text{process theory} \leftrightarrow \text{symmetric monoidal category}
		\]
	When needed, we shall introduce additional structure to the string diagrams.}
\end{definition}

The following theorem justifies using string diagrams for symmetric monoidal categories. 

\begin{theorem}[Theorem~3.49 in~\cite{CKbook}]
	[String] diagrams are sound and complete for symmetric monoidal categories. That is, two morphisms $f$ and $g$ are provably equal using the equations of a symmetric monoidal category if and only if they can be expressed as the same [string] diagram.
\end{theorem}

\subsection{Processes and Composition}

\begin{defn}
A \textit{process} can be defined as anything that has zero or more inputs and zero or more outputs. A process is diagrammatically represented as a box with wires.
\[
\input{interpretations/boxa1.tikz}
\]     
\end{defn}
The wires represent the input and output systems or system-types. The input types are below the box while the output types are above the box.

\begin{convention}
	Throughout this thesis---except in Chapter~\ref{chapdiscocirc}---all diagrams should be read from bottom to top and left to right.
\end{convention}

\begin{example}
A single-input, single-output process can be represented as  
\[
\input{interpretations/boxa2.tikz}
\]     	
\end{example}

Composition is an important part of process theory. Once there are two or more processes, how are they to be composed? Diagrammatically, the series composition of $n$ single-input, single-output processes is given by
\[
\input{interpretations/boxa3.tikz}
\]
This kind of composition usually corresponds to composition in time. 

For a process theory to be non-trivial, there must be at least one process $f$ that is non-{\tiny }separable (in time):
\[
\input{interpretations/separatea2.tikz}
\]
This equation means that the single-input, single-output process $f$ is not equal to a series composition of single-input, zero-output ($\psi_A$) and zero-input, single-output ($\psi_B$) processes. In other words, the whole $f$ is not completely describable in terms of its parts composed in series. This is called a \textit{non-trivial series composition}.

If this condition is not satisfied by a process theory, there can be no interesting behaviour in it since each process gives the same output regardless of the input.  

The parallel composition of $n$ single-input, single-output processes is given by   
\[
\input{interpretations/boxa5.tikz}
\]
This kind of composition usually corresponds to composition in space. 
A simpler example of parallel composition would be that of two zero-input, single-output processes
\[
\input{interpretations/boxa4.tikz}
\]
These two processes could, for example, be wavefunctions of two quantum systems. 
Following Schr\"{o}dinger's observation that entanglement is `\textit{the} characteristic trait of quantum mechanics'~\cite{schrodinger1935discussion},
    \[
    \ket{\psi_{AB}} \neq \ket{\psi_{A}} \otimes \ket{\psi_{B}},
    \]
we get a \textit{non-trivial parallel composition}: a situation in which the whole cannot be completely described in terms of the parts composed in parallel. 
\[
\input{interpretations/separatea3.tikz}
\]

\subsection{Diagrams}

As described earlier, a process is diagrammatically represented by a box with wires. Composition of these boxes results in bigger diagrams. To put it differently, diagrams comprise boxes and wires, labelled by processes and system-types respectively~\cite{CKbook}. For instance, the three processes
\[ \input{interpretations/box4.tikz} \ \ , \ \ \  \input{interpretations/box5.tikz} \ \ , \ \text{and} \ \ \  \input{interpretations/box6.tikz}\]
can be composed to form the diagram
\[
\input{interpretations/box7.tikz}
\]

When composing processes, system-types should match. Otherwise, the diagrams are not defined. For example, for the processes
$ \ \ \input{interpretations/box4.tikz} \ \ \text{and} \ \ \  \input{interpretations/box10.tikz} \ $, the following diagrams are defined: 
\[
\input{interpretations/box11.tikz} \ \ , \ \ \ \input{interpretations/box12.tikz}
\]
whereas the following are undefined:
\[
\input{interpretations/box13.tikz} \ \ , \ \ \ \input{interpretations/box14.tikz}
\]
Generally, we shall skip wire labels unless there is an ambiguity.

There are two ways to get new diagram equations from old ones. First, two diagrams are equal as long as they contain the same processes that are connected in the same way, and have the same order of input and output wires. In other words, deforming a diagram while retaining its connectivity and the order of its wires results in the same diagram. For example,
\[
\input{interpretations/box9.tikz} \ \ \ = \ \ \ \input{interpretations/box8.tikz}
\]

\begin{convention}
This property of diagrams---only connectivity matters (OCM)---will be frequently used in this thesis. In this context, the retention of the order of input and output wires is considered part of OCM.
\end{convention}

Second, diagrams can be substituted like algebra. For example,
\[
\input{interpretations/box8b.tikz} \ \ \ = \ \ \ \input{interpretations/box8c.tikz} \ \ \ \ \ \  \implies \ \ \ \ \  \input{interpretations/box8.tikz} \ \ \ = \ \ \ \input{interpretations/box8d.tikz}
\]

\subsection{Special Processes}

Some special kinds of processes appear a lot in process theories; these include states, effects and numbers~\cite{CKbook}. 

\begin{defn}
A \textit{state} is a process with no input. For example, the state
    \[\pointmap{\psi}\] 
     has no input and a single output. This diagram is interpreted as the preparation of a system in a particular state, in a context in which we are not interested in what came before the preparation. 
\end{defn}

\begin{defn}
An \textit{effect} (or a \textit{test}) is a process with no output. For example, the effect
    \[
\copointmap{\phi} 
\]    
has a single input and no output. This diagram is interpreted as a test of whether the system is in a particular state, in a context in which we are not interested in what happens after the test. 
\end{defn}

\begin{defn}
	A number (or a scalar) is a process with no input and no output.
    \[
    \scalar{\lambda}
    \]    
\end{defn}
A number can be obtained by composing a state $\psi$ and an effect $\phi$. 
    \[
    \input{interpretations/statetestpaper.tikz}
    \]
This is known as the \textit{generalised Born rule}.

In the context of quantum theory, diagrams of states, tests and numbers have the following correspondence with the \textit{Dirac notation}:

\[\pointmap{\psi} \xleftrightarrow \ \ket{\psi } \ \ \ , \ \ \  \copointmap{\phi} \xleftrightarrow \ \bra{\phi} \ \ \ , \ \ \ \begin{tikzpicture}
	\begin{pgfonlayer}{nodelayer}
		\node [style=point] (0) at (0, -0.75) {$\psi$};
		\node [style=copoint] (1) at (0, 0.75) {$\phi$};
	\end{pgfonlayer}
	\begin{pgfonlayer}{edgelayer}
		\draw (0) to (1);
	\end{pgfonlayer}
\end{tikzpicture}
   \xleftrightarrow \  \braket{\phi | \psi}\]

\subsection{Separability}

\begin{defn}[Definition~4.5 in~\cite{CKbook}] 
A process $f$ is \textit{$\circ$-separable} if it can be expressed as the series composition of a state $\phi$ and an effect $\pi$:
\[
\boxmap{\, f \,}\ \  = \ \ 
\begin{array}{c}
\raisebox{1.9mm}{\pointmap{\,\phi}}\vspace{0.5mm}\\
\raisebox{-1.9mm}{\copointmap{\,\pi\,}}
\end{array}
\]    
When this process is applied to any state $\psi$, the same state $\phi$ (up to a scalar) is obtained as the output:
\[
\begin{array}{c}
	\raisebox{0mm}{\boxmap{\, f \,}}	\vspace{-4mm}\\
	\raisebox{0.9mm}{{\pointmap{\,\psi}}}
\end{array}\ \  = \ \ 
\begin{array}{c}
	\raisebox{9mm}{\pointmap{\,\phi}}\vspace{-5mm}\\
	\raisebox{0mm}{\copointmap{\,\pi\,}}\vspace{-4mm}\\
	\raisebox{0mm}{{\pointmap{\,\psi}}}
\end{array}\ \  = \ \ 
\begin{array}{c}
	\raisebox{0mm}{\copointmap{\,\pi\,}}\vspace{-4mm}\\
	\raisebox{0mm}{{\pointmap{\,\psi}}}
\end{array} \pointmap{\,\phi}
\]    
In Dirac notation, the above calculation is $f \ket{\psi} = \ket{\phi} \braket{\pi| \psi} = \braket{\pi| \psi} \ket{\phi}$.
\end{defn}

States can be in one or more than one system, denoted by the number of wires:
\[
\begin{tikzpicture}
	\begin{pgfonlayer}{nodelayer}
		\node [style=none, yshift=0.8mm] (4) at (3.5, 0.75) {};
		\node [style=point, yshift=0.8mm] (5) at (3.5, -0.25) {$\pi$};
		\node [style=none] (6) at (3.5, 0) {};
	\end{pgfonlayer}
	\begin{pgfonlayer}{edgelayer}
		\draw [style=swap] (4.center) to (6.center);
	\end{pgfonlayer}
\end{tikzpicture}
 \quad, \quad \input{interpretations/state2.tikz} \quad, \quad \input{interpretations/statemult.tikz}
\]    
 
\begin{defn}[Definition~4.2 in~\cite{CKbook}] 
A two-system or bipartite state $\psi$ is called \textit{$\otimes$-separable} if it can be expressed as the parallel composition of two one-system states $\psi_1$ and $\psi_2$:
    \[\input{interpretations/State2split.tikz}\]
    In Dirac notation, this equation corresponds to $\ket{\psi} = \ket{\psi_{1}} \otimes \ket{\psi_{2}}$. This is the case if $\psi_1$ and $\psi_2$ are completely independent of each other, which is always the case in classical physics. 
   
\noindent States that are not $\otimes$-separable are called \textit{non-separable states}.	
\end{defn}

\begin{definition}
   A two-system state $\psi$ is called a \textit{cup state} if there is a two-system effect $\phi$, called a \textit{cap effect}, such that the following equations hold.
\begin{equation}\label{snaky}
\input{interpretations/CupCapNew.tikz}
\end{equation}
    \end{definition}

\begin{convention}
 For each system-type, we fix a specific pair of cup state $\psi$ and cap effect $\phi$ satisfying Eq.~(\ref{snaky}). Using the notation $\ \raisebox{-1mm}{\begin{tikzpicture}
	\begin{pgfonlayer}{nodelayer}
		\node [style=none] (0) at (-0.75, 0.5) {};
		\node [style=none] (1) at (0.75, 0.5) {};
		\node [style=none] (2) at (0.75, 0.5) {};
		\node [style=none] (3) at (-0.75, 0.5) {};
	\end{pgfonlayer}
	\begin{pgfonlayer}{edgelayer}
		\draw [in=-90, out=-90, looseness=1.75] (0.center) to (1.center);
	\end{pgfonlayer}
\end{tikzpicture}} \ \ := \ \ \input{interpretations/CupStateType.tikz} \ ,\ \  \raisebox{1mm} {\begin{tikzpicture}
	\begin{pgfonlayer}{nodelayer}
		\node [style=none] (0) at (-0.75, -0.5) {};
		\node [style=none] (1) at (0.75, -0.5) {};
		\node [style=none] (2) at (-0.75, -0.5) {}; 
		\node [style=none] (4) at (0.75, -0.5) {};
	\end{pgfonlayer}
	\begin{pgfonlayer}{edgelayer}
		\draw [in=90, out=90, looseness=1.75] (0.center) to (1.center);
	\end{pgfonlayer}
\end{tikzpicture}}\ \  :=\ \  \input{interpretations/CapEffectType.tikz}\ $, Eq.~(\ref{snaky}) becomes:
 \[
 \input{interpretations/CupCapyank.tikz}
 \]
\noindent These are sometimes called the \textit{snake equations}.

In diagrams, the symbols $\ \raisebox{-1mm}{\begin{tikzpicture}
	\begin{pgfonlayer}{nodelayer}
		\node [style=none] (0) at (-0.75, 0.5) {};
		\node [style=none] (1) at (0.75, 0.5) {};
		\node [style=none] (2) at (0.75, 0.5) {};
		\node [style=none] (3) at (-0.75, 0.5) {};
	\end{pgfonlayer}
	\begin{pgfonlayer}{edgelayer}
		\draw [in=-90, out=-90, looseness=1.75] (0.center) to (1.center);
	\end{pgfonlayer}
\end{tikzpicture}}\ $ and $\ \raisebox{1mm} {\begin{tikzpicture}
	\begin{pgfonlayer}{nodelayer}
		\node [style=none] (0) at (-0.75, -0.5) {};
		\node [style=none] (1) at (0.75, -0.5) {};
		\node [style=none] (2) at (-0.75, -0.5) {}; 
		\node [style=none] (4) at (0.75, -0.5) {};
	\end{pgfonlayer}
	\begin{pgfonlayer}{edgelayer}
		\draw [in=90, out=90, looseness=1.75] (0.center) to (1.center);
	\end{pgfonlayer}
\end{tikzpicture}} \ $ always denote the chosen cup and cap for the type of the wires they are attached to; they are not intended to range over all possible cup and cap morphisms.

We additionally assume that every system-type is self-dual, so both legs of each cup or cap carry the same type. Under this assumption, we impose the further structural choice that cups and caps are symmetric:
 \[
 \input{interpretations/cupType2.tikz} \ \ = \ \ \raisebox{-1mm}{\begin{tikzpicture}
	\begin{pgfonlayer}{nodelayer}
		\node [style=none] (0) at (-0.75, 0.5) {};
		\node [style=none] (1) at (0.75, 0.5) {};
		\node [style=none] (2) at (0.75, 0.5) {};
		\node [style=none] (3) at (-0.75, 0.5) {};
	\end{pgfonlayer}
	\begin{pgfonlayer}{edgelayer}
		\draw [in=-90, out=-90, looseness=1.75] (0.center) to (1.center);
	\end{pgfonlayer}
\end{tikzpicture}} \ \ \ , \ \ \  \  \input{interpretations/capType2.tikz} \ \ = \ \ \raisebox{1mm}{\begin{tikzpicture}
	\begin{pgfonlayer}{nodelayer}
		\node [style=none] (0) at (-0.75, -0.5) {};
		\node [style=none] (1) at (0.75, -0.5) {};
		\node [style=none] (2) at (-0.75, -0.5) {}; 
		\node [style=none] (4) at (0.75, -0.5) {};
	\end{pgfonlayer}
	\begin{pgfonlayer}{edgelayer}
		\draw [in=90, out=90, looseness=1.75] (0.center) to (1.center);
	\end{pgfonlayer}
\end{tikzpicture}}
 \]

\end{convention}

\begin{remark}
	Process theories that have cup states and cap effects may have diagrams in which wires connect outputs to inputs of the same process. 
	\[ \input{interpretations/arbitrary_kleinx.tikz}\]
\end{remark}

\noindent The equations characterising cup states and cap effects 
\[
\input{interpretations/CupCapyank2.tikz} \ \ \ , \ \ \  \ \input{interpretations/cupType2.tikz} \ \ = \ \ \raisebox{-1mm}{\begin{tikzpicture}
	\begin{pgfonlayer}{nodelayer}
		\node [style=none] (0) at (-0.75, 0.5) {};
		\node [style=none] (1) at (0.75, 0.5) {};
		\node [style=none] (2) at (0.75, 0.5) {};
		\node [style=none] (3) at (-0.75, 0.5) {};
	\end{pgfonlayer}
	\begin{pgfonlayer}{edgelayer}
		\draw [in=-90, out=-90, looseness=1.75] (0.center) to (1.center);
	\end{pgfonlayer}
\end{tikzpicture}} \ \ \ , \ \ \  \  \input{interpretations/capType2.tikz} \ \ = \ \ \raisebox{1mm}{\begin{tikzpicture}
	\begin{pgfonlayer}{nodelayer}
		\node [style=none] (0) at (-0.75, -0.5) {};
		\node [style=none] (1) at (0.75, -0.5) {};
		\node [style=none] (2) at (-0.75, -0.5) {}; 
		\node [style=none] (4) at (0.75, -0.5) {};
	\end{pgfonlayer}
	\begin{pgfonlayer}{edgelayer}
		\draw [in=90, out=90, looseness=1.75] (0.center) to (1.center);
	\end{pgfonlayer}
\end{tikzpicture}}
\]
are called the \textit{yanking equations}.

\begin{defn}
In string diagrams with cups and caps satisfying yanking equations, there exist states that are duals of processes and vice versa. This bijective correspondence between processes and states is called the \textit{process-state duality}. For example, for single-input, single-output processes and bipartite states:
\[\input{interpretations/smap.tikz} \ \  \mapsto \ \  \input{interpretations/smapstate.tikz} \ \ \ , \ \ \ \ \input{interpretations/state2b.tikz} \ \  \mapsto \ \  \input{interpretations/state2process.tikz} 
\]
\end{defn}
Applying the process-state duality twice, we get back the original states and processes, thanks to the snake equations:
\[\input{interpretations/smap.tikz} \ \  \mapsto \ \  \input{interpretations/smapstate.tikz} \ \  \mapsto \ \  \input{interpretations/smapstate2.tikz} \ \ = \ \ \input{interpretations/smap.tikz} \]
 \[ \input{interpretations/state2b.tikz} \ \  \mapsto \ \  \input{interpretations/state2process.tikz} \ \  \mapsto \ \  \input{interpretations/state2process2.tikz} \ \ = \ \ \input{interpretations/state2b.tikz} 
\]
\begin{convention}
In the above diagrams, we have used asymmetrical shapes---not rectangles and triangles anymore---for the process and state diagrams. This convention will help visually distinguish between a process and its transpose, adjoint and conjugate processes, described below. 
\end{convention}

\begin{defn}[Definition~4.23 in~\cite{CKbook}] 
	The \textit{transpose} $f^T$ of a process $f$ 
	\[
	\input{interpretations/smap.tikz}
	\]
	is defined by bending the input and output wires in the opposite directions. 
	\begin{equation}\label{eq:transpose}
	\begin{tikzpicture}
	\begin{pgfonlayer}{nodelayer}
		\node [style=maptrans] (0) at (0, 0) {$f$};
		\node [style=none] (1) at (0, -1.25) {};
		\node [style=none] (2) at (0, 1.25) {};
	\end{pgfonlayer}
	\begin{pgfonlayer}{edgelayer}
		\draw [style=swap] (0) to (1.center); 
		\draw [style=swap] (2.center) to (0);
	\end{pgfonlayer}
\end{tikzpicture}\ \  :=\ \ \input{interpretations/transmapyes.tikz}
	\end{equation}
    
\end{defn}

In other words, a cup-and-cap pair are needed to obtain the transpose of the process $f$. Another way to think of the transpose is that it is a process rotated by $180^{\circ}$, as is also implied by the notation on the left-hand side above. 
\[
\input{interpretations/smapAB.tikz} \ \ \overset{T}{\mapsto}  \ \ \input{interpretations/smaptransAB.tikz}  
\]

\begin{examples}
	The transpose of a state is obtained by bending its output wire to make it an input wire.
		\[	\begin{tikzpicture}
	\begin{pgfonlayer}{nodelayer}
		\node [style=none, yshift=0.8mm] (2) at (0, 0.5) {};
		\node [style=none] (3) at (0, -0.25) {};
		\node [style=kpoint] (7) at (0, -0.5) {$\psi$};
	\end{pgfonlayer}
	\begin{pgfonlayer}{edgelayer}
		\draw [style=swap] (2.center) to (3.center);
	\end{pgfonlayer}
\end{tikzpicture}
 \ \  \overset{T}{\mapsto} \ \ 	\input{interpretations/transpose2.tikz} \]
	The transpose of an effect is obtained by bending its input into an output.
		\[	\begin{tikzpicture}
	\begin{pgfonlayer}{nodelayer}
		\node [style=none, yshift=0.8mm] (2) at (0, -1) {};
		\node [style=none] (3) at (0, -0.25) {};
		\node [style=kpointadj] (8) at (0, 0) {$\phi$};
	\end{pgfonlayer}
	\begin{pgfonlayer}{edgelayer}
		\draw [style=swap] (2.center) to (3.center);
	\end{pgfonlayer}
\end{tikzpicture}
 \ \  \overset{T}{\mapsto} \ \ 	\input{interpretations/transpose4.tikz} \]
\end{examples}
Boxes can be slid along wires and along cups and caps~\cite{CKbook}.

	\begin{equation}\label{eq:translide1}
\scalebox{0.8}{\input{interpretations/sliding1.tikz}} \ \ = \ \ 
\scalebox{0.8}{\input{interpretations/sliding3.tikz}} 
\end{equation}
	\begin{equation}\label{eq:translide2}
  \scalebox{0.8}{\input{interpretations/sliding4.tikz}} \ \ = \ \ 
   \scalebox{0.8}{\input{interpretations/sliding6.tikz}}  
\end{equation}

\begin{remark}
	There is an operational way to think of a process and its transpose. Consider a cup state shared between two parties Alice and Bob that are at different spatial locations.
	\[
		\input{interpretations/transpose5.tikz}
	\]
	Suppose Alice applies the process $f$ to her system. The diagram becomes
		\[
	\input{interpretations/transpose6.tikz}
	\]
Using Eq.~(\ref{eq:translide1}), we have
		\[
\input{interpretations/transpose6.tikz} \ \ = \ \ \input{interpretations/transpose7.tikz}
\]
This equation can be interpreted as follows: Alice applying the process $f$ to her system is operationally the same as Bob applying the process $f^T$ to his system. 

As another example, suppose Alice applies an effect to test whether her system is in the state $\input{interpretations/transpose8b.tikz}$
		\[
\input{interpretations/transpose8.tikz}
\]
Using Eq.~(\ref{eq:translide1}) again, we have
		\[
\input{interpretations/transpose8.tikz} \ \ = \ \  \input{interpretations/transpose9.tikz} 
\]
This equation implies that if Alice and Bob share a cup state, Alice successfully testing her system to be in the state $\input{interpretations/transpose8b.tikz}$ is operationally the same as Alice having no system and Bob's system being in the state $\input{interpretations/transpose9b.tikz}$. 

\end{remark}

\begin{defn}[Definition~4.34 in~\cite{CKbook}] 
 For a process $f$ 
 	\[
 \input{interpretations/smap.tikz}
 \]
 with the same input and output type, the \textit{trace} is defined as 
  	\[
 \input{interpretations/smaptrace.tikz}
 \]
 The trace maps a process to a number. 
\end{defn}

\begin{defn}[Definition~4.34 in~\cite{CKbook}] 
For a process $g$ 
	\[
	\input{interpretations/smap2.tikz}
	\]
	with one of the inputs (say, the first one) having the same type as one of the outputs (say, the first one), the \textit{partial trace} is defined as 
	\[
	\input{interpretations/smap2trace.tikz}
	\]
\end{defn}

\begin{defn}\cite[p.~104]{CKbook} 
	The \textit{adjoint} $f^{\dagger}$ of $f$ is given by its vertical reflection
	\[
	\input{interpretations/smapAB.tikz} \ \ \mapsto \ \ \input{interpretations/smapdagAB.tikz}  
	\]
\end{defn}

The adjoint of a state is the effect that tests for it.

\begin{example}
The adjoint of the state $\psi$ yields the effect that tests for $\psi$.
			\[	\begin{tikzpicture}
	\begin{pgfonlayer}{nodelayer}
		\node [style=none, yshift=0.8mm] (2) at (0, 0.5) {};
		\node [style=none] (3) at (0, -0.25) {};
		\node [style=kpoint] (7) at (0, -0.5) {$\psi$};
	\end{pgfonlayer}
	\begin{pgfonlayer}{edgelayer}
		\draw [style=swap] (2.center) to (3.center);
	\end{pgfonlayer}
\end{tikzpicture}
 \ \  \overset{\dagger}{\mapsto} \ \ 	\begin{tikzpicture}
	\begin{pgfonlayer}{nodelayer}
		\node [style=none, yshift=0.8mm] (2) at (0, -1.25) {};
		\node [style=none] (3) at (0, -0.25) {};
		\node [style=kpointadj] (8) at (0, 0) {$\psi$};
	\end{pgfonlayer}
	\begin{pgfonlayer}{edgelayer}
		\draw [style=swap] (2.center) to (3.center);
	\end{pgfonlayer}
\end{tikzpicture}
 \]
\end{example}
\noindent In Dirac notation, it is $(\ket{\psi})^{\dagger} = \bra{\psi}$.

\begin{defn}\cite[p.~143]{CKbook} 
   The \textit{conjugate} $\bar{f}$ of $f$ is given by its horizontal reflection
	\[
	\input{interpretations/smapAB.tikz} \ \ \mapsto \ \ \input{interpretations/smapconj.tikz}  
	\]
\end{defn}

\begin{remark}
The conjugate can be obtained by taking the transpose of the adjoint of $f$,
\[
\input{interpretations/smap.tikz}\ \  \overset{\dagger}{\mapsto} \ \  \begin{tikzpicture}
	\begin{pgfonlayer}{nodelayer}
		\node [style=mapdag] (0) at (0, 0) {$f$};
		\node [style=none] (1) at (0, -1.5) {};
		\node [style=none] (2) at (0, 1.5) {};
	\end{pgfonlayer}
	\begin{pgfonlayer}{edgelayer}
		\draw [style=swap] (0) to (1.center);
		\draw [style=swap] (2.center) to (0);
	\end{pgfonlayer}
\end{tikzpicture}
\ \  \overset{T}{\mapsto} \ \ \input{interpretations/smapdagtrans.tikz} 
\]
where 
\[ \begin{tikzpicture}
	\begin{pgfonlayer}{nodelayer}
		\node [style=none] (3) at (0, 1.25) {};
		\node [style=none] (4) at (0, -1.25) {};
		\node [style=mapconj] (5) at (0, 0) {$f$};
	\end{pgfonlayer}
	\begin{pgfonlayer}{edgelayer}
		\draw [style=swap] (5) to (4.center);
		\draw [style=swap] (3.center) to (5);
	\end{pgfonlayer}
\end{tikzpicture}
 \ \  := \input{interpretations/smapdagtrans.tikz} \     
\]
or by taking the adjoint of the transpose of $f$; \textit{i.e.} $\bar{f} = (f^{\dagger})^T = (f^T)^{\dagger}$.
\end{remark}

Using adjoints of processes, the inner product can be defined.

\begin{defn}[Definition~4.48 in~\cite{CKbook}] 
	If two states $\psi$ and $\phi$ have the same type, their \textit{inner product} is given by $\ \begin{tikzpicture}
	\begin{pgfonlayer}{nodelayer}
		\node [style=kpoint] (0) at (-7.5, -0.5) {$\psi$};
		\node [style=kpointadj] (1) at (-7.5, 1) {$\phi$};
	\end{pgfonlayer}
	\begin{pgfonlayer}{edgelayer}
		\draw (0) to (1);
	\end{pgfonlayer}
\end{tikzpicture}
\ $.
	The states are \textit{orthogonal} if $\ \begin{tikzpicture}
	\begin{pgfonlayer}{nodelayer}
		\node [style=kpoint] (0) at (-7.5, -0.5) {$\psi$};
		\node [style=kpointadj] (1) at (-7.5, 1) {$\phi$};
	\end{pgfonlayer}
	\begin{pgfonlayer}{edgelayer}
		\draw (0) to (1);
	\end{pgfonlayer}
\end{tikzpicture}
 \ \ = \ \ 0\ $.
	A state $\psi$ is \textit{normalised} if $\ \begin{tikzpicture}
	\begin{pgfonlayer}{nodelayer}
		\node [style=kpoint] (0) at (-7.5, -0.5) {$\psi$};
		\node [style=kpointadj] (1) at (-7.5, 1) {$\psi$};
	\end{pgfonlayer}
	\begin{pgfonlayer}{edgelayer}
		\draw (0) to (1);
	\end{pgfonlayer}
\end{tikzpicture}
 \ \ = \ \ \emptydiag\ $.
\end{defn}
\noindent The empty diagram, denoted by a dashed box, represents the number 1. 

In Dirac notation, the inner product of states $\ket{\psi}$ and $\ket{\phi}$ is given by the `Dirac braket' $\braket{\phi|\psi}$. For orthogonal states, $\braket{\phi|\psi} = 0$ and for a normalised state $\ket{\psi}$, $\braket{\psi|\psi} = 1$.

The diagram $\ \begin{tikzpicture}
	\begin{pgfonlayer}{nodelayer}
		\node [style=kpoint] (0) at (-7.5, -0.5) {$\psi$};
		\node [style=kpointadj] (1) at (-7.5, 1) {$\phi$};
	\end{pgfonlayer}
	\begin{pgfonlayer}{edgelayer}
		\draw (0) to (1);
	\end{pgfonlayer}
\end{tikzpicture}
\ $ represents testing `whether the state $\psi$ is the state $\phi$'. In other words, the inner product gives the overlap between the states $\psi$ and $\phi$ in the form of a number. If there is no overlap, the inner product is $0$ and the states are orthogonal. For a normalised state, the inner product with itself is $1$.

\begin{defn}
	An \textit{orthonormal basis (ONB)} is a set of states 
	\[\mathcal{A} \ \ := \ \ \left\{ \begin{tikzpicture}
	\begin{pgfonlayer}{nodelayer}
		\node [style=none, yshift=0.8mm] (1) at (0, 0.75) {};
		\node [style=none] (2) at (0, 0) {};
		\node [style=kpoint] (3) at (0, -0.25) {$i$};
	\end{pgfonlayer}
	\begin{pgfonlayer}{edgelayer}
		\draw [style=swap] (1.center) to (2.center);
	\end{pgfonlayer}
\end{tikzpicture}
 \right\}_i \]
	that is orthonormal, \textit{i.e.,} the states obey the equation
	\[
	\begin{tikzpicture}
	\begin{pgfonlayer}{nodelayer}
		\node [style=kpoint] (0) at (-7.5, -0.5) {$i$};
		\node [style=kpointadj] (1) at (-7.5, 1) {$j$};
	\end{pgfonlayer}
	\begin{pgfonlayer}{edgelayer}
		\draw (0) to (1);
	\end{pgfonlayer}
\end{tikzpicture}
 \ \ = \ \ \delta_i^j\ 
	\]
	and form a basis---the minimal set of states such that for all processes $f$ and $g$:
	\[
	\left(\forall \ \  \begin{tikzpicture}
	\begin{pgfonlayer}{nodelayer}
		\node [style=none, yshift=0.8mm] (1) at (0, 0.75) {};
		\node [style=none] (2) at (0, 0) {};
		\node [style=kpoint] (3) at (0, -0.25) {$i$};
	\end{pgfonlayer}
	\begin{pgfonlayer}{edgelayer}
		\draw [style=swap] (1.center) to (2.center);
	\end{pgfonlayer}
\end{tikzpicture}
\ \ : \ \ \input{interpretations/ONB1.tikz} \ \ = \ \ \input{interpretations/ONB2.tikz} \right) \ \ \implies \ \ \input{interpretations/smap.tikz} \ \ = \ \  \input{interpretations/smapg.tikz}
	\]
\end{defn}

\noindent In terms of any ONB $\left\{ \begin{tikzpicture}
	\begin{pgfonlayer}{nodelayer}
		\node [style=none, yshift=0.8mm] (1) at (0, 0.75) {};
		\node [style=none] (2) at (0, 0) {};
		\node [style=kpoint] (3) at (0, -0.25) {$i$};
	\end{pgfonlayer}
	\begin{pgfonlayer}{edgelayer}
		\draw [style=swap] (1.center) to (2.center);
	\end{pgfonlayer}
\end{tikzpicture}
 \right\}_i $, the identity wire decomposes as follows:
\[
\idwire \ \ = \ \ \sum_i  \input{interpretations/ONB3.tikz}
\]
and from this decomposition, we can obtain the decomposition of cups and caps, respectively, as
\[
\ \ = \ \ \sum_i  \input{interpretations/ONB4.tikz} \ \ \text{and} \ \ \ \ = \ \ \sum_i  \input{interpretations/ONB5.tikz} \ \ .
\]

\begin{defn}[Definitions~4.54 and~4.56 in~\cite{CKbook}] 
A process $U$ is  an \textit{isometry} if it obeys the left equation only, and is
a \textit{unitary} if it also obeys the right equation:
\begin{equation}\label{eq:isometry}
\input{interpretations/isometry.tikz} \ \ \ \ \ \ \ \ \ , \ \ \ \ \ \ \ \ \ \  \input{interpretations/unitary.tikz}
\end{equation}
\end{defn}
\noindent Algebraically, the left equation is $U^{\dagger} U = I$ whereas the right equation is $U U^{\dagger} = I$.

The \textit{inverse} of a unitary process $\begin{tikzpicture}
	\begin{pgfonlayer}{nodelayer}
		\node [style=map] (2) at (0, 0) {$U$};
		\node [style=none] (5) at (0, -1.25) {};
		\node [style=none] (6) at (0, 1.25) {};
		\node [style=none] (7) at (0, 0) {};
	\end{pgfonlayer}
	\begin{pgfonlayer}{edgelayer}
		\draw [style=swap] (2) to (5.center);
		\draw (7.center) to (6.center);
	\end{pgfonlayer}
\end{tikzpicture}
$ is $\begin{tikzpicture}
	\begin{pgfonlayer}{nodelayer}
		\node [style=none] (6) at (0, -1.25) {};
		\node [style=none] (7) at (0, 0) {};
		\node [style=none] (9) at (0, 1.25) {};
		\node [style=mapdag] (10) at (0, 0) {$U$};
	\end{pgfonlayer}
	\begin{pgfonlayer}{edgelayer}
		\draw (7.center) to (6.center);
		\draw [style=swap] (9.center) to (10);
	\end{pgfonlayer}
\end{tikzpicture}
$.

\section{Towards a Quantum Process Theory}\label{sectoqth}

Having discussed process theories and the various string-diagrammatic definitions of `quantum-like' concepts, one may be tempted to formulate quantum theory as a process theory of \textbf{linear maps}.

\subsection{Linear Maps}
	Consider the process theory of \textbf{linear maps} in which system-types are finite-dimensional Hilbert spaces and processes are complex linear maps.
	
	In this process theory, the numbers are complex. This implies that the generalised Born rule in this theory yields complex numbers. 

\[    \begin{tikzpicture}
	\begin{pgfonlayer}{nodelayer}
		\node [style=kpoint] (0) at (-7.5, -0.5) {$\psi$};
		\node [style=kpointadj] (1) at (-7.5, 1) {$\phi$};
	\end{pgfonlayer}
	\begin{pgfonlayer}{edgelayer}
		\draw (0) to (1);
	\end{pgfonlayer}
\end{tikzpicture}
 \in 	\mathbb{C}
\]

\noindent Therefore, the `Born rule' for linear maps does not give a probability, since a probability must be a real number. This problem can be resolved by multiplying the number with its conjugate.

\[
\begin{tikzpicture}
	\begin{pgfonlayer}{nodelayer}
		\node [style=kpoint] (0) at (-7.5, -0.5) {$\psi$};
		\node [style=kpointadj] (1) at (-7.5, 1) {$\phi$};
	\end{pgfonlayer}
	\begin{pgfonlayer}{edgelayer}
		\draw (0) to (1);
	\end{pgfonlayer}
\end{tikzpicture}
 \quad \leadsto \quad \begin{tikzpicture}
	\begin{pgfonlayer}{nodelayer}
		\node [style=kpointtrans] (5) at (2, 1) {$\phi$};
		\node [style=kpointconj] (6) at (2, -0.5) {$\psi$};
		\node [style=kpointadj] (7) at (3.5, 1) {$\phi$};
		\node [style=kpoint] (8) at (3.5, -0.5) {$\psi$};
	\end{pgfonlayer}
	\begin{pgfonlayer}{edgelayer}
		\draw (6) to (5);
		\draw (8) to (7);
	\end{pgfonlayer}
\end{tikzpicture}

\]
In Dirac notation, this is $\braket{\phi|\psi}  \leadsto  \overline{\braket{\phi|\psi}}  \braket{\phi|\psi} = |\braket{\phi|\psi}|^2$.

\noindent For normalised states $\psi$ and $\phi$,
\[
0 \leq \ \ \begin{tikzpicture}
	\begin{pgfonlayer}{nodelayer}
		\node [style=kpointtrans] (5) at (2, 1) {$\phi$};
		\node [style=kpointconj] (6) at (2, -0.5) {$\psi$};
		\node [style=kpointadj] (7) at (3.5, 1) {$\phi$};
		\node [style=kpoint] (8) at (3.5, -0.5) {$\psi$};
	\end{pgfonlayer}
	\begin{pgfonlayer}{edgelayer}
		\draw (6) to (5);
		\draw (8) to (7);
	\end{pgfonlayer}
\end{tikzpicture}
 \ \ \leq \ \ 1
\]
thereby giving the standard Born rule of quantum theory.

Fixing the numbers by `doubling' leads to a discrepancy since the states and effects were described with respect to `un-doubled' numbers. 
\[
\input{interpretations/singleprod6.tikz}
\]
We need a new process theory in which doubling is consistently incorporated rather than introduced at the end in an ad hoc manner.

\subsection{Doubling}

Doubling all the processes in the theory of \textbf{linear maps},
\[
\input{interpretations/doubleprocess1to1z.tikz}
\]
we get the process theory of \textbf{pure quantum maps.} In this process theory, states and effects are respectively given by 
\[
    \input{interpretations/doubledstate.tikz} \quad  , \quad \quad \input{interpretations/doubledtest.tikz}
\]
and they are related to probability by the Born rule as follows:
    \[
    \input{interpretations/doubledstatetest.tikz}
    \]

\begin{defn}[Definition~6.8 in~\cite{CKbook}]
	The process theory of \textbf{pure quantum maps} has doubled finite-dimensional Hilbert spaces $\hat{H} = \mathcal{H} \otimes \mathcal{H}$ as system-types
	\[
	\input{interpretations/doubledwire.tikz}
	\]
	and doubled linear maps (called \textit{pure quantum maps}) as processes
	\[
	\input{interpretations/doubleprocess1to1z.tikz}
	\]
\end{defn}
\noindent In this process theory, there are doubled cups and caps, related to single cups and caps as follows:
\[
\begin{tikzpicture}
	\begin{pgfonlayer}{nodelayer}
		\node [style=none] (0) at (-1, 0.5) {}; 
		\node [style=none] (1) at (1, 0.5) {};  
	\end{pgfonlayer}
	\begin{pgfonlayer}{edgelayer}
		\draw [style=boldedge, in=-90, out=-90, looseness=1.75] (0.center) to (1.center);    
	\end{pgfonlayer}
\end{tikzpicture}\ \ :=\ \input{interpretations/double_cup_derivation1.tikz}\vspace{-2mm} \ \ \ , \ \ \ \ \ \ 
\begin{tikzpicture}
	\begin{pgfonlayer}{nodelayer}
		\node [style=none] (0) at (-1, -0.5) {};
		\node [style=none] (1) at (1, -0.5) {};
	\end{pgfonlayer}
	\begin{pgfonlayer}{edgelayer}
		\draw [style=boldedge, in=90, out=90, looseness=1.75] (0.center) to (1.center);
	\end{pgfonlayer}
\end{tikzpicture}
\ \ :=\ \input{interpretations/double_cap_derivation1.tikz}\vspace{-2mm}
\]

\begin{defn}[Definition~6.15 in~\cite{CKbook}]
	 An \textit{entangled state} is a pure quantum state that is not $\otimes$-separable. A bipartite entangled state is given by 
			\[\input{interpretations/perspective10.tikz} \ \ \neq \ \ \begin{tikzpicture}
	\begin{pgfonlayer}{nodelayer}
		\node [style=dkpoint] (35) at (0, -0.75) {$\widehat \psi_1$};
		\node [style=none] (37) at (0, -0.75) {};
		\node [style=none] (38) at (0, 1) {};
	\end{pgfonlayer}
	\begin{pgfonlayer}{edgelayer}
		\draw [style=boldedge] (37.center) to (38.center);
	\end{pgfonlayer}
\end{tikzpicture}
  \ \ \begin{tikzpicture}
	\begin{pgfonlayer}{nodelayer}
		\node [style=dkpoint] (35) at (0, -0.75) {$\widehat \psi_2$};
		\node [style=none] (37) at (0, -0.75) {};
		\node [style=none] (38) at (0, 1) {};
	\end{pgfonlayer}
	\begin{pgfonlayer}{edgelayer}
		\draw [style=boldedge] (37.center) to (38.center);
	\end{pgfonlayer}
\end{tikzpicture}
  \]
\end{defn}

\begin{example}
A well-known example of an entangled state is the Bell state $\frac{\ket{00} + \ket{11}}{\sqrt{2}}$, diagrammatically given by a doubled cup multiplied by a scalar.
\[
\begin{tikzpicture}
	\begin{pgfonlayer}{nodelayer}
		\node [style=none] (0) at (-1, 0.5) {};
		\node [style=none] (1) at (1, 0.5) {};
		\node [style=none] (2) at (-1.75, 0) {$\frac{1}{2}$};
	\end{pgfonlayer}
	\begin{pgfonlayer}{edgelayer}
		\draw [style=boldedge, in=-90, out=-90, looseness=1.75] (0.center) to (1.center);
	\end{pgfonlayer}
\end{tikzpicture}
 
\]
The corresponding Bell test is given in terms of a doubled cap.
\end{example}

\begin{remark}\cite[p.~630]{CKbook} 
Doubling in a process theory preserves diagrams and diagrammatic equations
\[
\input{interpretations/arbitrary_kleinx.tikz} \ \ = \ \  \input{interpretations/arbitrary_kleinx2.tikz}  \ \ \ \implies \ \ \ \input{interpretations/double_arbitrary_klein.tikz}  \ \ = \ \ \input{interpretations/double_arbitrary_klein2.tikz}    
\]
\noindent but it does not preserve global phases
\[
\text{double} \left(\scalar{\lambda} \input{interpretations/smap.tikz} \right) \ \ = \ \ \begin{tikzpicture}
	\begin{pgfonlayer}{nodelayer}
		\node [style=none] (3) at (0, 1.25) {};
		\node [style=none] (4) at (0, -1.25) {};
		\node [style=mapconj] (5) at (0, 0) {$f$};
	\end{pgfonlayer}
	\begin{pgfonlayer}{edgelayer}
		\draw [style=swap] (5) to (4.center);
		\draw [style=swap] (3.center) to (5);
	\end{pgfonlayer}
\end{tikzpicture}

\scalar{\overline{\lambda}} \scalar{\lambda} \input{interpretations/smap.tikz} \ \ = \ \ \begin{tikzpicture}
	\begin{pgfonlayer}{nodelayer}
		\node [style=none] (3) at (0, 1.25) {};
		\node [style=none] (4) at (0, -1.25) {};
		\node [style=mapconj] (5) at (0, 0) {$f$};
	\end{pgfonlayer}
	\begin{pgfonlayer}{edgelayer}
		\draw [style=swap] (5) to (4.center);
		\draw [style=swap] (3.center) to (5);
	\end{pgfonlayer}
\end{tikzpicture}
 \ \input{interpretations/smap.tikz}  \ \ = \ \ \begin{tikzpicture}
	\begin{pgfonlayer}{nodelayer}
		\node [style=dmap] (1) at (0, 0) {$\widehat f$};
		\node [style=none] (2) at (0, 1.25) {};
		\node [style=none] (25) at (0, -1.25) {};
	\end{pgfonlayer}
	\begin{pgfonlayer}{edgelayer}
		\draw [style=boldedge] (2.center) to (1);
		\draw [style=boldedge] (1) to (25.center);
	\end{pgfonlayer}
\end{tikzpicture}
 
\]
since \[
\scalar{\lambda} \scalar{\overline{\lambda}} =  \emptydiag \] where $\lambda = e^{i\theta}$, the global phase.
\end{remark}

\begin{defn}[Definition~6.29 in~\cite{CKbook}] 
	Using doubling, we can define a process called \textit{discarding}, by connecting the two single wires in a doubled wire. 
	\[
	\begin{tikzpicture}
	\begin{pgfonlayer}{nodelayer}
		\node [style=upground] (2) at (0, 0.75) {};
		\node [style=none] (3) at (0, 0.5) {};
		\node [style=none] (4) at (0, -0.5) {};
	\end{pgfonlayer}
	\begin{pgfonlayer}{edgelayer}
		\draw [style=boldedge] (3.center) to (4.center);
	\end{pgfonlayer}
\end{tikzpicture}
 \ \  := \ \ \input{interpretations/trace_def.tikz}
	\] 
\end{defn}
\noindent Notice that the discarding map is actually a cap in disguise, which means it has the same decomposition in terms of an ONB:
\begin{equation}\label{discarddecomp}
		\begin{tikzpicture}
	\begin{pgfonlayer}{nodelayer}
		\node [style=upground] (2) at (0, 0.75) {};
		\node [style=none] (3) at (0, 0.5) {};
		\node [style=none] (4) at (0, -0.5) {};
	\end{pgfonlayer}
	\begin{pgfonlayer}{edgelayer}
		\draw [style=boldedge] (3.center) to (4.center);
	\end{pgfonlayer}
\end{tikzpicture}
 \ \ = \ \ \sum_i \ \input{interpretations/ONB5.tikz} \ \ = \ \ \sum_i \ \begin{tikzpicture}
	\begin{pgfonlayer}{nodelayer}
		\node [style=dkpointadj] (36) at (0, 0.5) {$i$};
		\node [style=none] (58) at (0, -1) {};
		\node [style=none] (59) at (0, 0.5) {};
	\end{pgfonlayer}
	\begin{pgfonlayer}{edgelayer}
		\draw [style=boldedge] (58.center) to (59.center);
	\end{pgfonlayer}
\end{tikzpicture}

\end{equation}
\noindent The discarding operation gets rid of a system. For instance, for a normalised state $\widehat{\psi}$, it gives
\[
\begin{tikzpicture}
	\begin{pgfonlayer}{nodelayer}
		\node [style=dkpoint] (0) at (0, -0.5) {$\widehat \psi$};
		\node [style=upground] (1) at (0, 0.75) {};
	\end{pgfonlayer}
	\begin{pgfonlayer}{edgelayer}
		\draw [style=boldedge] (0) to (1);
	\end{pgfonlayer}
\end{tikzpicture}
 \ \  =\ \   \input{interpretations/discard_prop_proof_paper.tikz}\ \  = \ 
\kpointbraket{\psi}{\psi}\ =\ \ \emptydiag
\]
The discarding operation is a linear map but it is not a pure quantum map, since it cannot be obtained by doubling a linear map. It is a useful process to have in a quantum process theory. It is particularly important when we want to ignore or destroy (a part of) a system. 

\subsection{Quantum Maps}

    \begin{definition}[Definition~6.44 in~\cite{CKbook}] 
    The process theory of \textbf{quantum maps} includes doubled finite-dimensional Hilbert spaces as system-types, and diagrams of pure quantum maps and discarding as processes.
     \[
\input{interpretations/qmaps.tikz}
\]
    \end{definition}
\noindent In this theory, we get the maximally mixed state by taking the transpose of the  discarding process.
\[
\begin{tikzpicture}
	\begin{pgfonlayer}{nodelayer}
		\node [style=none] (9) at (0, 1.25) {};
		\node [style=none] (10) at (0, 0.25) {};
		\node [style=downground] (11) at (0, -0.5) {};
		\node [style=none] (12) at (0, 1.25) {};
	\end{pgfonlayer}
	\begin{pgfonlayer}{edgelayer}
		\draw [style=boldedge] (12.center) to (11);
	\end{pgfonlayer}
\end{tikzpicture}
 \ \ := \input{interpretations/discardtrans2.tikz}
\]
\begin{defn}[Definition~6.34 in~\cite{CKbook}] 
	The \textit{maximally mixed state} $\frac{\ket{0}\bra{0} + \ket{1}\bra{1}}{2}$ is given by 
     \[
\input{interpretations/purification3.tikz}
\]	
\end{defn}
Any impure quantum map can be described in terms of a bigger pure map and discarding. This is achieved through a technique called purification. 
\begin{defn}[Definition~6.47 in~\cite{CKbook}] 
Any quantum map $\Phi$, which may or may not be impure, can be expressed as a combination of a pure map $\widehat f$ and discarding :
\begin{equation}\label{purification}
\dmap{\Phi} \ =\ \ \input{interpretations/purificationI.tikz}	
\end{equation}
If such a relation exists between the processes $\widehat f$ and $\Phi$, $\widehat f$ is called a \textit{purification} of $\Phi$.
\end{defn}

\begin{example}
	The Bell state is a purification of the maximally mixed state. 
	     \[
	\input{interpretations/purification3.tikz}  \ \ = \ \  	\input{interpretations/purification2.tikz}
	\]	
This equation implies that if one of the systems of the Bell state is discarded, the other system becomes a maximally mixed state. 
\end{example}

\subsection{Causality}     

Causality is defined in terms of the discarding operation. 

\begin{defn}[Definition~6.52 in~\cite{CKbook}] 
A process $\phi$ is \textit{causal} if it obeys the equation
\[
\begin{tikzpicture}
	\begin{pgfonlayer}{nodelayer}
		\node [style=none] (0) at (0, -1.25) {};
		\node [style=dmap] (1) at (0, 0) {${\Phi}$};
		\node [style=upground] (2) at (0, 1.5) {};
	\end{pgfonlayer}
	\begin{pgfonlayer}{edgelayer}
		\draw [style=boldedge] (1) to (0.center);
		\draw [style=boldedge] (1) to (2);
	\end{pgfonlayer}
\end{tikzpicture}
\ \ = \ \,\begin{tikzpicture}
	\begin{pgfonlayer}{nodelayer}
		\node [style=upground] (2) at (0, 0.75) {};
		\node [style=none] (3) at (0, 0.5) {};
		\node [style=none] (4) at (0, -0.5) {};
	\end{pgfonlayer}
	\begin{pgfonlayer}{edgelayer}
		\draw [style=boldedge] (3.center) to (4.center);
	\end{pgfonlayer}
\end{tikzpicture}

\]
\textit{i.e., } performing a process and then discarding it is essentially the same as plain discarding. That is, if the output of a process is discarded, it is as if the process never happened.
\end{defn}

Here, a causal process means one that can be deterministically realised in the physical world. Discarding is the only causal quantum effect~\cite{CKbook}, others are non-deterministic. 

\begin{theorem}[Theorem~6.61 in~\cite{CKbook}] 
For every causal quantum map $\Phi$, there exists a purification by an isometry $\widehat U$:
\[
\dmap{\Phi} \ =\ \ \input{interpretations/quantummapcaus.tikz}   
\]
This is known as \textit{Stinespring dilation}.
\end{theorem}

\begin{remark}\label{stinespringrem}
Another version of Stinespring dilation is as follows.	For any causal quantum map $\phi$, there exists a unitary $\widehat V$ that has a pure causal quantum state $\hat \psi$ as one of its inputs and one of its outputs is discarded~\cite{CKbook}.
\begin{equation}\label{stinespring2}
\dmap{\Phi} \ =\ \ \input{interpretations/quantummapcaus2.tikz}   
\end{equation}
This version of Stinespring dilation is sometimes used to argue that the only physical processes in the real world are the pure, unitary processes, and what we observe are impure processes because we can access only a small part of the complete process.
\end{remark}

\begin{defn}\label{qprocessdef}
A \textit{quantum instrument} is essentially a collection or list of quantum maps
$\dmap{\Phi^1}\ , \ \  \dmap{\Phi^2}\ ,\ ... \ , \ \ \dmap{\Phi^n}$ also denoted as
\begin{equation}\label{eq:qprocesslist}
	\left( \dmap{\Phi^i} \right)^i 
\end{equation}
for which the following (generalised) causality condition holds:
\begin{equation}\label{gencausal}
	\sum_i \  \begin{tikzpicture}
	\begin{pgfonlayer}{nodelayer}
		\node [style=none] (0) at (0, -1.25) {};
		\node [style=dmap] (1) at (0, 0) {${\Phi^i}$};
		\node [style=upground] (2) at (0, 1.5) {};
	\end{pgfonlayer}
	\begin{pgfonlayer}{edgelayer}
		\draw [style=boldedge] (1) to (0.center);
		\draw [style=boldedge] (1) to (2);
	\end{pgfonlayer}
\end{tikzpicture}
\ \ = \ \, \begin{tikzpicture}
	\begin{pgfonlayer}{nodelayer}
		\node [style=upground] (2) at (0, 0.75) {};
		\node [style=none] (3) at (0, 0.5) {};
		\node [style=none] (4) at (0, -0.5) {};
	\end{pgfonlayer}
	\begin{pgfonlayer}{edgelayer}
		\draw [style=boldedge] (3.center) to (4.center);
	\end{pgfonlayer}
\end{tikzpicture}

\end{equation}
\noindent The quantum maps $\Phi^i$ in list (\ref{eq:qprocesslist}) are called \textit{branches}. If the quantum instrument includes only one branch, it is \textit{deterministic}. If it includes more than one branch, it is \textit{non-deterministic}. When a system is subjected to a quantum instrument, one of the branches is realised. 
\end{defn}

\begin{remark}
	Note that the abovementioned quantum instrument is more precisely qualified as \textit{uncontrolled}, because it does not depend on any outcome of earlier processes~\cite{CKbook}. 
\end{remark}
\noindent When a system $\boldsymbol \rho$ is subjected to a quantum instrument, the probability of occurrence of each branch $\Phi^i$ is given by the Born rule.
\begin{equation}\label{bornrule}
	P(\Phi^i| \boldsymbol \rho) \ \ \vcentcolon= \ \ \  \input{interpretations/born.tikz}
\end{equation}

\subsection{Classical Systems}

Most discussions of quantum theory involve not only quantum systems but also classical systems and how quantum systems interact with classical data. This is especially the case when measurement or a controlled process is involved. In many accounts of measurement, classical data are obtained as a result of the measurement of quantum systems. In controlled processes, classical data control which quantum processes are applied, thereby affecting quantum systems. 
        
To recap, doubling \textbf{linear maps}---the process theory of finite-dimensional Hilbert spaces and complex linear maps---resulted in the process theory of \textbf{pure quantum maps}. Introducing the discarding operation yielded impure quantum maps, leading to the theory of \textbf{quantum maps}. Now, adding single (\emph{i.e.} non-doubled) wires/systems to \textbf{quantum maps} adds classical data to the process theory. The doubled wires denote finite-dimensional quantum systems while the single wires denote finite-dimensional classical systems/data. 

\[  
\left( \ \ \textrm{quantum} \ \ := \ \ \didwire \ \  \right)
\quad \neq \quad
\left( \ \ \textrm{classical} \ \ := \ \ \idwire \ \ \right)
\]        

\begin{defn}
\textit{Classical states} are defined by fixing an ONB $\left\{ \pointmap{i} \right\}_i$ and then taking the basis states as classical values of data. A classical state $\pointmap{i}$ denotes `providing the classical value $i$'.
\end{defn}

\begin{defn}
	The corresponding \textit{classical effects} are given by $\copointmap{i}$ which represents `testing for the classical value $i$'.
\end{defn}

\noindent This definition of classical states and effects is consistent with orthonormality: 
\[
\input{interpretations/orthonormality.tikz} \ \ = \ \ \delta_i^j \ 
\]

\begin{convention}
All the definitions in this section from this point onwards will be with respect to the same fixed ONB $\mathcal{A} \ \ := \ \  \left\{ \pointmap{i} \right\}_i$.
\end{convention}

\begin{example}
	For a classical \textit{bit}, the ONB can be fixed as
	\[\mathcal{B} \ \ := \ \  \left\{ \pointmap{0}\ ,\ \pointmap{1}\ \right\}\]
	and the two basis states give classical values of data.
\end{example}


\begin{defn}
A \textit{probability distribution} $\{p^i\}_i$ can be associated with the state
\underline{}	 \[  \pointmap{p} \ \ := \ \ \sum_i p^i  \pointmap{i}\]
	 The states $\pointmap{i}$ then represent probability distributions called \textit{point distributions}.
\end{defn}
The state $\pointmap{p}$ can be thought of as the most general state of a classical system, whose value is not exactly known and hence a probability distribution is needed to describe it. 

\begin{defn}[Definition~8.17 in~\cite{CKbook}] 
The \textit{copy} and \textit{delete} operations are defined respectively as
\[
\input{interpretations/copy.tikz} \ :=\  \sum_i 
\begin{array}{c}
	\pointmap{i} \pointmap{i}\\
	\copointmap{i}
\end{array}
\qquad\qquad\qquad\quad
\begin{tikzpicture}
	\begin{pgfonlayer}{nodelayer}
		\node [style=white dot] (0) at (0, 0.5) {};
		\node [style=none] (1) at (0, -0.5) {};
	\end{pgfonlayer}
	\begin{pgfonlayer}{edgelayer}
		\draw [style=swap] (0) to (1.center); 
	\end{pgfonlayer}
\end{tikzpicture}  \ :=\   \sum_i \copointmap{i}
\]
\end{defn}
\noindent One can verify that
\[
\input{interpretations/copy2.tikz} \ =\  \pointmap{i} \pointmap{i} \qquad\qquad\qquad\quad \begin{tikzpicture}
	\begin{pgfonlayer}{nodelayer}
		\node [style=white dot] (0) at (0, 0.5) {};
		\node [style=none] (1) at (0, -0.5) {};
		\node [style=point] (2) at (0, -0.75) {$i$};
	\end{pgfonlayer}
	\begin{pgfonlayer}{edgelayer}
		\draw [style=swap] (0) to (1.center);
	\end{pgfonlayer}
\end{tikzpicture}
 \ = \ \emptydiag
\]
Note that the copy map works only for the ONB states. In other words, a state can be copied by the copy operation if and only if it is a basis state~\cite{CKbook}. The general probability distribution state $\pointmap{p}$ cannot be copied.

Other important classical maps are as follows.
\[
\input{interpretations/cocopy.tikz}\ := \ \sum_i 
\begin{array}{c}
	\pointmap{i} \\
	\copointmap{i} \copointmap{i}
\end{array}  \qquad\qquad\qquad\quad \begin{tikzpicture}
	\begin{pgfonlayer}{nodelayer}
		\node [style=white dot] (0) at (0, -0.5) {};
		\node [style=none] (1) at (0, 0.5) {};
	\end{pgfonlayer}
	\begin{pgfonlayer}{edgelayer}
		\draw [style=swap] (0) to (1.center);
	\end{pgfonlayer}
\end{tikzpicture}
  \ :=\  \sum_i \pointmap{i}
\]
\[
\begin{tikzpicture}
	\begin{pgfonlayer}{nodelayer}
		\node [style=white dot] (0) at (0, -0.5) {};
		\node [style=none] (1) at (1, 0.5) {};
		\node [style=none] (2) at (-1, 0.5) {};
	\end{pgfonlayer}
	\begin{pgfonlayer}{edgelayer}
		\draw [style=swap, in=165, out=-90] (2.center) to (0);
		\draw [style=swap, in=-90, out=15] (0) to (1.center);
	\end{pgfonlayer}
\end{tikzpicture}
 \ := \  \sum_i \pointmap{i} \pointmap{i} \qquad\qquad\qquad\quad \begin{tikzpicture}
	\begin{pgfonlayer}{nodelayer}
		\node [style=white dot] (0) at (0, 0.5) {};
		\node [style=none] (1) at (1, -0.5) {};
		\node [style=none] (2) at (-1, -0.5) {};
	\end{pgfonlayer}
	\begin{pgfonlayer}{edgelayer}
		\draw [style=swap, in=-165, out=90] (2.center) to (0);
		\draw [style=swap, in=90, out=-15] (0) to (1.center);
	\end{pgfonlayer}
\end{tikzpicture}
  \ := \  \sum_i \copointmap{i} \copointmap{i}
\]
The first map is called \textit{matching}. It takes two ONB states as input; if they are equal, it gives the same state as its output. Otherwise, it gives a zero. The second map  gives a uniform probability distribution. The last two maps are the cup and cap in disguise (as can be checked from their ONB expression; they give the state and effect for two classical systems that are perfectly correlated).

\begin{defn}[Definition~8.16 in~\cite{CKbook}] 
The \textit{uniform probability distribution} is defined as the state
\[
_{\frac{1}{D}} \  \begin{tikzpicture}
	\begin{pgfonlayer}{nodelayer}
		\node [style=white dot] (0) at (0, -0.5) {};
		\node [style=none] (1) at (0, 0.5) {};
	\end{pgfonlayer}
	\begin{pgfonlayer}{edgelayer}
		\draw [style=swap] (0) to (1.center);
	\end{pgfonlayer}
\end{tikzpicture}

\]
where $D$ represents the dimension of the ONB.
\end{defn}

\begin{defn}
\textit{Perfect classical correlations} are defined as the state
\[
_{\frac{1}{D}} \  \begin{tikzpicture}
	\begin{pgfonlayer}{nodelayer}
		\node [style=white dot] (0) at (0, -0.5) {};
		\node [style=none] (1) at (1, 0.5) {};
		\node [style=none] (2) at (-1, 0.5) {};
	\end{pgfonlayer}
	\begin{pgfonlayer}{edgelayer}
		\draw [style=swap, in=165, out=-90] (2.center) to (0);
		\draw [style=swap, in=-90, out=15] (0) to (1.center);
	\end{pgfonlayer}
\end{tikzpicture}

\]
where $D$ represents the dimension of the ONB.
\end{defn}

The aforementioned classical maps are special cases of a `spider'.
\begin{defn}[Definition~8.31 in~\cite{CKbook}] 
A \textit{spider} is defined as
	\begin{equation}\label{defspider}
		\input{interpretations/spidernm.tikz} \ \  := \ \  \input{interpretations/spiderdefNEW.tikz}
	\end{equation}
In Dirac notation, the right hand side of the equation corresponds to $\sum_i \ket{\underbrace{i \cdots i}_n} \bra{\underbrace{i \cdots i}_m}$. Spiders can be composed using the following rule.
	\begin{equation}\label{spfuse}
	\input{interpretations/bastards-compose-bis.tikz} 
	\end{equation}
\end{defn}
\noindent A spider can be thought of as a `generalised wire'. A wire has a single input and a single output. A spider can have multiple inputs and outputs. The spider in the definition above has $m$ inputs and $n$ outputs. 

\subsection{Classical-quantum Interaction}

Doubling \textit{classical spiders} yields \textit{quantum spiders}:
\[ 
\input{interpretations/zon2.tikz}\ \ := \ \  \input{interpretations/qspider2s.tikz} 
\]
Key examples of quantum spiders are the Bell state and the Greenberger–Horne–Zeilinger (GHZ) state, respectively given by
\[
\input{interpretations/qspidera.tikz} \ = \ \text{double}\left(\sum_i \pointmap{i}\pointmap{i} \right) , \input{interpretations/qspiderb.tikz} \ = \ \text{double}\left(\sum_i \pointmap{i}\pointmap{i}\pointmap{i} \right)   
\]
In Dirac notation, these states are $\frac{\ket{00} + \ket{11}}{\sqrt{2}}$ and $\frac{\ket{000} + \ket{111}}{\sqrt{2}}$ respectively.

Just as classical spiders are used to represent and reason about classical systems, so quantum spiders do for quantum systems. 
Classical and quantum spiders combine to give rise to a more general form of a spider, called the \textit{bastard spider}~\cite{CKbook}:
\begin{equation}\label{bast}
\input{interpretations/CQMII-3.tikz}
\end{equation}

\begin{remark}
It is pertinent to note that, due to the subtle differences between classical and quantum spiders in diagrams, some extra care is warranted. Classical spiders are represented by thin wires and dots, whereas quantum spiders are depicted using thicker wires and dots. For instance, consider the bastard spider on the left-hand side of Eq. (\ref{bast}). The wires on the left are thick, and therefore quantum, while the wires on the right are thin, indicating they are classical. The dot is also classical.
\end{remark}

Bastard spiders are important for representing and reasoning about interactions between classical and quantum systems. Three important examples in this regard are the \textit{encode}, \textit{measure} and \textit{decoherence} processes.

\begin{defn}
The \textit{encode} process encodes classical data into a quantum system. 
\[ 
\input{interpretations/encode.tikz} \ = \   \sum_i 
\begin{array}{c}
	\pointmap{i} \pointmap{i}\\
	\copointmap{i}
\end{array} \ = \   \sum_i 
\begin{array}{c}
\dpointmap{i}\vspace{1mm}\\
\copointmap{i}
\end{array}
\]
\end{defn}
\noindent In other words, it takes a classical state as an input and gives a quantum state as an output. For instance, a probability distribution, which is a general classical state, is encoded as
\[
\input{interpretations/encode2.tikz} \ = \ \input{interpretations/encode3.tikz} \ = \ \sum_i p^i \pointmap{i} \pointmap{i} \ = \ \sum_i p^i \dpointmap{i} 
\]
\noindent Note that the encode process is actually the copy process in disguise. Measuring is the adjoint of encoding and is therefore matching in disguise.
\begin{defn}
The \textit{measure} process 
\[ 
\input{interpretations/measure.tikz} \ = \  \sum_i 
\begin{array}{c}
 \pointmap{i}\\
	\copointmap{i} 	\copointmap{i}
\end{array} \ = \   \sum_i 
\begin{array}{c}
	\pointmap{i}\vspace{1mm}\\
	\dcopointmap{i}
\end{array}
\] 
measures the state of a quantum system with respect to the basis 
\[
\left\{ \dpointmap{i} \right\}_i
\]
and the outcome (classical data) of the measurement is a probability distribution.
\end{defn}
\noindent For example, consider the general quantum state
\[
\dkpoint{\boldsymbol\rho} \ = \ \sum_{ij} p^{ij} \pointmap{i} \pointmap{j} 
\]
that is measured as follows:
\[
\input{interpretations/measure2.tikz} \ = \ \input{interpretations/measure3.tikz} \ = \ \sum_i p^{ii} \pointmap{i}
\]
where the probabilities $\{p^{ii}\}_i$ are given by the Born rule:
\[
P(i | \boldsymbol\rho) \ := \ p^{ii} \ = \ \input{interpretations/measure4.tikz} \ = \ \begin{tikzpicture}
	\begin{pgfonlayer}{nodelayer}
		\node [style=dcopoint] (6) at (0, 0.75) {$i$};
		\node [style=dkpoint] (10) at (0, -0.75) {$\boldsymbol\rho$};
	\end{pgfonlayer}
	\begin{pgfonlayer}{edgelayer}
		\draw [style=boldedge] (6) to (10);
	\end{pgfonlayer}
\end{tikzpicture}

\]

\begin{defn}[Definition~8.61 in~\cite{CKbook}] 
Decoherence can be defined as a measurement followed by encoding:
\[
\input{interpretations/decoherence1-CQMII.tikz} \ := \ \input{interpretations/decoherence0.tikz} 
\]
\end{defn}
\noindent To illustrate, consider the general quantum state
\[
\dkpoint{\boldsymbol\rho} \ = \ \sum_{ij} p^{ij} \pointmap{i} \pointmap{j} 
\]
that undergoes decoherence:
\[
\input{interpretations/decoherence0b.tikz} \ = \ \input{interpretations/decoherence0c.tikz}  \ = \ \input{interpretations/decoherence0d.tikz} \ = \ \sum_i p^{ii} \pointmap{i} \pointmap{i} \ = \ \sum_i p^{ii} \dpointmap{i} 
\]

\section{Categorical Quantum Mechanics}\label{secCQM}

Armed with all of the diagrammatic tools we have discussed so far, we are now ready to jump into full-blown diagrammatic quantum theory. 

\subsection{Quantum Processes}

First, we describe a process theory that incorporates both quantum and classical maps. 
    \begin{definition}[Definition~8.3 in~\cite{CKbook}] 
   The process theory of \textbf{classical-quantum maps} (abbreviated as \textbf{cq-maps}) comprises finite-dimensional Hilbert spaces and doubled finite-dimensional Hilbert spaces as system-types and the processes are diagrams composed of quantum maps, the encode process, and the measure process:
\[
\input{interpretations/cqmaps.tikz}\]
    \end{definition}
\noindent Requiring the processes in this theory to be causal gives us quantum theory.

\begin{defn}[Definition~8.10 in~\cite{CKbook}]\label{quantumtheorydef}
The theory of \textbf{quantum processes}, which is \textit{quantum theory}, comprises \textbf{cq-maps} 
\[
\input{interpretations/classquantmap1.tikz}\]
that are causal: 
\[ 
\input{interpretations/classquantmapcausality.tikz}
\] 
\end{defn}
\noindent If we restrict this theory to processes with no classical inputs and outputs
\[ 
\input{interpretations/quantummap1.tikz}
\] 
we get the process theory of \textbf{causal quantum maps}.
\noindent Conversely, if we restrict to processes with no quantum inputs and outputs,
\[ 
 \map{f}\ \, :=\ \   \input{interpretations/classicalmap1.tikz}
\] 
we get the process theory of causal classical maps. This theory is also called \textbf{classical processes} or \textbf{stochastic processes}.

\subsection{Quantum Measurement}

Having finally described a quantum process theory, we are now ready to look at the different notions of measurement that show up in standard quantum theory. Previously, we mentioned the simplest form of measurement, given by the process
\begin{equation}\label{meas}
	\input{interpretations/measure.tikz}
 \end{equation}   
defined with respect to a fixed ONB. A more general version of this process is obtained by composition with a unitary.

\begin{defn}
	A \emph{demolition ONB} measurement is given by
\begin{equation}\label{demonb}    
\input{interpretations/dem-onb.tikz}
\end{equation}
where $\widehat U$ is a unitary quantum process.
\end{defn}
\noindent The measure process (\ref{meas}) is a special case of demolition ONB measurement in which the unitary process is the identity. 

The \textit{eigenstates} of ONB measurements are given by doubled ONB states. The ONB measurement (\ref{meas}) has eigenstates $\left\{\begin{tikzpicture}
	\begin{pgfonlayer}{nodelayer}
		\node [style=none] (1) at (0, 1) {};
		\node [style=none] (2) at (0, -0.5) {};
		\node [style=dkpoint] (11) at (0, -0.5) {$i$};
	\end{pgfonlayer}
	\begin{pgfonlayer}{edgelayer}
		\draw [style=boldedge] (2.center) to (1.center);
	\end{pgfonlayer}
\end{tikzpicture}
\right\}_i$, whereas the more general measurement (\ref{demonb}) has eigenstates $\left\{\input{interpretations/copenhagen0d3.tikz}\right\}_i$. The role of the unitary $\widehat{U}$ is to obtain an arbitrary ONB measurement from the ONB measurement (\ref{meas}). Any state that is not an eigenstate of the measurement is called a \textit{superposition} state with respect to that measurement~\cite{CKbook}.

These measurements are called \emph{demolition} measurements because the quantum state is destroyed as a result of these processes. A \emph{non-demolition} measurement is one in which the quantum state is not destroyed. An example of a non-demolition measurement is 
\begin{equation}\label{nondemmeas2}
 \input{interpretations/CQMII-7.tikz} 
\end{equation}

\begin{remark}\label{discarddec}
We get a demolition measurement by discarding the quantum output of a non-demolition measurement.
\[
\input{interpretations/nondem-disc-quantum1.tikz}
\]
Conversely, discarding the classical outcome results in decoherence:
\[
\input{interpretations/CQMII-8.tikz} \ \ \ = \ \  \ \input{interpretations/decoherence1-CQMII.tikz}
\] 
\noindent In other words, a non-demolition measurement followed by discarding the outcome of the measurement has the same effect on the quantum state as decoherence.
\end{remark}

\begin{defn}
 A \textit{non-demolition measurement} is defined as
\begin{equation}\label{nondemmeas}
\input{interpretations/non-dem-map.tikz} 
\end{equation}
where $\widehat U$ is a unitary. The process (\ref{nondemmeas2}) is a non-demolition measurement in which the unitary is the identity.
\end{defn}

What happens to a quantum state after a non-demolition measurement, say the process (\ref{nondemmeas2}), is performed and the classical outcome $i$ is obtained? Diagrammatically, this non-demolition measurement of the quantum state $\boldsymbol{\rho}$ is given by
\[
  \input{interpretations/CQMII-9-1.tikz}
\] 
The measurement results in changing the state $\rho$ to the state $i$. This is usually called a \emph{collapse}:
\[
\dkpoint{\boldsymbol\rho}\ \ \mapsto\ \ \raisebox{1mm}{\dpoint{i}}  
\]

Measurements based on spiders and unitaries only are known as \textit{degenerate}~\cite{coecke2016categorical}. There exists a more general variant of measurement, called a \textit{von Neumann measurement} or \textit{projective measurement}. 

\begin{defn}\cite[p.~23]{coecke2016categorical} 
A (non-demolition) \textit{von Neumann measurement} is defined as a quantum process $\ \input{interpretations/measdiag1.tikz}\ $ that obeys the equation $\ \input{interpretations/diagramsUUorth3.tikz}\ $.
\end{defn}

In other words, making a von Neumann measurement twice is equal to making the measurement once and then copying the outcome of the measurement. This implies that once a von Neumann measurement is made, making the same measurement again will result in the same outcome, and the quantum state will stay the same. This characteristic of von Neumann measurements is commonly known as the \textit{von Neumann projection postulate}. One can verify that measurements defined by processes (\ref{nondemmeas}) and (\ref{nondemmeas2}) obey the projection postulate and are hence von Neumann measurements. If the quantum output of the von Neumann measurement is discarded, a demolition von Neumann measurement is obtained. 

More generally, a demolition quantum measurement can be defined as a quantum process from a quantum input to a classical output.
\begin{definition}
A \textit{demolition positive operator-valued measure (POVM) measurement} is defined as
\begin{equation}\label{nondempovm}
\input{interpretations/demo-povm.tikz}
\end{equation}
where $\Phi$ is a quantum process with pure quantum maps as branches. A \textit{non-demolition POVM measurement} is obtained by purifying (\ref{nondempovm}):
\begin{equation}
\input{interpretations/POVMg.tikz}
\end{equation}
Here $\widehat U$ is an isometry. 
\end{definition}

\begin{remark}
In fact, every non-demolition POVM measurement is obtained by composing one output of an isometry $\widehat U$ with an ONB measurement~\cite{CKbook}:
\[
\input{interpretations/CQMII-11.tikz}
\]
This is known as \textit{Naimark dilation}.
\end{remark}

So far, we have defined measurements and processes/operations that are independent of the earlier measurement outcomes. We can also define processes controlled by measurement outcomes or classical inputs; for instance, we can define controlled unitaries. 

\begin{defn}
	A \textit{controlled unitary} is a quantum-classical process 
	\[
	\input{interpretations/controliso1.tikz}
	\]
satisfying the equations:
	\[
	\input{interpretations/diagramsWW3copy.tikz}\qquad\qquad\quad\input{interpretations/diagramsWW4copy.tikz}
	\]
\end{defn}

\subsection{Mixtures}

There are situations in which we are not certain about which one of a known set of possible quantum processes actually took place. 
In such a case, we can represent the set of processes as a single quantum process controlled by a classical input. This classical input is supplied by a probability distribution to take into account our ignorance.

\begin{defn}\label{statmix}
A \textit{mixture} of quantum processes is defined as
\[
\input{interpretations/CQMII-4.tikz}
\]    
\end{defn}
\noindent In fact, we can represent any mixture of quantum processes like this~\cite{CKbook}:
\[
\sum_{i} p^i \dmap{\Phi^i} \ = \ \input{interpretations/bekan13.tikz}
\]
\noindent where $\{p^i\}_i$ represents a probability distribution.
The components of the mixture can be obtained by plugging ONB states into the classical input~\cite{CKbook}: 
\[
\input{interpretations/bekana.tikz}  \ = \ \dmap{\Phi^i}
\]

\begin{remark}
For a mixture (Definition~\ref{statmix}), it is known that one of the processes from a certain collection occurs, but it is not known which one. The probability distribution supplied to the classical input reflects our lack of knowledge about which process actually took place. On the other hand, a quantum instrument with multiple branches (as described in Definition~\ref{qprocessdef}) is inherently non-deterministic---that is, the non-determinism does not stem from our ignorance.

\end{remark}

\section{The ZX-calculus}\label{secZX}

We have discussed most of the key concepts in quantum theory using our diagrammatic formalism. What is left is a description of \textit{phases} and \textit{complementarity}. Incorporating these aspects of quantum theory requires us to look inside the boxes of the process theory we have described. This granular description of quantum processes involves spiders corresponding to two complementary ONBs and `decorated' with phases. In this section, we restrict the core of our discussion to qubits.

\subsection{Complementarity}

We saw that spiders were described with respect to a fixed ONB. For two complementary bases, we shall need spiders of two colours. 
\[
	\input{interpretations/Zspider.tikz} \quad \quad\quad  \input{interpretations/Xspider.tikz}
\]
We fix the complementary bases $\{\ket{0}, \ket{1}\}$ and $\{\ket{+}, \ket{-}\}$ for the white and gray spiders respectively. This gives us 
\[
\input{interpretations/Zspider.tikz} \ := \  \begin{array}{c}
	\pointmap{0} \cdots \pointmap{0}  \vspace{1mm}\\
	\copointmap{0} \cdots \copointmap{0} 
\end{array} + \begin{array}{c}
\pointmap{1} \cdots \pointmap{1}  \vspace{1mm}\\
\copointmap{1} \cdots \copointmap{1} 
\end{array}  \ = \ \ketbra{0 \cdots 0}{0 \cdots 0} + \ketbra{1 \cdots 1}{1 \cdots 1}
\]
\[
\input{interpretations/Xspider.tikz}  \ := \  \begin{array}{c}
	\pointmap{+} \cdots \pointmap{+}  \vspace{1mm}\\
	\copointmap{+} \cdots \copointmap{+} 
\end{array} + \begin{array}{c}
	\pointmap{-} \cdots \pointmap{-}  \vspace{1mm}\\
	\copointmap{-} \cdots \copointmap{-} 
\end{array}  \ =  \ \ketbra{+ \cdots +}{+ \cdots +} + \ketbra{- \cdots -}{- \cdots -}
\]
Since the chosen bases lie on the Z- and X-axes of the Bloch sphere, these are called Z and X bases; this is where the Z and X of the ZX-calculus come from. The corresponding spiders are called Z and X spiders. The diagrams made of Z and X spiders are called ZX-diagrams.

How do we make sense of complementarity using these two kinds of spiders? Intuitively, complementarity is the idea that if the information is encoded in one basis and measured in a complementary basis, there must be no transfer of information. It can be checked that this is the case for the Z and X spiders:
\begin{equation}\label{compex}
\input{interpretations/complementarity.tikz} \ = \  \input{interpretations/complementarity2.tikz} \ = \ \begin{tikzpicture}
	\begin{pgfonlayer}{nodelayer}
		\node [style=white dot] (0) at (0, -0.75) {};
		\node [style=none] (4) at (0, -1.5) {};
		\node [style=gray dot] (5) at (0, 0.75) {};
		\node [style=none] (7) at (0, 1.5) {};
	\end{pgfonlayer}
	\begin{pgfonlayer}{edgelayer}
		\draw [style=swap] (4.center) to (0);
		\draw [style=swap] (7.center) to (5);
	\end{pgfonlayer}
\end{tikzpicture}

\end{equation}
This is interpreted as follows: encoding in one basis and then measuring in a complementary basis allows for no flow of information.

In the last step of Eq.~(\ref{compex}), the symbol $=$ should be read as equality up to a non-zero global scalar factor. As we shall be mostly concerned with how the processes are composed and not with the exact numbers, we ignore such global scalars in diagrammatic calculations. 

\begin{convention}
     In diagrammatic equations and calculations in what follows, all equalities are up to non-zero global scalars. 
\end{convention}

A more general version of the complementarity of Z and X spiders is given by
\begin{equation}\label{compex2}
\input{interpretations/complementarity4.tikz} 
\end{equation}

\subsection{Phases}

A \textit{decorated spider} is a spider decorated with a phase. A general Z spider with the phase $\alpha$ is given by
\begin{equation}\label{eq:Z-spider-def}
	\input{interpretations/Zsp-a.tikz} \ := \  \begin{array}{c}
		\pointmap{0} \cdots \pointmap{0}  \vspace{1mm}\\
		\copointmap{0} \cdots \copointmap{0} 
	\end{array} + e^{i \alpha} \begin{array}{c}
		\pointmap{1} \cdots \pointmap{1}  \vspace{1mm}\\
		\copointmap{1} \cdots \copointmap{1} 
	\end{array} \ = \  \ketbra{0 \cdots 0}{0 \cdots 0} + e^{i \alpha} \ketbra{1 \cdots 1}{1 \cdots 1}
\end{equation}
For instance, the decorated spider $\begin{tikzpicture}
	\begin{pgfonlayer}{nodelayer}
		\node [style=white phase dot] (0) at (0, -0.75) {$\alpha$};
		\node [style=none] (1) at (0, 0.75) {};
	\end{pgfonlayer}
	\begin{pgfonlayer}{edgelayer}
		\draw (0) to (1.center);
	\end{pgfonlayer}
\end{tikzpicture}
$ corresponds to the state $\ket{0} + e^{i\alpha}\ket{1}$, and its conjugate is $\begin{tikzpicture}
	\begin{pgfonlayer}{nodelayer}
		\node [style=white phase dot] (0) at (0, -0.75) {$-\alpha$};
		\node [style=none] (1) at (0, 0.75) {};
	\end{pgfonlayer}
	\begin{pgfonlayer}{edgelayer}
		\draw (0) to (1.center);
	\end{pgfonlayer}
\end{tikzpicture}
 =   \ket{0} + e^{-i\alpha}\ket{1}$.

The old spider fusion rule (\ref{spfuse}) does not apply to decorated spiders, for which there is the following fusion rule:
\begin{equation}\label{spfuse2}
	\input{interpretations/fuseZ.tikz}
\end{equation}

Roughly speaking, phase is the information in the quantum system that gets lost when a quantum state is measured. This can be verified for decorated spiders using the fusion rule (\ref{spfuse2}):
\[
	\input{interpretations/phase1.tikz} \ = \  	\input{interpretations/phase2.tikz}  \  = \ \begin{tikzpicture}
	\begin{pgfonlayer}{nodelayer}
		\node [style=white dot] (5) at (0, -0.5) {};
		\node [style=none] (7) at (0, 0.75) {};
	\end{pgfonlayer}
	\begin{pgfonlayer}{edgelayer}
		\draw [style=swap] (7.center) to (5);
	\end{pgfonlayer}
\end{tikzpicture}
 
\]

A general X spider with the phase $\alpha$ is given by 
\begin{equation}
	\input{interpretations/Xsp-a.tikz} \  :=\  \begin{array}{c}
		\pointmap{+} \cdots \pointmap{+}  \vspace{1mm}\\
		\copointmap{+} \cdots \copointmap{+} 
	\end{array} + e^{i \alpha}  \begin{array}{c}
		\pointmap{-} \cdots \pointmap{-}  \vspace{1mm}\\
		\copointmap{-} \cdots \copointmap{-} 
	\end{array} \ = \  \ketbra{+ \cdots +}{+ \cdots +} + e^{i \alpha} \ketbra{- \cdots -}{- \cdots -}
\end{equation}
and it has a fusion rule similar to rule (\ref{spfuse2}).

\begin{remark}
Representing spiders with round dots — rather than asymmetric boxes — makes intuitive sense, since spiders exhibit symmetries such as `leg-swapping' and `leg-flipping'.
 
 \[
 	\input{interpretations/Zsp-a2.tikz} \ = \ 	\input{interpretations/Zsp-a.tikz}  \ \ \ \ \ , \ \ \ \ \ \     	\input{interpretations/Xsp-a2.tikz} \ = \ 	\input{interpretations/Xsp-a.tikz} 
 \]
 
 \[
  	\input{interpretations/Zsp-a3.tikz} \ = \ 	\input{interpretations/Zsp-a4.tikz}  \ \ \ \ \ , \ \ \ \ \ \     	\input{interpretations/Xsp-a3.tikz} \ = \ 	\input{interpretations/Xsp-a4.tikz} 
 \]
 
\end{remark}

\subsection{Rules of the ZX-calculus}

We have seen some rules about how spiders of the same colour and different colours interact. Adding few more rules, we get the complete set for the stabiliser fragment of ZX-calculus~\cite{van2020zx}:
\[
\input{interpretations/ZX-rules.tikz}
\]
The letters above the equality signs are shorthand notation for the rule names and stand for spider $(f)$usion, ($\pi$)-commute, $(c)$opy, $(b)$ialgebra, $(h)$adamard,  $(hh)$-cancellation and $(id)$entity. All these rules hold up to  a global scalar factor and for phases $\alpha, \beta \in \{0, \pi/2, \pi, -\pi/2\}$. The Hadamard process is a convenient notation for a frequently-used ZX diagram:
\[
\begin{tikzpicture}
	\begin{pgfonlayer}{nodelayer}
		\node [style=none] (97) at (0, 1.25) {};
		\node [style=none] (98) at (0, -1.25) {};
		\node [style=small hadamard] (134) at (0, 0) {};
	\end{pgfonlayer}
	\begin{pgfonlayer}{edgelayer}
		\draw (98.center) to (97.center);
	\end{pgfonlayer}
\end{tikzpicture}
 \ \ = \ \ \input{interpretations/hadamard2.tikz}
\]
and, as illustrated in the $(h)$ rule above, changes spider colour. The $(h)$ and $(hh)$ rules imply that the above rules also hold for spiders with their colours inverted.  

Complementarity (Eq.~(\ref{compex})) can be derived using the $(f)$, $(id)$, $(b)$, and $(c)$ rules. The decorated-spider version of the general complementarity rule (Eq.~(\ref{compex2})
\begin{equation}\label{spchop}
	\input{interpretations/chop.tikz}
\end{equation}
can be derived using Eq.~(\ref{compex}) and the $(f)$-rule. Another rule we shall use later is derived as follows:
\begin{equation}\label{picopyrule}
	\input{interpretations/picopy1.tikz}
\end{equation}

A reader interested in a more thorough exposition of the ZX-calculus is referred to the comprehensive tutorial~\cite{van2020zx} or the textbook~\cite{KissingerWetering2024Book}; a high-level crash course for professionals can be found in Ref.~\cite{coecke2023basic}. Since 2023, there has been a book on the ZX-calculus targeted to the general public~\cite{CGbook} as well!

It has been a century since the inception of quantum theory, whereas the ZX-calculus was proposed only in 2007~\cite{coecke2008interacting, coecke2011interacting}. When presenting the ZX-calculus as an alternative to the standard Hilbert space formalism, three desirable criteria are typically considered: \textit{universality}, \textit{soundness} and \textit{completeness}. Universality requires that any linear map between qubits be representable as a ZX-diagram. Soundness means that any equality derivable in the ZX-calculus also holds in the Hilbert space formalism. The ZX-calculus is indeed both universal and sound for qubit quantum theory~\cite{coecke2008interacting, coecke2011interacting}. 

\begin{theorem}\label{zxuniversalthm}
	The ZX-calculus is universal and sound for linear maps between qubits.
\end{theorem}

Completeness of the ZX-calculus means that its rules are sufficient to derive any equality that can be derived using the Hilbert space formalism in qubit quantum mechanics. Proving completeness was a non-trivial task and it was achieved in multiple stages, starting with a small fragment of quantum theory and gradually increasing its size~\cite{backens2014zx, backens2014zxb, hadzihasanovic2015diagrammatic, jeandel2018complete}. This culminated in the proof of completeness of the ZX-calculus for qubit quantum theory in 2017~\cite{hadzihasanovic2018two, jeandel2018diagrammatic, vilmart2018near}. 

\begin{theorem}\label{zxcompletethm}
	The ZX-calculus is complete for linear maps between qubits. 
\end{theorem}
\noindent In other words, the ZX-calculus is as expressive as the Hilbert space formalism and hence is a viable substitute of the latter for qubit quantum theory. Research aimed at proving completeness of ZX and similar graphical calculi (such as ZW~\cite{de2024minimality}, ZH~\cite{backens2018zh} and ZXW~\cite{poor2023completeness, wang2023completeness}) for bigger fragments of quantum theory is in full swing. In fact, very recently, there has been a presentation of \textit{finite-dimensional ZX-calculus} and its completeness for finite-dimensional Hilbert spaces~\cite{poor2024zx}.

The ZX-calculus has been extensively employed in quantum computing and technology research~\cite{KissingerWetering2024Book}.\footnote{As of writing this chapter, the website \url{https://zxcalculus.com/publications.html} lists 300+ papers related to the ZX-calculus.} Some examples of research topics include quantum circuit optimisation~\cite{duncan2020graph, de2020fast, kissinger2020reducing, van2024optimal}, quantum error correction~\cite{de2020zx, kissinger2022phase, huang2023graphical, cowtan2024css}, measurement-based quantum computing~\cite{duncan2010rewriting, backens2021there, mcelvanney2022complete}, and fusion-based quantum computing~\cite{bombin2023unifying, pankovich2023flexible}.

\pgfplotsset{compat=1.17}
\usetikzlibrary{trees}
\usetikzlibrary{topaths}
\usetikzlibrary{decorations.pathmorphing}
\usetikzlibrary{fadings}
\usetikzlibrary{decorations.pathreplacing}
\usetikzlibrary{decorations.markings}
\usetikzlibrary{matrix,backgrounds,folding}
\usetikzlibrary{chains,scopes,positioning,fit}
\usetikzlibrary{arrows,shadows}
\usetikzlibrary{calc} 
\usetikzlibrary{chains}
\usetikzlibrary{shapes,shapes.geometric,shapes.misc}
\usetikzlibrary{circuits.ee.IEC}
\usetikzlibrary{decorations.markings}

\tikzstyle{every picture}=[baseline=-0.5em,scale=1]

\pgfdeclarelayer{edgelayer}
\pgfdeclarelayer{nodelayer}
\pgfsetlayers{background,edgelayer,nodelayer,main}
\tikzstyle{none}=[inner sep=0mm]
\tikzstyle{every loop}=[]
\tikzstyle{mark coordinate}=[inner sep=0pt,outer sep=0pt,minimum size=3pt,fill=black,circle]

\tikzstyle{smalldotb}=[fill=black, inner sep=0mm,minimum width=1mm,minimum height=1mm,draw,shape=circle]
\tikzstyle{H}=[-, style=dashed]

\chapter{\label{chapconstructors}Constructors, Processes, and Quantum Theory}

\epigraph{`How often have I said to you that when you have eliminated the impossible, whatever remains, \textit{however improbable}, must be the truth?'}  
{Sherlock Holmes to Dr Watson, \\ The Sign of Four~\cite[p.~93]{doyle2010sign}}

\textit{The novel contributions of this chapter are presented in Sections~\ref{secconsproc}, ~\ref{seclocquant} and~\ref{secconstCQM}. Section~\ref{secconsproc} is adapted from the publication~\cite{Gogioso2023Constructors}, based on work carried out in collaboration with Stefano Gogioso, Vincent Wang-Ma{\'s}cianica, Carlo Maria Scandolo, and Bob Coecke. The author of this thesis is the third author of the publication~\cite{Gogioso2023Constructors}. Only the parts of this work to which the author directly contributed are described here. Sections~\ref{seclocquant} and~\ref{secconstCQM} are unpublished and represent original work carried out solely by the author of this thesis.}

\section{Introduction}

Constructor theory~\cite{deutsch2013constructor, deutsch2015constructor, marletto2022science} is a framework to formulate fundamental scientific theories in terms of the \textit{possibility} and \textit{impossibility} of \textit{tasks}. 
It particularly aims to provide an alternative to the prevailing approach in physics that characterises physical phenomena in terms of initial conditions and dynamical laws. While the prevailing approach formulates theories by specifying a dichotomy between \textit{what happens} and \textit{what does not happen}, constructor theory seeks to do so by specifying a dichotomy between \textit{what can happen} (\emph{i.e.}, what is possible) and \textit{what cannot happen} (\emph{i.e.}, what is impossible)~\cite{deutsch2015constructor}.\footnote{The eponymous popular science book~\cite{marletto2022science} dubs constructor theory `the science of can and can't'.}

In constructor theory, tasks are defined as \textit{transformations} between \textit{systems}. Auxiliary systems may also be required as inputs by these transformations. For instance, consider the task of transforming black shoes into brown shoes. This task of painting shoes requires an amount of brown paint in addition to black shoes that are to be painted. The paint plays the role of a catalyst.

Tasks transform \textit{states} of systems into other states, and \textit{attributes} of systems, like the blackness of a shoe, into other attributes. In the physical world, the amount of brown paint is limited and the aforementioned task can be performed only a finite number of times before one runs out of brown paint. Ideally, if an unlimited amount of brown paint were available, the painting task could be performed as many times as needed. Such inexhaustible catalysts are called \textit{constructors}. 

A task is called \textit{possible} if there is a constructor that makes the task performable an arbitrary number of times. Otherwise, the task is called \textit{impossible}. Even though constructors and tasks are abstract, they offer explanatory value. The aim of constructor theory is to characterise physical theories in terms of which tasks are possible and which are impossible. Constructor theory has been described as a deeper (meta)theory than other theories of physics.\footnote{In an article in Quanta Magazine, it was called `a master theory---a set of ideas so fundamental that all other theories would spring from it.'~\cite{quantagefter}} In the seminal article by David Deutsch, it is pitched as `the ultimate generalisation of the theory of computation'~\cite[p.~12]{deutsch2013constructor} and also a theory that could `underlie all other theories including relativity and quantum theory'~\cite[p.~5]{deutsch2013constructor}. The relationship between constructor theory and other theories (called \textit{subsidiary theories} by Deutsch~\cite{deutsch2013constructor}) was described as follows:

``Other theories specify what substrates and tasks exist, and provide the multiplication tables for serial and parallel composition of tasks, and state that some of the tasks are impossible, and explain why. Constructor theory provides a unifying formalism in which other theories can do this, and its principles constrain their laws, and in particular, require certain types of task to be possible.''~\cite[p.~5]{deutsch2013constructor}

Being a metatheory, constructor theory is implementation-agnostic. That is, the user of constructor theory is free to pick a formal system of mathematics of their liking as a concrete language to interpret it. One such language is that of process theories. We contend that process theories represent a good choice of mathematical language to formally incarnate constructor theory because they are especially suitable for describing the composition of processes in space-time. Furthermore, they are expressible as string diagrams, enabling the user to work at an abstraction-level of their choosing. 

This chapter is organised as follows. It comprises three main sections, briefly summarised below. 

In Section~\ref{secconsproc}, the ideas and jargon of constructor theory are formally interpreted using the string-diagrammatic language of process theories. The aim is to use string diagrams to bridge work between the research areas of constructor theory~\cite{deutsch2013constructor, deutsch2015constructor, marletto2015constructor, marletto2016constructor, marletto2016constructorb} and process theories~\cite{coecke2010universe, coecke2011categories, SelingerSurvey, coecke2013alternative, CKbook}. 
The target is to provide a rigorous and intuitive mathematical language for constructor theory, which may help constructor theorists develop, express and communicate their ideas. This work is also intended to invite constructor theorists to the research area of process theories---potentially an extremely fruitful arena within which the ramifications of constructor theory can be explored. 

In Section~\ref{seclocquant}, it is argued that a constructor-theoretic formulation of non-relativistic quantum theory is afflicted with a fundamental problem: there is an inconsistency between the restrictions imposed by principles of locality and composition. Unless there is a fully compositional and local formulation of quantum theory (which, to the best of our knowledge, does not exist), either the principle of locality or the composition principle must be discarded from constructor theory. 

Finally, in Section~\ref{secconstCQM}, it is demonstrated via examples that categorical quantum mechanics (CQM) is a viable constructor theory of non-relativistic quantum physics---one that is completely compositional but non-local. 

Section~\ref{sumconst} concludes the chapter.

\section{Constructor Theory as a Process Theory}\label{secconsproc}

This section is adapted from our publication~\cite{Gogioso2023Constructors}. The expository presentation of this section is intended for those, particularly constructor theorists, who are not familiar with process theories. But it serves another purpose as well: it offers a dictionary to understand the mathematics of constructor theory, transliterated into diagrams with the fewest possible interpretational choices and without any bells and whistles.

\subsection{Conceivable Tasks}

From a constructor-theoretic point of view, a physical theory is characterised by answering the question `which \textit{tasks} can be performed within this physical theory?'. There are two notions of tasks in constructor theory: the abstract \textit{conceivable tasks}, and the concrete \textit{possible tasks}. Conceivable tasks are independent of particular scientific theories. They provide a formal setting for formulating principles or deriving constraints. Possible tasks, on the other hand, are dependent on the particular theory under discussion and are induced by constructors physically available to implement the tasks.

According to the seminal paper on constructor theory~\cite{deutsch2013constructor}, it is required of conceivable tasks that they be composable in parallel and sequence (series). 
In other words, conceivable tasks form a symmetric monoidal category (SMC). In the same paper, Deutsch used relations between sets to model tasks in constructor theory. The literature of constructor theory~\cite{deutsch2013constructor, deutsch2015constructor, marletto2015constructor, marletto2016constructor, marletto2016constructorb} has followed this modelling choice ever since. In our formalisation of constructor theory as a process theory, we also stick to this choice.

\begin{remark}
In this chapter, monoidal categories are taken to be \textit{strict}.\footnote{Strict monoidal categories are defined in Definition~\ref{moncat}.} Particularly, it is assumed that in a monoidal category $\smc{C}$, the objects obj($\smc{C}$) form a strict monoid. For SMC $\Rel$, this means a choice of a singleton set $1 := \{*\}$ that acts as a strict unit for the Cartesian product, \emph{i.e.,} $X \times 1 = X = 1 \times X$.\footnote{SMCs are defined in Definition~\ref{symmoncat}.}

This also implies that we can write tuples without worrying about using parentheses or nesting. For instance, we can write triples as $X \times Y \times Z = \suchthat{(x,y,z)}{x \in X, y \in Y, z \in Z}$.
It must be noted that strictness does not apply to symmetry isomorphisms; \emph{i.e.,} for symmetry isomorphisms, we have $X \times Y \cong Y \times X$ but this does not imply $X \times Y = Y \times X$
\end{remark}

\begin{definition}[Definition~4.97 in~\cite{CKbook}]
	A dagger symmetric monoidal category ($\dagger$-SMC) is a symmetric monoidal category with a \textit{dagger functor} $\dagger$ that 
\begin{itemize}
\item does not alter objects: $A^{\dagger} \coloneqq A$
\item reverses morphisms: $(f : A \rightarrow B)^{\dagger} \coloneqq f^{\dagger} : B \rightarrow A$
\item is involutive: $(f^{\dagger})^{\dagger} = f$
\item and respects the symmetric monoidal category structure:
\[
(g \circ f)^{\dagger} = f^{\dagger} \circ g^{\dagger} \quad \quad \quad (f \otimes g)^{\dagger} = f^{\dagger} \otimes g^{\dagger }  \quad \quad \quad \sigma^{\dagger}_{A, B} = \sigma_{B, A} \ 
\]
\end{itemize}
\end{definition}


In this work, the theory of \textit{conceivable tasks} is taken to be $\Rel$, which is the $\dagger$-SMC of sets and relations. The components of this theory are described below.

A task or relation $\task{A} \subseteq X \times Y$ is denoted by $\task{A}: X \rightarrow Y$ where sets $X$ and $Y$ comprise legitimate input and output states respectively of the task. In order to avoid confusion, the pair or tuple notation is used for pairs or tuples of a Cartesian product, and the \textit{maplet} notation is used for pairs of domain and codomain elements in a relation.
\[
x \mapsto y \;\;:\equiv\;\; (x,y)
\hspace{2.5cm}
x \tmapsto{A} y \;\;:\equiv\;\; (x,y) \in \task{A}
\]
The $\task{A}$ symbol is omitted from $\tmapsto{A}$ when there is no ambiguity and the context is clear.
Diagrammatically, this task can be represented as
\[
\input{constructors/task.tikz}
\]
The arrow shape of the box represents the direction of the task from input to output.

As mentioned before, conceivable tasks can be composed in sequence and parallel. 

The \textit{sequential composition} of tasks $\task{A} : X \rightarrow Y$ and $\task{B} : Y \rightarrow Z$, denoted by $\task{B} 
\circ \task{A} : X \rightarrow Z$, is defined as $
\task{B} \circ \task{A}
:= \suchthat{x\mapsto z}{\exists y \in Y.\, x \stackrel{\task{A}}{\mapsto} y \text{ and } y \stackrel{\task{B}}{\mapsto} z}$. 
Here task $\task{A}$ happens first, followed by $\task{B}$. Diagrammatically, this sequential composition is represented as
\[
\input{constructors/seqcomp.tikz}
\]

Using two sets of states $X$ and $Y$, a composite set of states can be obtained using the Cartesian product $X \times Y$, as $
X \times Y
:= \suchthat{(x,y)}{x \in X \text{ and }y \in Y}$. The \textit{parallel composition} of two tasks $\task{A}: X \rightarrow Y$ and $\task{B}: Z \rightarrow W$, denoted by $\task{A} \times \task{B}: X \times Z \rightarrow Y \times W$, is defined as $\task{A} \times \task{B}
:= \suchthat{ (x,z) \mapsto (y,w) }{ x \tmapsto{A} y \text{ and } z \tmapsto{B} w }$. 
Diagrammatically, the parallel composition is given by
\[
\input{constructors/parcomp.tikz}
\]

The \textit{transpose} of a task $\task{A} : X \rightarrow Y$, denoted by $\task{A}^{\dagger} : Y \rightarrow X$, is defined as $\task{A}^\dagger
:= \suchthat{ y \mapsto x}{x \tmapsto{A} y }$
and diagrammatically represented as
\[
\input{constructors/transpose.tikz}
\]

\textit{Symmetry isomorphisms}, which are also known as \textit{swaps}, are denoted by $\sigma_{X,Y}: X \times Y \stackrel{\cong}{\rightarrow} Y \times X$ and defined as $\sigma_{X,Y} := \suchthat{ (x,y) \mapsto (y,x) }{ x \in X \text{ and } y \in Y }$.
Diagrammatically, swaps are represented as 
\[
\input{constructors/syms.tikz}
\]
Swaps are a structural feature of the category and have the following properties. 
\[
\input{constructors/symeqs.tikz}
\]
Thanks to swaps and their nice properties shown above, relations can be composed into acyclic networks, in which outputs of relations can be connected to inputs of other relations.

\begin{remark}
 The sets $X$ and $Y$ are taken to be distinct for the sake of generality. Restricting these sets to be always equal means restricting the theory of conceivable tasks to be the $\dagger$-SMC $\EndoRel$ of sets and endo-relations $R: X \rightarrow X$, which is a sub-$\dagger$-SMC of $\Rel$.
\end{remark}

The \textit{copy map}, denoted as $\delta_X : X \rightarrow X \times X$, is defined as follows on a set $X$:
\begin{align*}
	\delta_X &:= \suchthat{x \mapsto (x,x)}{ x \in X}\\
\end{align*}	
Diagrammatically, the copy map is given by
\[
\input{constructors/copy.tikz}
\]
The \textit{discarding map}, denoted as $\epsilon_X : X \rightarrow 1$, is defined as follows on a set $X$:
\begin{align*}	
	\epsilon_X &:= \suchthat{x \mapsto *}{ x \in X}
\end{align*}
diagrammatically represented as
\[
\input{constructors/del.tikz}
\]

\noindent Copy and discarding maps have some useful properties. Copies are indistinguishable under swaps
\[
\scalebox{0.8}{$
	\input{constructors/copydeleqs1.tikz}
	$}
\]
and under repeated copying
\[
\scalebox{0.8}{$
	\input{constructors/copydeleqs2.tikz}
	$}
\]
Moreover, copying and then discarding one of the copies results in the identity:
\[
\scalebox{0.8}{$
	\input{constructors/copydeleqs3.tikz}
	$}
\]

The transpose $\delta_X^\dagger : X \times X \rightarrow X$ of the copy map is a partial function that 
gives the common value of its inputs as the output when the inputs are equal, and is not defined otherwise. 
\[
\delta_X^\dagger := \suchthat{(x,x) \mapsto x}{ x \in X}
\]
It is called the \textit{match map}. Diagrammatically, it is represented as
\[
\input{constructors/match.tikz}
\]

It is important to distinguish between states and attributes: considering the set $X$, its elements $x \in X$ are referred to as \textit{states} while its subsets $S \subseteq X$ are referred to as \textit{attributes.}.

Relations $S: 1 \rightarrow X$ 
can be identified with all possible \textit{attributes} of states in $X$, \textit{i.e.,} with all possible subsets $S \subseteq X$: $S \cong \suchthat{* \mapsto x}{ x \in S}$.

In diagrams, states and attributes are represented by the same notation. This is because states $x \in X$ can be identified with singleton sets $\{x\}\subseteq X$. Thus, states or attributes are diagrammatically given by
\[
\input{constructors/attr.tikz}
\]

The transpose $\epsilon_X^{\dagger}: 1 \rightarrow X$ of the discarding map is the \textit{trivial attribute}, corresponding to the subset $X \subseteq X$. It is defined as $\eta_X := \epsilon_X^\dagger = \suchthat{* \mapsto x}{ x \in X}$ and is represented diagrammatically as
\[
\input{constructors/unit.tikz}
\]

Tasks can be conditioned to specific input states using attributes.

\begin{definition}
Let $\task{A}: X \times Z \rightarrow Y$ be a task and let $S \subseteq Z$ be an attribute on states in $Z$. The task obtained by forgetting all information about the $Z$ input of task $\task{A}$ except the fact that the input state has attribute $S$ is the following \textit{pre-conditioned task}.
\[    
\substack{
	\input{constructors/preconditioned.tikz}\\
	\ \\
	\\
		\task{A} \circ (\id{X} \times S) \ \ 
}
= \suchthat{x \mapsto y}{\exists z \in S.\, (x,z) \tmapsto{A} y}
\]
As a special case, the $Z$ input can be discarded completely. This is obtained by pre-conditioning against the trivial attribute $\eta_Z$, as follows:
\[
\substack{
	\input{constructors/prediscard.tikz}\\
	\ \\
	\\
		\task{A} \circ (\id{X} \times \eta_Z)\ \ 
}
= \suchthat{x \mapsto y}{\exists z \in Z.\, (x,z) \tmapsto{A} y}
\]
\end{definition}

The relations $X \rightarrow 1$ are exactly the transposes $S^{\dagger}: X \rightarrow 1$ of the attributes $S: 1 \rightarrow X$. Explicitly, they are the constant partial functions with the attribute $S$ as their domain: $S^\dagger := \suchthat{ x \mapsto * }{ x \in S }$. The transposes of attributes are called \textit{tests}. Tasks can be conditioned to specific output states using tests.

\begin{definition}
Let $\task{A}: X \rightarrow Y \times Z$ be a task and let $S \subseteq Z$ be an attribute on states in $Z$. The task obtained by forgetting all information about the $Z$ output of task $\task{A}$ except the fact that the output state has attribute $S$ is the following \textit{post-conditioned task}.
\[
\substack{
	\input{constructors/postconditioned.tikz}\\
	\ \\
	\\
		(\id{Y} \times S^\dagger) \circ \task{A} \ \ 
}
= \suchthat{x \mapsto y}{\exists z \in S.\, x \tmapsto{A} (y, z)}
\]
As a special case, the $Z$ output can be discarded completely. This is obtained by post-conditioning against the trivial attribute $\eta_Z$, as follows:
\[
\substack{
	\input{constructors/postdiscard.tikz}\\
	\ \\
	\\
		(\id{Y} \times \epsilon_Z) \circ \task{A} \ \ 
}
= \suchthat{x \mapsto y}{\exists z \in Z.\, x \tmapsto{A} (y, z)}
\]
\end{definition}

\begin{remark}
A task $\task{A}: X \times Z \rightarrow Y \times W$ can simultaneously be pre-conditioned against an attribute $P \subseteq Z$ and post-conditioned against an attribute  $Q \subseteq W$ as follows:
\[
\substack{
	\input{constructors/prepost.tikz}\\ 
	\ \\
	\\
	(\id{Y} \times Q^\dagger) \circ \task{A} \circ (\id{X} \times P) \ \ 	
}
= \suchthat{x \mapsto y}{\exists z \in P, w \in Q.\, (x, z) \tmapsto{A} (y, w)}
\]
\end{remark}

\subsection{Possible Tasks}

As mentioned earlier, conceivable tasks are theory-independent whereas possible tasks are theory-dependent. A choice of substrates within a theory of processes is needed to determine possible tasks. 

\begin{defn}\label{chsub}
	A \textit{choice of substrates}, represented as $\left( \smc{C}, \Sigma, \Gamma \right)$, consists of the following components:
	\begin{enumerate}
		\item  A reference \textit{theory of processes}, in the form of a strict SMC $\smc{C}$ with monoidal product $\otimes$ and unit $I$.\footnote{For instance, this could be the theory of finite-dimensional quantum systems and unitary transformations. However, as far as we know, the literature on constructor theory only considers $\Rel$, in which the tensor product $\otimes$ is the Cartesian product. In contrast, quantum theory requires the tensor product to be the Kronecker product of Hilbert spaces.} This specifies the theory of conceivable tasks.
		\item A choice of \textit{substrates}, in the form of a subset $\Sigma \subseteq \obj{\smc{C}}$ of systems in the theory of processes. 
		\item A choice of \textit{sets of substrate states}, in the form of a family $\Gamma = \left( \Gamma_{\substr{H}} \right)_{\substr{H} \in \Sigma}$ where $\Gamma_\substr{H} \subseteq \states{\smc{C}}{\substr{H}}$ is a set of states in $\smc{C}$ for each substrate $\substr{H} \in \Sigma$.
	\end{enumerate}
	The choice of substrates is required to be closed under parallel composition: $I \in \Sigma$ and  $\substr{H} \otimes \substr{K} \in \Sigma$ for all $\substr{H}, \substr{K} \in \Sigma$. Moreover, the set of substrate states is required to respect parallel composition of substrates: $\Gamma_I = 1$ and $\Gamma_{\substr{H}\otimes\substr{K}} = \Gamma_\substr{H} \times \Gamma_\substr{K}$ for all $\substr{H}, \substr{K} \in \Sigma$.\footnote{This requirement states that the state of a combined system is the ordered set of the subsystem states. This requirement on substrate states to be `Cartesian' is called `the principle of locality'~\cite{deutsch2013constructor,deutsch2015constructor} in the literature of constructor theory. This requirement ensures that there is explicitly no entanglement arising out of parallel composition of substrates.} 
\end{defn}

For two substrates $\substr{H}, \substr{K} \in \Sigma$, considering the tasks $\Gamma_{\substr{H}} \rightarrow \Gamma_{\substr{K}}$, one is interested in the question `which of these tasks are possible within the given theory of processes?'. A task is possible if there exists a \textit{constructor} that can enable the task to be performed. To elaborate and make this statement precise, we need the following definitions.

\begin{definition}
        Let $\left( \smc{C}, \Sigma, \Gamma \right)$ be a choice of substrates and consider two substrates $\substr{H}, \substr{K} \in \Sigma$. A process $f: \substr{H} \rightarrow \substr{K}$
        that maps states in $\Gamma_\substr{H}$ to states in $\Gamma_\substr{K}$ is called a
         \textit{task-inducing} process:
\[
\forall \rho \in \Gamma_\substr{H}.\;
f(\rho) \in \Gamma_\substr{K}
\]
The task \textit{induced by $f$} is denoted by $\indtask{f}$:
\[
\indtask{f} := \suchthat{\rho \mapsto f(\rho)}{\rho \in \Gamma_\substr{H}}
\]
\end{definition}

\begin{defn}\label{defpossible}	
	Let $\left( \smc{C}, \Sigma, \Gamma \right)$ be a choice of substrates and consider a task $\task{A}: \Gamma_\substr{H} \rightarrow \Gamma_\substr{K}$. The task $\task{A}$ is \textit{possible} if there are:
	\begin{itemize}
		\item[(i)] a substrate \substr{C} (acting as a \textit{constructor} for the task),
		\item[(ii)] an attribute $P \subseteq \Gamma_\substr{C}$ (singling out the relevant constructor states), and
		\item[(iii)] a task-inducing process $f: \substr{H} \otimes \substr{C} \rightarrow \substr{K} \otimes \substr{C}$ (actually performing the task)
	\end{itemize}
	such that the following two conditions are satisfied:
	\begin{enumerate}
		\item The task $\task{A}$ is obtained from the induced task $\indtask{f}$ by requiring that the input constructor state has attribute $P$ and discarding the constructor output:
\begin{align*}
	\task{A} &=
	\substack{
		\input{constructors/def28.tikz}\\
		\ \\
		\\
		(\id{\Gamma_\substr{K}} \times \epsilon_{\Gamma_{\substr{C}}})
		\circ \indtask{f}
		\circ (\id{\Gamma_\substr{H}} \times P)
	}= \suchthat{
		\rho \mapsto \rho'
	}{
		\exists \gamma \in P, \gamma' \in \Gamma_\substr{C}.\,
		f(\rho \otimes \gamma) = \rho' \otimes \gamma'
	}
\end{align*}
		\item The attribute $P$ is preserved by the induced task $\indtask{f}$. While a particular constructor state $\gamma \in P$ may be modified to become $\gamma'$ by the underlying process of the induced task $\indtask{f}$, $\gamma'$ remains a constructor state for the same induced task $\indtask{f}$, \emph{i.e.} $\gamma' \in P$. In $\Rel$, this constraint is equivalently expressed as the induced task $\indtask{f}$ sending the set of constructors $P$ to a subset of itself, regardless of the input and output on the substrates $\substr{H},\substr{K}$:
\[
\substack{
	\input{constructors/def282.tikz}\\
	\ \\
	\\
	(\epsilon_{\Gamma_\substr{K}} \times \id{\Gamma_{\substr{C}}})
	\circ \indtask{f}
	\circ (\eta_{\Gamma_\substr{H}} \times P) \ \ \subseteq \ \ P
}
\]
	\end{enumerate}
 The set of possible tasks under the given choice of substrates is denoted by $\possibletasks{\left( \smc{C}, \Sigma, \Gamma \right)}$.
\end{defn}

We give the main result of this section in the following proposition. The result shows that for a given choice of substrates, the set of possible tasks forms a sub-SMC of $\Rel$, the theory of conceivable tasks. In other words, the set of possible tasks is closed under composition in arbitrary (acyclic) networks.

\begin{proposition}
	For a given choice of substrates $\left( \smc{C}, \Sigma, \Gamma \right)$, the set of possible tasks $\possibletasks{\left( \smc{C}, \Sigma, \Gamma \right)}$ forms a sub-SMC of $\Rel$.
\end{proposition}

\begin{proof}
The identity tasks for all systems are made possible with the identity isomorphisms of $\smc{C}$:
\[
\input{constructors/idtask.tikz}
\]
enabled by the trivial constructor $\substr{C} := I$.

Likewise, the swap tasks for all systems are made possible by the symmetry isomorphisms of $\smc{C}$, using the trivial constructor $\substr{C} := I$:
\[
\input{constructors/swaptask.tikz}
\]
For two possible tasks $\task{A}$ and $\task{B}$, having constructors $C$ and $D$ respectively, the sequential composition $\task{B} \circ \task{A}$ is made possible with the constructor $C \otimes D$, as follows:
\[\input{constructors/seqtaskcomp.tikz}\] 
For two possible tasks $\task{A}$ and $\task{B}$, having constructors $C$ and $D$ respectively, the parallel composition $\task{A} \times \task{B}$ is made possible with the constructor $C \otimes D$, as follows:
\[\input{constructors/partaskcomp.tikz} \]
This concludes the proof. \qedhere
\end{proof}

\begin{remark}
		The possible tasks form a sub-SMC of $\Rel$ because all 
		structural operations are inherited from $\Rel$. The objects 
		are the sets $\Gamma_\substr{H}$, which are objects of $\Rel$. 
		The morphisms are relations between these sets, which are 
		morphisms of $\Rel$. Sequential composition of tasks is 
		relational composition, which is sequential composition in 
		$\Rel$. The identity tasks are identity relations in $\Rel$, 
		and the swap tasks are the symmetry isomorphisms of $\Rel$. 
		For parallel composition, the key point is the locality 
		condition $\Gamma_{\substr{H} \otimes \substr{L}} = 
		\Gamma_\substr{H} \times \Gamma_\substr{L}$. This ensures 
		that the parallel composition of tasks 
		$\task{A}: \Gamma_\substr{H} \rightarrow \Gamma_\substr{K}$ 
		and $\task{B}: \Gamma_\substr{L} \rightarrow \Gamma_\substr{M}$ 
		is a relation 
		$\task{A} \times \task{B}: \Gamma_\substr{H} \times 
		\Gamma_\substr{L} \rightarrow \Gamma_\substr{K} \times 
		\Gamma_\substr{M}$, which is exactly the monoidal product 
		of $\Rel$ applied to the relevant objects. Since all 
		operations coincide with those of $\Rel$ and the proof 
		establishes closure under these operations, the possible 
		tasks form a sub-SMC of $\Rel$.
\end{remark}

\begin{remark}
While closure under composition may be expected, since Deutsch required conceivable tasks to be composable in parallel and sequence, the result shows something more. For any subsidiary theory $\smc{C}$, the possible tasks form a sub-SMC of $\Rel$, the theory of conceivable tasks. The subsidiary theory can be anything, be it quantum theory, classical mechanics, or any other theory of physics, but the possible tasks it gives rise to are always relations between sets of states, living in $\Rel$. In this sense, $\Rel$ serves as the universal arena for constructor theory, regardless of the underlying physics. This universality hinges on the locality condition, which aligns the monoidal product of tasks with that of $\Rel$.
\end{remark}

To recap, taking the theory of conceivable tasks to be $\Rel$, we showed that, for a given choice of substrates, the set of possible tasks forms a sub-SMC of $\Rel$. Our choice of $\Rel$ aligns with the constructor theory literature, wherein tasks are modelled by relations between sets~\cite{deutsch2013constructor, deutsch2015constructor, marletto2015constructor, marletto2016constructor, marletto2016constructorb}. This work constitutes, to the best of our knowledge, the first process-theoretic formalisation of constructor theory.

\subsection{Attributes as states}\label{attvsst}

As originally proposed, constructor theory defined tasks as transformations of states~\cite{deutsch2013constructor}. Later literature defined tasks as acting on the attributes of states instead of on the underlying states~\cite{deutsch2015constructor, marletto2015constructor, marletto2016constructor, marletto2016constructorb}. This was motivated by the idea that the abstract specification of (possible) tasks should be based on the observable `macrostates' (\emph{i.e.,} attributes, or subsets of a set) of a physical system, instead of on the unobservable `microstates' (\emph{i.e.,} states, or elements of a set) that constitute them. In Ref.~\cite{Gogioso2023Constructors}, it was shown that the attribute-focused perspective is derivable from the state-focused perspective in a compositionally sound way via a suitable coarse-graining. We summarise the main insights here.\footnote{We do not include all details and provide only a summary, as the author of this thesis was not much involved in this part of the work.} 

A notion of `coarse-graining' is defined for tasks. This coarse graining is in order to obtain tasks transforming attributes (macrostates) from tasks transforming states (microstates). The attributes are allowed to have non-trivial overlap. More specifically, the attributes do not need to form a partition, but nesting of attributes is not allowed. That is, no attribute can be a proper superset of another attribute.\footnote{This ensures that the attributes are distinguishable in a strong sense, as required by constructor theory~\cite{deutsch2015constructor}.} With these requirements, coarse-grained tasks are defined and the following results are obtained. 

First, it is proved that for any process theory of tasks, the coarse-grained tasks can also be arranged into a process theory. 
This implies that the tasks defined on attributes are just as compositionally sound as those defined on states. Second, the tasks defined on states can be compositionally embedded into the universe of coarse-grained tasks. This proves that coarse-grained tasks are a sound generalisation of the tasks originally defined on states. Third, the coarse-grained tasks can be embedded back into the universe of ordinary tasks. This shows that the ordinary tasks---\emph{i.e.} those defined on states---are as expressive as coarse-grained tasks.

To summarise, a distinction between states (microstates) and attributes (macrostates) is made in the constructor theory literature for dealing with theories such as thermodynamics and information theory~\cite{deutsch2015constructor, marletto2016constructorb}. In categorical semantics, as discussed above, this corresponds to the distinction between ordinary tasks defined on states and coarse-grained tasks defined on attributes. The distinction may or may not make sense in a physical setting. However, mathematically speaking, there is no distinction between the two approaches since they are equivalent as far as their expressivity is concerned.



\definecolor{hexcolor0xa9a9a9}{rgb}{0.663,0.663,0.663} 
\tikzstyle{GrayLine}=[dashed,draw=hexcolor0xa9a9a9] 
\tikzstyle{gray}=[dashed,draw=hexcolor0xa9a9a9]

\section{Locality in Quantum Theory}\label{seclocquant}

There is an overlap between the research areas of constructor theory and process theories as applied to quantum theory \emph{i.e.} CQM. However, there also exists a bone of contention between the two areas: the principle of locality, which lies at the core of constructor theory~\cite{deutsch2013constructor, deutsch2015constructor} but is rejected by CQM. In this section, we discuss the principle of locality and how it appears in the Deutsch-Hayden approach to quantum theory~\cite{deutsch2000information}. We also argue that Deutsch-Hayden locality is not compositional. This means that employing Deutsch-Hayden locality leads to a conflict between two main principles of constructor theory: composition and locality.

Constructor theory, since its conception, has stuck by `the principle of locality'~\cite{deutsch2013constructor}, according to which the state of a joint system is the ordered set of the states of its subsystems. This is implemented in categorical semantics by imposing the `Cartesian' requirement on substrate states: $\Gamma_{\substr{H}\otimes\substr{K}} = \Gamma_\substr{H} \times \Gamma_\substr{K}$ for all $\substr{H}, \substr{K} \in \Sigma$ in Definition~\ref{chsub}. The need for keeping locality in a physical theory is debatable. However, from a mathematical point of view, the generic tensor products for symmetric monoidal categories conservatively and very fruitfully generalise Cartesian products of sets. This is especially important in physical theories like quantum theory. 

Below we provide more background on the principle of locality, and how it is made to work in quantum theory. Finally, we show how implementing the principle of locality renders quantum theory non-compositional. 

\subsection{The Principle of Locality}

Constructor theory takes seriously Einstein's notion of locality~\cite{einstein1949albert}, which comprises two components as described by Don Howard:

``The first, which I call the `separability principle', asserts that any
two spatially separated systems possess their own separate real states.
The second, the `locality principle' asserts that all physical effects are
propagated with finite, subluminal velocities, so that no effects can be
communicated between systems separated by a space-like interval.''~\cite[p.~3]{howard1985einstein}

The separability condition in constructor theory is stronger than the one in the passage quoted above. According to this condition, the whole is separable into independent parts, and these parts when combined make the whole. In other words, `the whole is not more than the union of the parts'~\cite[p.~15]{tibau2023locality}.

Focusing on the separability part of Einstein's notion of locality, Deutsch posits that it is `a necessary condition for tasks to be composable into networks'~\cite[p.~24]{deutsch2013constructor} according to the composition principle~\cite{deutsch2013constructor}. 

In constructor-theoretic terms, separability means that the joint state of any two substrates can always be expressed as an ordered pair of individual states of the two substrates; in Definition~\ref{chsub}, this was imposed as the following requirement on substrate states: $\Gamma_{\substr{H}\otimes\substr{K}} = \Gamma_\substr{H} \times \Gamma_\substr{K}$ for all $\substr{H}, \substr{K} \in \Sigma$.

Einstein locality (read: separability) is uncontroversial in classical physics. However, in the standard formulation of quantum theory, there are entangled states that are non-separable by definition. Deutsch and Hayden came up with a formulation of quantum theory that is arguably consistent with Einstein locality~\cite{deutsch2000information}. The existence of this so-called `local' formulation~\cite{deutsch2000information, deutsch2012vindication} is offered as an evidence of the compatibility of quantum theory with constructor theory. Below, we give a brief overview of this formulation~\cite{deutsch2000information}.

\subsection{Deutsch-Hayden Descriptors}

The Deutsch-Hayden formulation of quantum theory describes qubits not in terms of states but in terms of \textit{descriptors}. Roughly speaking, a descriptor of a qubit is a set of operators that determines all the observables of that system. Since the Deutsch-Hayden descriptors are based on the Heisenberg picture of quantum theory, we first describe how it differs from the more commonly used Schr\"{o}dinger picture. Our exposition follows that of Ref.~\cite{bedard2021abc}.

Consider a pure state $\ket{\psi}$. Such a state could be prepared by applying a unitary operation $U$ to the state $\ket{0}$, where the latter is deemed an initial state. In standard quantum theory, the expectation of measurement outcomes is computed using the formula
\[
\braket{0| U^{\dagger} O U |0}
\]
where $O \equiv \sum_i \lambda_i \ket{\phi_i} \bra{\phi_i} $ represents a general observable. The set of vectors $\{\ket{\phi_i}\}_i$ represents the measurement basis and $\lambda_i$ are the corresponding eigenvalues. The \textit{Schr\"{o}dinger picture} reads the above formula as the evolution of states while the observable does not change: $\left(\bra{0} U^{\dagger} \right) O \left( U \ket{0} \right)$. In other words, the state $\ket{0}$ evolves to the state $U\ket{0}$ while $O$ does not evolve. In contrast, the \textit{Heisenberg picture} views the observable as evolving whereas the state remains fixed: $\braket{0| \left( U^{\dagger} O U \right) |0}$. That is, $O$ changes to $U^{\dagger} O U$ while the state remains fixed to $\ket{0}$. As the state remains fixed, it is called the \textit{reference vector}. 

In short, the Heisenberg picture describes quantum systems not in terms of the evolution of their states but in terms of the evolution of their observables. 

The Deutsch-Hayden approach describes a qubit via a mathematical object that encapsulates information about \textit{all} the evolved observables~\cite{deutsch2000information, bedard2021abc}. This task is made tractable by the fact that observables are linear operators, and linear operators form a vector space. As the evolution of the general observable $O$ to $U^{\dagger} O U$ is linear, one needs to track the evolution of only the basis operators of $O$. That is, if $O = \sum_k a_k A_k$ where $\{A_k\}_k$ are the basis operators of $O$, we have $U^{\dagger} O U = \sum_k a_k U^{\dagger} A_k U $. Therefore, it is sufficient to track evolution of each basis operator $A_k$ to determine the evolution of any observable by $U$.

\subsubsection{Descriptor of a Single Qubit}

Observables of a single qubit are given by $2\times 2$ matrices, for which a suitable basis is provided by the Pauli matrices along with the identity:
\[
\boldsymbol{\sigma} = (\sigma_x, \sigma_y, \sigma_z) = \left(  \left[ {\begin{array}{cc}
		0 & 1 \\
		1 & 0 \\
\end{array} } \right],   \left[ {\begin{array}{cc}
		0 & -i \\
		i & 0 \\
\end{array} } \right],   \left[ {\begin{array}{cc}
		1 & 0 \\
		0 &-1 \\
\end{array} } \right] \right) \ \text{and}\  \sigma_0 = \mathbb{1} =   \left[ {\begin{array}{cc}
		1 & 0 \\
		0 & 1 \\
\end{array} } \right] 
\]
The evolution of $\mathbb{1}$, \emph{i.e.}, $U^{\dagger} \mathbb{1} U = \mathbb{1} $ is trivial and can be ignored. Hence, one needs to keep track of the evolution of $\boldsymbol{\sigma}$ to determine any evolved observable of the single qubit. So, the \textit{descriptor} of the single qubit after unitary evolution $U$ is given by 
\[
\boldsymbol{q} = U^{\dagger} \boldsymbol{\sigma} U  
\]

\begin{example}\cite{bedard2021abc}\label{DHex1}
	Consider a circuit in which a state is initialised as $\ket{0}$ and then subjected to a Hadamard gate:
	\[
	\input{constructors/DH1.tikz}
	\]

	\noindent We consider the state $\ket{0}$ to be fixed and give a description of the quantum system in terms of the evolution of its descriptor. The unitary of the Hadamard is given by $H = \frac{1}{\sqrt{2}} \left[ {\begin{array}{cc}
			1 & 1 \\
			1 & -1 \\
	\end{array} } \right] $. The descriptor is initially $\boldsymbol{q}^0 = \boldsymbol{\sigma} = (\sigma_x, \sigma_y, \sigma_z)$, which after the application of the Hadamard becomes $\boldsymbol{q}^1 = H^{\dagger} \boldsymbol{\sigma} H = H^{\dagger} (\sigma_x, \sigma_y, \sigma_z) H = (\sigma_z, -\sigma_y, \sigma_x) $. The evolution of the observable $\ket{0}\bra{0}$ can be determined using the descriptors. After the Hadamard is applied, the observable becomes
\[
H^{\dagger} \ket{0}\bra{0} H = H^{\dagger}  \left[ {\begin{array}{cc}
		1 & 0 \\
		0 & 0 \\
\end{array} } \right] H = H^{\dagger} \left(\frac{\mathbb{1} + \sigma_z}{2}\right) H = H^{\dagger} \left(\frac{\mathbb{1} + q_{z}^0}{2}\right) H = \frac{\mathbb{1} + q_{z}^1}{2} = \frac{\mathbb{1} + \sigma_x}{2}
\]
The expectation of this observable can be calculated using the reference vector $\ket{0}$ as follows:
\[
\bra{0}H^{\dagger} \ket{0}\bra{0} H \ket{0} = \bra{0} \frac{\mathbb{1} + \sigma_x}{2} \ket{0} = \frac{1}{2}
\]
which is the probability of obtaining the measurement outcome $\ket{0}$.	
\end{example}

\subsubsection{Descriptor of $n$ Qubits}

Consider $n$ qubits prepared in the state $\ket{0}^{\otimes n}$ and subjected to a unitary operation $U$. We are interested in describing the evolution of all the observables of this quantum system. For $n$ qubits, the observables are given by $2^n \times 2^n = 4^n$ dimensional complex matrices. A suitable basis for these is given by products of Pauli operators 
\[
\mathcal{A} \equiv \{\sigma_{\mu_1} \otimes \sigma_{\mu_2} \otimes \cdots \sigma_{\mu_n}\  | \  \mu_i \in \{0, x, y, z\} \},
\]
which are linearly independent and $4^n$ in number.

This implies that by tracking the evolution of each operator in the basis $\mathcal{A}$ from $\sigma_{\mu_1} \otimes \sigma_{\mu_2} \otimes \cdots \sigma_{\mu_n}$ to $U^{\dagger} \sigma_{\mu_1} \otimes \sigma_{\mu_2} \otimes \cdots \sigma_{\mu_n} U$, one can determine the evolution of any observable.

Instead of tracking all the $4^n$ basis observables, one may track the evolution of the following set of observables
\[
\boldsymbol{q}_{i}^0 = \mathbb{1}^{i-1} \otimes \boldsymbol{\sigma} \otimes \mathbb{1}^{n-i}, \quad \quad \quad i=1, \cdots, n, 
\]
where $\mathbb{1}^{k}$ represents the tensor product of $k$ copies of the $2\times2$ identity matrix. $\boldsymbol{q}_{i}^0$ has three components for each $i$. These $3n$ observables defined above can be multiplied to obtain any of the $4^n$ basis observables~\cite{tibau2023locality, bedard2021abc}. $\boldsymbol{q}_{i}^0$ is the \textit{descriptor} of qubit $i$ at time $0$. The time-evolved descriptor at time $t$ after unitary evolution is given by
\[
\boldsymbol{q}_{i}^t =  U^{\dagger} \boldsymbol{q}_{i}^0 U
\]
The $n$-tuple with $\boldsymbol{q}_{i}^0$ as components is denoted as $\boldsymbol{q}^0$. It is the joint descriptor of $n$ qubits. The operators in $\boldsymbol{q}^0$ satisfy the Lie algebra $\mathfrak{su}(2)^{\otimes n}$:
\[
[q_{iw}^0, q_{jw'}^0] = 0  \quad \quad ( i \neq j\  \text{and}\ \forall w, w' \in \{0, x, y, z\}  )\]
\[
q_{ix}^0 q_{iy}^0 = iq_{iz}^0 \quad \quad \text{(and its cyclic permutations)} \]
\[
(q_{iw}^0)^2 = \mathbb{1} \quad \quad (\forall w \in \{0, x, y, z\} )
\]

\noindent which is preserved by unitary evolution. 

These algebraic relations imply that for any $i$, there is a redundancy in the observables. That is, to completely determine any triple of observables $(q_{ix}^0, q_{iy}^0, q_{iz}^0)$, one needs only two components as the third can be calculated from the other two using the algebraic relations. Therefore, one may avoid tracking the $y$ component of each descriptor. Hence, one needs to track evolution of $2n$ observables to completely determine the evolution of an $n$-qubit system.

To summarise, for an $n$-qubit system, the Deutsch-Hayden approach focuses on tracking the evolution of \textit{all} the observables from $O$ to $U^{\dagger} O U$. This is achieved by tracking the evolution of $2n$ basis observables (each of dimension $2^n \times 2^n$) from $\boldsymbol{q}^0$ to $\boldsymbol{q}^t = U^{\dagger} \boldsymbol{q}^0 U$, from which the $4^n$ observables of the basis $\mathcal{A}$ can be determined. The $4^n$ observables can, in turn, be used to determine the evolution of any observable.

\begin{example}\label{DHex2}
	Consider again the circuit of Example~\ref{DHex1}
	\[
	\input{constructors/DH1.tikz}
	\]
	where a state is initialised as $\ket{0}$ and then subjected to a Hadamard gate.

	Based on the foregoing discussion, after fixing $\ket{0}$ as the reference vector, one needs to track evolution of only two basis observables $(\sigma_x, \sigma_z)$ to fully characterise the circuit. In other words, $\boldsymbol{q}^0 = (\sigma_x, \sigma_z)$ represents the descriptor for the single qubit before the application of the Hadamard gate.

After the Hadamard is applied, the descriptor becomes $\boldsymbol{q}^1 = H^{\dagger} (\sigma_x, \sigma_z) H = (\sigma_z, \sigma_x) $.
\end{example}

\begin{example}\label{DHex3}
Consider the following circuit, which is usually used to create the Bell state $\frac{\ket{00} + \ket{11}}{\sqrt{2}}$: 
\[
\input{constructors/DH2.tikz}
\]
In this circuit, two qubits are initialised in the $\ket{00}$ state. The first qubit is subjected to a Hadamard gate. This is followed by an application of a CNOT gate to the two qubits, where the first qubit acts as the control and the second the target. 

We fix $\ket{00}$ as the reference vector, and track evolution of two descriptors initially given by $\boldsymbol{q}_{1}^0 = (q_{1x}, q_{1z}) = (\sigma_{x} \otimes \mathbb{1}, \sigma_{z} \otimes \mathbb{1})$ and $\boldsymbol{q}_{2}^0 = (q_{2x}, q_{2z}) =(\mathbb{1} \otimes \sigma_{x}, \mathbb{1} \otimes \sigma_{z})$. The joint description of the two systems is given by the ordered pair of the descriptors of the two qubits $\boldsymbol{q}_{1}^0$ and $\boldsymbol{q}_{2}^0$, \emph{i.e.,} $\boldsymbol{q}^0 = (\boldsymbol{q}_{1}^0, \boldsymbol{q}_{2}^0)$. After the application of the Hadamard but before that of the CNOT, the descriptors are given by $\boldsymbol{q}_{1}^1 =  (H \otimes \mathbb{1})^{\dagger} (\sigma_{x} \otimes \mathbb{1}, \sigma_{z} \otimes \mathbb{1}) (H \otimes \mathbb{1}) = (\sigma_{z}\otimes \mathbb{1}, \sigma_{x} \otimes \mathbb{1}) = (q_{1z}, q_{1x})$ and  $\boldsymbol{q}_{2}^1 =  (H \otimes \mathbb{1})^{\dagger} (\mathbb{1} \otimes \sigma_{x}, \mathbb{1} \otimes \sigma_{z}) (H \otimes \mathbb{1}) = (\mathbb{1} \otimes \sigma_{x}, \mathbb{1} \otimes \sigma_{z}) = (q_{2x}, q_{2z})$. As expected, there is no change in the descriptor of the second qubit because the Hadamard acts on the first qubit only. At this point, the joint description of the two systems is given by $\boldsymbol{q}^1 = (\boldsymbol{q}_{1}^1, \boldsymbol{q}_{2}^1)$. After the application of the CNOT gate, the unitary of which is $U_{\text{CNOT}} = \left[ {\begin{array}{cccc}
		1 & 0 & 0 & 0 \\
		0 & 1 & 0 & 0 \\
		0 & 0 & 0 & 1 \\
		0 & 0 & 1 & 0 \\
\end{array} } \right]$, the qubit descriptors are given by $\boldsymbol{q}_{1}^2 =  (H \otimes \mathbb{1})^{\dagger}  U_{\text{CNOT}}^{\dagger}  (\sigma_{x} \otimes \mathbb{1}, \sigma_{z} \otimes \mathbb{1}) U_{\text{CNOT}} (H \otimes \mathbb{1}) = (\sigma_{z} \otimes \sigma_{x}, \sigma_{x} \otimes \mathbb{1}) = (q_{1z} q_{2x}, q_{1x}) $ and $\boldsymbol{q}_{2}^2 =  (H \otimes \mathbb{1})^{\dagger}  U_{\text{CNOT}}^{\dagger}  (\mathbb{1} \otimes \sigma_{x}, \mathbb{1} \otimes \sigma_{z}) U_{\text{CNOT}} (H \otimes \mathbb{1}) = (\mathbb{1} \otimes \sigma_{x}, \sigma_{x} \otimes \sigma_{z}) = (q_{2x}, q_{1x} q_{2z}) $. The joint descriptor of the two-qubit system is $\boldsymbol{q}^2 = (\boldsymbol{q}_{1}^2, \boldsymbol{q}_{2}^2)$.
\end{example}


\subsection{Locality and Compositionality}\label{loccomp}

Here, we see how locality as implemented in the Deutsch-Hayden approach interacts with the composition of circuits. 

\begin{remark}\label{DHrem1}
In Examples~\ref{DHex2} and~\ref{DHex3}, we characterised qubits in terms of their descriptors. Example~\ref{DHex2} involved a single qubit and thus required only one descriptor, the unitary evolution of which was tracked. Example~\ref{DHex3} was about two qubits and hence there were two descriptors to be tracked. Moreover, the descriptor of the joint two-qubit system was obtained by taking the ordered pair of the individual descriptors of the two qubits. 

Consider the parallel composition of the circuits from Examples~\ref{DHex2} and~\ref{DHex3}. 
\[
\input{constructors/DH3b.tikz}
\]
If we want to characterise this circuit in terms of qubit descriptors, we cannot use those from Examples~\ref{DHex2} and~\ref{DHex3}. The above circuit is a three-qubit system and requires that all the descriptors be of dimension $2^3 \times 2^3$ whereas the descriptors in Examples~\ref{DHex2} and~\ref{DHex3} are of dimensions $2\times2$ and $2^2 \times 2^2$ respectively.

More explicitly, from Example~\ref{DHex2}, we have the evolved descriptor of the first qubit after the application of the Hadamard gate:  $\boldsymbol{q}_{1}^2 = (\sigma_z, \sigma_x)$.  From Example~\ref{DHex3}, we have the descriptor of the joint system comprising the second and third qubits after the application of the Hadamard and CNOT gates:  $(\boldsymbol{q}_{2}^2, \boldsymbol{q}_{3}^2) = ((\sigma_{z} \otimes \sigma_{x}, \sigma_{x} \otimes \mathbb{1}), (\mathbb{1} \otimes \sigma_{x}, \sigma_{x} \otimes \sigma_{z}) )$. The descriptor of the joint system comprising the three qubits after the application of all the unitary gates is \textit{not} given by the tuple $(\boldsymbol{q}_{1}^2, \boldsymbol{q}_{2}^2, \boldsymbol{q}_{3}^2) = ((\sigma_z, \sigma_x), (\sigma_{z} \otimes \sigma_{x}, \sigma_{x} \otimes \mathbb{1}), (\mathbb{1} \otimes \sigma_{x}, \sigma_{x} \otimes \sigma_{z}) ) $. Instead, it is given by $((\sigma_z \otimes \mathbb{1} \otimes \mathbb{1}, \sigma_x \otimes \mathbb{1} \otimes \mathbb{1}), (\mathbb{1} \otimes \sigma_{z} \otimes \sigma_{x},\mathbb{1} \otimes  \sigma_{x} \otimes \mathbb{1}), (\mathbb{1} \otimes \mathbb{1} \otimes \sigma_{x}, \mathbb{1} \otimes \sigma_{x} \otimes \sigma_{z}))$, which is not the ordered tuple of the descriptors from Examples~\ref{DHex2} and~\ref{DHex3}. In this sense, the descriptors of the Deutsch-Hayden approach are not compositional.
\end{remark}

Example~\ref{DHex3} and Remark~\ref{DHrem1} show that the descriptors of two systems can be composed using the ordered tuple only if they are defined with reference to the same circuit. If two circuits, the descriptors of which we have individually tracked, are composed in parallel, the ordered tuple of their descriptors does not give the descriptor of the overall circuit. In fact, as soon as the two circuits are put together in parallel, the old descriptors no longer apply to the individual circuits. Instead, one needs to define and track new descriptors. This shows that the descriptors are not compositional.

\begin{remark}
The presence of another circuit, or even just a qubit, changes the descriptors of the original circuit even if there is no interaction between the two circuits. In this sense, the Deutsch-Hayden descriptors are not \textit{local}. 
\end{remark}

\begin{remark}
	Two key principles of constructor theory are the composition principle and the principle of locality~\cite{deutsch2013constructor, deutsch2015constructor}. The composition principle states that every regular network of possible tasks is a possible task, whereas the principle of locality states that the complete description of the system is given by the ordered tuple of the description of its parts. For the latter principle to be applicable to quantum theory, the Deutsch-Hayden formulation~\cite{deutsch2000information, deutsch2012vindication} is quoted as a supporting framework. Our analysis in Example~\ref{DHex3} and Remark~\ref{DHrem1} has shown that there is a clash between the two principles if the Deutsch-Hayden approach is used. In quantum theory, preparing qubits in $\ket{0}$ states and applying unitary operations on them, like we did in Examples~\ref{DHex2} and~\ref{DHex3}, are possible tasks. The parallel composition of these tasks, like the one in Remark~\ref{DHrem1}, is a regular network and, hence, a possible task according to the composition principle.
	However, the example in Remark~\ref{DHrem1} shows that the locality principle is violated---the descriptor of the joint system is not obtainable by taking the ordered tuple of the descriptors of its parts. 
		\end{remark}
	
		\begin{remark}
	Another way to describe the tension between composition and locality is as follows. Preparing circuits like those in Examples~\ref{DHex2} and~\ref{DHex3} are possible tasks. The `states' of these tasks are given by the descriptors. The principle of locality states that the ordered tuple of these `states' gives the `state' of the joint task obtained by composing the individual tasks in parallel. However, taking the ordered tuple of the `states' does not result in any legitimate `state'.  
\end{remark}

The foregoing discussion shows that constructor theorists need to reject either of the two key principles, namely locality or composition, or find a completely local formulation of quantum theory that is compatible with the principle of composition. Otherwise, the claim that constructor theory is a deeper or more fundamental theory than quantum theory does not hold.

\section{Constructors in Categorical Quantum Mechanics}\label{secconstCQM}

The central idea in constructor theory is that of \textit{possibility}, understood as what \textit{could} happen. Constructor theory and process theories are alike in the sense that they have a common lineage of counter-factual reasoning, which can be traced back to the distinction between `actuality' and `potentiality' by Aristotle~\cite{cohen2000aristotle}.

In constructor theory, theories are characterised by defining the dichotomy between possible and impossible tasks. The impossibility of cloning a general state has been suggested to yield quantum theory, at least in a broad sense~\cite{deutsch2015constructor}. In CQM~\cite{CKbook, CGbook, abramsky2004categorical}, back in at least 2006, classicality was indeed defined by the ability to clone~\cite{coecke2007quantum}. This led to the development of spiders~\cite{coecke2008classical} and the ZX-calculus~\cite{coecke2008interacting, coecke2011interacting}. As discussed in Chapter~\ref{chapinterpretations1}, the latter is now a prominent formalism in quantum computing and technologies\footnote{The website \url{https://zxcalculus.com/publications.html} lists 300+ publications related to the ZX-calculus in some way.} as well as in quantum science education~\cite{dundar2023quantum, dundarcoecke2025makingquantumworldaccessible}.

When it is formulated as a concrete SMC, a process theory is about possible and impossible processes that obey the axioms of the corresponding category. The reconstructions of quantum theory in terms of process theories convert these categorical axioms into physical postulates that are considered reasonable~\cite{hardy2011foliable, selby2021reconstructing}.

In the previous section, we showed that the Deutsch-Hayden formulation of quantum theory obeys Einstein locality at the expense of compositionality, thereby weakening the claim to fundamentality of constructor theory. Dispensing with the principle of locality in constructor theory, one can formulate quantum theory in constructor-theoretic terms. In this section, we show that CQM achieves just that. Here, we conceive CQM as a fully compositional theory of possible tasks in which parallel composition is given by the Kronecker product. We present examples and constructor-theoretic explanations of possible tasks in quantum theory and quantum computation within the diagrammatic language of CQM. Some of these examples correspond to those found in the literature of constructor theory.

\begin{convention}
The notation for CQM/ZX-diagrams in this section follows that introduced in Chapter~\ref{chapinterpretations1}. In diagrammatic equations and calculations, all equalities are up to non-zero global scalars. 
\end{convention}

\begin{remark}
	As discussed in Section~\ref{attvsst}, state-based and attribute-based approaches to constructor theory are equally expressive. Here, we stick to the state-based perspective.
\end{remark}

\subsection{Tasks with Trivial Constructors}

We start with examples that involve trivial constructors. In Ref.~\cite{deutsch2015constructor}, swapping two states is mentioned as an example of a reversible computation task.

\begin{example}\label{exswap}
	Swapping two quantum states is a possible task. Such a task comes for free in a process theory such as CQM. It is represented by 
	\[
	\input{constructors/CQMswap.tikz}
	\]
	Applying this task to any two states $\widehat{\psi}$ and $\widehat{\phi}$ gives
	\[
	\input{constructors/CQMswap2.tikz} \ = \ 	\input{constructors/CQMswap3.tikz} 
	\]
	This task requires two substrates that are the two quantum systems whose states are to be swapped, but does not require any explicit constructor. In other words, this task is made possible by the trivial constructor $I$.	
\end{example}

\begin{example}\label{nottask}
	Transforming a $\widehat 0$ state into a $\widehat 1$ state or vice versa is a possible task, performed by a NOT gate, represented by
	\[ \begin{tikzpicture}
	\begin{pgfonlayer}{nodelayer}
		\node [style=none] (15) at (0, 1.25) {};
		\node [style=none] (16) at (0, -1.25) {};
		\node [style=none] (22) at (0, 0) {};
		\node [style=gray ddot] (23) at (0, 0) {$\ \pi \ $};
	\end{pgfonlayer}
	\begin{pgfonlayer}{edgelayer}
		\draw [style=boldedge] (16.center) to (15.center);
	\end{pgfonlayer}
\end{tikzpicture}
\]
	This task requires no explicit constructor. That is, this task is made possible by the trivial constructor $I$.	
\end{example}

The Deutsch-Hayden approach~\cite{deutsch2000information} was originally introduced to provide a local formulation of scenarios involving entangled states, such as the Einstein-Podolsky-Rosen (EPR) experiment and the quantum teleportation protocol. In these scenarios, a Bell state is generated by preparing two systems, each in the $\widehat 0$ state, and applying the Hadamard and CNOT gates to them. It is another example of a possible task that does not require an explicit constructor. 

\begin{example}\label{exconsbell}
	The task of transforming two states $\widehat 0$ and $\widehat 0$ into the Bell state is given by 
\begin{equation}\label{belltask}
	\input{constructors/CQMbell1.tikz} 
\end{equation}
where the first system undergoes a Hadamard transformation. This is followed by the application of a CNOT gate to the two systems. The first system is the control whereas the second is the target. Supplying the input states to this task, we get
	\[
\input{constructors/CQMbell2.tikz} \ \overset{(h)}{=}  \  \input{constructors/CQMbell3.tikz}  \ \overset{(f)}{=}  \ \input{constructors/CQMbell4.tikz}  \ \overset{(id)}{=}  \ \begin{tikzpicture}
	\begin{pgfonlayer}{nodelayer}
		\node [style=none] (0) at (-1.25, 1) {};
		\node [style=none] (1) at (1.25, 1) {};
	\end{pgfonlayer}
	\begin{pgfonlayer}{edgelayer}
		\draw [style=boldedge, in=-90, out=-90, looseness=2.25] (0.center) to (1.center);
	\end{pgfonlayer}
\end{tikzpicture}
 
\]	
In CQM, the Bell state is non-separable and hence CQM does not obey the principle of locality. However CQM is a completely compositional theory, unlike the Deutsch-Hayden formulation of quantum theory~\cite{deutsch2000information}. This aspect shall be exemplified and further explained later. 
\end{example}

\begin{remark}
Process (\ref{belltask}) is itself a composite process. It is obtained by the parallel composition of the Hadamard and identity gates, followed by a serial composition with the CNOT gate. 
\end{remark}

\subsection{Tasks with Explicit Constructors}

In Ref.~\cite{violaris2022irreversibility}, there is an example of a possible task requiring an explicit constructor, which we revisit using CQM below.

\begin{example}\label{expconst1}
	Consider the task of transforming a $\widehat 0$ state into a $\widehat 1$ state using a CNOT gate. This task is made possible by the state $\widehat 1$ which acts as a constructor.
	
	\noindent Here, the CNOT gate acts as a task-inducing process.
	\[
	\input{constructors/CQMcnot.tikz}
	\]
	The first input is the target and takes the substrate state to be transformed. The second input is the control and takes the constructor state. The task is obtained from the task-inducing process by requiring that the input constructor state is $\widehat 1$ and discarding the constructor output. This is given by the diagram
	 \begin{equation}\label{cnottask}
	 \input{constructors/CQMcnot2.tikz}
	 \end{equation}
	 In agreement with the conditions of a possible task according to Definition~\ref{defpossible}, it can be verified that this task actually transforms a substrate state from $\widehat 0$ into $\widehat 1$
\[
\input{constructors/CQMcnot3.tikz} \ \overset{(\ref{picopyrule})}{=}  \ \input{constructors/CQMcnot4.tikz} 
\ \overset{(f)}{=}  \  \begin{tikzpicture}
	\begin{pgfonlayer}{nodelayer}
		\node [style=none] (15) at (0, 1.25) {};
		\node [style=none] (16) at (0, -0.75) {};
		\node [style=none] (22) at (0, -0.75) {};
		\node [style=gray ddot] (23) at (0, -0.75) {$\ \pi \ $};
	\end{pgfonlayer}
	\begin{pgfonlayer}{edgelayer}
		\draw [style=boldedge] (16.center) to (15.center);
	\end{pgfonlayer}
\end{tikzpicture}
 
\]	 
and that the constructor state is preserved by the induced task:
	 \[
\input{constructors/CQMcnot8.tikz} \ \overset{(f)}{=}  \ \input{constructors/CQMcnot9.tikz}  \ \overset{(id)}{=}  \ \input{constructors/CQMcnot10.tikz} \ \overset{(f)}{=}  \ \ \ \ \ \ \input{constructors/CQMcnot11.tikz} \ \overset{(f)}{=}  \ \ \ \ \ \ \input{constructors/CQMcnot12.tikz} \ \overset{(\ref{picopyrule})}{=}  \ \ \ \ \ \ \input{constructors/CQMcnot13.tikz} \ \overset{(\ref{picopyrule})}{=} \ \ \ \ \ \  \begin{tikzpicture}
	\begin{pgfonlayer}{nodelayer}
		\node [style=none] (15) at (0, 1.25) {};
		\node [style=none] (16) at (0, -0.75) {};
		\node [style=none] (22) at (0, -0.75) {};
		\node [style=gray ddot] (23) at (0, -0.75) {$\ \pi \ $};
	\end{pgfonlayer}
	\begin{pgfonlayer}{edgelayer}
		\draw [style=boldedge] (16.center) to (15.center);
	\end{pgfonlayer}
\end{tikzpicture}
 
\]
\end{example} 

\begin{remark}
	Task (\ref{cnottask}) can be simplified as follows
	\[
		\input{constructors/CQMcnot2.tikz}  \ \overset{(\ref{picopyrule})}{=} \ \input{constructors/CQMcnot4b.tikz} 
		\ \overset{(f)}{=} \ \begin{tikzpicture}
	\begin{pgfonlayer}{nodelayer}
		\node [style=none] (15) at (0, 1.25) {};
		\node [style=none] (16) at (0, -1.25) {};
		\node [style=none] (22) at (0, 0) {};
		\node [style=gray ddot] (23) at (0, 0) {$\ \pi \ $};
	\end{pgfonlayer}
	\begin{pgfonlayer}{edgelayer}
		\draw [style=boldedge] (16.center) to (15.center);
	\end{pgfonlayer}
\end{tikzpicture}

	\]
	which gives the NOT gate, discussed in Example~\ref{nottask}. This shows that the task of transforming $\widehat{0}$ to $\widehat{1}$ using a CNOT gate requires an explicit constructor while the same task does not require any explicit constructor when one uses a NOT gate. Once the constructor is supplied to a CNOT gate at the input and the constructor output is discarded, we get a NOT gate. 	
\end{remark}

\begin{example}\label{expconst2}
	Consider the task of transforming a $\widehat{0}$ state into the maximally mixed state. This task can be performed by a CNOT gate 
	\[
	\input{constructors/CQMcnotb.tikz}
	\]
	and is made possible by a maximally mixed state acting as the constructor. Here, the CNOT gate is the task-inducing process and the required task is obtained by requiring that the input constructor state is the maximally mixed state and discarding the constructor output. This is given by the diagram
	 	\[
	 \input{constructors/CQMcnotb2.tikz}
	 \]

	 The conditions of a possible task according to Definition~\ref{defpossible} can be verified. We plug the state $\widehat{0}$ at the task input and simplify
	 	\[
\input{constructors/CQMcnotb3.tikz} \ \overset{(f)}{=}  \ \input{constructors/CQMcnotb4.tikz}  \ \overset{(id)}{=}  \ \input{constructors/CQMcnotb5.tikz} \ \overset{(f)}{=}  \ 
\begin{tikzpicture}
	\begin{pgfonlayer}{nodelayer}
		\node [style=none] (18) at (0, 1) {};
		\node [style=none] (19) at (0, -0.75) {};
		\node [style=none] (21) at (0, -0.75) {};
		\node [style=downground] (24) at (0, -1) {};
	\end{pgfonlayer}
	\begin{pgfonlayer}{edgelayer}
		\draw [style=boldedge, in=-90, out=90] (19.center) to (18.center);
	\end{pgfonlayer}
\end{tikzpicture}
 
\]	 
confirming that the task performs the required transformation.

To check whether the constructor is preserved by the task-inducing process, we have
	 	\[
\input{constructors/CQMcnotb9.tikz} \ \overset{(f)}{=}  \ \input{constructors/CQMcnotb10.tikz}  \ \overset{(id)}{=}  \ \input{constructors/CQMcnotb11.tikz}  \ \overset{(f)}{=}  \ \ \ \ \ \  
\begin{tikzpicture}
	\begin{pgfonlayer}{nodelayer}
		\node [style=none] (15) at (0, 1) {};
		\node [style=none] (16) at (0, -0.75) {};
		\node [style=downground] (17) at (0, -1) {};
	\end{pgfonlayer}
	\begin{pgfonlayer}{edgelayer}
		\draw [style=boldedge] (16.center) to (15.center);
	\end{pgfonlayer}
\end{tikzpicture}
  
\]	 
This shows that the constructor is preserved by the process, and hence the task under discussion is a possible task.
\end{example}

\subsection{Composition of Tasks}

According to the principle of composition~\cite{deutsch2013constructor, deutsch2015constructor}, series and parallel compositions of possible tasks result in possible tasks. We provide examples of series and parallel composition below.

\begin{example}
	Consider two possible tasks: (a) that of transforming a $\widehat{1}$ state into $\widehat{0}$ using a CNOT gate. This task is similar to that of Example~\ref{expconst1}, except that the input state to be transformed is $\widehat{1}$ instead of $\widehat{0}$; and (b) the task from Example~\ref{expconst2}, which transforms $\widehat{0}$ into the maximally mixed state.
	
	We are interested in the serial composition of task (a) followed by task (b). The serial composition of the corresponding task-inducing processes is given by
	\[
	\input{constructors/CQMseries.tikz}  
	\]
	The composite task is obtained by supplying the required constructor inputs and discarding the constructor outputs.
	\[
\input{constructors/CQMseries2.tikz}  
\]	
Applying the state $\widehat{1}$ as the input to this task, we have
	\[
\input{constructors/CQMseries3.tikz}  \ \overset{(f)}{=}  \ \input{constructors/CQMseries4.tikz} \ \overset{(\ref{picopyrule})}{=}  \ \input{constructors/CQMseries5.tikz}    
 \ \overset{(f)}{=}  \ \input{constructors/CQMseries10.tikz}    
\]	
\[
\overset{(id)}{=}  \ \input{constructors/CQMseries11.tikz} \ \ \ \ \   
\ \overset{(f)}{=} \ \input{constructors/CQMseries13.tikz}       
\]		
which shows that the composite task	transforms $\widehat{1}$ into the maximally mixed state.
	
We can also check that the constructor for this task is preserved by the induced task:
\[
\input{constructors/CQMseries14.tikz}    \ \overset{(\ref{picopyrule})}{=}  \ \input{constructors/CQMseries15.tikz}    
 \overset{(f)}{=}  \ \input{constructors/CQMseries18.tikz}    \ \overset{(\ref{spchop})}{=} \ 
\ \ \ \ \ \input{constructors/CQMseries21.tikz}   
\]
In this example, the constructor of the composite process is the tensor product of $\widehat{1}$ and the maximally mixed state.
\end{example}

\begin{example}
	Consider the parallel composition of tasks in Examples~\ref{expconst1} and~\ref{expconst2}. The parallel composition of the corresponding task-inducing processes is given by 
		\[
	\input{constructors/CQMparallel.tikz}  
	\]
	The composite task is obtained by supplying the required constructor inputs and discarding the constructor outputs.
	\[
	\input{constructors/CQMparallel2.tikz}  
	\]	
To verify that the task performs the desired transformation, we plug the states $\widehat{0}$ and  $\widehat{0}$ at the two inputs of the composite task and simplify the diagram as follows:
	\[
\input{constructors/CQMparallel3.tikz}  \  \overset{(f)}{=}  \  \input{constructors/CQMparallel4.tikz}   \  \overset{(\ref{picopyrule})}{=}  \ \input{constructors/CQMparallel5.tikz}  
\]		
	\[
 \  \overset{(f)}{=}  \ \input{constructors/CQMparallel8.tikz}  
 \overset{(id)}{=}  \ \input{constructors/CQMparallel9.tikz}  \  \overset{(f)}{=}  \  
  \input{constructors/CQMparallel12.tikz}
\]		
which is the desired output: the tensor product of $\widehat{1}$ and the maximally mixed state. To verify that the constructor for this task is preserved by the induced task, we apply maximally mixed states to the two inputs of the task and discard the task outputs.
	\[
\input{constructors/CQMparallel13.tikz}  \  \overset{(f)}{=}  \  \input{constructors/CQMparallel14.tikz}   \  \overset{(\ref{picopyrule})}{=}  \ \input{constructors/CQMparallel15.tikz}  
\]		
	\[
 \overset{(f)}{=}  \ \ \ \ \ \ \input{constructors/CQMparallel17.tikz}   \  \overset{(\ref{spchop})}{=}  \ 
\ \ \ \ \ \ \ \ \ \  \input{constructors/CQMparallel12.tikz}
\]
The simplified diagram is the tensor product of $\widehat{1}$ and the maximally mixed state---the constructor for this task. Hence, the constructor is preserved. 	
\end{example}

Tasks with explicit constructors can be composed with tasks with none. We give an example of one such parallel composition below.

\begin{example}
	Consider a parallel composition of tasks in Examples~\ref{exconsbell} and~\ref{expconst1}. The task-inducing process for this composite task is given by
\[
\input{constructors/CQMparallelb0.tikz}  
\]	
The constructor for this task is the tensor product of the trivial constructor $I$ and $\widehat{1}$ (\emph{i.e.}, $I \otimes \widehat{1} = \widehat{1}$). The composite task is obtained by supplying the required state to the constructor input and discarding the constructor output. This is given by the parallel composition of diagrams (\ref{belltask}) and (\ref{cnottask}):
	\[
	\input{constructors/CQMparallelb.tikz}  
	\]
Supplying states $\widehat{0}$, $\widehat{0}$, and $\widehat{0}$ to the three inputs, we have
	\[
\input{constructors/CQMparallelb2.tikz}  \  \overset{(f)}{=}  \  \input{constructors/CQMparallelb3.tikz} \  \overset{(h)}{=}  \ \input{constructors/CQMparallelb4.tikz}
\]
	\[
 \overset{(f)}{=}  \ \input{constructors/CQMparallelb5.tikz}  \  \overset{(\ref{picopyrule})}{=}  \  \input{constructors/CQMparallelb6.tikz} \ 
 \overset{(f)}{=}  \ \input{constructors/CQMparallelb8.tikz}
\]
which shows that the task transforms the states $\widehat{0}$, $\widehat{0}$, and $\widehat{0}$ into the tensor product of the Bell state and $\widehat{1}$.
\end{example}

\begin{remark}
The output of the task discussed in the above example is a parallel composition of the Bell state and $\widehat{1}$:
\begin{equation}\label{remlocal}
\input{constructors/CQMparallelb8.tikz}
\end{equation}
If one subscribes to the principle of locality in constructor theory~\cite{deutsch2013constructor, deutsch2015constructor}, the Bell state must be expressible in terms of two separable parts. Using the Deutsch-Hayden formalism~\cite{deutsch2000information} (which has been quoted in favour of the principle of locality~\cite{deutsch2015constructor}), one can indeed find such an expression for the Bell state. However, as soon as another quantum state is introduced in parallel to the Bell state (such as in diagram (\ref{remlocal})), the old Deutsch-Hayden descriptors cannot be used anymore, and one needs to use new Deutsch-Hayden descriptors for the Bell state. In this sense, the so-called `local' description of the Bell state is not compositional. In another sense, it is also non-local: the local descriptors of a system are dependent on the presence of other quantum systems which are not part of the local system to be described. We discussed this in detail in Section~\ref{loccomp}.

So, if we desire a compositionally sound constructor theory of quantum physics and quantum information/computation, we cannot use the Deutsch-Hayden formalism for quantum theory as it clashes with the principle of composition~\cite{deutsch2013constructor, deutsch2015constructor}. In contrast, CQM offers a constructor theory of quantum physics but one which rejects the principle of locality, as it includes non-separable states such as the Bell state as legitimate inputs and outputs of possible tasks.
\end{remark}

\subsection{Reversibility of Tasks}

The tasks described in Examples~\ref{exswap},~\ref{nottask},~\ref{exconsbell}, and~\ref{expconst1} are reversible from a constructor-theoretic point of view. That is, the reverse of these tasks (transforming outputs of the original tasks into the corresponding inputs) are made possible by the constructors of the original tasks. Tasks in Examples~\ref{exswap},~\ref{nottask}, and~\ref{exconsbell} are trivially reversible as the processes involved are reversible and do not require explicit constructors. The corresponding reverse tasks are obtained by vertically flipping the diagrams. We provide two examples of reversible tasks below: one with an explicit constructor and one with a trivial constructor.

\begin{example}
Consider Example~\ref{expconst1}, in which a task transforms a state $\widehat{0}$ into $\widehat{1}$. Its reverse task would transform $\widehat{1}$ into $\widehat{0}$, and is made possible by the same constructor as the one that enabled the original task. It is obtained by vertically flipping the task-inducing process (which changes inputs into outputs and vice versa) and then plugging the constructor state at the constructor input and discarding the constructor output. 
\[
	 \input{constructors/CQMcnot2.tikz}
\]
It can be verified that the above process achieves the desired task.
\end{example}

\begin{example}
	Consider Example~\ref{exconsbell} in which a task transforms the tensor product of two states $\widehat 0$ and $\widehat 0$ into the Bell state. This task is made possible by the trivial constructor $I$. Its reverse task would transform the Bell state into the tensor product of $\widehat 0$ and $\widehat 0$, and is enabled by the same constructor $I$. This is obtained by vertically flipping the task-inducing process.
	\[
	\input{constructors/CQMbell1b.tikz} 
	\]
	We verify that this process performs the desired reverse task:
		\[
	\input{constructors/CQMbell1b2.tikz} \ \overset{(id)}{=}  \  \input{constructors/CQMbell1b3.tikz}  \ \overset{(f)}{=}  \ \input{constructors/CQMbell1b4.tikz}  \ \overset{(\ref{spchop})}{=}  \ \input{constructors/CQMbell1b5.tikz} \ \overset{(h)}{=}  \ \input{constructors/CQMbell1b6.tikz} 
	\]	 
\end{example}

A constructor for a task is not necessarily a constructor for its reverse task. This is called \textit{constructor-based irreversibility}~\cite{marletto2022emergence, violaris2022irreversibility}. We provide an example of constructor-based irreversibility below.

\begin{example}
 In Example~\ref{expconst2}, a task transforms $\widehat{0}$ into the maximally mixed state; this task is enabled by a maximally mixed state acting as the constructor. We can check whether the reverse task is made possible by the same constructor. We vertically flip the task-inducing process (which gives the same process as in the original example), plug the constructor state at the constructor input and discard the constructor output. We have
 \[
 \input{constructors/CQMcnotb2.tikz}
 \]
 Applying a maximally mixed state as the input to this process, we have
  \[
 \input{constructors/CQMcnotc.tikz} \  \overset{(f)}{=}  \  
  \input{constructors/CQMcnotc3.tikz}  \  \overset{(\ref{spchop})}{=}  \  
  \begin{tikzpicture}
	\begin{pgfonlayer}{nodelayer}
		\node [style=none] (18) at (0, 1) {};
		\node [style=none] (19) at (0, -0.75) {};
		\node [style=none] (21) at (0, -0.75) {};
		\node [style=downground] (33) at (0, -1) {};
	\end{pgfonlayer}
	\begin{pgfonlayer}{edgelayer}
		\draw [style=boldedge, in=-90, out=90] (19.center) to (18.center);
	\end{pgfonlayer}
\end{tikzpicture}
 
 \]
 which is not the desired output: the state $\widehat{0}$. Therefore, this is an example of constructor-based irreversibility. 
\end{example}

In the above example, the task is not reversible but the task-inducing process is reversible.

\begin{remark}
Note that in Example~\ref{expconst2} and, therefore, the above example as well, the task-inducing process is reversible. If we forget the distinction between constructors and regular substrates, and treat $\widehat{0}$ and the maximally mixed state as merely two inputs of the task-inducing process, we get
\[
 \input{constructors/CQMcnotd.tikz} \  \overset{(f)}{=}  \   \input{constructors/CQMcnotd2.tikz} \  \overset{(id)}{=}  \   \input{constructors/CQMcnotd3.tikz} \  \overset{(f)}{=}  \   \input{constructors/CQMcnotd4.tikz} \  \overset{(f)}{=}  \ \input{constructors/CQMcnotd5.tikz} 
\]
Now, applying this output state as the input to the vertically-flipped task-inducing process, we have
\[
\input{constructors/CQMcnotd6.tikz} \  \overset{(f)}{=}  
\ \input{constructors/CQMcnotd8.tikz} \  \overset{(\ref{spchop})}{=}  \ 
\input{constructors/CQMcnotd10.tikz}
\]
This shows that the output of the reverse process is the input of the original process, confirming that the task-inducing process is reversible.

In summary, constructor-based irreversibility is the idea that even if the task-inducing process is reversible, the corresponding possible task may not be. A possible task is reversible if its constructor also makes its reverse task possible.
\end{remark}

\begin{remark}
A foundational question in physics concerns how irreversibility arises from time-reversal symmetric physical laws. Conventional approaches to reconciling irreversibility with reversibility involve statistical mechanical methods~\cite{andrews1965statistical, wallace2015quantitative}. Constructor theory offers a fresh perspective on the emergence of irreversibility in terms of the possibility and impossibility of tasks. Task-inducing processes correspond to time-reversal symmetric physical laws, but that alone is not sufficient for the reversibility of tasks. If the constructor of a possible task does not also act as a constructor for its reverse task (meaning the latter task is impossible), then the task is irreversible---even if the underlying task-inducing process is reversible. Unlike the statistical approaches to irreversibility, which are restricted to macroscopic physics and rely on approximations, constructor-theoretic irreversibility is scale-independent and exact~\cite{violaris2022irreversibility, marletto2022emergence}.
\end{remark}

\section{Summary and Outlook}\label{sumconst}

In this chapter, we presented categorical semantics for constructor theory. Our focus was on the desired mathematical foundations explained in Deutsch's seminal paper~\cite{deutsch2013constructor}. We took the theory of conceivable tasks to be $\Rel$, and then showed that, for a given choice of substrates, the set of possible tasks forms a sub-SMC of $\Rel$. While the result may not be surprising, it constitutes, to the best of our knowledge, the first process-theoretic formalisation of constructor theory. This formalisation offers a rigorous yet intuitive language to express and explore constructor-theoretic ideas.

We then identified a key inconsistency in the constructor-theoretic formulation of non-relativistic quantum theory: a conflict between the principles of locality and composition. Our analysis of the Deutsch–Hayden approach~\cite{deutsch2000information} demonstrated that locality comes at the cost of compositionality. Finally, we argued that CQM can be considered a viable constructor-theory of quantum physics---one that is fully compositional but non-local---and illustrated this with examples grounded in quantum theory and computation.

The string diagrammatic syntax we employed to formally represent constructor theory can also be interpreted in other SMCs. This provides an opportunity to investigate and explore the ramifications of constructor theory beyond $\Rel$, which has long been the default modelling choice in the constructor theory literature.

Our work suggests promising avenues for future research: namely, developing process-theoretic counterparts of constructor theories of information~\cite{deutsch2015constructor}, thermodynamics~\cite{marletto2016constructorb}, probability~\cite{marletto2016constructor}, and life~\cite{marletto2015constructor}, as well as exploring the connection between constructor theory and resource theories~\cite{coecke2016mathematical}.

\usetikzlibrary{backgrounds}
\usetikzlibrary{arrows}
\usetikzlibrary{shapes,shapes.geometric,shapes.misc}

\tikzstyle{tikzfig}=[baseline=-0.25em,scale=0.5]

\pgfkeys{/tikz/tikzit fill/.initial=0}
\pgfkeys{/tikz/tikzit draw/.initial=0}
\pgfkeys{/tikz/tikzit shape/.initial=0}
\pgfkeys{/tikz/tikzit category/.initial=0}

\pgfdeclarelayer{edgelayer}
\pgfdeclarelayer{nodelayer}
\pgfsetlayers{background,edgelayer,nodelayer,main}

\tikzstyle{none}=[inner sep=0mm]

\tikzstyle{every loop}=[]

\usetikzlibrary{decorations.markings}
\usetikzlibrary{shapes.geometric}

\pgfdeclarelayer{edgelayer}
\pgfdeclarelayer{nodelayer}
\pgfsetlayers{edgelayer,nodelayer,main}

\tikzstyle{none}=[inner sep=0pt]


\tikzstyle{every picture}=[baseline=-0.25em,scale=0.6]
\tikzstyle{dotpic}=[] 



\tikzstyle{red dot}=[fill=red, draw=black, shape=circle]
\tikzstyle{green dot}=[fill={rgb,255: red,0; green,131; blue,0}, draw=black, shape=circle]
\tikzstyle{new style 0}=[fill=white, draw=black, shape=circle]
\tikzstyle{dot}=[fill=white, draw=black, shape=circle, inner sep=1pt]

\chapter{\label{chapwavelogic}Waves, Interference, and Computation}

\epigraph{`Civilization advances by extending the number of important operations which we can perform without thinking about them.'}  
{Alfred North Whitehead, \\ An Introduction to Mathematics~\cite[p.~61]{whitehead1911introduction}}

\textit{This chapter is adapted from the publication~\cite{waseem2023stringwave}, based on joint work with Alexy Karenowska. The author of this thesis is the first author of the publication~\cite{waseem2023stringwave}, and played a leading role in all the work described in this chapter. The novel contributions are presented in Sections~\ref{secwavestring} and~\ref{secappwave}.}

\section{Introduction}

The prevalent paradigm of classical computing technology is based on encoding information in electrical voltages or currents and performing computation by manipulating the movement of charge. The building block of classical computers is the complementary metal-oxide-semiconductor (CMOS) transistor. CMOS transistors can be employed as switches and amplifiers~\cite{boylestad2018electronic} with which one can physically implement logic gates to perform any Boolean operation~\cite{jiang2019theory, uyemura1999cmos}. 

The practical success of conventional computing has relied on the continuing minituarisation of the CMOS technology~\cite{dennard1974design, mahmoud2020introduction}, increasing the packing density of transistors in integrated circuits (ICs). This trend of progressive miniaturisation and the corresponding increase in the number of transistors on ICs has been famously described in an observation~\cite{moore1965cramming, moore1975progress} now known as `Moore's law'. The original observation stated that the number of transistors per circuit chip doubled every year.

Recently, there has been a realisation that Moore's law is not going to hold anymore because of fundamental physical scaling limits~\cite{mahmoud2020introduction, waldrop2016chips, theis2017end}. Therefore, there is interest in novel, `beyond-CMOS' technologies and paradigms~\cite{irds2022}. One such paradigm is wave-based computation, which is based on encoding information in the phase and/or amplitude of waves and performing computation through phase shifting and wave interference. There exist a number of physical platforms for implementing wave-based computation. These include optical computing~\cite{imai1986optical, ambs2010optical, kazanskiy2022optical}, neuromorphic computing~\cite{rahman2015wave}, and spin-wave computing~\cite{mahmoud2020introduction}. 

Wave-based computation is a form of analogue computation, and therefore inherits the main limitations that have long affected analogue systems, such as limited precision, challenges in miniaturisation, poor scalability, restricted programmability, lack of robust memory, limited support for general-purpose computation, and most importantly, noise susceptibility and the unavailability of an effective approach for fault tolerance~\cite{zangeneh2021analogue, maclennan2014promise, ostrove2019improving, belostotski2025survey}. While recent advances in metamaterials~\cite{zangeneh2021analogue} and metasurfaces~\cite{xu2022metasurface} offer prospects for improving scalability and miniaturisation, these developments do not resolve the fundamental issue of noise susceptibility.

In analogue computing, information is represented by continuous physical quantities. These quantities are inherently sensitive to disturbances caused by thermal noise, device mismatch, and fabrication imperfections. Unlike digital computing, which repeatedly restores signals to well-defined states, analogue-only computing cannot remove noise in this way. As a result, small perturbations accumulate as signals propagate and interact, leading to progressively increasing errors over the course of a computation~\cite{zangeneh2021analogue}. Noise therefore remains a fundamental limitation rather than an issue that can be eliminated through engineering alone.

Nevertheless, despite these constraints, wave-based computation remains an interesting and valuable paradigm to study. It offers potential advantages, including parallel computing~\cite{khitun2012multi}, reversible computing~\cite{cuykendall1987reversible, balynskiy2018reversible}, and low energy consumption~\cite{atulasimha2016nanomagnetic, chumak2015magnon}, which may be beneficial for specific computational tasks. Recently, it was experimentally demonstrated that classical wave-based computing (with phase encoding) can execute certain quantum computing algorithms that do not require entanglement~\cite{balynsky2021quantum}. Although this approach cannot compete with quantum computing, it may still provide an efficient alternative to classical digital computing for certain computational tasks~\cite{balynsky2021quantum}.

There have been a number of implementations of Boolean logic gates and circuits using different wave-based computing platforms. However, a general theoretical (syntactic) framework to design, analyse and optimise these circuits has thus far been lacking.

In this chapter, we offer such a syntactic framework. We present a string-diagrammatic formalism for wave-based computation with phase encoding. We motivate the approach and demonstrate its application with reference to spin-wave or `magnonic' circuits. However, the same syntax is applicable to other wave-based technologies that use phase encoding~\cite{imre2006majority, iwamura1998single, vaysset2016toward, li2013three, vacca2014nanomagnet, khokhriakov2022multifunctional, rose2007programmable}. Through the example of spin-wave circuits, we demonstrate the usage of the formalism in designing, analysing and simplifying Boolean logic circuits. 

The remainder of the chapter is structured as follows. Section~\ref{spcomput} provides a brief introduction to spin-wave computing. Section~\ref{secwavestring} presents a string-diagrammatic formalism for spin-wave Boolean circuits. Section~\ref{secappwave} discusses its applications in the design, analysis, and optimisation of wave-based logic circuits. Section~\ref{sumwave} offers concluding remarks.

\section{Spin-wave Computing}\label{spcomput}

Spin wave dynamics (also known as magnonics) is the study of \textit{spin waves}---collective excitations of the electronic spin lattice of magnetic materials~\cite{demokritov2012magnonics, rezende2020fundamentals}. Spin waves behave classically at room temperature but demonstrate quantum behaviour at millikelvin temperatures~\cite{chumak2019magnonics}. The spin-wave quantum is called the \textit{magnon}. Magnonics research is concerned with generating, manipulating and detecting spin waves in order to characterise magnetic materials; or study interactions between magnons, other quasiparticles like photons or phonons, and superconducting qubits~\cite{chumak2015magnon, chumak2019magnonics, lachance2019hybrid, flebus20242024}. Moreover, certain magnonic systems allow for the possibility of propagating spin waves that can be harnessed to design logic devices~\cite{flebus20242024, serga2010yig, khitun2010magnonic, khitun2016magnonic}. 

Magnonic or spin-wave logic devices rely on propagating spin waves, instead of electronic charges, as information carriers~\cite{mahmoud2020introduction, chumak2019magnon}. Spin-wave computing is considered attractive for two reasons. Firstly, it offers generation of spin waves of gigahertz frequencies at nanometer wavelengths, leading to the prospect of miniaturised devices with high clock speeds~\cite{chumak2015magnon}. Secondly, the propagation of spin waves through certain magnetic insulators involves relatively low energy dissipation~\cite{chumak2019magnon}, as compared to the movement of charges in CMOS devices.

There have been multiple implementations of spin-wave-based Boolean logic gates.\footnote{Both Boolean~\cite{khitun2010magnonic, kostylev2005spin, schneider2008realization} and non-Boolean~\cite{papp2017nanoscale, khitun2013magnonic, macia2011spin, papp2014spin} computing have been explored using spin-wave systems. In this chapter, we focus on the Boolean case.} A NOT gate was first experimentally demonstrated in 2005~\cite{kostylev2005spin}. It was soon followed by implementations of XNOR and NAND gates~\cite{schneider2008realization}. As the NAND gate is a universal gate, its realisation opened the possibility of implementing any Boolean logical circuit within the spin-wave platform. In the last few years, designs of more complex spin-wave circuits have also been proposed~\cite{mahmoud2021spinb, mahmoud2020spinb, mahmoud2022would}, including half adders~\cite{wang2019integrated}, full adders~\cite{mahmoud2021spin}, and a 32-bit ripple-carry adder~\cite{garlando2023numerical}.

Wave-based computing, and especially spin-wave computing, offers a very economical implementation of the majority gate~\cite{khitun2010magnonic, fischer2017experimental, zografos2017non, reuben2020rediscovering}. The simplest majority gate has three Boolean inputs and one output. The output is equal to the value which is in majority at the input. In conventional CMOS-based computing, the implementation of a majority gate requires a combination of AND and OR gates, equivalent to $\sim10$ transistors~\cite{chumak2015magnon}. In contrast, (spin-)wave computing offers the majority gate as a primitive logic gate. Additionally, controlling one of the inputs of the majority gate results in an AND gate or an OR gate. Spin-wave computing also offers the XOR gate as a primitive. This can potentially replace the CMOS XOR gate that usually comprises multiple AND, OR and NOT gates, or equivalently $\sim8$ transistors~\cite{chumak2015magnon}. 

Although spin-wave computing shows great promise for energy-efficient and compact logic devices and circuits, several key hurdles remain. Like other forms of analogue computation~\cite{zangeneh2021analogue, maclennan2014promise}, it is inherently susceptible to noise, signal degradation, and accumulated errors, making reliable operation challenging~\cite{mahmoud2020introduction}. Error correction is essential but difficult, particularly because it often involves non-linear processes that can significantly increase energy consumption, thereby undermining the benefits of spin-wave systems~\cite{verba2019correction}. Cascading gates---necessary for constructing complex circuits---requires amplitude normalisation to preserve signal integrity. However, amplitude normalisation itself is an energy-intensive, non-linear process~\cite{dutta2015compact}. While spin-wave devices are amenable to miniaturisation and scaling, a major challenge at smaller dimensions is crosstalk noise~\cite{dutta2015compact}. Efforts are underway to overcome these barriers by integrating spin-wave computing with complementary platforms, including optics, superconducting qubits, and CMOS circuits~\cite{chumak2015magnon, li2020hybrid, lachance2019hybrid}. Having acknowledged these experimental challenges, we now restrict our discussion to the modelling of spin-wave circuits for the remainder of this chapter.

We use string diagrams~\cite{penrose1971applications, joyal1991geometry, SelingerSurvey} to develop a formalism for wave-based Boolean logic circuits, based on the architecture of the spin-wave majority gate. However, the formalism can be applied to other physical platforms and technologies that are based on wave-based computation or majority gate networks~\cite{vemuru2014majority, parhami2020majority, zhang2023survey}, including quantum-dot cellular automata~\cite{imre2006majority, walus2004circuit}, single-electron circuits~\cite{iwamura1998single}, nanomagnet logic~\cite{vacca2014nanomagnet}, graphene spin circuits~\cite{khokhriakov2022multifunctional}, spin torque majority gates~\cite{vaysset2016toward}, molecular scale electronics~\cite{rose2007programmable}, and DNA logic circuits~\cite{li2013three, fan2020propelling}.

In wave-based computation, information is encoded in any measurable property of the wave. The encoding one chooses has a bearing on the physical structures one can use to perform computation, and on the robustness of signals~\cite{mahmoud2020introduction}. In our work, we choose phase encoding as it allows composition of logic gates to form more complex circuits without requiring any components that are not wave-based. That is, phase-encoded wave-based circuits are compositional. Moreover, phase also tends to be favoured over the other common candidate, amplitude, because it is often more robust to noise~\cite{mahmoud2020introduction}. The phase information of the waves is extracted either from time-domain measurements obtained using oscilloscopes with high sampling rates~\cite{kanazawa2017role, fischer2017experimental}, or from frequency-domain measurements performed with vector network analysers~\cite{mahmoud2020introduction}.

To illustrate phase encoding, consider a sinusoidal signal of the form $A \sin(\omega t + \phi)$, where $A$ denotes the amplitude, $\omega$ the angular frequency, $t$ the time, and $\phi$ the phase. Phase $\phi = 0$ can be defined to represent the logical bit $0$ and, likewise, $\phi = \pi$ to represent the logical bit~$1$. As an example, if a computation is performed in which two waves corresponding to the bit $1$ and one wave corresponding to the bit $0$ are superposed, the resultant $A \sin(\omega t +\pi) + A \sin(\omega t +\pi) + A \sin(\omega t) = A \sin(\omega t +\pi)$ corresponds to a logical $1$. Note that each wave has the same amplitude $A$. If the output had a non-zero amplitude which was not $A$, it would be normalised to $A$ before using it as an input to another gate. Amplitude normalisation is performed using devices like directional couplers~\cite{mahmoud2020spinb} that exploit non-linear wave interactions~\cite{mahmoud2020introduction}. 

In order to ensure the compositionality of circuits, situations involving total destructive interference must be prevented. For example, one configuration to avoid is a circuit that superposes the waves $A \sin(\omega t)$ and $A \sin(\omega t + \pi)$, as it produces an output with zero amplitude and an undefined phase, making it unsuitable as an input to another circuit. Essentially, since inputs are normalised and only two phases are allowed in Boolean computation, total destructive interference can arise only if an even number of waves interfere. This is why, in majority gates, it is generally preferable to use an odd number of inputs, ensuring that interference of an even number of waves never occurs~\cite{dutta2017proposal,mahmoud2020introduction}.

In this work, we restrict ourselves to using only three-input majority gates for 
interfering waves. AND and OR gates are derivatives of the majority gate and 
thus also involve interference of three waves---more on this in 
Section~\ref{loggatcirc}. Because input amplitudes are always normalised, 
three-input majority gates preclude total destructive interference.

\section{Wave-based Computation with String Diagrams}\label{secwavestring}

\subsection{Logic Gates and Circuits}\label{loggatcirc}

The basic structures through which waves or signals propagate in wave-based circuits are called \textit{waveguides}. We denote a waveguide string-diagrammatically by a wire:
\begin{equation}\label{wire}
	\scalebox{0.75}{\begin{tikzpicture}
	\begin{pgfonlayer}{nodelayer}
		\node [style=none] (13) at (1.25, 2) {};
		\node [style=none] (15) at (1.25, -2) {};
	\end{pgfonlayer}
	\begin{pgfonlayer}{edgelayer}
		\draw (13.center) to (15.center);
	\end{pgfonlayer}
\end{tikzpicture}
}
\end{equation}
The wave is transmitted through the wire from bottom to top. A wire is essentially an \textit{identity gate}. 

\begin{convention}\label{sigconv}
	In the remainder of the chapter, a wave $\sin(\omega t)$ (with amplitude $A=1$ and phase $\phi =0$) is assumed to enter each circuit from the bottom and exit at the top. If the exiting wave has the phase $\pi$, this corresponds to the logical output $1$; if it is $0$, to the logical output $0$. 
\end{convention}

A waveguide may impose a \textit{phase shift}, say $\theta$, on the wave that is transmitted through it. This is denoted by a wire with the phase shift $\theta$:
\begin{equation}\label{phase}
	\scalebox{0.75}{\begin{tikzpicture}
	\begin{pgfonlayer}{nodelayer}
		\node [style=none] (2) at (-2.75, 2) {};
		\node [style=new style 0] (5) at (-2.75, 0) {$\theta$};
		\node [style=none] (7) at (-2.75, -2) {};
	\end{pgfonlayer}
	\begin{pgfonlayer}{edgelayer}
		\draw (5) to (2.center);
		\draw (7.center) to (5);
	\end{pgfonlayer}
\end{tikzpicture}
}
\end{equation}
The phase shift has a single parameter $\theta \in \{0, \pi\}$ where $\theta = p \pi$ ($\theta=\pi$ corresponds to $p=1$ and $\theta=0$ corresponds to $p=0$). A single phase shift along a wire initialises a Boolean logic \textit{variable}. Note that, in accordance with Convention~\ref{sigconv}, this initialisation of the variable is with respect to the reference wave $\sin(\omega t)$, which enters the circuit with phase 0. If $\theta = \pi$, the wave becomes $\sin(\omega t + \pi)$, and therefore bit 1 is initialised; conversely, if $\theta = 0$, bit 0 is initialised. Logical bits 0 and 1 are respectively represented as
\[
\scalebox{0.75}{\input{fig2c.tikz}} \ \ , \ \ \ \ \ \ \ \  \scalebox{0.75}{\begin{tikzpicture}
	\begin{pgfonlayer}{nodelayer}
		\node [style=none] (2) at (-2.75, 2) {};
		\node [style=new style 0] (5) at (-2.75, 0) {$\pi$};
		\node [style=none] (7) at (-2.75, -2) {};
	\end{pgfonlayer}
	\begin{pgfonlayer}{edgelayer}
		\draw (5) to (2.center);
		\draw (7.center) to (5);
	\end{pgfonlayer}
\end{tikzpicture}
}
\]
More complex circuits can be created by composing copies of circuits (\ref{wire}) and (\ref{phase}). For example, series or sequential composition of two copies of circuit (\ref{phase}) with phase shifts $\alpha$ ($=a\pi$) and $\beta$ ($=b\pi$) respectively gives the diagram

\begin{equation}\label{series1}
	\scalebox{0.75}{\begin{tikzpicture}
	\begin{pgfonlayer}{nodelayer}
		\node [style=new style 0] (1) at (0, -1) {$\alpha$};
		\node [style=none] (2) at (0, -3) {};
		\node [style=none] (4) at (0, 3) {};
		\node [style=new style 0] (5) at (0, 1) {$\beta$};
	\end{pgfonlayer}
	\begin{pgfonlayer}{edgelayer}
		\draw (2.center) to (1);
		\draw (5) to (4.center);
		\draw (5) to (1);
	\end{pgfonlayer}
\end{tikzpicture}
}
\end{equation}
This circuit modifies the phase of the input wave by the shift $\alpha$ followed by the shift $\beta$. 

Setting the second phase shift to $0$, the circuit becomes
\[
\scalebox{0.75}{\input{series0.tikz}}
\]
which is the series composition of the initialised variable $\alpha$ and an identity gate. This means that the circuit outputs $\alpha$.

Conversely, setting the second phase shift of circuit (\ref{series1}) to $\pi$ yields
\[
\scalebox{0.75}{\begin{tikzpicture}
	\begin{pgfonlayer}{nodelayer}
		\node [style=new style 0] (1) at (0, -1) {$\alpha$};
		\node [style=none] (2) at (0, -3) {};
		\node [style=none] (4) at (0, 3) {};
		\node [style=new style 0] (5) at (0, 1) {$\pi$};
	\end{pgfonlayer}
	\begin{pgfonlayer}{edgelayer}
		\draw (2.center) to (1);
		\draw (5) to (4.center);
		\draw (5) to (1);
	\end{pgfonlayer}
\end{tikzpicture}
}
\]
which is the series composition of the variable $\alpha$ and phase shift $\pi$. It is instructive to determine the input-output relations of this logic gate. Input $\alpha = 0$ gives the output $\pi$ while input $\alpha = \pi$ yields the output $0$. Therefore, the truth table for this gate is
\[
\begin{tabular}{c|c}
	$a$ & $out$ \\
	\hline
	0 & 1 \\
	1 & 0 \\
\end{tabular}
\]
which means that it is a NOT gate~\cite{kostylev2005spin}. Symbolically, it can be denoted by $a'$. In other words, $a'$ is the Boolean algebraic interpretation of this diagram. The NOT gate gives the complement of the input variable. An abbreviated version of the diagram for the complement (Comp) of the variable $a$ is  
\[
\scalebox{0.75}{\input{comp.tikz}}
\] 

Let's consider circuit (\ref{series1}) again
\[
\scalebox{0.75}{\begin{tikzpicture}
	\begin{pgfonlayer}{nodelayer}
		\node [style=new style 0] (1) at (0, -1) {$\alpha$};
		\node [style=none] (2) at (0, -3) {};
		\node [style=none] (4) at (0, 3) {};
		\node [style=new style 0] (5) at (0, 1) {$\beta$};
	\end{pgfonlayer}
	\begin{pgfonlayer}{edgelayer}
		\draw (2.center) to (1);
		\draw (5) to (4.center);
		\draw (5) to (1);
	\end{pgfonlayer}
\end{tikzpicture}
}
\]
and find its truth table. It has two input variables $\alpha$ and $\beta$. Setting both inputs to $0$ or $\pi$ yields the output $0$. Complementary inputs give the output $\pi$. Therefore, this circuit has the truth table
\[
\begin{tabular}{c|c|c}
	$a$ & $b$ & $out$ \\
	\hline
	0 & 0 & 0 \\
	0 & 1 & 1 \\
	1 & 0 & 1 \\
	1 & 1 & 0 \\        
\end{tabular}
\]
which is that of an exclusive OR (XOR) gate~\cite{schneider2008realization}. 

Composing a NOT gate (which is a $\pi$ phase shift) in series with an XOR gate gives the diagram
\[
\scalebox{0.75}{\input{fig3d.tikz}}
\]
which has the truth table
\[
\begin{tabular}{c|c|c}
	$a$ & $b$ & $out$ \\
	\hline
	0 & 0 & 1 \\
	0 & 1 & 0 \\
	1 & 0 & 0 \\
	1 & 1 & 1 \\        
\end{tabular}
\]
and therefore is an XNOR gate. 

So far, we have discussed only the series composition of circuits. One can compose circuits in parallel too. For instance, two wires having phase shifts $\alpha$ and $\beta$ can be composed in parallel as
\[
\scalebox{0.75}{\input{parallel.tikz}}
\]
This circuit initialises two Boolean variables. 

In phase-encoded wave-based computation, a wave signal may be required to be split or copied into multiple waveguides. For a wave copied into three waveguides, the \textit{copy} operation is represented as
\[
\scalebox{0.75}{\input{copy.tikz}}
\]
For an input wave entering the bottom wire, three copies are produced at the output. 

Interference of waves is performed by combining signals from multiple waveguides into a single waveguide. Diagrammatically, it is given by the \textit{merge} operation:
\[
\scalebox{0.75}{\input{merge.tikz}}
\]
The output of a merge is obtained by adding the waves entering the three input wires and then normalising the resultant amplitude to $1$.

We use only one-to-three copy and three-to-one merge operations here because, in addition to phase shifts, they are all we need to create all other standard logic gates. Three phase shifts ($\alpha$, $\beta$, and $\gamma$) taken in parallel, along with a copy and a merge, give the circuit
\begin{equation}\label{majabc}
	\scalebox{0.75}{\input{fig4a.tikz}} 
\end{equation}
In keeping with Convention~\ref{sigconv}, 
$\sin(\omega t)$ enters the input wire at the bottom, and is copied into three wires. The three waves are phase shifted by $\alpha$, $\beta$, $\gamma$ (corresponding to the logical inputs $a, b, c$
) and are then merged into the single output wire. The output is obtained by normalising the amplitude of the sum $\sin(\omega t + \alpha) +  \sin(\omega t + \beta) +  \sin(\omega t + \gamma)$. An output wave $\sin(\omega t)$ corresponds to logical $0$ whereas $\sin(\omega t + \pi)$ corresponds to logical~$1$. 
The truth table for this gate is as follows:
\[
\begin{tabular}{c|c|c|c}
	$a$ & $b$ & $c$ & $out$ \\
	\hline
	0 & 0 & 0 & 0\\
	0 & 0 & 1 & 0 \\
	0 & 1 & 0 & 0 \\
	0 & 1 & 1 & 1 \\  
	1 & 0 & 0 & 0 \\
	1 & 0 & 1 & 1 \\
	1 & 1 & 0 & 1 \\
	1 & 1 & 1 & 1 \\   
\end{tabular}
\]
\emph{i.e.,} it is a majority (MAJ) gate~\cite{khitun2010magnonic, fischer2017experimental, zografos2017non}. In Boolean algebra, it is interpreted as $(a \land b)  \lor (b \land c) \lor (c \land a)$.

If one of the inputs, say $c$, is set to $0$, the following truth table is obtained:
\[
\begin{tabular}{c|c|c}
	$a$ & $b$ & $out$ \\
	\hline
	0 & 0 & 0 \\
	0 & 1 & 0 \\
	1 & 0 & 0 \\
	1 & 1 & 1 \\        
\end{tabular}
\]
That is, the resultant is effectively an AND gate~\cite{parhami2020majority}.

Setting the input $c$ to $1$ yields the truth table of an OR gate~\cite{parhami2020majority}:
\[
\begin{tabular}{c|c|c}
	$a$ & $b$ & $out$ \\
	\hline
	0 & 0 & 0 \\
	0 & 1 & 1 \\
	1 & 0 & 1 \\
	1 & 1 & 1 \\        
\end{tabular}
\]

In string diagrams, these gates are given by 
\[
\scalebox{0.75}{\input{fig4b.tikz}}  \ \ , \ \ \ \ \ \    \scalebox{0.75}{\input{fig4c.tikz}} 
\]
which are interpreted in Boolean algebra as $a \land b$ and $a \lor b$, respectively. Note that we have reordered the phases using the commutativity rule (\ref{eq:CM}):
\[
\scalebox{0.75}{\input{ocm1.tikz}} \ = \ \scalebox{0.75}{\input{ocm4.tikz}} 
\]

\begin{convention}
	We follow the convention of drawing the input phase shifts of AND and OR gates on the sides. This helps to visually keep track of the input variables, especially in more complex diagrams.
\end{convention}

More complex circuits can be obtained by composing logic gates. For instance, an AND gate and an OR gate can be composed such that the output of the AND gate is used as an input of the OR gate. Consider an AND gate $a \land b$ the output of which is used an input of an OR gate to implement the circuit $(a \land b) \lor c$. Diagrammatically, this circuit is obtained by plugging the $a \land b$ gate in as an input inside an $\lor$ gate with $c$ as the other input:
\[
\scalebox{0.75}{\input{nesting1.tikz}}
\] 

Another example of composing logic gates is as follows. A NOT gate can be composed with an AND gate $a \land b$ to yield a circuit that inverts the output of the AND gate: $(a \land b)'$. This is actually a NAND gate. In diagrams, it is generated by composing a NOT gate (\textit{i.e.} a $\pi$ phase shift) in series with the AND gate:
\[
\scalebox{0.75}{\input{andnot.tikz}}
\]

\subsection{Diagrammatic Reasoning}\label{diagreas}

For diagrammatic reasoning, we need to define a notion of equality and then a set of diagram substitution or rewrite rules based on this notion of equality. In this chapter, two diagrams are considered to be (operationally) equal if they represent the same logical operation, \emph{i.e.} produce the same truth table.  Using this notion of equality, we define rewrite rules as follows.

Two definitional rules were already discussed in the previous section: the identity (ID) rule and the complement (Comp) rule: 
\begin{equation}
\scalebox{0.75}{\input{fig2c.tikz}} 
	\tag{ID}\label{eq:ID}
\end{equation}
\begin{equation}
\scalebox{0.75}{\input{comp.tikz}} 
	\tag{Comp}\label{eq:Comp}
\end{equation}
The remaining rules are the fusion (F) rule, the copy (C) rule, the commutativity (CM) rule, the distributivity (D) rule, the majority (M) rule, the associativity (A) rule, and the chopping (CH) rule:
\begin{equation}
\scalebox{0.75}{\input{fusion.tikz}}
	\tag{F}\label{eq:F}
\end{equation}
\begin{equation}
\scalebox{0.75}{\input{copyrule3.tikz}} = \scalebox{0.75}{\input{copyrule4.tikz}} \ , \ \scalebox{0.75}{\input{copyrule1.tikz}} = \scalebox{0.75}{\input{copyrule2.tikz}} \ \ \ \ 
	\tag{C}\label{eq:C}
\end{equation}
\begin{equation}
	\scalebox{0.75}{\input{ocm1.tikz}} = \scalebox{0.75}{\input{ocm4.tikz}} = \ \cdots \ = \scalebox{0.75}{\input{ocm6.tikz}}
	\tag{CM}\label{eq:CM}
\end{equation}
\begin{equation}
\scalebox{0.75}{\input{dist1.tikz}} 
	\tag{D}\label{eq:D}
\end{equation}
\begin{equation}
\scalebox{0.75}{\input{maj1.tikz}} = \scalebox{0.75}{\input{idema6b.tikz}} 
	\tag{M}\label{eq:M}
\end{equation}
\begin{equation}
\scalebox{0.75}{\input{ass1.tikz}} 
	\tag{A}\label{eq:A}
\end{equation}

\begin{equation}
\scalebox{0.75}{\input{rule31.tikz}} = \scalebox{0.75}{\input{rule32.tikz}} 
	\tag{CH}\label{eq:CH}
\end{equation}
where $\alpha, \beta, \gamma, \phi, \theta \in \{0, \pi\}$, $\phi \neq \theta$, and $+$ is the sum modulo $2\pi$.

As an example of diagrammatic reasoning, we shall now use the rules defined above to derive another chopping rule (CH2) that will be of use later.

\begin{proposition}
For all $\phi, \alpha \in \{0, \pi\}$,

\begin{equation}
\scalebox{0.75}{\input{rule31b.tikz}} = \scalebox{0.75}{\input{rule32b.tikz}}
	\tag{CH2}\label{eq:CH2}
\end{equation}
\end{proposition}
\begin{proof}
\[
\scalebox{0.75}{\input{rule31b.tikz}} \overset{\text{(Comp)}}{=} \scalebox{0.75}{\input{rule31b2.tikz}} \overset{\text{(C)}}{=} \scalebox{0.75}{\input{rule31b3.tikz}} 
\overset{\text{(F)}}{=} \scalebox{0.75}{\input{rule31b4.tikz}}  \overset{\text{}}{=} \scalebox{0.75}{\input{rule31b5.tikz}} 
\]
\[ \overset{\text{(ID)}}{=} \scalebox{0.75}{\input{rule31b6.tikz}} 
\overset{\text{(CH)}}{=} \scalebox{0.75}{\input{rule31b7.tikz}}  \overset{\text{(F)}}{=} \scalebox{0.75}{\input{rule31b8.tikz}}  \overset{\text{}}{=} \scalebox{0.75}{\input{rule31b9.tikz}}  \overset{\text{(ID)}}{=} \scalebox{0.75}{\input{rule31b10.tikz}}   
\]
\end{proof}

An important question to ask is whether this string-diagrammatic formalism is universal, sound and complete for Boolean algebra. Here, universality means that the string diagrams can represent every valid Boolean algebraic expression; soundness means that any derivation in string diagrams is correct when interpreted as Boolean algebra; and completeness means that any valid Boolean algebraic equation can be derived using string diagrams.

Any Boolean algebraic expression comprises Boolean variables connected by 
$\land$, $\lor$ and $'$ operations. String diagrammatic formulae for these variables (phase shifts) and operations (logic gates) have already been demonstrated. Taking the 
$\land$ or $\lor$ of a Boolean expression with another one diagrammatically means nesting the corresponding diagrams in an 
AND or OR circuit. Taking the 
$'$ (negation) of an expression corresponds to composing the circuit with a NOT gate (\emph{i.e.} a $\pi$ phase shift). Our string-diagrammatic formalism is therefore universal by construction.

Proving soundness and completeness of this formalism for Boolean algebra will be the subject of future work. However, we offer some remarks below. 

The axioms of Boolean algebra are
\begin{enumerate}[label=(\alph*)]
	\item $\forall x, y \in B,\ x \land y = y \land x$
	\item $\forall x, y \in B,\ x \lor y = y \lor x$
	\item $\forall x \in B,\ x \lor 0 = x$
	\item $\forall x \in B,\ x \land 1 = x$
	\item $\forall x \in B,\ x \lor x' = 1$
	\item $\forall x \in B,\ x \land x' = 0$
	\item $\forall x, y, z \in B,\ x \land (y \lor z) = (x \land y) \lor (x \land z)$ 
	\item $\forall x, y, z \in B,\ x \lor (y \land z) = (x \lor y) \land (x \lor z)$
\end{enumerate}
where $B = \{0, 1\}$.    

Axioms (a) and (b) correspond diagrammatically to the (CM) rule, (c) and (d) to the (CH) rule, (e) and (f) to the (CH2) rule, and (g) and (h) to the (D) rule. The (CH2) rule is derived string-diagrammatically using the (CH) rule. Notice that associativity is not one of the algebraic axioms since it can be derived from the eight above. In fact, its corresponding rule, (A), is derivable string-diagrammatically as well---see Appendix~\ref{appendwaves} for details. The rest of the rewrite rules are purely string-diagrammatic, \emph{i.e.} they come for free with the formalism.

\section{Application to Wave Logic Circuits}\label{secappwave}

One can use the diagrammatic formalism for design, analysis and optimisation of wave-logic circuits. Some examples are presented below.

\subsection{Design}

\begin{example}
A half adder is characterised by the 
truth table: 
\[
\begin{tabular}{c|c|c|c}
	$a$ & $b$ & $carry$ & $sum$ \\
	\hline
	0 & 0 & 0 & 0\\
	0 & 1 & 0 & 1 \\
	1 & 0 & 0 & 1 \\
	1 & 1 & 1 & 0 \\            \end{tabular}
\]
The implementation of an AND gate has already been discussed. Hence the design of the carry circuit is straightforward. The algebraic expression of the sum circuit is that of an XOR gate, which (as discussed before) is implemented by composing phase shifts (corresponding to the input variables) in series. This means the complete half adder circuit is given by 
\[
\scalebox{0.75}{\input{half_add.tikz}}
\]

\end{example}

\begin{example}
For a slightly more complex example, consider a full adder. Its truth table is as follows:
\[
\begin{tabular}{c|c|c|c|c}
	$c\_in$ & $a$ & $b$ & $c\_out$ & $sum$ \\
	\hline
	0 & 0 & 0 & 0 & 0 \\
	0 & 0 & 1 & 0 & 1 \\
	0 & 1 & 0 & 0 & 1 \\
	0 & 1 & 1 & 1 & 0 \\    
	1 & 0 & 0 & 0 & 1 \\
	1 & 0 & 1 & 1 & 0 \\
	1 & 1 & 0 & 1 & 0 \\
	1 & 1 & 1 & 1 & 1 \\          
\end{tabular}
\]
The truth table shows that the $c\_out$ output is equal to that of the MAJ gate, whose implementation has already been discussed. On the other hand, the first half of $\text{sum}$ (when $c\_in=0$) is the output of an XOR gate for inputs $a$ and  $b$ whereas the second half (when $c\_in=1$) is its complement. This means that the $\text{sum}$ circuit can be obtained by composing the XOR of $a$ and $b$ with the phase shift $c\_in$ in series. Hence, the full adder is given by the diagram
\[
\scalebox{0.75}{\input{full_add2.tikz}}
\]
where we represent the phase shift corresponding to $c\_in$ by $\gamma_i$.
\end{example}

\begin{remark}
	In diagrams such as the half adder and full adder, a Boolean parameter may appear in more than one location. This repetition is represented at the level of the diagram's interface: several phase shifts may carry the same label, indicating multiple uses of the same external Boolean input.
	
	We do not realise these repetitions via the internal copy operation, because they do not correspond to any physical splitting of a wave inside the device; only to multiple references to the same external parameter, exactly as in ordinary Boolean algebraic syntax.
	
	This light form of external bookkeeping does not affect the universality of the diagrammatic language. Universality here concerns the ability to represent any Boolean algebraic expression diagrammatically, and repeated variables in such expressions are realised by repeated labels.
	
	Whenever internal duplication of a wave is required, the calculus already provides a copy operation; its absence in the half adder and full adder examples is simply a modelling choice that keeps the diagrams faithful to the physical architecture of those specific circuits.
\end{remark}

\subsection{Analysis}

\begin{example}
Considering the circuit $\scalebox{0.75}{\input{analyse1.tikz}}$, 
perhaps we are interested in its behaviour if the input $a$ is $0$. This corresponds to replacing $\scalebox{0.75}{\input{analyse6.tikz}}$ by $\scalebox{0.75}{\input{analyse5.tikz}}$ in the circuit. The resulting circuit can then be simplified using the rewrite rules, as follows:
\[
\scalebox{0.75}{\input{analyse2.tikz}} 
\overset{\text{(M)}}{=}
\scalebox{0.75}{\input{analyse3.tikz}} \overset{\text{(CH)}}{=}
\scalebox{0.75}{\input{analyse4.tikz}} 
\]
This shows that if input $a$ is set to $0$, the output of the circuit is given by the variable $c$, corresponding to the phase $\gamma$.
\end{example}

\subsection{Optimisation}

\begin{example}
Consider the circuit $ \scalebox{0.75}{\input{circ1.tikz}}$
for which an optimal implementation is desired. Such an implementation would achieve the same logical operation with the simplest possible circuit and the minimum possible number of logic gates.

The above circuit can be simplified using the diagram rewrite rules as follows.
\[
\scalebox{0.55}{\input{circ1.tikz}} \overset{\text{(C)}}{=} \scalebox{0.55}{\input{circ2.tikz}}
\]
\[\overset{\text{(F), (ID), (Comp)}}{=} \scalebox{0.55}{\input{circ3.tikz}} \overset{\text{(D)}}{=} \scalebox{0.55}{\input{circ4.tikz}}
\]
\[\overset{\text{(CH2)}}{=} \scalebox{0.55}{\input{circ5.tikz}} \overset{\text{(CH)}}{=} \scalebox{0.55}{\input{circ6.tikz}}
\]
\[\overset{\text{(D)}}{=} \scalebox{0.55}{\input{circ7.tikz}} 
\overset{\text{(A), (CM)}}{=} \scalebox{0.55}{\input{circ8.tikz}}
\]
\[\overset{\text{(M), (A)}}{=} \scalebox{0.55}{\input{circ9.tikz}} \overset{\text{(M)}}{=} \scalebox{0.55}{\input{circ10.tikz}}
\]
\[\overset{\text{(CM)}}{=} \scalebox{0.55}{\input{circ11.tikz}} \overset{\text{(M)}}{=} \scalebox{0.55}{\input{circ12.tikz}}
\]
\end{example}

\begin{remark}
When a logic circuit is required to be simplified, the conventional approach is to translate it into its corresponding symbolic algebraic expression, simplify the expression, and then translate the resultant expression back into the circuit. 
\[
\text{Circ}_1 \longrightarrow \text{Alg}_1 \longrightarrow \text{Alg}_2 \longrightarrow \text{Circ}_2
\]

The string-diagrammatic formalism introduced in this chapter can be exploited to simplify circuits without translation into algebraic expressions. This means that the diagrams offer an alternative method for simplifying Boolean circuits.
\[
\text{Circ}_1 \longrightarrow \text{Circ}_2
\]
\end{remark}

\begin{remark}
	In fact, it is, in principle, possible to use the diagrams as an alternative to Boolean algebra: by mapping the expressions to diagrams, performing graphical simplification, and then converting the simplified diagrams into algebraic expressions again. 
	\[
	\text{Alg}_1 \longrightarrow 
	\text{Circ}_1 \longrightarrow \text{Circ}_2
	\longrightarrow 
	\text{Alg}_2 
	\]
	Although as a syntax, Boolean algebra is more concise as compared with string diagrams, the latter are more amenable to visual intuition which is useful for a circuit designer or analyst.
\end{remark}

\section{Summary and Outlook}\label{sumwave}
	
In this chapter, we introduced a string-diagrammatic formalism for wave-based computation with phase encoding. Through the example of spin-wave circuits, we demonstrated the usage of the formalism in designing, analysing and simplifying Boolean logic circuits. The approach is device-independent and hence useful as a visual, intuitive aid to design, optimise and analyse circuits in different physical platforms that are based on wave-based or majority-logic-based computation~\cite{imai1986optical, ambs2010optical, kazanskiy2022optical, rahman2015wave, mahmoud2020introduction}. 

A key advantage of the approach lies in its immediate applicability to spin-wave logic circuits~\cite{mahmoud2020introduction, khitun2010magnonic}, in the sense that the components of the circuit diagrams map directly to their physical counterparts. In other words, `what you see is what you get'. The formalism, moreover provides a theoretical framework to represent and characterise spin-wave logic implementations such as those in Refs.~\cite{khitun2010magnonic, kostylev2005spin, schneider2008realization, fischer2017experimental, zografos2017non}.
	
In addition, diagrammatic rewrite rules allow graphical simplification of circuits which potentially lend themselves to automation and/or software-based optimisation. Indeed, broadly, the approach is very expressive in the flexibility it lends to circuit design and implementation. For instance, one may tweak a particular circuit to obtain the desired design (possibly subject to hardware constraints) by moving the phase shifts along the wires and across other phase shifts, and using copy and merge operations (applying the corresponding rules). More concretely, if one is to compose a NOT gate with a circuit, it may be cascaded at the top or the bottom of the circuit (or the circuit may be modified internally in a more elaborate manner), giving the same logical outcome in either case.
	\[
	\scalebox{0.75}{\input{andnot.tikz}} 
	\overset{\text{}}{=}
	\scalebox{0.75}{\input{andnot5.tikz}} \overset{\text{}}{=}
	\scalebox{0.75}{\input{andnot4.tikz}} 
	\]
	\[
	\overset{\text{}}{=} \scalebox{0.75}{\input{andnot3.tikz}} 
	\overset{\text{}}{=}
	\scalebox{0.75}{\input{andnot2.tikz}} \overset{\text{}}{=}
	\scalebox{0.75}{\input{andnot6.tikz}} 
	\]
	
An obvious avenue for future work is providing formal semantics for the string diagrams presented here, and proving that the formalism is sound and complete for Boolean algebra. Another line of research is to conceive of the formalism as a diagrammatic alternative to the traditional symbolic approach, wherein lies the possibility of a new axiomatisation of Boolean algebra. 
	
This study focuses on the theoretical modelling of spin-wave circuits and introduces a formalism that abstracts away many hardware-level challenges. However, since spin-wave computing inherits key limitations of analogue systems---such as noise sensitivity, signal degradation, error accumulation, and energy-intensive error-correction mechanisms---these factors are not fully captured in the current framework. In particular, future extensions of the formalism may need to incorporate models for error correction, which often relies on non-linear processes that can compromise energy efficiency~\cite{verba2019correction}. Similarly, cascading gates requires amplitude normalisation to maintain signal integrity, yet this too is a non-linear and energy-intensive operation~\cite{dutta2015compact}. Modelling these aspects---along with crosstalk noise at small scales and potential integration with complementary platforms such as optics, superconducting qubits, or CMOS circuits~\cite{chumak2015magnon, li2020hybrid, lachance2019hybrid}---will be crucial for improving the practical relevance of the formalism. Doing so would enhance its applicability to real-world implementations and support the development of more robust, scalable spin-wave computing systems.

\usetikzlibrary{decorations.markings}
\usetikzlibrary{shapes.geometric}

\pgfdeclarelayer{edgelayer}
\pgfdeclarelayer{nodelayer}
\pgfsetlayers{edgelayer,nodelayer,main}

\tikzstyle{none}=[inner sep=0pt]
\definecolor{hexcolor0xff0000}{rgb}{1.000,0.000,0.000}
\definecolor{hexcolor0x000000}{rgb}{0.000,0.000,0.000}
\definecolor{hexcolor0x00ff00}{rgb}{0.000,1.000,0.000}
\definecolor{hexcolor0x000000}{rgb}{0.000,0.000,0.000}
\definecolor{hexcolor0xffff00}{rgb}{1.000,1.000,0.000}
\definecolor{hexcolor0xffffff}{rgb}{1.000,1.000,1.000}

\tikzstyle{rn}=[circle,fill=hexcolor0xff0000,draw=hexcolor0x000000,line width=0.8 pt]
\tikzstyle{gn}=[circle,fill=hexcolor0x00ff00,draw=hexcolor0x000000,line width=0.8 pt]
\tikzstyle{yn}=[circle,fill=hexcolor0xffff00,draw=hexcolor0x000000,line width=0.8 pt]
\tikzstyle{wn}=[circle,fill=hexcolor0xffffff,draw=hexcolor0x000000,line width=0.8 pt]
\tikzstyle{wnthick}=[circle,fill=hexcolor0xffffff,draw=hexcolor0x000000,line width=2.500]

\tikzstyle{simple}=[-,draw=hexcolor0x000000,line width=2.000]
\tikzstyle{arrow}=[-,draw=hexcolor0x000000,postaction={decorate},decoration={markings,mark=at position .5 with {\arrow{>}}},line width=2.000]
\tikzstyle{tick}=[-,draw=hexcolor0x000000,postaction={decorate},decoration={markings,mark=at position .5 with {\draw (0,-0.1) -- (0,0.1);}},line width=2.000]
\tikzstyle{halfthickness}=[-,draw=hexcolor0x000000,line width=0.500]
\tikzstyle{thick}=[-,draw=hexcolor0x000000,line width=2.500]
\tikzstyle{thicker}=[-,draw=hexcolor0x000000,line width=4.000]

\tikzstyle{env}=[copoint,regular polygon rotate=0,minimum width=0.2cm, fill=black]

\tikzstyle{probs}=[shape=semicircle,fill=white,draw=black,shape border rotate=180,minimum width=1.2cm]

%
%


\tikzstyle{every picture}=[baseline=-0.5em,scale=0.7]

\tikzstyle{dotpic}=[] 
\tikzstyle{diredges}=[every to/.style={diredge}]
\tikzstyle{math matrix}=[matrix of math nodes,left delimiter=(,right delimiter=),inner sep=2pt,column sep=1em,row sep=0.5em,nodes={inner sep=0pt},text height=1.5ex, text depth=0.25ex]


\tikzstyle{inline text}=[text height=1.5ex, text depth=0.25ex,yshift=0.5mm]
\tikzstyle{label}=[font=\footnotesize,text height=1.5ex, text depth=0.25ex,yshift=0.5mm]
\tikzstyle{left label}=[label,anchor=east,xshift=1.5mm]
\tikzstyle{right label}=[label,anchor=west,xshift=-1.5mm]

\tikzstyle{braceedge}=[decorate,decoration={brace,amplitude=2mm,raise=-1mm}]
\tikzstyle{small braceedge}=[decorate,decoration={brace,amplitude=1mm,raise=-1mm}]

\tikzstyle{doubled}=[line width=1.6pt] 
\tikzstyle{boldedge}=[doubled,shorten <=-0.17mm,shorten >=-0.17mm]
\tikzstyle{boldedgegray}=[doubled,gray,shorten <=-0.17mm,shorten >=-0.17mm]

\tikzstyle{semidoubled}=[line width=1.4pt] 
\tikzstyle{semiboldedgegray}=[semidoubled,gray,shorten <=-0.17mm,shorten >=-0.17mm]

\tikzstyle{boldedgedashed}=[very thick,dashed,shorten <=-0.17mm,shorten >=-0.17mm]
\tikzstyle{vboldedgedashed}=[doubled,dashed,shorten <=-0.17mm,shorten >=-0.17mm]
\tikzstyle{left hook arrow}=[left hook-latex]
\tikzstyle{right hook arrow}=[right hook-latex]
\tikzstyle{sembracket}=[line width=0.5pt,shorten <=-0.07mm,shorten >=-0.07mm]

\tikzstyle{causal edge}=[->,thick,gray]
\tikzstyle{causal nondir}=[thick,gray]
\tikzstyle{timeline}=[thick,gray, dashed]

\tikzstyle{cedge}=[<->,thick,gray!70!white]

\tikzstyle{empty diagram}=[draw=gray!40!white,dashed,shape=rectangle,minimum width=1cm,minimum height=1cm]
\tikzstyle{empty diagram small}=[draw=gray!50!white,dashed,shape=rectangle,minimum width=0.6cm,minimum height=0.5cm]

\newcommand{\measurement}{\tikz[scale=0.6]{ \draw [use as bounding box,draw=none] (0,-0.1) rectangle (1,0.7); \draw [fill=white] (1,0) arc (0:180:5mm); \draw (0,0) -- (1,0) (0.5,0) -- +(60:7mm);}}


\tikzstyle{dot}=[inner sep=0mm,minimum width=2mm,minimum height=2mm,draw,shape=circle]
\tikzstyle{ddot}=[inner sep=0mm, doubled, minimum width=2.5mm,minimum height=2.5mm,draw,shape=circle]

\tikzstyle{black dot}=[dot,fill=black]
\tikzstyle{white dot}=[dot,fill=white,,text depth=-0.2mm]
\tikzstyle{green dot}=[white dot] 
\tikzstyle{gray dot}=[dot,fill=gray!40!white,,text depth=-0.2mm]
\tikzstyle{red dot}=[gray dot] 


\tikzstyle{black ddot}=[ddot,fill=black]
\tikzstyle{white ddot}=[ddot,fill=white]
\tikzstyle{gray ddot}=[ddot,fill=gray!40!white]

\tikzstyle{gray edge}=[gray!40!white]

\tikzstyle{small dot}=[inner sep=0.5mm,minimum width=0pt,minimum height=0pt,draw,shape=circle]

\tikzstyle{small black dot}=[small dot,fill=black]
\tikzstyle{small white dot}=[small dot,fill=white]
\tikzstyle{small gray dot}=[small dot,fill=gray!40!white]

\tikzstyle{causal dot}=[inner sep=0.4mm,minimum width=0pt,minimum height=0pt,draw=white,shape=circle,fill=gray!40!white]


\tikzstyle{phase dimensions}=[minimum size=5mm,font=\footnotesize,rectangle,rounded corners=2.5mm,inner sep=0.2mm,outer sep=-2mm]
\tikzstyle{dphase dimensions}=[minimum size=5mm,font=\footnotesize,rectangle,rounded corners=2.5mm,inner sep=0.2mm,outer sep=-2mm]

\tikzstyle{white phase dot}=[dot,fill=white,phase dimensions]
\tikzstyle{white phase ddot}=[ddot,fill=white,dphase dimensions]
\tikzstyle{green phase ddot}=[ddot,fill=green,dphase dimensions]

\tikzstyle{white rect ddot}=[draw=black,fill=white,doubled,minimum size=5mm,font=\footnotesize,rectangle,rounded corners=2.5mm,inner sep=0.2mm]
\tikzstyle{gray rect ddot}=[draw=black,fill=gray!40!white,doubled,minimum size=6mm,font=\footnotesize,rectangle,rounded corners=3mm]

\tikzstyle{gray phase dot}=[dot,fill=gray!40!white,phase dimensions]
\tikzstyle{gray phase ddot}=[ddot,fill=gray!40!white,dphase dimensions]
\tikzstyle{red phase ddot}=[ddot,fill=red,dphase dimensions]

\tikzstyle{grey phase dot}=[gray phase dot]
\tikzstyle{grey phase ddot}=[gray phase ddot]

\tikzstyle{small phase dimensions}=[minimum size=4mm,font=\tiny,rectangle,rounded corners=2mm,inner sep=0.2mm,outer sep=-2mm]
\tikzstyle{small dphase dimensions}=[minimum size=4mm,font=\tiny,rectangle,rounded corners=2mm,inner sep=0.2mm,outer sep=-2mm]

\tikzstyle{small gray phase dot}=[dot,fill=gray!40!white,small phase dimensions]
\tikzstyle{small gray phase ddot}=[ddot,fill=gray!40!white,small dphase dimensions]


\tikzstyle{small map}=[draw,shape=rectangle,minimum height=4mm,minimum width=4mm,fill=white]

\tikzstyle{cnot}=[fill=white,shape=circle,inner sep=-1.4pt]

\tikzstyle{asym hadamard}=[fill=white,draw,shape=NEbox,inner sep=0.6mm,font=\footnotesize,minimum height=4mm]
\tikzstyle{asym hadamard conj}=[fill=white,draw,shape=NWbox,inner sep=0.6mm,font=\footnotesize,minimum height=4mm]
\tikzstyle{asym hadamard dag}=[fill=white,draw,shape=SEbox,inner sep=0.6mm,font=\footnotesize,minimum height=4mm]

\tikzstyle{hadamard}=[fill=white,draw,inner sep=0.6mm,font=\footnotesize,minimum height=4mm,minimum width=4mm]
\tikzstyle{small hadamard}=[fill=white,draw,inner sep=0.6mm,minimum height=1.5mm,minimum width=1.5mm]
\tikzstyle{dhadamard}=[hadamard,doubled]
\tikzstyle{small dhadamard}=[small hadamard,doubled]
\tikzstyle{small dhadamard rotate}=[small hadamard,doubled,rotate=45]
\tikzstyle{antipode}=[white dot,inner sep=0.3mm,font=\footnotesize]

\tikzstyle{scalar}=[diamond,draw,inner sep=0.5pt,font=\small]
\tikzstyle{dscalar}=[diamond,doubled, draw,inner sep=0.5pt,font=\small]

\tikzstyle{small box}=[rectangle,inline text,fill=white,draw,minimum height=5mm,yshift=-0.5mm,minimum width=5mm,font=\small]
\tikzstyle{small gray box}=[small box,fill=gray!30]
\tikzstyle{medium box}=[rectangle,inline text,fill=white,draw,minimum height=5mm,yshift=-0.5mm,minimum width=10mm,font=\small]
\tikzstyle{square box}=[small box] 
\tikzstyle{medium gray box}=[small box,fill=gray!30]
\tikzstyle{semilarge box}=[rectangle,inline text,fill=white,draw,minimum height=5mm,yshift=-0.5mm,minimum width=12.5mm,font=\small]
\tikzstyle{large box}=[rectangle,inline text,fill=white,draw,minimum height=5mm,yshift=-0.5mm,minimum width=15mm,font=\small]
\tikzstyle{large gray box}=[small box,fill=gray!30]

\tikzstyle{Bayes box}=[rectangle,fill=black,draw, minimum height=3mm, minimum width=3mm]

\tikzstyle{gray square point}=[small box,fill=gray!50]

\tikzstyle{dphase box white}=[dhadamard]
\tikzstyle{dphase box gray}=[dhadamard,fill=gray!50!white]

\tikzstyle{point}=[regular polygon,regular polygon sides=3,draw,scale=0.75,inner sep=-0.5pt,minimum width=9mm,fill=white,regular polygon rotate=180]
\tikzstyle{copoint}=[regular polygon,regular polygon sides=3,draw,scale=0.75,inner sep=-0.5pt,minimum width=9mm,fill=white]
\tikzstyle{dpoint}=[point,doubled]
\tikzstyle{dcopoint}=[copoint,doubled]

\tikzstyle{wide copoint}=[fill=white,draw,shape=isosceles triangle,shape border rotate=90,isosceles triangle stretches=true,inner sep=0pt,minimum width=1.5cm,minimum height=6.12mm]
\tikzstyle{wide point}=[fill=white,draw,shape=isosceles triangle,shape border rotate=-90,isosceles triangle stretches=true,inner sep=0pt,minimum width=1.5cm,minimum height=6.12mm,yshift=-0.0mm]
\tikzstyle{wide point plus}=[fill=white,draw,shape=isosceles triangle,shape border rotate=-90,isosceles triangle stretches=true,inner sep=0pt,minimum width=1.74cm,minimum height=7mm,yshift=-0.0mm]

\tikzstyle{wide dpoint}=[fill=white,doubled,draw,shape=isosceles triangle,shape border rotate=-90,isosceles triangle stretches=true,inner sep=0pt,minimum width=1.5cm,minimum height=6.12mm,yshift=-0.0mm]
\tikzstyle{wide dcopoint}=[fill=white,doubled,draw,shape=isosceles triangle,shape border rotate=90,isosceles triangle stretches=true,inner sep=0pt,minimum width=1.5cm,minimum height=6.12mm,yshift=-0.0mm]

\tikzstyle{tinypoint}=[regular polygon,regular polygon sides=3,draw,scale=0.55,inner sep=-0.15pt,minimum width=6mm,fill=white,regular polygon rotate=180]

\tikzstyle{white point}=[point]
\tikzstyle{white dpoint}=[dpoint]
\tikzstyle{green point}=[white point] 
\tikzstyle{white copoint}=[copoint]
\tikzstyle{gray point}=[point,fill=gray!40!white]
\tikzstyle{gray dpoint}=[gray point,doubled]
\tikzstyle{red point}=[gray point] 
\tikzstyle{gray copoint}=[copoint,fill=gray!40!white]
\tikzstyle{gray dcopoint}=[gray copoint,doubled]

\tikzstyle{white point guide}=[regular polygon,regular polygon sides=3,font=\scriptsize,draw,scale=0.65,inner sep=-0.5pt,minimum width=9mm,fill=white,regular polygon rotate=180]

\tikzstyle{black point}=[point,fill=black,font=\color{white}]
\tikzstyle{black copoint}=[copoint,fill=black,font=\color{white}]

\tikzstyle{tiny gray point}=[tinypoint,fill=gray!40!white]

\tikzstyle{diredge}=[->]
\tikzstyle{ddiredge}=[<->]
\tikzstyle{rdiredge}=[<-]
\tikzstyle{thickdiredge}=[->, very thick]
\tikzstyle{pointer edge}=[->,very thick,gray]
\tikzstyle{pointer edge part}=[very thick,gray]
\tikzstyle{dashed edge}=[dashed]
\tikzstyle{thick dashed edge}=[very thick,dashed]
\tikzstyle{thick gray dashed edge}=[thick dashed edge,gray!40]
\tikzstyle{thick map edge}=[very thick,|->]


\makeatletter

\renewcommand{\boxshape}[3]{%
\pgfdeclareshape{#1}{
\inheritsavedanchors[from=rectangle] 
\inheritanchorborder[from=rectangle]
\inheritanchor[from=rectangle]{center}
\inheritanchor[from=rectangle]{north}
\inheritanchor[from=rectangle]{south}
\inheritanchor[from=rectangle]{west}
\inheritanchor[from=rectangle]{east}
\backgroundpath{
\southwest \pgf@xa=\pgf@x \pgf@ya=\pgf@y
\northeast \pgf@xb=\pgf@x \pgf@yb=\pgf@y

\@tempdima=#2
\@tempdimb=#3

\pgfpathmoveto{\pgfpoint{\pgf@xa - 5pt + \@tempdima}{\pgf@ya}}
\pgfpathlineto{\pgfpoint{\pgf@xa - 5pt - \@tempdima}{\pgf@yb}}
\pgfpathlineto{\pgfpoint{\pgf@xb + 5pt + \@tempdimb}{\pgf@yb}}
\pgfpathlineto{\pgfpoint{\pgf@xb + 5pt - \@tempdimb}{\pgf@ya}}
\pgfpathlineto{\pgfpoint{\pgf@xa - 5pt + \@tempdima}{\pgf@ya}}
\pgfpathclose
}
}}

\boxshape{NEbox}{0pt}{5pt}
\boxshape{SEbox}{0pt}{-5pt}
\boxshape{NWbox}{5pt}{0pt}
\boxshape{SWbox}{-5pt}{0pt}
\boxshape{EBox}{-3pt}{3pt}
\boxshape{WBox}{3pt}{-3pt}
\makeatother

\tikzstyle{cloud}=[shape=cloud,draw,minimum width=1.5cm,minimum height=1.5cm]

\tikzstyle{map}=[draw,shape=NEbox,inner sep=2pt,minimum height=6mm,fill=white]
\tikzstyle{dashedmap}=[draw,dashed,shape=NEbox,inner sep=2pt,minimum height=6mm,fill=white]
\tikzstyle{mapdag}=[draw,shape=SEbox,inner sep=2pt,minimum height=6mm,fill=white]
\tikzstyle{mapadj}=[draw,shape=SEbox,inner sep=2pt,minimum height=6mm,fill=white]
\tikzstyle{maptrans}=[draw,shape=SWbox,inner sep=2pt,minimum height=6mm,fill=white]
\tikzstyle{mapconj}=[draw,shape=NWbox,inner sep=2pt,minimum height=6mm,fill=white]

\tikzstyle{langmap}=[draw,shape=NEbox,inner sep=2pt,minimum height=2.4mm,minimum width=3.2mm,fill=white]
\tikzstyle{langmaptrans}=[draw,shape=SWbox,inner sep=2pt,minimum height=2.4mm,minimum width=3.2mm,fill=white]

\tikzstyle{medium map}=[draw,shape=NEbox,inner sep=2pt,minimum height=6mm,fill=white,minimum width=7mm]
\tikzstyle{medium map dag}=[draw,shape=SEbox,inner sep=2pt,minimum height=6mm,fill=white,minimum width=7mm]
\tikzstyle{medium map adj}=[draw,shape=SEbox,inner sep=2pt,minimum height=6mm,fill=white,minimum width=7mm]
\tikzstyle{medium map trans}=[draw,shape=SWbox,inner sep=2pt,minimum height=6mm,fill=white,minimum width=7mm]
\tikzstyle{medium map conj}=[draw,shape=NWbox,inner sep=2pt,minimum height=6mm,fill=white,minimum width=7mm]
\tikzstyle{semilarge map}=[draw,shape=NEbox,inner sep=2pt,minimum height=6mm,fill=white,minimum width=9.5mm]
\tikzstyle{semilarge map trans}=[draw,shape=SWbox,inner sep=2pt,minimum height=6mm,fill=white,minimum width=9.5mm]
\tikzstyle{semilarge map adj}=[draw,shape=SEbox,inner sep=2pt,minimum height=6mm,fill=white,minimum width=9.5mm]
\tikzstyle{semilarge map dag}=[draw,shape=SEbox,inner sep=2pt,minimum height=6mm,fill=white,minimum width=9.5mm]
\tikzstyle{semilarge map conj}=[draw,shape=NWbox,inner sep=2pt,minimum height=6mm,fill=white,minimum width=9.5mm]
\tikzstyle{large map}=[draw,shape=NEbox,inner sep=2pt,minimum height=6mm,fill=white,minimum width=12mm]
\tikzstyle{large map conj}=[draw,shape=NWbox,inner sep=2pt,minimum height=6mm,fill=white,minimum width=12mm]
\tikzstyle{very large map}=[draw,shape=NEbox,inner sep=2pt,minimum height=6mm,fill=white,minimum width=17mm]

\tikzstyle{medium dmap}=[draw,doubled,shape=NEbox,inner sep=2pt,minimum height=6mm,fill=white,minimum width=7mm]
\tikzstyle{medium dmap dag}=[draw,doubled,shape=SEbox,inner sep=2pt,minimum height=6mm,fill=white,minimum width=7mm]
\tikzstyle{medium dmap adj}=[draw,doubled,shape=SEbox,inner sep=2pt,minimum height=6mm,fill=white,minimum width=7mm]
\tikzstyle{medium dmap trans}=[draw,doubled,shape=SWbox,inner sep=2pt,minimum height=6mm,fill=white,minimum width=7mm]
\tikzstyle{medium dmap conj}=[draw,doubled,shape=NWbox,inner sep=2pt,minimum height=6mm,fill=white,minimum width=7mm]
\tikzstyle{semilarge dmap}=[draw,doubled,shape=NEbox,inner sep=2pt,minimum height=6mm,fill=white,minimum width=9.5mm]
\tikzstyle{semilarge dmap trans}=[draw,doubled,shape=SWbox,inner sep=2pt,minimum height=6mm,fill=white,minimum width=9.5mm]
\tikzstyle{semilarge dmap adj}=[draw,doubled,shape=SEbox,inner sep=2pt,minimum height=6mm,fill=white,minimum width=9.5mm]
\tikzstyle{semilarge dmap dag}=[draw,doubled,shape=SEbox,inner sep=2pt,minimum height=6mm,fill=white,minimum width=9.5mm]
\tikzstyle{semilarge dmap conj}=[draw,doubled,shape=NWbox,inner sep=2pt,minimum height=6mm,fill=white,minimum width=9.5mm]
\tikzstyle{large dmap}=[draw,doubled,shape=NEbox,inner sep=2pt,minimum height=6mm,fill=white,minimum width=12mm]
\tikzstyle{large dmap conj}=[draw,doubled,shape=NWbox,inner sep=2pt,minimum height=6mm,fill=white,minimum width=12mm]
\tikzstyle{large dmap trans}=[draw,doubled,shape=SWbox,inner sep=2pt,minimum height=6mm,fill=white,minimum width=12mm]
\tikzstyle{large dmap adj}=[draw,doubled,shape=SEbox,inner sep=2pt,minimum height=6mm,fill=white,minimum width=12mm]
\tikzstyle{large dmap dag}=[draw,doubled,shape=SEbox,inner sep=2pt,minimum height=6mm,fill=white,minimum width=12mm]
\tikzstyle{very large dmap}=[draw,doubled,shape=NEbox,inner sep=2pt,minimum height=6mm,fill=white,minimum width=19.5mm]

\tikzstyle{muxbox}=[draw,shape=rectangle,minimum height=3mm,minimum width=3mm,fill=white]
\tikzstyle{dmuxbox}=[muxbox,doubled]

\tikzstyle{box}=[draw,shape=rectangle,inner sep=2pt,minimum height=6mm,minimum width=6mm,fill=white]
\tikzstyle{dbox}=[draw,doubled,shape=rectangle,inner sep=2pt,minimum height=6mm,minimum width=6mm,fill=white]
\tikzstyle{dmap}=[draw,doubled,shape=NEbox,inner sep=2pt,minimum height=6mm,fill=white]
\tikzstyle{dmapdag}=[draw,doubled,shape=SEbox,inner sep=2pt,minimum height=6mm,fill=white]
\tikzstyle{dmapadj}=[draw,doubled,shape=SEbox,inner sep=2pt,minimum height=6mm,fill=white]
\tikzstyle{dmaptrans}=[draw,doubled,shape=SWbox,inner sep=2pt,minimum height=6mm,fill=white]
\tikzstyle{dmapconj}=[draw,doubled,shape=NWbox,inner sep=2pt,minimum height=6mm,fill=white]

\tikzstyle{ddmap}=[draw,doubled,dashed,shape=NEbox,inner sep=2pt,minimum height=6mm,fill=white]
\tikzstyle{ddmapdag}=[draw,doubled,dashed,shape=SEbox,inner sep=2pt,minimum height=6mm,fill=white]
\tikzstyle{ddmapadj}=[draw,doubled,dashed,shape=SEbox,inner sep=2pt,minimum height=6mm,fill=white]
\tikzstyle{ddmaptrans}=[draw,doubled,dashed,shape=SWbox,inner sep=2pt,minimum height=6mm,fill=white]
\tikzstyle{ddmapconj}=[draw,doubled,dashed,shape=NWbox,inner sep=2pt,minimum height=6mm,fill=white]

\boxshape{sNEbox}{0pt}{3pt}
\boxshape{sSEbox}{0pt}{-3pt}
\boxshape{sNWbox}{3pt}{0pt}
\boxshape{sSWbox}{-3pt}{0pt}
\tikzstyle{smap}=[draw,shape=sNEbox,fill=white]
\tikzstyle{smapdag}=[draw,shape=sSEbox,fill=white]
\tikzstyle{smapadj}=[draw,shape=sSEbox,fill=white]
\tikzstyle{smaptrans}=[draw,shape=sSWbox,fill=white]
\tikzstyle{smapconj}=[draw,shape=sNWbox,fill=white]

\tikzstyle{dsmap}=[draw,dashed,shape=sNEbox,fill=white]
\tikzstyle{dsmapdag}=[draw,dashed,shape=sSEbox,fill=white]
\tikzstyle{dsmaptrans}=[draw,dashed,shape=sSWbox,fill=white]
\tikzstyle{dsmapconj}=[draw,dashed,shape=sNWbox,fill=white]

\boxshape{mNEbox}{0pt}{10pt}
\boxshape{mSEbox}{0pt}{-10pt}
\boxshape{mNWbox}{10pt}{0pt}
\boxshape{mSWbox}{-10pt}{0pt}
\tikzstyle{mmap}=[draw,shape=mNEbox]
\tikzstyle{mmapdag}=[draw,shape=mSEbox]
\tikzstyle{mmaptrans}=[draw,shape=mSWbox]
\tikzstyle{mmapconj}=[draw,shape=mNWbox]

\tikzstyle{mmapgray}=[draw,fill=gray!40!white,shape=mNEbox]
\tikzstyle{smapgray}=[draw,fill=gray!40!white,shape=sNEbox]

\makeatletter
\pgfdeclareshape{cornerpoint}{
\inheritsavedanchors[from=rectangle] 
\inheritanchorborder[from=rectangle]
\inheritanchor[from=rectangle]{center}
\inheritanchor[from=rectangle]{north}
\inheritanchor[from=rectangle]{south}
\inheritanchor[from=rectangle]{west}
\inheritanchor[from=rectangle]{east}
\backgroundpath{
\southwest \pgf@xa=\pgf@x \pgf@ya=\pgf@y
\northeast \pgf@xb=\pgf@x \pgf@yb=\pgf@y

\pgfmathsetmacro{\pgf@shorten@left}{\pgfkeysvalueof{/tikz/shorten left}}
\pgfmathsetmacro{\pgf@shorten@right}{\pgfkeysvalueof{/tikz/shorten right}}

\pgfpathmoveto{\pgfpoint{0.5 * (\pgf@xa + \pgf@xb)}{\pgf@ya - 5pt}}
\pgfpathlineto{\pgfpoint{\pgf@xa - 8pt + \pgf@shorten@left}{\pgf@yb - 1.5 * \pgf@shorten@left}}
\pgfpathlineto{\pgfpoint{\pgf@xa - 8pt + \pgf@shorten@left}{\pgf@yb}}
\pgfpathlineto{\pgfpoint{\pgf@xb + 8pt - \pgf@shorten@right}{\pgf@yb}}
\pgfpathlineto{\pgfpoint{\pgf@xb + 8pt - \pgf@shorten@right}{\pgf@yb - 1.5 * \pgf@shorten@right}}
\pgfpathclose
}
}

\pgfdeclareshape{cornercopoint}{
\inheritsavedanchors[from=rectangle] 
\inheritanchorborder[from=rectangle]
\inheritanchor[from=rectangle]{center}
\inheritanchor[from=rectangle]{north}
\inheritanchor[from=rectangle]{south}
\inheritanchor[from=rectangle]{west}
\inheritanchor[from=rectangle]{east}
\backgroundpath{
\southwest \pgf@xa=\pgf@x \pgf@ya=\pgf@y
\northeast \pgf@xb=\pgf@x \pgf@yb=\pgf@y

\pgfmathsetmacro{\pgf@shorten@left}{\pgfkeysvalueof{/tikz/shorten left}}
\pgfmathsetmacro{\pgf@shorten@right}{\pgfkeysvalueof{/tikz/shorten right}}

\pgfpathmoveto{\pgfpoint{0.5 * (\pgf@xa + \pgf@xb)}{\pgf@yb + 5pt}}
\pgfpathlineto{\pgfpoint{\pgf@xa - 8pt + \pgf@shorten@left}{\pgf@ya + 1.5 * \pgf@shorten@left}}
\pgfpathlineto{\pgfpoint{\pgf@xa - 8pt + \pgf@shorten@left}{\pgf@ya}}
\pgfpathlineto{\pgfpoint{\pgf@xb + 8pt - \pgf@shorten@right}{\pgf@ya}}
\pgfpathlineto{\pgfpoint{\pgf@xb + 8pt - \pgf@shorten@right}{\pgf@ya + 1.5 * \pgf@shorten@right}}
\pgfpathclose
}
}

\pgfdeclareshape{langpoint}{
\inheritsavedanchors[from=rectangle] 
\inheritanchorborder[from=rectangle]
\inheritanchor[from=rectangle]{center}
\inheritanchor[from=rectangle]{north}
\inheritanchor[from=rectangle]{south}
\inheritanchor[from=rectangle]{west}
\inheritanchor[from=rectangle]{east}
\backgroundpath{
\southwest \pgf@xa=\pgf@x \pgf@ya=\pgf@y
\northeast \pgf@xb=\pgf@x \pgf@yb=\pgf@y

\pgfmathsetmacro{\pgf@shorten@left}{\pgfkeysvalueof{/tikz/shorten left}}
\pgfmathsetmacro{\pgf@shorten@right}{\pgfkeysvalueof{/tikz/shorten right}}

\pgfpathmoveto{\pgfpoint{0.5 * (\pgf@xa + \pgf@xb)}{\pgf@ya - 2pt}}
\pgfpathlineto{\pgfpoint{\pgf@xa - 8pt}{\pgf@yb - 3 * \pgf@shorten@left + 5pt}} 
\pgfpathlineto{\pgfpoint{\pgf@xa - 8pt}{\pgf@yb -1pt}}
\pgfpathlineto{\pgfpoint{\pgf@xb + 8pt}{\pgf@yb -1pt}}
\pgfpathlineto{\pgfpoint{\pgf@xb + 8pt}{\pgf@yb - 3 * \pgf@shorten@left + 5pt}}
\pgfpathclose
}
}

\pgfdeclareshape{langcopoint}{
\inheritsavedanchors[from=rectangle] 
\inheritanchorborder[from=rectangle]
\inheritanchor[from=rectangle]{center}
\inheritanchor[from=rectangle]{north}
\inheritanchor[from=rectangle]{south}
\inheritanchor[from=rectangle]{west}
\inheritanchor[from=rectangle]{east}
\backgroundpath{
\southwest \pgf@xa=\pgf@x \pgf@ya=\pgf@y
\northeast \pgf@xb=\pgf@x \pgf@yb=\pgf@y

\pgfmathsetmacro{\pgf@shorten@left}{\pgfkeysvalueof{/tikz/shorten left}}
\pgfmathsetmacro{\pgf@shorten@right}{\pgfkeysvalueof{/tikz/shorten right}}

\pgfpathmoveto{\pgfpoint{0.5 * (\pgf@xa + \pgf@xb)}{\pgf@yb +0pt}}
\pgfpathlineto{\pgfpoint{\pgf@xa - 8pt}{\pgf@ya + 3 * \pgf@shorten@left - 5pt}} 
\pgfpathlineto{\pgfpoint{\pgf@xa - 8pt}{\pgf@ya + 1pt}}
\pgfpathlineto{\pgfpoint{\pgf@xb + 8pt}{\pgf@ya + 1pt}}
\pgfpathlineto{\pgfpoint{\pgf@xb + 8pt}{\pgf@ya + 3 * \pgf@shorten@left - 5pt}}
\pgfpathclose
}
}

\pgfdeclareshape{langrect}{
\inheritsavedanchors[from=rectangle] 
\inheritanchorborder[from=rectangle]
\inheritanchor[from=rectangle]{center}
\inheritanchor[from=rectangle]{north}
\inheritanchor[from=rectangle]{south}
\inheritanchor[from=rectangle]{west}
\inheritanchor[from=rectangle]{east}
\backgroundpath{
\southwest \pgf@xa=\pgf@x \pgf@ya=\pgf@y
\northeast \pgf@xb=\pgf@x \pgf@yb=\pgf@y

\pgfmathsetmacro{\pgf@shorten@left}{\pgfkeysvalueof{/tikz/shorten left}}
\pgfmathsetmacro{\pgf@shorten@right}{\pgfkeysvalueof{/tikz/shorten right}}

\pgfpathmoveto{\pgfpoint{\pgf@xa - 8pt}{\pgf@ya + 3 * \pgf@shorten@left - 5pt}} 
\pgfpathlineto{\pgfpoint{\pgf@xa - 8pt}{\pgf@ya + 1pt}}
\pgfpathlineto{\pgfpoint{\pgf@xb + 8pt}{\pgf@ya + 1pt}}
\pgfpathlineto{\pgfpoint{\pgf@xb + 8pt}{\pgf@ya + 3 * \pgf@shorten@left - 5pt}}
\pgfpathclose
}
}

\makeatother

\pgfkeyssetvalue{/tikz/shorten left}{0pt}
\pgfkeyssetvalue{/tikz/shorten right}{0pt}

\tikzstyle{kpoint common}=[draw,fill=white,inner sep=1pt,minimum height=4mm]

\tikzstyle{langstate}=[shape=langcopoint,shorten left=5pt,kpoint common,font=\footnotesize]
\tikzstyle{langeffect}=[shape=langpoint,shorten left=5pt,kpoint common,font=\footnotesize]
\tikzstyle{langstatedash}=[shape=langcopoint,dashed, shorten left=5pt,kpoint common,font=\footnotesize]
\tikzstyle{langeffectdash}=[shape=langpoint,dashed, shorten left=5pt,kpoint common,font=\footnotesize]
\tikzstyle{langbox}=[shape=langrect,shorten left=5pt,kpoint common,font=\footnotesize] 

\tikzstyle{kpoint}=[shape=cornerpoint,shorten left=5pt,kpoint common]
\tikzstyle{kpoint adjoint}=[shape=cornercopoint,shorten left=5pt,kpoint common]

\tikzstyle{kpoint conjugate}=[shape=cornerpoint,shorten right=5pt,kpoint common]
\tikzstyle{kpoint transpose}=[shape=cornercopoint,shorten right=5pt,kpoint common]
\tikzstyle{kpoint symm}=[shape=cornerpoint,shorten left=5pt,shorten right=5pt,kpoint common]

\tikzstyle{black kpoint}=[shape=cornerpoint,shorten left=5pt,kpoint common,fill=black,font=\color{white}]
\tikzstyle{black kpoint adjoint}=[shape=cornercopoint,shorten left=5pt,kpoint common,fill=black,font=\color{white}]
\tikzstyle{black kpointadj}=[shape=cornercopoint,shorten left=5pt,kpoint common,fill=black,font=\color{white}]

\tikzstyle{black dkpoint}=[shape=cornerpoint,shorten left=5pt,kpoint common,fill=black, doubled,font=\color{white}]
\tikzstyle{black dkpoint adjoint}=[shape=cornercopoint,shorten left=5pt,kpoint common,fill=black, doubled,font=\color{white}]
\tikzstyle{black dkpointadj}=[shape=cornercopoint,shorten left=5pt,kpoint common,fill=black, doubled,font=\color{white}]

\tikzstyle{kpointdag}=[kpoint adjoint]
\tikzstyle{kpointadj}=[kpoint adjoint]
\tikzstyle{kpointconj}=[kpoint conjugate]
\tikzstyle{kpointtrans}=[kpoint transpose]

\tikzstyle{big kpoint}=[kpoint, minimum width=1.2 cm, minimum height=8mm, inner sep=4pt, text depth=3mm]

\tikzstyle{wide kpoint}=[kpoint, minimum width=1 cm, inner sep=2pt]
\tikzstyle{wide kpointdag}=[kpointdag, minimum width=1 cm, inner sep=2pt]
\tikzstyle{wide kpointconj}=[kpointconj, minimum width=1 cm, inner sep=2pt]
\tikzstyle{wide kpointtrans}=[kpointtrans, minimum width=1 cm, inner sep=2pt]

\tikzstyle{gray kpoint}=[kpoint,fill=gray!50!white]
\tikzstyle{gray kpointdag}=[kpointdag,fill=gray!50!white]
\tikzstyle{gray kpointadj}=[kpointadj,fill=gray!50!white]
\tikzstyle{gray kpointconj}=[kpointconj,fill=gray!50!white]
\tikzstyle{gray kpointtrans}=[kpointtrans,fill=gray!50!white]

\tikzstyle{gray dkpoint}=[kpoint,fill=gray!50!white,doubled]
\tikzstyle{gray dkpointdag}=[kpointdag,fill=gray!50!white,doubled]
\tikzstyle{gray dkpointadj}=[kpointadj,fill=gray!50!white,doubled]
\tikzstyle{gray dkpointconj}=[kpointconj,fill=gray!50!white,doubled]
\tikzstyle{gray dkpointtrans}=[kpointtrans,fill=gray!50!white,doubled]

\tikzstyle{white label}=[draw,fill=white,rectangle,inner sep=0.7 mm]
\tikzstyle{gray label}=[draw,fill=gray!50!white,rectangle,inner sep=0.7 mm]
\tikzstyle{black label}=[draw,fill=black,rectangle,inner sep=0.7 mm]

\tikzstyle{dkpoint}=[kpoint,doubled]
\tikzstyle{wide dkpoint}=[wide kpoint,doubled]
\tikzstyle{dkpointdag}=[kpoint adjoint,doubled]
\tikzstyle{wide dkpointdag}=[wide kpointdag,doubled]
\tikzstyle{dkcopoint}=[kpoint adjoint,doubled]
\tikzstyle{dkpointadj}=[kpoint adjoint,doubled]
\tikzstyle{dkpointconj}=[kpoint conjugate,doubled]
\tikzstyle{dkpointtrans}=[kpoint transpose,doubled]

\tikzstyle{kscalar}=[kpoint common, shape=EBox, inner xsep=-1pt, inner ysep=3pt,font=\small]
\tikzstyle{kscalarconj}=[kpoint common, shape=WBox, inner xsep=-1pt, inner ysep=3pt,font=\small]


 \tikzstyle{upground}=[circuit ee IEC,ground,rotate=90,scale=2.5]
 \tikzstyle{downground}=[circuit ee IEC,ground,rotate=-90,scale=2.5]
 \tikzstyle{bigground}=[regular polygon,regular polygon sides=3,draw=gray,scale=0.50,inner sep=-0.5pt,minimum width=10mm,fill=gray]


\tikzstyle{arrs}=[-latex,font=\small,auto]
\tikzstyle{arrow plain}=[arrs]
\tikzstyle{arrow dashed}=[dashed,arrs]
\tikzstyle{arrow bold}=[very thick,arrs]
\tikzstyle{arrow hide}=[draw=white!0,-]
\tikzstyle{arrow reverse}=[latex-]
\tikzstyle{cdnode}=[]


\tikzstyle{H}=[-, style=dashed]
\tikzstyle{K}=[-, line width=1pt]
\tikzstyle{Kv}=[-, line width=1pt, ->]
\tikzstyle{Kv<>}=[-,line width=1pt,{<->}]
\tikzstyle{gF}=[-, draw=none, fill={rgb,255: red,191; green,191; blue,191}]
\tikzstyle{KB}=[-, draw=blue, line width=1pt]
\tikzstyle{KR}=[-, draw=red, line width=1pt] 
\tikzstyle{KP}=[-, draw={rgb,255: red,128; green,0; blue,128}, line width=1pt] 
\tikzstyle{KO}=[-, draw={rgb,255: red,255; green,128; blue,0}, line width=1pt]
\tikzstyle{KL}=[-, draw={rgb,255: red,191; green,255; blue,0}, line width=1pt]
\tikzstyle{KHO}=[-, draw={rgb,255: red,255; green,128; blue,0}, style=dashed, line width=1pt]
\tikzstyle{KTG}=[-, draw={rgb,255: red,128; green,128; blue,128}, style=dotted, line width=1pt]
\tikzstyle{KTlG}=[-, draw={rgb,255: red,191; green,191; blue,191}, style=dotted, line width=1pt]
\tikzstyle{KBv}=[-, draw=blue, ->]
\tikzstyle{KOv}=[-, draw={rgb,255: red,255; green,128; blue,0}, line width=1pt, ->]
\tikzstyle{KLv}=[-, draw={rgb,255: red,191; green,255; blue,0}, ->]
\tikzstyle{T}=[-, style=dotted]
\tikzstyle{wF}=[-, fill=white, draw=none]
\tikzstyle{KH}=[-, style=dashed, line width=1pt]
\tikzstyle{Hv}=[->, style=dashed]
\tikzstyle{cv}=[-,right hook->]
\tikzstyle{vv}=[-,->>]
\tikzstyle{v}=[-,->]
\tikzstyle{<>}=[-,<->]
\tikzstyle{Hvv}=[-,->>,style=dashed]
\tikzstyle{Kvp}=[->]
\tikzstyle{KTB}=[-, draw=blue, style=dotted, line width=1pt]
\tikzstyle{KBgF}=[-, fill={rgb,255: red,191; green,191; blue,191}, draw=blue, line width=1 pt]
\tikzstyle{KBggF}=[-, fill={rgb,255: red,128; green,128; blue,128}, draw=blue, line width=1 pt]
\tikzstyle{b-wf}=[-, fill=white]
\tikzstyle{b-gf}=[-, fill={rgb,255: red,191; green,191; blue,191}, draw=black]
\tikzstyle{KHB}=[-, draw=blue, style=dashed, line width=1pt]

\tikzstyle{bigunit}=[dot,fill=white,text depth=-0.2mm]
\tikzstyle{smallblackdot}=[fill=black, inner sep=0mm,minimum width=1mm,minimum height=1mm,draw,shape=circle]
\tikzstyle{smallorangedot}=[fill={rgb,255: red,255; green,128; blue,0}, inner sep=0mm,minimum width=1mm,minimum height=1mm,draw,shape=circle]
\tikzstyle{smallbluedot}=[fill=blue, inner sep=0mm,minimum width=1mm,minimum height=1mm,draw,shape=circle]

\chapter{\label{chapdiscocirc}Distributional Compositional Circuits for English and Urdu Text}

\epigraph{`A sentence is not a state, but a process.... Text is a process that alters meanings of words.'}  
{Bob Coecke,\\ The Mathematics of Text Structure~\cite[pp.~198--199]{coecke2021mathematics}}

\textit{This chapter is adapted from the publication~\cite{DiscocircUrdu}, based on work carried out in collaboration with Vincent Wang-Ma\'{s}cianica, Jonathon Liu, and Bob Coecke. The author of this thesis is the first author of the publication~\cite{DiscocircUrdu}, and played a leading role in all the work described in this chapter. The novel contributions are presented in Sections~\ref{urdugram} and~\ref{urducircuits}.}

\section{Introduction}

Distributional Compositional Circuits (DisCoCirc) is a framework that incorporates both distributional (\emph{i.e.} vectorial) and compositional aspects to model the meaning of natural languages~\cite{coecke2021mathematics}. DisCoCirc gives a process-theoretic account of language and builds further on the earlier work on the Categorical Distributional Compositional (DisCoCat) model for combining grammar and meaning~\cite{CSC, FrobMeanI}. An important distinction between the two is that DisCoCirc is able to model not only the meaning of individual sentences, but also the interaction of sentences giving rise to the meaning of texts generally. The central idea is that the information associated with noun entities appearing in the text (encoded in circuit wires) is updated by sentences (modelled as gates) as the text progresses. 

DisCoCirc admits a two-dimensional string-diagrammatic formalism, inspired by quantum circuits/networks, and therefore offers prospects for quantum-computational natural language processing for texts comprising multiple sentences or paragraphs~\cite{laakkonen2024quantum, duneau2024scalable}. It also has been applied to a number of problems including spatio-temporal models of language meaning~\cite{TalkSpace}, logical and conversational negation in natural language~\cite{rodatz2021conversational, shaikh2021composing}, and solving logical puzzles~\cite{semspace_tiffany_2020, duneau2021parsing}. The relationship between DisCoCirc and discourse-representation theory~\cite{kamp2013discourse} has also been briefly explored~\cite{wang2023distilling}. 

Recently, Wang-Ma\'{s}cianica, Liu and Coecke proposed a method for generating DisCoCirc diagrams for a significant fragment of English~\cite{wang2023distilling}. They started by creating a \textit{hybrid grammar} for English text, incorporating \textit{phrase structure}, \textit{pronominal links}, and \textit{phrase regions}.\footnote{To provide some intuition for these linguistic concepts, here are a few examples.
	
Consider the sentence ``Alice likes Bob''. The \textit{phrase structure} provides a breakdown of the sentence into its basic constituents: ``Alice'' and ``Bob'' are noun phrases, while ``likes'' is a verb phrase.
	
If we have two sentences, ``Alice likes Bob'' and ``Bob hates Charlie'', with ``Bob'' referring to the same person, a \textit{pronominal link} keeps track of this information. When the two sentences are combined into ``Alice likes Bob who hates Charlie'', the pronominal link resolves the reference of the relative pronoun ``who'', establishing a clear link back to its antecedent ``Bob'' and ensuring semantic coherence.
	
Finally, \textit{phrase regions} help show how groups of words belong together. For instance, in the sentence ``Dan sees Eve laugh'', they make it clear that ``Eve laugh'' falls within the scope of ``sees''. In other words, the region headed by ``sees'' includes both ``Eve'' and ``laugh'', showing that Dan is observing the entire action of Eve laughing.} It is called a hybrid grammar because it integrates key ideas from different grammatical formalisms. The main structure is based on transformational phrase structure grammars~\cite{peters1973generative}, with modifications inspired by pregroup grammars~\cite{lambek2008word}. Pronominal links are derived from discourse representation theory~\cite{kamp2013discourse}, whereas phrase regions are based on dependency grammars~\cite{tesniere2015elements}.

These hybrid grammar representations of text were translated into DisCoCirc \textit{text circuits}, via an intermediate structure called \textit{text diagrams} that involve string diagrams. The entire translation process preserves the compositionality and connectedness of text meaning. 

The same work describes how in the reverse direction, text circuits can be used as a generative grammar. For each freely generated text circuit, we can write some corresponding text. In this chapter, we use `text' to refer not only to a string of words forming sentences, but to this string of words further endowed with a hybrid grammar structure. The text corresponding to a given text circuit is not unique owing to the loss of grammatical `bureaucracy'~\cite{wang2023distilling} in the passage from one-dimensional syntax to two-dimensional text circuits.\footnote{Some of this was already observed in Ref.~\cite{GramEqs}.}
Here, grammatical bureaucracy is used in a broad sense to refer to all of the stylistic choices one must make when communicating some desired meaning in the form of a text.
That is, whenever two different texts communicate what is essentially the same meaning, we attribute the differences in the structure of these texts to grammatical bureaucracy.
Thus, grammatical bureaucracy includes syntactic rules like those governing word order, but also choices like the use of pronouns, whether to use a single long sentence with multiple clausal constructions or multiple short sentences, etc.

It was, in particular, suggested that because of eliminating grammatical bureaucracy and stylistic choices embedded in a particular natural language, text circuits are to some degree language-independent~\cite{wang2023distilling}. We take this suggestion seriously in this chapter, and show how DisCoCirc does indeed undo the grammatical bureaucracy related to word and phrase order for restricted fragments of English and Urdu. English and Urdu are both Indo-European languages and share similarities in grammatical rules and idioms. While they both follow the noun-verb distinction, they differ in sentence structure due to differences in word order. For instance, English uses a subject–verb–object (SVO) order, whereas Urdu follows a subject–object–verb (SOV) order.

Our main argument is structured as follows.

\begin{itemize}
	\item In Ref.~\cite{wang2023distilling}, a hybrid grammar was developed for English. A surjection from the set of all English text generated with the English hybrid grammar to the set of all English text circuits was demonstrated:
	\[ \text{English text} \twoheadrightarrow \text{English text circuits} \]
	\item In a similar vein, in this chapter, we describe how the hybrid grammar can be adapted for Urdu. We then provide rules for its translation into text diagrams and text circuits, which is essentially the same as in Ref.~\cite{wang2023distilling}.\footnote{We expect texts in other Indo-European languages to be amenable to translation into text diagrams and circuits. However, our approach may not be directly applicable to languages like Arabic, which lie outside this language family. We will return to this point later.} We show that this gives a surjective map from the set of all Urdu text generated with Urdu hybrid grammar to the set of all Urdu text circuits:
	\[ \text{Urdu text} \twoheadrightarrow \text{Urdu text circuits} \]
	\item For these restricted fragments, there is a clear isomorphism between the hybrid grammars for English and Urdu: 
	\[ \text{English hybrid grammar} \simeq \text{Urdu hybrid grammar} \] 
	\item Given the above correspondence, it turns out that text circuits for English and Urdu become the same, up to translating the labels on the gates. In other words, the following diagram commutes:
	\[
	\begin{tikzcd}
		\text{English text} \arrow[r, twoheadrightarrow] & \text{English circuits} \arrow[d, leftrightarrow, "Dictionary\qquad\ \ "] \\ 
		\text{Urdu text} \arrow[u, leftrightarrow, "\begin{matrix}
			Grammar\  \&\
			Dictionary
		\end{matrix}"] \arrow[r, twoheadrightarrow] & \text{Urdu circuits}
	\end{tikzcd}
	\]
	Now, formally speaking, our texts consist of a hybrid grammar structure with labels (words) provided by a dictionary. If we restrict to just the grammatical structures and forget the language-specific word labels, the circuits for English and Urdu become literally the same. In this case, the above diagram reduces to:
	\[
	\begin{tikzcd}
		\text{English text} \arrow[dr, twoheadrightarrow] \arrow[dd, leftrightarrow] &  \\
		& \text{Circuits} \\
		\text{Urdu text} \arrow[ur, twoheadrightarrow] &
	\end{tikzcd}
	\]
\end{itemize}

The chapter is organised as follows. The hybrid grammar, text diagrams and text circuits for English are introduced in Sections~\ref{enggram},~\ref{textdiagrams} and~\ref{textcircuits} respectively. A hybrid grammar for Urdu is introduced in Section~\ref{urdugram}, followed by the main results related to Urdu text diagrams and circuits in Section~\ref{urducircuits}. Section~\ref{conclusion} concludes the chapter.

\section{A Hybrid Grammar for English Text}\label{enggram}

In this section, we introduce the hybrid (generative) grammar for English developed in Ref.~\cite{wang2023distilling}. We focus on modelling the language fragment that includes nouns, adjectives, verbs (intransitive and transitive), adverbs (modifying verbs), adpositions (for intransitive verbs), verbs with sentential complements, conjunctions, and pronominal links (subject relative pronouns and object relative pronouns). The technical terms will be explained later.
	
\begin{disclaimer}	
For ease of analysis, the modelled language fragment is limited to propositional statements in a single tense. It is also assumed that adverbs and adjectives can be stacked indefinitely and appear to the left of their respective nouns and verbs. Additionally, determiners are not explicitly modelled.
\end{disclaimer}

In  Ref.~\cite{wang2023distilling}, first, \textit{simple} sentences were modelled, containing verbs, adverbs, adjectives and adpositions. Then, \textit{pronominal links} were introduced to account for recurring nouns and pronoun-referent pairs. Rewrite rules were introduced that allowed for the fusion of simple sentences into more complex ones, and the introduction of relative pronouns. Rules for modelling verbs with sentential complement and \textit{phrase scope boundary} were introduced to accommodate compound sentences formed of components that are themselves sentences. We shall provide detailed descriptions and examples later.

The hybrid grammar begins with a standard \textit{phrase structure} grammar that generates the simple sentences, based on a \textit{string rewrite system} with finitely many \textit{production rules} of the form $\alpha \mapsto \beta$~\cite{hopcroft2001introduction}.
These are valid transformations of strings of symbols represented by $\alpha$ and $\beta$. Individual symbols may be phrase components or entire words of the language we are modelling. In the latter case, the symbol is \textit{terminal}, meaning no more rewrite rules can be further applied. 
In a grammar such as ours, a particular language comprises all strings of terminal symbols---forming grammatically correct sentences---that can be generated by applying finitely many production rules (associated with that language) to a start symbol, $\texttt{S}$. Only those productions that yield grammatical sentences are considered valid. For example, using the rules:
\begin{equation*}
	\begin{split}
		\texttt{S} &\mapsto \texttt{NP}_1 \cdot \texttt{TVP} \cdot \texttt{NP}_2 \\
		\texttt{NP}_1 &\mapsto \texttt{\underline{John}} \\
		\texttt{TVP} &\mapsto \texttt{\underline{reads}} \\
		\texttt{NP}_2 &\mapsto \texttt{\underline{books}}
	\end{split}    
\end{equation*}
where $\texttt{NP}$ and $\texttt{TVP}$ represent noun phrase and transitive verb phrase respectively, underlining denotes terminals, and $\cdot$ denotes string concatenation; we can generate the sentence \texttt{\underline{John reads books.}}: 
\begin{equation*}
	\begin{split}
		\texttt{S} &\mapsto \texttt{NP}_1 \cdot \texttt{TVP} \cdot \texttt{NP}_2 \\
		&\mapsto  \texttt{\underline{John}}  \cdot \texttt{TVP} \cdot \texttt{NP}_2 \\
		&\mapsto  \texttt{\underline{John}}  \cdot  \texttt{\underline{reads}}\cdot \texttt{NP}_2 \\
		&\mapsto  \texttt{\underline{John}}  \cdot  \texttt{\underline{reads}}\cdot \texttt{\underline{books}}
	\end{split}    
\end{equation*}

The rewriting of strings can be represented as two-dimensional tree diagrams, read from top to bottom. Our example sentence can be represented as
\[\scalebox{0.7}{\input{urdu/johnreadsbook.tikz}}\]

\begin{convention}
	In this chapter, generative production rules $\alpha \mapsto \beta$ are depicted from top to bottom.
\end{convention}

\begin{disclaimer}
	We restrict to a version of hybrid grammar~\cite{wang2023distilling} generated by the explicitly context-free rules. A context-free rule is one that applies to an individual non-terminal symbol without depending on the surrounding symbols. 
	We also exclude reflexive pronouns. 	
\end{disclaimer}	

Symbolic labels can be dropped for intermediate edges and, instead, a colour coding can be used. See Table~\ref{tab:colourcode} for the colour coding employed in this chapter. 

\begin{table}[H]
	\centering
	\scalebox{1}{
		\begin{tabular}{|c|c|c|}
			\hline
			\text{Grammatical type} & Symbol/word & \text{Colour code} \\ \hline
			\text{Noun phrase} & NP & \input{urdu/engdiags/black.tikz}   \\ \hline
			\text{Intransitive verb phrase} & IVP & \input{urdu/engdiags/green.tikz}  \\ \hline
			\text{Transitive verb phrase} & TVP & \input{urdu/engdiags/green.tikz}   \\ \hline
			\text{Sentential complement verb} & SCV & \input{urdu/engdiags/green.tikz}   \\ \hline			
			\text{Adjective} & ADJ & \input{urdu/engdiags/orangedash.tikz}   \\ \hline					
			\text{Adverb} & ADV & \input{urdu/engdiags/orange.tikz}   \\ \hline		
			\text{Conjunction} & CNJ & \input{urdu/engdiags/blue.tikz}   \\ \hline					
			\text{Adposition} & ADP & \input{urdu/engdiags/purple.tikz}   \\ \hline												
			\text{Copula} & \texttt{\underline{is}} (English), \texttt{\underline{hai}} (Urdu) & \input{urdu/engdiags/red.tikz}   \\ \hline				
			\text{Phrase scope boundary} & $\textcolor{blue}{\texttt{(}}$, $\textcolor{blue}{\texttt{)}}$ & \input{urdu/engdiags/blue.tikz}   \\ \hline									
		\end{tabular}
	}
	\caption{Colour coding for grammatical types}
	\label{tab:colourcode}
\end{table}

In this section, a \textit{phrase structure grammar} is introduced by providing the tree fragments for various grammatical types. We start with \textit{simple} sentences, which include single verbs that do not take sentential complements. 

\subsection{Simple Sentences}

The \textit{verb} in a simple sentence may be \textit{intransitive} or \textit{transitive}. If the verb is intransitive, the sentence comprises a noun phrase followed by an intransitive verb phrase. For example, \texttt{\underline{Bob laughs.}} The production rule for generating this sentence is $\texttt{S} \mapsto \texttt{NP} \cdot \texttt{IVP}$ and the tree fragment is
\begin{equation}\label{IVPrule}
\input{urdu/engdiags/IVcsg.tikz}
\end{equation}
On the other hand, if the verb is transitive, the sentence comprises a noun phrase, followed by a transitive verb phrase and another noun phrase. For instance, \texttt{\underline{Alice loves Bob.}} The production rule for this case is $\texttt{S} \mapsto \texttt{NP}_1 \cdot \texttt{TVP} \cdot \texttt{NP}_2$  and the tree fragment is
\begin{equation}\label{TVPrule}
\input{urdu/engdiags/TVcsg.tikz}
\end{equation}

\begin{remark}
In the language fragment considered here, we require the subject and object of a transitive verb to be distinct.
\end{remark}

Apart from the aforementioned production rules, there are terminal rules for each verb type: $\texttt{IVP} \mapsto \texttt{\underline{IV}}$ and $\texttt{TVP} \mapsto \texttt{\underline{TV}}$, respectively, where terminals are denoted by underlining.
\[
\input{urdu/engdiags/IVterminal.tikz} \quad\quad\quad \input{urdu/engdiags/TVterminal.tikz} 
\]

\begin{convention}
	From here onward, we do not explicitly show the terminal rules in the diagrams. The terminal symbols can be directly inferred from the finished derivations. 
\end{convention}

A simple sentence may consist of \textit{adjectives} in two ways: adjectives can appear to the left of the noun phrase that they are modifying, for example, \texttt{\underline{tall Bob}}; alternatively, a single adjective can appear to the right of a noun phrase after the copular verb \texttt{\underline{is}}, for example, \texttt{\underline{Bob is tall.}} The corresponding production rules are $\texttt{NP} \mapsto \texttt{ADJ} \cdot \texttt{NP}$ and $\texttt{NP} \mapsto \texttt{NP} \cdot \texttt{\underline{is}} \cdot \texttt{ADJ}$, respectively. The tree fragments are as follows:
\begin{equation}\label{ADJrule}
\input{urdu/engdiags/ADJcsg.tikz} \quad \quad \quad \input{urdu/engdiags/ADJiscsg.tikz} 
\end{equation}

In simple sentences, \textit{adverbs} can appear to the left of verbs. An example sentence is: \texttt{\underline{Alice happily talks.}} For intransitive and transitive verb phrases, the adverb production rules are $\texttt{IVP} \mapsto \texttt{ADV} \cdot \texttt{IVP}$ and $\texttt{TVP} \mapsto \texttt{ADV} \cdot \texttt{TVP}$ respectively. The tree fragments are:
\[
\input{urdu/engdiags/ADVIVcsg.tikz} \quad \quad \quad \input{urdu/engdiags/ADVTVcsg.tikz} 
\]

A simple sentence may contain an \textit{adposition} that appears to the right of an intransitive verb phrase, followed by a noun phrase. For example, adding an adposition to the sentence \texttt{\underline{Alice smiles.}}, we can get \texttt{\underline{Alice smiles at Bob.}} The production rule for an adposition is $\texttt{IVP} \mapsto \texttt{IVP} \cdot \texttt{ADP} \cdot \texttt{NP}$, with the tree fragment given by
\begin{equation}\label{ADPrule}
\input{urdu/engdiags/ADPivcsg.tikz} 
\end{equation}

\begin{example}
	Using the \texttt{IVP}-production rule yields $\texttt{NP} \cdot \texttt{IVP}$, corresponding to the grammatical structure of \texttt{\underline{Alice laughs.}} 
	\[
\input{urdu/engdiags/treeex2.tikz} 
\]	
	Applying the \texttt{ADP}-production rule to \texttt{IVP}---before the application of the terminal rule---gives the grammatical structure of \texttt{\underline{Alice laughs at Bob.}}
	\[
	\input{urdu/engdiags/treeex.tikz} 
	\]
\end{example}

\subsection{Compound Sentences}

\textit{Compound} sentences contain more than one verb. Compound sentences can be obtained by joining two simple sentences using \textit{relative pronouns}. Pronoun information is modelled using \textit{pronominal links} identifying nouns in the same or different sentences. These are represented by arrows pointing at the identified nouns drawn below the trees. For instance, the two instances of  \texttt{\underline{Bob}} in the sentences  \texttt{\underline{Bob is kind.}} and \texttt{\underline{Bob smiles at Alice.}} can be identified as follows:
	\[
\input{urdu/engdiags/treeex3.tikz} 
\]
If there is a pronominal link, the later occurrences of a noun can be replaced by pronouns referring to the earlier occurrence.
	\[
\input{urdu/engdiags/treeex4.tikz} 
\]
In certain settings of pronominal links, two consecutive sentences (\emph{i.e.,} their parse trees) can be joined together using relative pronouns. Consider the configuration in which the subject noun of the second tree $\texttt{S}_2$ points to the object noun of the first tree $\texttt{S}_1$. Then, a \textit{subject relative pronoun} can be used to replace the subject noun in the second tree $\texttt{S}_2$, and the two sentences can be fused. 
Diagrammatically, it can be represented by the rewrite rule:
\begin{equation}\label{esrp}
\input{urdu/engdiags/srprule.tikz} 
\end{equation}
On the right hand side, the big triangle denotes the single sentence formed by the fusion of the two sentences. The above rewrite rule modifies phrase structure trees and hence is an example of a \textit{transformational grammar rule}.

\begin{example}
	In the sentences \texttt{\underline{Bob likes Alice.}} and \texttt{\underline{Alice hates Eve.}}, the two occurrences of \texttt{\underline{Alice}} can be identified using a pronominal link. 
		\[
	\input{urdu/treeex5.tikz} 
	\]
	The two trees can be joined by replacing the second occurrence of \texttt{\underline{Alice}} with \texttt{\underline{who}}. This results in the following diagram: 
		\[
\input{urdu/treeex6.tikz} 
\]	
\end{example}

A subject relative pronoun can also be used to relate subject nouns of two sentences in the following way. If the subject noun of the second parse tree $\texttt{S}_2$ refers to the subject noun of the first parse tree $\texttt{S}_1$, the transformation rule yields: the noun phrase of $\texttt{S}_1$, followed by $\texttt{S}_1$ with its noun phrase replaced by the relative pronoun, followed by $\texttt{S}_2$ without its subject noun. Diagrammatically, it can be denoted as
\begin{equation}\label{essrp}
\input{urdu/engdiags/ssrprule.tikz} 
\end{equation}
where the multiarrow identifies the pronominally linked nouns.

\begin{example}
Consider the following example again.
\[
\input{urdu/engdiags/treeex3.tikz} 
\]
Applying the transformational rule, we get the sentence
		\[
\input{urdu/treeex7.tikz} 
\]	
The artefacts $\texttt{!}$ (denoting isolated nouns) and \textvisiblespace  \  (denoting blank labels) will be eliminated in the setting of text diagrams, introduced in Section~\ref{textdiagrams}.
\end{example}

If the object noun of the second sentence $\texttt{S}_2$ refers to the object noun of the first sentence $\texttt{S}_1$, the two sentences can be fused using an \textit{object relative pronoun}. The transformational rule gives $\texttt{S}_1$, followed by the object relative pronoun, followed by $\texttt{S}_2$ with its object noun removed. Diagrammatically, it can be represented by
\begin{equation}\label{eorp}
\input{urdu/engdiags/orprule.tikz} 
\end{equation}

\begin{example}
	Consider the sentences: \texttt{\underline{Bob likes Alice.}} and \texttt{\underline{Eve hates Alice.}} We identify the two occurrences of the object noun \texttt{\underline{Alice}} using a pronominal link. 
			\[
	\input{urdu/engdiags/treeex8.tikz} 
	\]	
	Applying the transformational rule, we have 
				\[
	\input{urdu/engdiags/treeex9.tikz} 
	\]	
\end{example}

\begin{remark}
	The language fragment considered in this chapter does not include reflexive pronouns. Therefore, we do not discuss intra-sentence pronominal links, introduced in Ref.~\cite{wang2023distilling}.
\end{remark}

Compound sentences can also be obtained by including in sentences phrases that are themselves sentences. Two such kinds of compound sentences are introduced here.

A compound sentence can be formed using a \textit{verb with a sentential complement}. Examples of such verbs are \texttt{\underline{sees}} and \texttt{\underline{thinks}}. The compound sentence includes a noun phrase, followed by a verb, followed by a sentence that is the sentential complement of the verb. For example, in \texttt{\underline{Alice sees Bob laugh.}}, the sentential complement is \texttt{\underline{Bob laughs.}} The production rule is given by $\texttt{S} \mapsto \texttt{NP} \cdot \texttt{SCV} \cdot \textcolor{blue}{\texttt{(}} \cdot \texttt{S} \cdot \textcolor{blue}{\texttt{)}}$ where the types $\textcolor{blue}{\texttt{(}}$ and $\textcolor{blue}{\texttt{)}}$ represent the boundaries of the phrase scope. Its tree fragment is as follows.
\[
 \input{urdu/engdiags/SCVscopecsg.tikz} 
\]
\begin{example}
	The tree diagram of the sentence \texttt{\underline{Alice sees Bob laugh.}} is given by
				\[
	\input{urdu/engdiags/treeex10.tikz} 
	\]	
\end{example}

A compound sentence can also be formed by joining two sentences using \textit{conjunctions}. The production rule is $\texttt{S} \mapsto \textcolor{blue}{\texttt{(}} \cdot \texttt{S} \cdot \textcolor{blue}{\texttt{)}} \cdot \texttt{CNJ} \cdot \textcolor{blue}{\texttt{(}} \cdot \texttt{S} \cdot \textcolor{blue}{\texttt{)}}$ and its tree fragment is given by.
\[
\input{urdu/engdiags/CNJscopecsg.tikz} 
\]

\begin{example}
	The phrase structure of the sentence \texttt{\underline{Alice laughs and Bob cries.}} is as follows.
				\[
\input{urdu/engdiags/treeex11.tikz} 
\]			
\end{example}

\section{Text Diagrams}\label{textdiagrams}

\textit{Text diagrams} were introduced in~\cite{wang2023distilling} as an intermediate representation between hybrid grammar and text circuits. In a single mathematical framework, text diagrams capture the connectedness of meaning or grammatical structure, including not just the grammatical relations of simple sentences but also pronominal links as well as phrase scope.

Text diagrams do away with artefacts such as those handling pronominal links. Based on string diagrams~\cite{SelingerSurvey, BaezLNP, CKbook}, text diagrams offer a structural framework for boxes with inputs and outputs that allow for parallel and sequential composition.

\begin{convention}
		In this chapter, all string diagrams must be read from top to bottom and left to right.
\end{convention}

\subsection{Simple Sentences}

When moving from hybrid grammar to text diagrams, the phrase structure rules of the previous section are modified in the following way: each $\texttt{S}$-type is replaced by $\texttt{NP}$-types, the number of which depends on the sentence. It is ensured that every rule that includes $\texttt{NP}$-types has the same number of input and output $\texttt{NP}$-types. The modified rules, corresponding respectively to rules (\ref{IVPrule}), (\ref{TVPrule}), (\ref{ADJrule}), and (\ref{ADPrule}), are given below:
\[
\input{urdu/IVgraphl.tikz} \quad \quad 
\input{urdu/TVgraphl.tikz} \quad \quad  \input{urdu/ADJisgraphl.tikz}  \quad \quad \input{urdu/ADPIVgraphl.tikz}
\]

\begin{convention}
	To guide the eye, we vertically align the ends of input wires with those of the corresponding output wires in text diagrams (and, later, in text circuits as well). Moreover, we bend the wires to make them point upward or downward, reflecting their roles as inputs or outputs string-diagrammatically. We also drop labels that appear more than once.
	\[
	\input{urdu/IVgraphb.tikz} \quad \quad 
	\input{urdu/TVgraphb.tikz} \quad \quad  \input{urdu/ADJisgraphb.tikz}  \quad \quad \input{urdu/ADPIVgraphb.tikz}
	\]
	The copular verb $\texttt{\underline{is}}$ can be eliminated via the rewrite
	\[
	\input{urdu/ADJistransform.tikz} 
	\]
	converting the postpositonal adjective into a prepositional one.
\end{convention}

\begin{example}
	The sentence  \texttt{\underline{Alice laughs at Bob.}} has the following grammatical structure
		\[
	\input{urdu/engdiags/treeex.tikz} 
	\]
which, in text diagrams, is given by
				\[
\input{urdu/textdiag5b.tikz} 
\]		
\end{example}

This modification allows composition of tree fragments (while respecting the grammatical types). In text diagrams, pronominal links and phrase scope constructions are part of the same unified mathematical framework. Moreover, unlike in the hybrid grammar tree diagrams, wires in text diagrams may cross one another.

\subsection{Compound Sentences}

In the passage from hybrid grammar to text diagrams, the (multi)arrows become wires linking the pronominally-linked noun wires.
\begin{equation}\label{link2d}
\scalebox{0.9}{\input{urdu/link2diagram.tikz}} 
\end{equation}
The pronominal link wires can be eliminated using the diagram rewrite rules:
							\[
\input{urdu/prontransforms.tikz} 
\]
This implies that for two text diagrams $\mathcal{D}_1$ and $\mathcal{D}_1$, we have the following rewrite rules.
\begin{equation}\label{pronlink}
\input{urdu/pronlinkrule.tikz} 
\end{equation}
\begin{equation}\label{pronlink2}
\input{urdu/pronlinkrule2.tikz} 
\end{equation}

\begin{example}
	Consider the following hybrid grammar text.
			\[
	\input{urdu/treeex5.tikz} 
	\]
	Each tree can be replaced with the corresponding text diagram. 
			\[
\input{urdu/textdiag2.tikz} 
\]	
The pronominal link arrow can be replaced by a wire.
			\[
\input{urdu/textdiag3.tikz} 
\]		
Notice that there is only one available output wire corresponding to \texttt{\underline{Alice}}, and it is not labelled by the pronoun $\texttt{\underline{who}}$. Using the rewrite rule for pronoun wire elimination, we get the following diagram. 
			\[
\input{urdu/textdiag4.tikz} 
\]		
\end{example}

\begin{example}
	Consider the following example sentence again. 
	\[
	\input{urdu/treeex7.tikz} 
	\]	
	The artefacts $\texttt{!}$ and \textvisiblespace \  were required in the hybrid grammar to account for relative pronouns. In text diagrams, these artefacts are eliminated in the following way. Replacing the grammar trees with the corresponding text diagrams
				\[
	\input{urdu/textdiag5.tikz} 
	\]		
	and replacing the multiarrows with pronominal link wires, we get
				\[
\input{urdu/textdiag6.tikz} 
\]			
	which can be rewritten as follows.
				\[
\input{urdu/textdiag7.tikz} 
\]				
\end{example}

Moving from hybrid grammar to text diagrams, the sentence type $\texttt{S}$ was replaced by a sentence-dependent number of $\texttt{NP}$ wires. Therefore, string diagrams with regions~\cite{wang2023distilling} are needed to accommodate phrase scope. 

\begin{convention}
	\textit{Phrase regions} delineate phrase scope in text diagrams. These regions act as planar obstacles to everything except $\texttt{NP}$ wires and pronominal link wires. This means that only $\texttt{NP}$ and pronominal link wires can enter or exit the phrase regions.
	$\texttt{\underline{NP}}$ labels must always be placed outside the phrase regions.  
\end{convention}

The text diagrams corresponding to verbs with sentential complements and conjunctions are respectively given by
\[
\input{urdu/SCVtranslate.tikz} \quad \quad \quad \input{urdu/CNJtranslate.tikz}
\]
where $i$ and $j$ denote the number of $\texttt{NP}$ wires.

\begin{example}
	The sentence \texttt{\underline{Alice sees Bob laugh.}} corresponds to the following text diagram.
			\[
\input{urdu/textdiag8.tikz} 
\]		
\end{example}

\begin{example}
	The text diagram of the sentence \texttt{\underline{Alice laughs and Bob cries.}} is as follows.
				\[
	\input{urdu/textdiag9.tikz} 
	\]		
\end{example}

\section{Text Circuits}\label{textcircuits}

\textit{Text circuits} are composed of wires (or types), boxes (or first-order gates), and boxes with holes (or second-order gates). Nouns are represented by wires with labels:
\[
\input{urdu/nounABC.tikz} 
\]
Here, each wire represents a distinct noun.

Adjectives are represented by single-input-single-output boxes or gates that update or act upon noun wires:
\[
\input{urdu/ADJgate.tikz}
\]

Like adjectives, intransitive verbs are also represented by single-input-single-output gates, whereas transitive verbs are given by double-input-double-output gates as they act on subject and object noun wires:
\[\input{urdu/IVgate.tikz} \quad\quad\quad \input{urdu/TVgate.tikz}
\]

Adverbs modify verbs, which are themselves boxes. Hence, adverbs are represented by second-order gates, which are boxes with holes that are to be filled with verb boxes. The shape of the expected verb box is denoted by the wires inside the holes.
\[
\input{urdu/ADVbox.tikz}
\]
When the holes in the adverb boxes are filled, gates are obtained.
\[
\input{urdu/ADVgate.tikz}
\]

Adpositons also act upon verbs and add another noun wire to the right of the circuit. An adposition acting on an intransitive verb is given by a box with a hole expecting an intransitive verb circuit.
\[
\input{urdu/ADPIVbox.tikz}
\]

Verbs with sentential complements are represented by second-order gates that add a noun wire to the left of the circuit. Since the sentential complement can have an arbitrary number of noun wires, these gates form a class of boxes to account for the different number of wires.
\[
\input{urdu/SCVbox.tikz}
\]
For conjunctions as well, there is a class of second-order gates. The conjunction boxes have two holes that accept boxes/gates that may have some overlapping noun wires.
\[
\input{urdu/CNJbox2.tikz}
\]
As special cases, the conjunction boxes may have disjoint or completely shared noun wires.
\[
\input{urdu/CNJbox.tikz} \qquad\qquad\qquad\qquad \input{urdu/CNJbox3.tikz}
\]

Gates can be composed in parallel or series/sequence to form \textit{text circuits}. If no noun labels are shared between the gates, they are composed in parallel. Conversely, if the noun labels are shared, the gates are composed in sequence with the shared noun wires matched. 

\begin{example}
	The text \texttt{\underline{Alice likes Bob. Bob hates Eve.}} can be drawn as the following circuit.
	\[
	\input{urdu/textcircuit1.tikz}
	\]
\end{example}

\begin{convention}
In text circuits, nouns are represented by wires, which may cross or twist past one another. Two circuits are considered the same if they have the same connectivity of gates. This allows us to eliminate unnecessary wire twists to work with simpler diagrams. An example is as follows.
	\[
\input{urdu/textcircuit6b.tikz} \ \ = \ \  \input{urdu/textcircuit6.tikz}
\]
\end{convention}

\begin{example}
	The text \texttt{\underline{Dan knows Bob likes Alice whom Eve smiles at.}} can be drawn as the following circuit.
		\[
	\input{urdu/textcircuit2.tikz}
	\]	
\end{example}

\begin{thesis}[Text Circuit Thesis; Thesis~5.17 in~\cite{wang2023distilling}]\label{circuitthesis}
	Equal text circuits correspond to equal text meanings.
\end{thesis}

\begin{theorem}[Text Circuit Theorem; Theorem~5.1 in~\cite{wang2023distilling}]\label{textcircth}
Let $\texttt{E}$ denote the set of generators of the hybrid English grammar.\footnote{In this chapter, for the sake of simplicity, we focus on a restricted version of the hybrid grammar developed in Ref.~\cite{wang2023distilling}. We do not include reflexive pronouns. Moreover, the following three hybrid grammar rules are excluded from our analysis.
			\[
	\input{urdu/engdiags/ADPTVcsg.tikz} \quad\quad 	\input{urdu/engdiags/SCV1.tikz} \quad \quad 	\input{urdu/engdiags/SCV2.tikz}
	\]
The first rule on the left involves an adposition (ADP) and a transitive verb phrase (TVP), while the remaining rules allow noun wires to exit phrase scope. Note that these rules are context-sensitive and do not apply to individual non-terminal symbols.} 
Let $\texttt{T}_\texttt{E}$ denote the set of all English text constructed with grammar $\texttt{E}$, and let $\texttt{C}_\texttt{E}$ denote the set of all text circuits for English. Then, there exists a surjection $\texttt{T}_\texttt{E} \twoheadrightarrow \texttt{C}_\texttt{E}$.
\end{theorem}

\begin{remark}
The significance of Theorem~\ref{textcircth} stems from Thesis~\ref{circuitthesis}. Theorem~\ref{textcircth} implies that the text corresponding to a given text circuit may not be unique. According to Thesis~\ref{circuitthesis}, all texts corresponding to the same text circuit share the same meaning. In this sense, DisCoCirc text circuits eliminate differences between texts that convey the same meaning.

When two different texts express essentially the same meaning, we attribute their structural variation to grammatical bureaucracy. This includes stylistic choices such as the use of pronouns, the preference for a single long sentence with multiple clauses versus several shorter ones, and similar considerations. This grammatical bureaucracy is eliminated in the transition from the linear syntax of text to the two-dimensional structure of text circuits.
\end{remark}

\begin{example}
		Consider the following hybrid grammar text.
	\[
\input{urdu/treeex36.tikz} 
\]	
Replacing the grammar trees with the corresponding text diagrams, we have the following diagram.
				\[
\input{urdu/textdiag6b.tikz} 
\]		
Composing the pronominal link wires, we get
				\[
\input{urdu/textdiag7b.tikz} 
\]
which can be expressed as gates and boxes to yield a circuit.
				\[
\input{urdu/textcircuit3b.tikz}
\]		

\end{example}

\section{A Hybrid Grammar for Urdu Text}\label{urdugram}

It was shown in Ref.~\cite{wang2023distilling} that English text generated from the English hybrid grammar is surjective to text circuits. In this chapter, we demonstrate a similar result for Urdu.

First, we develop a hybrid grammar for Urdu text. The language fragment we consider for Urdu is the same as that for English discussed in the previous sections.

Since a language is specified by set of production rules, different production rules lead to different languages. In English, the sentence \texttt{\underline{John reads books.}} has the following structure. 
\[\scalebox{0.7}{\input{urdu/johnreadsbook.tikz}}\]
Translating to Urdu, \texttt{\underline{John reads books.}}  becomes: 
\begin{center}
    \includegraphics[scale=0.22]{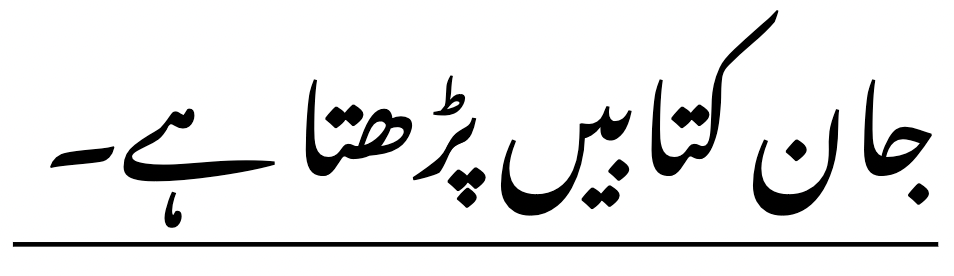}
\end{center}
It can be transliterated\footnote{Urdu script is written from right to left. Throughout this chapter, for ease of readability and linguistic analysis, we shall use English transliteration of Urdu text. However, to dispel any ambiguity which may arise from transliteration, we shall provide the Urdu script and the respective phrase meanings in English as well. 
} into English as
\begin{center}
	\texttt{\underline{John}} \ \  \ \texttt{\underline{kitabein}}\ \ \ 
 \texttt{\underline{parhta hai}} (Urdu) \\
	\ \texttt{\underline{John}} \ \ \ \ \ \   \texttt{\underline{books}} \ \  \ \ \  \  \texttt{\underline{reads}}  \ \ \ \ \ \ \ \ \ \ \ \ \ \ \ \ \ \ \ \  \  
\end{center}

\noindent Using the production rules
\begin{equation*}
	\begin{split}
		\texttt{S} &\mapsto \texttt{NP}_1 \cdot \texttt{NP}_2 \cdot \texttt{TVP} \\
		\texttt{NP}_1 &\mapsto \texttt{\underline{John}} \\
		\texttt{NP}_2 &\mapsto \texttt{\underline{kitabein}} \\
		\texttt{TVP} &\mapsto \texttt{\underline{parhta hai}} \\
	\end{split}    
\end{equation*}
we can generate the sentence under discussion:
\begin{equation*}
	\begin{split}
		\texttt{S} &\mapsto \texttt{NP}_1 \cdot \texttt{NP}_2 \cdot \texttt{TVP} \\
		&\mapsto \texttt{\underline{John}}  \cdot \texttt{NP}_2  \cdot \texttt{TVP} \\
		&\mapsto\texttt{\underline{John}} \cdot  \texttt{\underline{kitabein}}  \cdot \texttt{TVP}  \\  
		&\mapsto \texttt{\underline{John}}  \cdot \texttt{\underline{kitabein}}  \cdot \texttt{\underline{parhta hai}}
	\end{split}    
\end{equation*}
Its tree diagram is given by
\[\scalebox{0.7}{\input{urdu/johnkitaabe.tikz}}\]

From this example, we can immediately spot an obvious difference between English and Urdu---the order of subject, verb and object: in Urdu, the verb is placed at the end of the sentence. 
This contrast plays a significant role in differentiating Urdu and English grammars.

We now develop production rules and tree diagrams for the fragment of Urdu text (including verbs, adjectives, adverbs, adpositions, pronominal links, phrase scope) corresponding to that of English discussed in Section~\ref{enggram}. In so doing, we realise that many of the rules and tree fragments are in fact the same. The ones that are different differ mainly in the relative placement of the verb. See Table~\ref{tab:diff_gram}. The differences are summarised as follows: 
\begin{itemize}
	\item In contrast to subject-verb-object (SVO) as in English, Urdu has the subject-object-verb (SOV) order.
	\[
	\input{urdu/urdudiags/TVPgram.tikz} 
	\]
	\item The placement of adpositions differs from that in English; in Urdu, the verb simply comes last.   
	\[
	\input{urdu/urdudiags/ADPivgrammar.tikz} 	
	\] 
	\item Sentential complements precede verbs in Urdu.
	\[
		\input{urdu/urdudiags/SCVgram.tikz} 
	\]
	\item In Urdu, the copula $\texttt{\underline{hai}}$ in the postpositional adjectival construction appears to the right of the adjective. On the other hand, in English, the copula $\texttt{\underline{is}}$ appears to the left of the adjective. 
	\[
	\input{urdu/urdudiags/ADJpostgrammar.tikz} 	
	\]
\end{itemize}

The transformational grammar rules for pronominal links are the same for English and Urdu in terms of the sentence-order, and the connectivity of pronominal links. However, there are differences arising from the different word order within individual simple sentences.

The counterpart of rewrite (\ref{esrp}) for subject relative pronoun in Urdu is as follows.
\begin{equation}\label{usrp}
\input{urdu/urdudiags/srpruleurdu.tikz} 
\end{equation}
\begin{example}
	Consider the Urdu sentences 
	\begin{center}
		\includegraphics[scale=0.6]{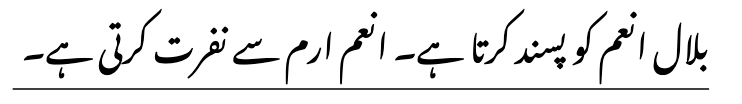}
	\end{center}
	\noindent transliterated as:
	\begin{center}
		\textit{\underline{Bilal} \ \ \underline{Anam (ko)} \ \ \  \underline{pasand karta hai.} \ \ \ \ \ \ \underline{Anam} \ \ \ \  \underline{Erum (se)} \ \ \underline{nafrat karti hai.} \ \ } (Urdu) \\
		\ \ \underline{Bilal} \ \ \ \ \ \ \  \underline{Anam} \ \ \ \ \ \  \ \underline{likes.} \ \ \ \ \ \ \ \ \ \ \ \ \  \   \underline{Anam} \ \ \ \ \ \ \   \underline{Erum} \ \  \ \ \ \ \ \ \ \ \underline{hates.} \ \ \ \ \ \ \ \ \ \ \ \ \ \ \ \  \ \ \ \ 
	\end{center}
\noindent The two occurrences of \texttt{\underline{Anam}} can be identified using a pronominal link. 
			\[
	\input{urdu/treeex30.tikz} 
	\]
\noindent The two trees can be joined by replacing the second occurrence of \texttt{\underline{Anam}} with \texttt{\underline{jo}}. This results in the following diagram: 	
				\[
	\input{urdu/treeex31.tikz} 
	\]
\noindent corresponding to the Urdu sentence
	\begin{center}
	\includegraphics[scale=0.6]{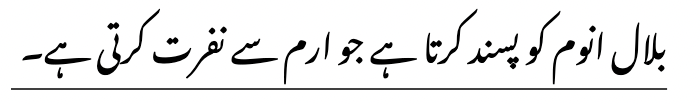}
\end{center}
	\noindent transliterated as:
\begin{center}
	\textit{\underline{Bilal} \ \ \underline{Anam (ko)} \ \ \  \underline{pasand karta hai} \ \ \ \ \ \ \underline{jo} \ \ \ \  \underline{Erum (se)} \ \ \underline{nafrat karti hai.} \ \ } (Urdu) \\
		\ \ \underline{Bilal} \ \ \ \ \ \ \  \underline{Anam} \ \ \ \ \ \  \ \underline{likes} \ \ \ \ \ \ \ \ \ \ \ \ \  \   \underline{who} \ \ \ \ \ \ \   \underline{Erum} \ \  \ \ \ \ \ \ \ \ \underline{hates.} \ \ \ \ \ \ \ \ \ \ \ \ \ \ \ \  \ \ \ \ 
\end{center}
\end{example}

The counterpart of rewrite (\ref{essrp}) for subject relative pronoun in Urdu is given by
\begin{equation}\label{ussrp}
\input{urdu/engdiags/ssrprule.tikz} 
\end{equation}

\begin{example}
	Consider the Urdu sentences 
\begin{center}
	\includegraphics[scale=0.6]{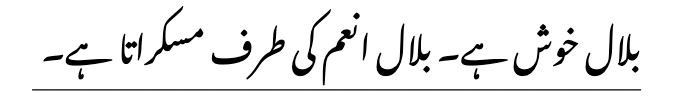}
\end{center}
\noindent transliterated as:
\begin{center}
	\textit{\underline{Bilal} \ \ \underline{khush} \ \ \  \underline{hai.} \ \ \ \ \ \ \underline{Bilal} \ \  \underline{Anam} \ \  \underline{ki taraf} \ \  \underline{muskurata hai.} \ \ } (Urdu) \\
\ \ \	\ \ \underline{Bilal} \ \ \  \  \underline{happy}  \ \ \ \ \ \underline{is.} \ \ \ \ \  \ \   \underline{Bilal} \ \ \  \underline{Anam} \ \  \ \ \ \   \underline{at} \ \ \ \ \ \ \  \underline{smiles.}  \ \ \ \ \ \ \ \  \ \ \  \ \  \ \ \ \ \ \ \ \ \ \ \ \ \  \ \ \ \  \ \ \ 
\end{center}
\noindent The two occurrences of \texttt{\underline{Bilal}} can be identified using a pronominal link. 
	\[
	\input{urdu/treeex32.tikz} 
	\]
Applying the transformational rule, we get the sentence	
	\[
	\input{urdu/treeex33.tikz} 
	\]
\noindent corresponding to the Urdu sentence
\begin{center}
	\includegraphics[scale=0.6]{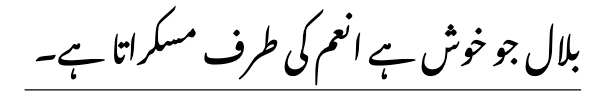}
\end{center}
\noindent transliterated as:
\begin{center}
	\textit{\underline{Bilal} \ \ \underline{jo} \ \  \underline{khush} \ \ \  \underline{hai} \ \ \ \ \ \  \underline{Anam} \ \  \underline{ki taraf} \ \  \underline{muskurata hai.} \ \ } (Urdu) \\
\ \ \ \  \underline{Bilal} \  \   \underline{who} \ \   \underline{happy}  \ \ \ \  \underline{is} \ \ \ \ \  \ \ \ \  \underline{Anam} \ \  \ \ \ \  \underline{at} \ \ \ \ \ \ \ \  \underline{smiles.}  \ \ \ \ \ \ \ \  \ \ \ \ \  \ \ \ \ \ \ \ \ \ \  \ \ \ \ \ \ \ \ \ \ 
\end{center}	
	
\end{example}

The counterpart of rewrite (\ref{eorp}) for object relative pronoun in Urdu is as follows.
\begin{equation}\label{uorp}
\input{urdu/urdudiags/orpruleurdu.tikz} 
\end{equation}

\begin{example}
	Consider the Urdu sentences 
\begin{center}
	\includegraphics[scale=0.6]{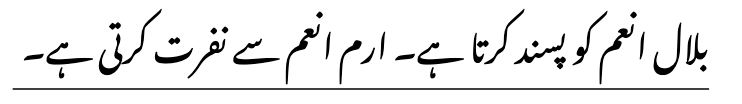}
\end{center}
\noindent transliterated as:
\begin{center}
	\textit{\underline{Bilal} \ \ \underline{Anam (ko)} \ \ \  \underline{pasand karta hai.} \ \ \ \ \ \ \underline{Erum} \ \ \ \  \underline{Anam (se)} \ \ \underline{nafrat karti hai.} \ \ } (Urdu) \\
		\ \ \underline{Bilal} \ \ \ \ \ \ \  \underline{Anam} \ \ \ \ \ \  \ \underline{likes.} \ \ \ \ \ \ \ \ \ \ \ \ \  \   \underline{Erum} \ \ \ \ \ \ \   \underline{Anam} \ \  \ \ \ \ \ \ \ \ \underline{hates.} \ \ \ \ \ \ \ \ \ \ \ \ \ \ \ \  \ \ \ \ 
\end{center}
\noindent The two occurrences of \texttt{\underline{Anam}} can be identified using a pronominal link. 	
	\[
	\input{urdu/treeex34.tikz} 
	\]
Applying the transformational rule, we get the sentence		
	\[
	\input{urdu/treeex35.tikz} 
	\]
\noindent corresponding to the Urdu sentence
\begin{center}
	\includegraphics[scale=0.6]{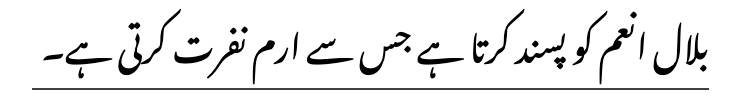}
\end{center}
\noindent transliterated as:
\begin{center}
	\textit{\underline{Bilal} \ \ \underline{Anam (ko)} \ \ \  \underline{pasand karta hai} \ \ \ \ \ \ \underline{jis (se)} \ \ \ \  \underline{Erum} \ \ \underline{nafrat karti hai.} \ \ } (Urdu) \\
			\ \ \underline{Bilal} \ \ \ \ \ \ \  \underline{Anam} \ \ \ \ \ \ \ \ \ \ \underline{likes}  \ \ \ \ \ \ \ \ \ \  \   \underline{whom} \ \ \ \ \ \ \   \underline{Erum} \ \  \ \ \ \ \ \ \ \ \underline{hates.} \ \ \ \ \ \ \ \ \ \ \ \ \ \ \ \  \ \ \ \ 
\end{center}	
\end{example}

\begin{remark}\label{pronconnect}
	Moving from the hybrid grammar to text diagrams (using (\ref{link2d}), (\ref{pronlink}) and (\ref{pronlink2})), the connectivity of pronominal links is retained, and is the same for English and Urdu. 
\end{remark}

\section{Text Diagrams and Circuits for Urdu Text}\label{urducircuits}

\subsection{Main result}

Let $\texttt{E}$ denote the set of generators of the hybrid English grammar, summarised in Table~\ref{tab:diff_gram}. 
Let $\texttt{T}_\texttt{E}$ denote the set of all English text constructed with grammar $\texttt{E}$, and let $\texttt{C}_\texttt{E}$ denote the set of all text circuits for English. Then, there exists a surjection $\texttt{T}_\texttt{E} \twoheadrightarrow \texttt{C}_\texttt{E}$~\cite{wang2023distilling}.

We have created a set of generators for Urdu grammar $\texttt{U}$ which closely correspond with the English generators $\texttt{E}$; see  Table~\ref{tab:diff_gram}. Let $\texttt{T}_\texttt{U}$ denote the set of all Urdu text generated with $\texttt{U}$. Let $\texttt{C}_\texttt{U}$ denote the set of all text circuits for Urdu. Then,
\begin{itemize}
\item there exists a surjection $\texttt{T}_\texttt{U}\twoheadrightarrow \texttt{C}_\texttt{U}$, and
\item $\texttt{C}_\texttt{U}$ is isomorphic to $\texttt{C}_\texttt{E}$, i.e., $\texttt{C}_\texttt{U}
\xrightarrow{\simeq}
\texttt{C}_\texttt{E}$ (up to word translations at the gate level).
\end{itemize}

\subsection{Urdu Text Surjects onto Circuits}

The method of turning hybrid grammar trees into text diagrams is the same in English and Urdu: each component of a hybrid grammar tree is modified so that the number of $\texttt{NP}$ wires for inputs and outputs is equal, and sentence types $\texttt{S}$ are eliminated.

Table~\ref{tab:diff_gram} illustrates the similarities and differences between Urdu and English grammar, as reflected in the hybrid grammar trees and text diagrams. 

The same text diagram reductions hold in Urdu as those in English. The following is the reduction of a postpositional adjectival construction using a copula $\texttt{\underline{hai}} (\simeq \texttt{\underline{is}})$ to a prepositional adjective that does not require a copula:
	\[
\input{urdu/ADJistransformurdu.tikz} 
\]

Just as in the case of English hybrid grammar, in the passage to text diagrams, the pronominal-link (multi)arrows become wires (\ref{link2d}) that can be eliminated using the diagram rewrite rules (\ref{pronlink}) and (\ref{pronlink2}) (see Remark~\ref{pronconnect}).

\begin{table}
	\centering
	\scalebox{0.75}{
		\begin{tabular}{|c|c|c|c|c|}
			\hline
			\text{Rule} & \text{English grammar} & \text{English diagram} & \text{Urdu grammar} & \text{Urdu diagram} \\ \hline
			Intrans.Verb & \input{urdu/engdiags/IVcsg.tikz} & \input{urdu/engdiags/IVgraph.tikz} & \input{urdu/engdiags/IVcsg.tikz} & \input{urdu/engdiags/IVgraph.tikz} \\ \hline
			\textcolor{red}{Trans.Verb} & \input{urdu/engdiags/TVcsg.tikz} & \input{urdu/engdiags/TVgraph.tikz} & \input{urdu/urdudiags/TVPgram.tikz} & \input{urdu/urdudiags/TVP.tikz} \\ \hline
			Adjective(Pre.) & \input{urdu/engdiags/ADJcsg.tikz} & \input{urdu/engdiags/ADJgraph.tikz} & \input{urdu/engdiags/ADJcsg.tikz} & \input{urdu/engdiags/ADJgraph.tikz} \\ \hline
			\textcolor{red}{Adjective(Post.)} &  \input{urdu/engdiags/ADJiscsg.tikz} &  \input{urdu/engdiags/ADJisgraph.tikz} & \input{urdu/urdudiags/ADJpostgrammar.tikz} & \input{urdu/urdudiags/ADJpostdiagram.tikz} \\ \hline
			Adverb(IV) & \input{urdu/engdiags/ADVIVcsg.tikz} & \input{urdu/engdiags/ADVIVgraph.tikz} & \input{urdu/engdiags/ADVIVcsg.tikz} & \input{urdu/engdiags/ADVIVgraph.tikz}  \\ \hline
			Adverb(TV) & \input{urdu/engdiags/ADVTVcsg.tikz} & \input{urdu/engdiags/ADVTVgraph.tikz} & \input{urdu/engdiags/ADVTVcsg.tikz} & \input{urdu/engdiags/ADVTVgraph.tikz} \\ \hline
			\textcolor{red}{Adposition(IV)} & \input{urdu/engdiags/ADPIVcsg.tikz} & \input{urdu/engdiags/ADPiveng.tikz} & \input{urdu/urdudiags/ADPivgrammar.tikz}  & \input{urdu/urdudiags/ADPivdiagram.tikz} \\ \hline
			\textcolor{red}{Sent.Comp.Verb} & \input{urdu/engdiags/SCVscopecsg.tikz} & \input{urdu/SCVtranslate.tikz} & \input{urdu/urdudiags/SCVgram.tikz} & \input{urdu/urdudiags/SCVtranslateurdu.tikz}  \\ \hline
			Conjunction & \input{urdu/engdiags/CNJscopecsg.tikz} & \input{urdu/CNJtranslate.tikz}  & \input{urdu/engdiags/CNJscopecsg.tikz} & \input{urdu/CNJtranslate.tikz}  \\ \hline
		\end{tabular}
	}
	\caption{Generators for hybrid Urdu and English grammar and the corresponding diagrams. Note the different constructions for English and Urdu in the rules labelled red: Trans.Verb, Adjective(Post.), Adposition(IV) and Sent.Comp.Verb.}
	\label{tab:diff_gram}
\end{table}

\begin{lemma}
	Let $\texttt{U}$ denote the set of generators of the hybrid Urdu grammar.
	Let $\texttt{T}_\texttt{U}$ denote the set of all Urdu text constructed with grammar $\texttt{U}$, and let $\texttt{C}_\texttt{U}$ denote the set of all text circuits for Urdu. Then, there exists a surjection $\texttt{T}_\texttt{U} \twoheadrightarrow \texttt{C}_\texttt{U}$.
\end{lemma}

\begin{proof}
	Table~\ref{tab:diff_gram} shows the generators of English and Urdu grammar, four of which are different for the two languages. Viewing the generators as rooted labelled trees (vertices are labels like $\texttt{NP}$ or $\underline{\texttt{hai}}$, and the root is the initial sentence type $\texttt{S}$), the generators of English and Urdu for each rule are isomorphic in the graph-theoretic sense. 
	
	\begin{itemize}
		\item Transitive verb
		\[
		\input{urdu/urdudiags/rootlabtree2.tikz} \cong 	 \input{urdu/urdudiags/rootlabtree1.tikz} 
		\]
		\item Postpositional adjective (where $\texttt{\underline{hai}} (\cong \texttt{\underline{is}})$)
		\[
		\input{urdu/urdudiags/rootlabtree7.tikz} \cong 	 \input{urdu/urdudiags/rootlabtree8.tikz} 
		\]
		
		\item Adposition 
		\[
		\input{urdu/urdudiags/rootlabtree3.tikz} \cong 	 \input{urdu/urdudiags/rootlabtree4.tikz} 
		\]
		
		\item Verb with a sentential complement
		\[
		\input{urdu/urdudiags/rootllabtree5.tikz} \cong 	 \input{urdu/urdudiags/rootlabtree6.tikz} 
		\]
	\end{itemize}
	
	The transformational grammar rules for Urdu (\ref{usrp},~\ref{ussrp},~\ref{uorp}) and English ((\ref{esrp},~\ref{essrp},~\ref{eorp}) are the same in terms of the connectivity of pronominal links and sentence-order. They differ with respect to the intra-sentence word orders. These differences are eliminated due to the isomorphisms of generators described above. Therefore, essentially the same transformational grammar rules and diagram rewrites apply to English and Urdu.
	
	The aforementioned modifications translate the formal claim of Theorem~\ref{textcircth}~\cite{wang2023distilling} for Urdu; the hybrid grammar text for Urdu surjects onto text circuits.
\end{proof}

\begin{remark}
The graph-theoretic isomorphism of the generators of English and Urdu grammars is not surprising, as both languages belong to the Indo-European language family~\cite{kapovic2017indo}. For all the grammatical rules considered, the generators of both languages have the same types, although they may appear in different orderings. For contrast, consider null-subject languages, in which subjects can be omitted~\cite{jaeggli2012null}. One example is Arabic, where pronouns are dropped because they can be inferred from verb morphology~\cite{kenstowicz1989null}. Therefore, the generators of Arabic grammar---or those of other null-subject languages---are not isomorphic to those of English and Urdu, even for the limited fragment of language considered in this chapter.
\end{remark}

\subsection{English and Urdu Give the Same Circuits}

In the case that we only consider context-free generators, the desired isomorphism between English and Urdu circuits essentially follows from the isomorphism between the trees generated by the English and Urdu hybrid grammar.


\begin{lemma}\label{syntaxiso}
Let $\texttt{T}_\texttt{E}$ denote the set of texts generated by the English production rules $\texttt{E}$, and let $\texttt{T}_\texttt{U}$ denote the set of texts generated by the Urdu production rules $\texttt{U}$.
Viewing the syntax trees as rooted labelled trees (vertices are labels like $\texttt{NP}$ or $\underline{\texttt{hai}}$, and the root is the initial sentence type $\texttt{S}$), there is an isomorphism $\texttt{T}_\texttt{E}\simeq \texttt{T}_\texttt{U}$ in the graph-theoretic sense.
\end{lemma}

\begin{proof}
Essentially, the Urdu generators $\texttt{U}$ correspond exactly to the English generators $\texttt{E}$, except some of them have the order of their outputs switched around.

So, given an English syntax tree generated by some sequence of applications of production rules in $\texttt{E}$, to translate this to an Urdu syntax tree we simply apply the corresponding rules in $\texttt{U}$ to the appropriate symbols. 
The exact proof of this is a simple induction. Note that in order to `apply the right rule to the right symbol', we must track the identities of different symbols, \textit{e.g.} distinguishing between different instances of $\texttt{NP}$ in our string. 
This tracking can be done by numbering the symbols with indices.

Refer to Table~\refeq{tab:diff_gram}. For the base case, the intransitive verb, the transitive verb, the adjective (postpositional), the sentential complement verb or the conjunction rule can be applied. Of these rules, two are exactly the same for English and Urdu. The rest of the rules generate syntax trees that when viewed as rooted labelled trees are isomorphic graph-theoretically.
	\begin{itemize}
	\item Transitive verb
	\[
	\input{urdu/urdudiags/rootlabtree2.tikz} \cong 	 \input{urdu/urdudiags/rootlabtree1.tikz} 
	\]
	\item Postpositional adjective (where $\texttt{\underline{hai}} (\cong \texttt{\underline{is}})$)
	\[
	\input{urdu/urdudiags/rootlabtree7.tikz} \cong 	 \input{urdu/urdudiags/rootlabtree8.tikz} 
	\]
	\item Verb with a sentential complement
	\[
	\input{urdu/urdudiags/rootllabtree5.tikz} \cong 	 \input{urdu/urdudiags/rootlabtree6.tikz} 
	\]
\end{itemize}
The inductive hypothesis is that each pair of English and Urdu syntax trees generated by $n$ applications of corresponding pairs of production rules is isomorphic. 
\[
\input{urdu/urdudiags/tree1.tikz} \ \cong \  \input{urdu/urdudiags/tree2.tikz}
\]

For the induction step, we apply another production rule to the syntax trees generated by the inductive hypothesis. Any of the possible production rules from Table~\refeq{tab:diff_gram} can be applied. All except four rules are different for English and Urdu. Of the four, three were shown to be isomorphic in the base case. The last one is also isomorphic in the same sense. 

Adposition 
\[
\input{urdu/urdudiags/rootlabtree3.tikz} \cong 	 \input{urdu/urdudiags/rootlabtree4.tikz} 
\]
Then, using the inductive hypothesis, English and Urdu syntax trees generated by $n+1$ applications of production rules are isomorphic. 
\[
\input{urdu/urdudiags/tree1b.tikz}\ \cong \  \input{urdu/urdudiags/tree2b.tikz}
\]
\end{proof}

\begin{lemma}
	$\texttt{C}_\texttt{U}$ is isomorphic to $\texttt{C}_\texttt{E}$, i.e., $\texttt{C}_\texttt{U}
	\xrightarrow{\simeq}
	\texttt{C}_\texttt{E}$ (up to word translations at the gate level).
\end{lemma}

\begin{proof}
According to Lemma~\ref{syntaxiso}, there is an isomorphism $\texttt{T}_\texttt{E}\simeq \texttt{T}_\texttt{U}$ between individual syntax trees.
We introduce pronominal links and the rewrite rules ((\ref{esrp},~\ref{essrp},~\ref{eorp}) for English; (\ref{usrp},~\ref{ussrp},~\ref{uorp}) for Urdu) that allow us to adjoin pronominally-linked trees into single sentences, thus attaining the full-blown hybrid grammar. The transformational rewrite rules for English and Urdu are the same in terms of the connectivity of pronominal links. The rewrite rules differ with respect to the word order in individual syntax trees. This difference is eliminated by Lemma~\ref{syntaxiso}. Therefore, the isomorphism between trees lifts to a kind of isomorphism between these full hybrid grammar structures. 

Next, we convert hybrid grammar to text diagrams. We apply the rules in Table~\ref{tab:diff_gram} for resolving the $\texttt{S}$ type into the constituent $\texttt{NP}$ wires and converting phrase scope into phrase regions. Then we turn the pronominal links into dashed wires (\ref{link2d}) in preparation for composition. After performing the composition using the rewrite rule (\ref{pronlink}), we recover text diagrams. 

The isomorphism of hybrid grammar structures simply lifts to an isomorphism of the structures at each intermediate step. At the step where we reach the level of text diagrams, the isomorphism becomes an exact equality (modulo the different word labels, and allowing a certain degree of topological deformation as we usually do in string diagrams).

With the (topologically) identical structure of the English and Urdu diagrams, we simply apply the same conversion map from text diagrams to text circuits to obtain the same text circuit up to word-translations at the level of individual gates.
\end{proof}

\begin{example}
Consider the English sentence \texttt{\underline{The young student who sees the honest}}\\ \texttt{\underline{teacher passionately teach smiles at him.}} which we translate into the Urdu sentence 
\begin{center}
	\includegraphics[scale=0.5]{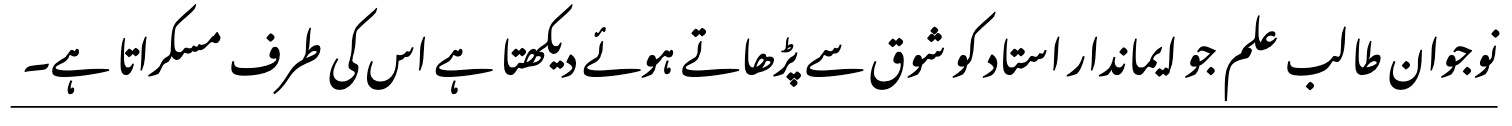}
\end{center}
\noindent transliterated as:
\begin{center}
	\textit{\underline{nojawan} \ \ \underline{talib-e-ilm} \ \ \  \underline{jo} \ \ \ \ \ \ \underline{imandar} \ \ \ \  \underline{ustad (ko)} \ \ \underline{shauq se} \ \ \underline{parhate huwe}} (Urdu) \\
	\ \ \underline{(the) young} \ \ \  \underline{student} \ \ \ \ \underline{who} \ \ \underline{(the) honest} \ \ \  \underline{teacher} \ \   \underline{passionately} \ \ \ \ \ \  \underline{teach} \ \ \ \ \ \ \ \ \ \ \ \ \ \ \ \ \ \ \  \\
	\ \\
	\textit{\underline{dekhta hai} \ \ \underline{us} \ \ \underline{ki taraf} \ \  \underline{muskurata hai}} (Urdu) \\
	\ \  \ \underline{sees}  \ \ \ \ \ \ \   \underline{him} \ \ \ \ \     \underline{at} \ \ \ \ \ \ \ \  \underline{smiles} \ \  \ \ \ \ \ \  \ \ \ \ \ \ \ \ \  \ \  
\end{center}

\noindent We start with the hybrid grammar structure of the English and Urdu sentences. 

\[\scalebox{0.7}{\input{urdu/text2circENG_1.tikz}}\]
\[\scalebox{0.65}{\input{urdu/circ2textURD_6.tikz}}\]

\noindent We replace the individual tree diagrams with the corresponding text diagrams, and the pronominal link (multi)arrows with dashed wires.

\[\scalebox{0.7}{\input{urdu/text2circENG_2.tikz}}\]
\[\scalebox{0.7}{\input{urdu/circ2textURD_5.tikz}}\]


\noindent Eliminating the dashed wires with diagram rewrite rules, we obtain text diagrams for the the whole sentences. 

\[\scalebox{0.7}{\input{urdu/text2circENG_3b.tikz} \quad\quad\quad\quad\quad\quad \input{urdu/circ2textURD_4b.tikz}}\]

\noindent At this point, we can see that the text diagrams for English and Urdu have the same topological structure. We, therefore, apply the same conversion map to express text diagrams in terms of gates and boxes, resulting in the same text circuits, up to gate-level translations.
\[\scalebox{0.7}{\input{urdu/textcircuit4.tikz} \quad\quad\quad\quad\quad\quad\quad\quad\quad\quad\quad\quad\quad \input{urdu/textcircuit5.tikz}}\]
\end{example}

\begin{remark}
The process from text to circuits can be (non-deterministically) reversed. This makes text circuits generative formalisms for English and Urdu text. Note that we are only concerned here with the fragment of language described in this chapter.
\end{remark}

\section{Summary and Outlook}\label{conclusion}

In this chapter, we demonstrated how DisCoCirc text circuits eliminate grammatical bureaucracy based on word or phrase order for restricted fragments of English and Urdu. Building on Ref.~\cite{wang2023distilling}, which established a surjection from English hybrid grammar to English text circuits, we adapted the hybrid grammar for Urdu and defined translation rules into text circuits that mirror those for English. This yields a surjective map from Urdu hybrid grammar to Urdu text circuits. For the language fragments considered, the English and Urdu hybrid grammars are isomorphic. Using this isomorphism, we show that their corresponding text circuits are structurally identical, differing only in gate labels. 

A major difference between grammars in different natural languages arises from different word orderings for, for example, subject, verb and object. For instance, in English the usual ordering is subject-verb-object (SVO), whereas in Urdu, it is subject-object-verb (SOV). These differences, in turn, exist because human verbal communication is restricted to one dimension and different cultures and demographics made different stylistic choices as languages evolved~\cite{wang2023distilling}. But there is no such restriction on machines. Two-dimensional grammars such as ours may be a suitable abstraction of text for computers, particularly quantum computers~\cite{laakkonen2024quantum, duneau2024scalable}, and may prove advantageous for natural language processing tasks, such as machine translation. This chapter takes a step forward in this direction.  
Our work represents a first step, and there is work to be done to expand the fragments of natural language that we can handle.

A limitation of this work is that both English and Urdu belong to the Indo-European language family and share similar grammatical structures, including a clear noun–verb distinction~\cite{kapovic2017indo}. While they differ in word order---English follows SVO and Urdu follows SOV---the generators of their grammars are graph-theoretically isomorphic for the fragment considered, making the main results of this work somewhat unsurprising. In contrast, extending this framework to typologically distinct languages, such as null-subject languages like Arabic~\cite{kenstowicz1989null, jaeggli2012null}, presents a greater challenge. These languages allow subject omission due to rich verb morphology, resulting in grammatical generators that are not isomorphic to those of English or Urdu. Addressing such structural differences is an important direction for future work.

\chapter{\label{chapconclusions}Conclusion}

\epigraph{`The compartments into which human thought is divided are not so water–tight that fundamental progress in one is a matter of indifference to the rest.'}  
{Arthur Eddington, \\The Philosophy of Physical Science~\cite[p.~8]{eddington1939philosophy}}

\textit{In this chapter, we summarise the key contributions described in this thesis and discuss avenues for future research. At the end of this chapter, we make a case for the prospects of process-relational philosophy in science.}

\section*{A process-theoretic approach to constructor theory and quantum physics}

In Chapter~\ref{chapconstructors}, we describe three main contributions. 

Firstly, we present categorical semantics for constructor theory while adhering to the desired mathematical foundations described in the seminal article~\cite{deutsch2013constructor} on constructor theory. Taking the theory of conceivable tasks to be the symmetric monoidal category (SMC) $\Rel$, we show that for a given choice of substrates, the set of possible tasks forms a sub-SMC of $\Rel$. 

Secondly, we argue that in the constructor-theoretic formulation of non-relativistic quantum theory, there is an inconsistency between the restrictions imposed by the principles of locality and composition. Constructor theorists argue that there exists a local formulation of quantum theory: the Deutsch-Hayden approach~\cite{deutsch2000information}. We demonstrate, using concrete examples, that this approach is not compositional.

Thirdly, we discuss that a key difference between constructor theory and process theories concerns the principle of locality. 
We argue that categorical quantum mechanics (CQM) can be conceived as a constructor theory of quantum physics that is compositional but non-local.

\subsection*{Future work}

We recommend several directions for future research:

\begin{itemize}
	\item A key future direction is to develop process-theoretic formulations of constructor theories in domains such as information~\cite{deutsch2015constructor}, thermodynamics~\cite{marletto2016constructorb}, probability~\cite{marletto2016constructor}, and life~\cite{marletto2015constructor}. This will clarify their foundational principles and interconnections, enabling a unified understanding of processes across these fields.
	\item We also propose investigating the relationship between constructor theory and resource theories~\cite{coecke2016mathematical, plesnik2024impossibility}. Exploring this connection could yield new insights into how resources are transformed and conserved in physical and informational processes, enriching both frameworks and opening new avenues for application.
	\item The string-diagrammatic syntax of process theories can also be interpreted in SMCs other than the one discussed in this thesis. This provides an opportunity to investigate and explore the ramifications of constructor theory beyond $\Rel$, which has been the default modelling choice in the constructor theory literature so far~\cite{deutsch2013constructor, marletto2022emergence}.
\end{itemize}

\section*{Wave-based computation in diagrams}

In Chapter~\ref{chapwavelogic}, we develop a string-diagrammatic formalism for wave-based logic circuits with phase encoding. The formalism is motivated with reference to spin-wave or `magnonic' circuits. Through the example of spin-wave circuits, the usage of the formalism in designing, analysing and simplifying Boolean logic circuits is demonstrated.

\subsection*{Future work}

Some directions for future research are outlined below:

\begin{itemize}
	\item To extend this work, we would anticipate developing formal semantics for the string diagrams presented here, and proving that the formalism is sound and complete for Boolean algebra.
	\item Another avenue for future work is to conceive of the formalism as a diagrammatic alternative to symbolic Boolean algebra, wherein lies the possibility of a new axiomatisation of Boolean algebra.
	\item This work has proposed a formalism for the theoretical modelling of spin-wave circuits, which abstracts away many hardware-level complexities. However, spin-wave computing inherits core limitations of analogue systems, such as noise susceptibility, signal degradation, error propagation, and energy-intensive error correction. These issues are not fully addressed in the current framework. Future extensions may need to incorporate error-correction mechanisms---typically involving non-linear operations~\cite{verba2019correction}---and amplitude normalization for cascaded logic gates, which is also non-linear~\cite{dutta2015compact}. Additionally, modelling crosstalk noise at smaller scales and integrating with platforms such as optical systems, superconducting qubits, or CMOS technologies~\cite{chumak2015magnon, li2020hybrid, lachance2019hybrid} will be crucial for improving the framework’s practical applicability. Addressing these aspects will enhance its relevance to real-world applications and contribute to more scalable and resilient spin-wave computing architectures.
\end{itemize}

\section*{DisCoCirc beyond English text}

In Chapter~\ref{chapdiscocirc}, we show how word- or phrase-order-based grammatical bureaucracy is eliminated by DisCoCirc text circuits for restricted fragments of English and Urdu.

We describe a restricted version of the hybrid grammar, text diagrams and text circuits for English developed in Ref.~\cite{wang2023distilling}. According to Ref.~\cite{wang2023distilling}, there is a surjection from the set of all English text generated with the English hybrid grammar to the set of all English text circuits. In a similar vein, in this thesis, we describe how the hybrid grammar can be adapted for Urdu. We then provide rules for its translation into text diagrams and text circuits, which are essentially the same as those in Ref.~\cite{wang2023distilling}. We show that this gives a surjective map from the set of all Urdu text generated with Urdu hybrid grammar to the set of all Urdu text circuits. 

Furthermore, for the language fragments considered in this thesis, there is a clear isomorphism between the hybrid grammars for English and Urdu. Using this isomorpism, we show that text circuits for English and Urdu become the same, up to translation of gate-labels.

\subsection*{Future work}

Our work can be extended in several directions: 

\begin{itemize}
	\item This thesis focused on a limited fragment of English and Urdu. Future work can extend this to broader language fragments by incorporating additional grammatical features such as multiple tenses, reflexive pronouns, determiners etc.
	
	\item The primary syntactic variation in the language fragments studied lies in word order: English follows SVO, while Urdu follows SOV.\footnote{Different natural languages can share the same grammatical structure, including word order. For instance, Hindi and Urdu have nearly identical grammars but different vocabularies~\cite{bhatt2018teaching}. The results obtained for Urdu in this thesis directly apply to Hindi for the considered fragment.} Any language that differs from English or Urdu only in the word order of simple sentences can potentially be modelled using DisCoCirc. Hybrid grammars for such languages can be constructed to fit within the DisCoCirc framework. Future research can pursue both a theoretical characterisation of all possible DisCoCirc languages and an empirical investigation into which of them correspond to real-world languages. DisCoCirc, as a two-dimensional framework, abstracts away from the linear order of words and thus supports language-independent grammatical modelling.
	
	\item A limitation of this work is that both English and Urdu are Indo-European languages with comparable grammatical structures, such as clear noun-verb distinctions~\cite{kapovic2017indo}. Although they differ in word order, their grammatical generators for the studied fragment are graph-theoretically isomorphic, resulting in English and Urdu texts mapping to the same circuits. Extending DisCoCirc to typologically distinct languages, such as null-subject languages like Arabic~\cite{kenstowicz1989null, jaeggli2012null}, presents additional challenges. These languages allow subject omission due to rich verbal morphology, leading to non-isomorphic grammatical structures. Addressing such typological variation is essential for broadening the applicability of the DisCoCirc framework.
	
	\item The relationship between DisCoCirc and other linguistic formalisms---especially discourse representation theory~\cite{kamp2010discourse}, which shares conceptual similarities---requires further exploration. Some initial work in this direction appears in Ref.~\cite{wang2023distilling}.
\end{itemize}

\section*{An invitation to process-relational philosophy}

As mentioned in Chapter~\ref{chapintro}, this thesis can be considered a proof of concept of applied process-relational philosophy. 

CQM, reviewed in Chapter~\ref{chapinterpretations1}, is a process-theoretic formulation of quantum physics. It suggests that the most natural ontology for quantum physics centres on processes and their interactions, rather than on states and their dynamics.

Chapter~\ref{chapconstructors} shows how constructor theory is essentially process-theoretic, from both a conceptual and a mathematical point of view. This implies that constructor theories of information~\cite{deutsch2015constructor}, thermodynamics~\cite{marletto2016constructorb}, probability~\cite{marletto2016constructor}, and life~\cite{marletto2015constructor} are potential process theories as well. Concrete details of this claim will be the subject of future work. With regards to quantum theory, Chapter~\ref{chapconstructors} argues how the principles of locality and compositionality cannot coexist. Furthermore, discarding the principle of locality, CQM achieves the desiderata of a constructor theory for quantum physics---which is, of course, process-theoretic by construction.  

Chapter~\ref{chapwavelogic} develops a formalism for wave-based logic circuits in which processes are primitive. Logic variables are initialised as phase-shifting processes. Composing these phase shifting processes gives rise to different logic gates depending on the topological structure of the composite processes.

Chapter~\ref{chapdiscocirc} describes two-dimensional circuits to model grammar, meaning, and their interaction in natural languages. The key assumption underlying DisCoCirc is that the meaning of a piece of text is determined by how it updates the meanings of the words within it. In DisCoCirc, this assumption is implemented process-theoretically using text circuits. The meanings of nouns (denoted by wires) are updated by processes like verbs and adjectives (represented by gates), the meanings of which are in turn altered by processes like adverbs and adpositions (represented by higher order gates). DisCoCirc offers prospects in quantum natural language processing~\cite{laakkonen2024quantum, duneau2024scalable}, where the work on quantum processes and computation discussed in Chapters~\ref{chapinterpretations1} and~\ref{chapconstructors} is of immediate relevance. 

The person most notably associated with process philosophy in recent history is Alfred North Whitehead. Whitehead's treatise \textit{Process and Reality}~\cite{whitehead2010process} was published in 1929, but was based on his 1927-28 lectures at the University of Edinburgh. Quantum mechanics, as we know it today, was born in the mid-1920s and undoubtedly influenced Whitehead's work. Whitehead was a mathematician first and a philosopher later; however, he did not have access to the tools of applied category theory---which are arguable particularly well-suited to process ontology, relationalism and compositionality. The absence of an appropriate accompanying mathematical framework may be one reason process philosophy has not yet caught on in scientific research.

Category theory, developed in the mid-1940s~\cite{eilenberg1945general, mac1978categories}, was initially of interest in pure mathematics only and was considered very abstract. In the last two decades, it has found application in scientific domains as wide-ranging as quantum physics and quantum computing~\cite{coecke2006kindergarten, coecke2022kindergarden, CKbook, heunen2019categories, CGbook, coecke2023basic}, chemistry~\cite{baez2022applied, baez2017compositional, genovese2020categorical}, information theory~\cite{bradley2021entropy}, electrical circuit theory~\cite{baez2015compositional, boisseau2021string}, photonics~\cite{de2022quantum, clement2022lov}, linguistics~\cite{clark2008compositional, coecke2021mathematics, wang2023distilling}, game theory~\cite{hedges2015string, hedges2016compositionality, ghani2018compositional}, control theory~\cite{baez2014categories, erbele2016categories}, thermostatics~\cite{baez2023compositional}, cryptography~\cite{broadbent2022categoricalb, pavlovic2014chasing}, probability theory~\cite{sturtz2014categorical}, functional programming~\cite{milewski2018category}, machine learning~\cite{shiebler2021category, gavranovic2024fundamental} etc. This recent surge of interest in applied category theory~\cite{fong2019invitation} is perhaps due to the development of graphical calculi like string diagrams that marry visual intuition with mathematical rigour. 

String diagrams are a boon not only to professional scientists but also to high school students. In a recent study~\cite{dundar2023quantum}, students aged 16 to 18 were taught quantum physics and quantum computing using the ZX-calculus~\cite{CGbook}. The excellent performance of the students (over 80\% pass-rate and around half the number of students earning a distinction) on a postgraduate-level quantum theory exam~\cite{dundarcoecke2025makingquantumworldaccessible} underscores the potential impact of the string-diagrammatic formalisms on science education. 

The author of this thesis invites you to the research area of process-relational philosophy and its application to the sciences. A powerful mathematical toolbox is available to support the accompanying conceptual change of worldview. At best, this research area will help make concrete progress by solving problems or making discoveries. At worst, it will provide a new perspective to understand our existing scientific theories and thereby also learn something about nature. There is nothing to lose.

\begin{appendices}
\chapter{\label{appendwaves}Wave Logic Circuits: Derivation of the Associativity Rule}

\begin{proposition}\label{propa}
	For all $\alpha, \beta \in \{0, \pi\}$, $\scalebox{0.75}{\input{thmiii1.tikz}}  = \scalebox{0.75}{\input{thmiii5.tikz}}$.
\end{proposition}
\begin{proof}
	\[
	\scalebox{0.55}{\input{thmiii1.tikz}} \overset{\text{(CH)}}{=} \scalebox{0.55}{\input{thmiii2.tikz}}
	\overset{\text{(D)}}{=} \scalebox{0.55}{\input{thmiii3.tikz}}
	\]
	\[
	\overset{\text{(M)}}{=} \scalebox{0.55}{\input{thmiii4.tikz}}
	\overset{\text{(CH)}}{=} \scalebox{0.55}{\input{thmiii5.tikz}} 
	\]
\end{proof}

\begin{proposition}\label{propb}
	For all $\alpha, \beta \in \{0, \pi\}$, $\scalebox{0.75}{\input{thmiii1b.tikz}}  = \scalebox{0.75}{\input{thmiii5.tikz}}$.
\end{proposition}
\begin{proof}

	\[
	\scalebox{0.55}{\input{thmiii1b.tikz}} \overset{\text{(CH)}}{=} \scalebox{0.55}{\input{thmiii2b.tikz}}
	\overset{\text{(D)}}{=} \scalebox{0.55}{\input{thmiii3c.tikz}}
	\]
	\[
	\overset{\text{(M)}}{=} \scalebox{0.55}{\input{thmiii4b.tikz}}
	\overset{\text{(CH)}}{=} \scalebox{0.55}{\input{thmiii5.tikz}} 
	\]
\end{proof}

\begin{proposition}\label{sameness}
	If $\scalebox{0.75}{\input{andgate4.tikz}} = 	\scalebox{0.75}{\input{andgate5.tikz}}$ and $\scalebox{0.75}{\input{andgate6.tikz}} = 	\scalebox{0.75}{\input{andgate7.tikz}}$,\\ then $\scalebox{0.75}{\input{thmv1.tikz}} = 	\scalebox{0.75}{\input{thmv8.tikz}}$.
\end{proposition}
\begin{proof}
	Assume $\scalebox{0.55}{\input{andgate4.tikz}} = 	\scalebox{0.55}{\input{andgate5.tikz}}$ and $\scalebox{0.55}{\input{andgate6.tikz}} = 	\scalebox{0.55}{\input{andgate7.tikz}}$. Then, we have	
	\[
	\scalebox{0.55}{\input{thmv1.tikz}} \overset{\text{(CH)}}{=} \scalebox{0.55}{\input{thmv2.tikz}}
	\overset{\text{(CH2)}}{=} \scalebox{0.55}{\input{thmv3.tikz}}
	\overset{\text{(D)}}{=} \scalebox{0.55}{\input{thmv4.tikz}} 
	\]  
	\[
	\overset{\text{}}{=} \scalebox{0.55}{\input{thmv5.tikz}}
	\overset{\text{(D)}}{=} \scalebox{0.55}{\input{thmv6.tikz}} 
	\overset{\text{(CH2)}}{=} \scalebox{0.55}{\input{thmv7.tikz}}
	\overset{\text{(CH)}}{=} \scalebox{0.55}{\input{thmv8.tikz}}
	\]
\end{proof}

\begin{proposition}\label{assoc1}
	For all $\alpha, \beta, \gamma \in \{0, \pi\}$, 
	\[\scalebox{0.75}{\input{ass2.tikz}}\]
\end{proposition}

\begin{proof}
	In Proposition (\ref{sameness}), replace $\scalebox{0.55}{\input{thmv1.tikz}}$ by $\scalebox{0.55}{\input{ass2b.tikz}}$ and $\scalebox{0.55}{\input{thmv8.tikz}}$ by $\scalebox{0.55}{\input{ass2c.tikz}}\ $. Now, we need to show:
	
	(a)	$\scalebox{0.55}{\input{thmvia1.tikz}} \overset{\text{}}{=} \scalebox{0.55}{\input{thmvia4.tikz}}$
	\\ 
	and (b) $ \scalebox{0.55}{\input{thmvib1.tikz}} \overset{\text{}}{=} \scalebox{0.55}{\input{thmvib6.tikz}}$

	\[
		\text{(a)}\ \ \  L.H.S. = \scalebox{0.55}{\input{thmvia1.tikz}} \overset{\text{(CM)}}{=} \scalebox{0.55}{\input{thmvia2.tikz}}
	\overset{\text{Prop. (\ref{propa})}}{=} \scalebox{0.55}{\input{thmvia3.tikz}}
	\]
	\[
	R.H.S. = \scalebox{0.55}{\input{thmvia4.tikz}} 
	\overset{\text{(CM)}}{=} \scalebox{0.55}{\input{thmvia5.tikz}}
	\]
	\[
	\overset{\text{(D)}}{=} \scalebox{0.55}{\input{thmvia6.tikz}} 
	\overset{\text{Prop. (\ref{propa})}}{=} \scalebox{0.55}{\input{thmvia7.tikz}}
	\overset{\text{Prop. (\ref{propb})}}{=} \scalebox{0.55}{\input{thmvia8.tikz}}
	\]

	\[
	\text{(b)} \ \ \  L.H.S = \scalebox{0.55}{\input{thmvib1.tikz}} \overset{\text{(CM)}}{=} \scalebox{0.55}{\input{thmvib2.tikz}}
	\]
	\[
	\overset{\text{(D)}}{=} \scalebox{0.55}{\input{thmvib3.tikz}}
	\overset{\text{(CH2)}}{=} 
	\scalebox{0.55}{\input{thmvib4.tikz}} 
	\]
	\[
	\overset{\text{(CH)}}{=} \scalebox{0.55}{\input{thmvib5.tikz}}
	\]
	\[
	R.H.S. = \scalebox{0.55}{\input{thmvib6.tikz}} 
	\overset{\text{(CM)}}{=} \scalebox{0.55}{\input{thmvib7.tikz}}
	\]
	\[
	\overset{\text{(D)}}{=} \scalebox{0.55}{\input{thmvib8.tikz}}
	\overset{\text{(D)}}{=} \scalebox{0.55}{\input{thmvib9.tikz}}
	\]
	\[
	\overset{\text{(CH2)}}{=} \scalebox{0.55}{\input{thmvib10.tikz}}
	\overset{\text{(CH)}}{=} \scalebox{0.55}{\input{thmvib11.tikz}}
	\overset{\text{(D)}}{=} \scalebox{0.55}{\input{thmvib12.tikz}} \\
	\]
\end{proof}

\begin{proposition}
	For all $\alpha, \beta, \gamma \in \{0, \pi\}$, 
	\[\scalebox{0.75}{\input{ass3.tikz}}\]
\end{proposition}

\begin{proof}
	Define complements of $\alpha$, $\beta$ and $\gamma$: 
	\[
	\scalebox{0.55}{\input{comp2.tikz}}\ , \ \ \scalebox{0.55}{\input{comp3.tikz}}\ , \ \ \scalebox{0.55}{\input{comp4.tikz}}.
	\]
	Then, from Proposition (\ref{assoc1}), we have
	\[
	\scalebox{0.55}{\input{ass2d.tikz}}
	\]
	Taking the complement of both sides, we get	
	\[
	\scalebox{0.55}{\input{ass2d2.tikz}}
	\]
	\[
	\overset{\text{(C)}}{\iff} \scalebox{0.55}{\input{ass2d3.tikz}}
	\]
	\[
	\overset{\text{(C)}}{\iff} \scalebox{0.55}{\input{ass2d4.tikz}}
	\]
	\[
	\overset{\text{(F)}, \text{(ID)}}{\iff} \scalebox{0.55}{\input{ass3.tikz}}
	\]
\end{proof}
\end{appendices}

\renewcommand*{\bibfont}{\fontsize{10.5}{14}\selectfont}

\printbibliography[title={References}]
\end{document}